\documentclass[11pt]{article} 

\usepackage[T1]{fontenc}        
\usepackage{textcomp, amssymb}  
\usepackage[english]{babel}

\usepackage{anysize}            
\usepackage{amsfonts}
\usepackage{amsmath}
\usepackage{dsfont}
\usepackage{MnSymbol}
\usepackage{mathtools}
\usepackage{epsfig}
\usepackage{epstopdf}
\usepackage{lmerge}
\usepackage{tikz}
\usepackage{amsthm}
\usepackage{ifpdf}

\usetikzlibrary{calc,decorations.pathmorphing,shapes,graphs}

\newcounter{sarrow}

\newcounter{sarrow1}
\newcommand\xnrsquigarrow[1]{%
\stepcounter{sarrow1}%
\mathrel{\begin{tikzpicture}[baseline= {( $ (current bounding box.south) + (0,-0.5ex) $ )}]
\node[inner sep=.5ex] (\thesarrow) {$\scriptstyle #1$};
\path[draw,<-,decorate,
  decoration={zigzag,amplitude=0.7pt,segment length=1.2mm,pre=lineto,pre length=4pt}]
    (\thesarrow1.south east) -- (\thesarrow1.south west);
    $\slashedarrowfill@\relbar\relbar/$
    \end{tikzpicture}}%
}

\makeatletter
\def\slashedarrowfill@#1#2#3#4#5{%
  $\m@th\thickmuskip0mu\medmuskip\thickmuskip\thinmuskip\thickmuskip
   \relax#5#1\mkern-7mu%
   \cleaders\hbox{$#5\mkern-2mu#2\mkern-2mu$}\hfill
   \mathclap{#3}\mathclap{#2}%
   \cleaders\hbox{$#5\mkern-2mu#2\mkern-2mu$}\hfill
   \mkern-7mu#4$%
}
\def\rightslashedarrowfillb@{%
  \slashedarrowfill@\relbar\relbar/\rightarrow}
\newcommand\xnrightarrow[2][]{%
  \ext@arrow 0055{\rightslashedarrowfillb@}{#1}{#2}}

\def\rightslashedarrowfille@{%
  \slashedarrowfill@\relbar\relbar/\twoheadrightarrow}
\newcommand\xntworightarrow[2][]{%
  \ext@arrow 0055{\rightslashedarrowfille@}{#1}{#2}}

\def\rightslashedarrowfillg@{%
  \slashedarrowfill@\relbar\relbar{\raisebox{.12em}{}}\twoheadrightarrow}
\newcommand\xtworightarrow[2][]{%
  \ext@arrow 0055{\rightslashedarrowfillg@}{#1}{#2}}

\def\rightslashedarrowfillx@{%
  \slashedarrowfill@\Relbar\Relbar/\rightrightarrows}
\newcommand\xnTworightarrow[2][]{%
  \ext@arrow 0055{\rightslashedarrowfillx@}{#1}{#2}}

\def\rightslashedarrowfilly@{%
  \slashedarrowfill@\Relbar\Relbar{\raisebox{.12em}{}}\rightrightarrows}
\newcommand\xTworightarrow[2][]{%
  \ext@arrow 0055{\rightslashedarrowfilly@}{#1}{#2}}

\pgfdeclareshape{slash underlined}
{
  \inheritsavedanchors[from=rectangle] 
  \inheritanchorborder[from=rectangle]
  \inheritanchor[from=rectangle]{north}
  \inheritanchor[from=rectangle]{north west}
  \inheritanchor[from=rectangle]{north east}
  \inheritanchor[from=rectangle]{center}
  \inheritanchor[from=rectangle]{west}
  \inheritanchor[from=rectangle]{east}
  \inheritanchor[from=rectangle]{mid}
  \inheritanchor[from=rectangle]{mid west}
  \inheritanchor[from=rectangle]{mid east}
  \inheritanchor[from=rectangle]{base}
  \inheritanchor[from=rectangle]{base west}
  \inheritanchor[from=rectangle]{base east}
  \inheritanchor[from=rectangle]{south}
  \inheritanchor[from=rectangle]{south west}
  \inheritanchor[from=rectangle]{south east}
  \inheritanchorborder[from=rectangle]
  \foregroundpath{
    \southwest \pgf@xa=\pgf@x \pgf@ya=\pgf@y
    \northeast \pgf@xb=\pgf@x \pgf@yb=\pgf@y
    \pgf@xc=\pgf@xa
    \advance\pgf@xc by .5\pgf@xb
    \pgf@yc=\pgf@ya
    \advance\pgf@xc by -1.3pt
    \advance\pgf@yc by -1.8pt
    \pgfpathmoveto{\pgfqpoint{\pgf@xc}{\pgf@yc}}
    \advance\pgf@xc by  2.6pt
    \advance\pgf@yc by  3.6pt
    \pgfpathlineto{\pgfqpoint{\pgf@xc}{\pgf@yc}}
    \pgfpathmoveto{\pgfqpoint{\pgf@xa}{\pgf@ya}}
    \pgfpathlineto{\pgfqpoint{\pgf@xb}{\pgf@ya}}
 }
}
\tikzset{nomorepostaction/.code=\let\tikz@postactions\pgfutil@empty}

\newtheorem{theorem}{Theorem}[section]
\newtheorem{definition}[theorem]{Definition}

\begin{document}

\begin{titlepage}
\thispagestyle{empty}

\hrule
\begin{center}
{\bf\LARGE Secure Process Algebra}

\vspace{0.7cm}
--- Yong Wang ---

\vspace{2cm}
\begin{figure}[!htbp]
 \centering
 \includegraphics[width=1.0\textwidth]{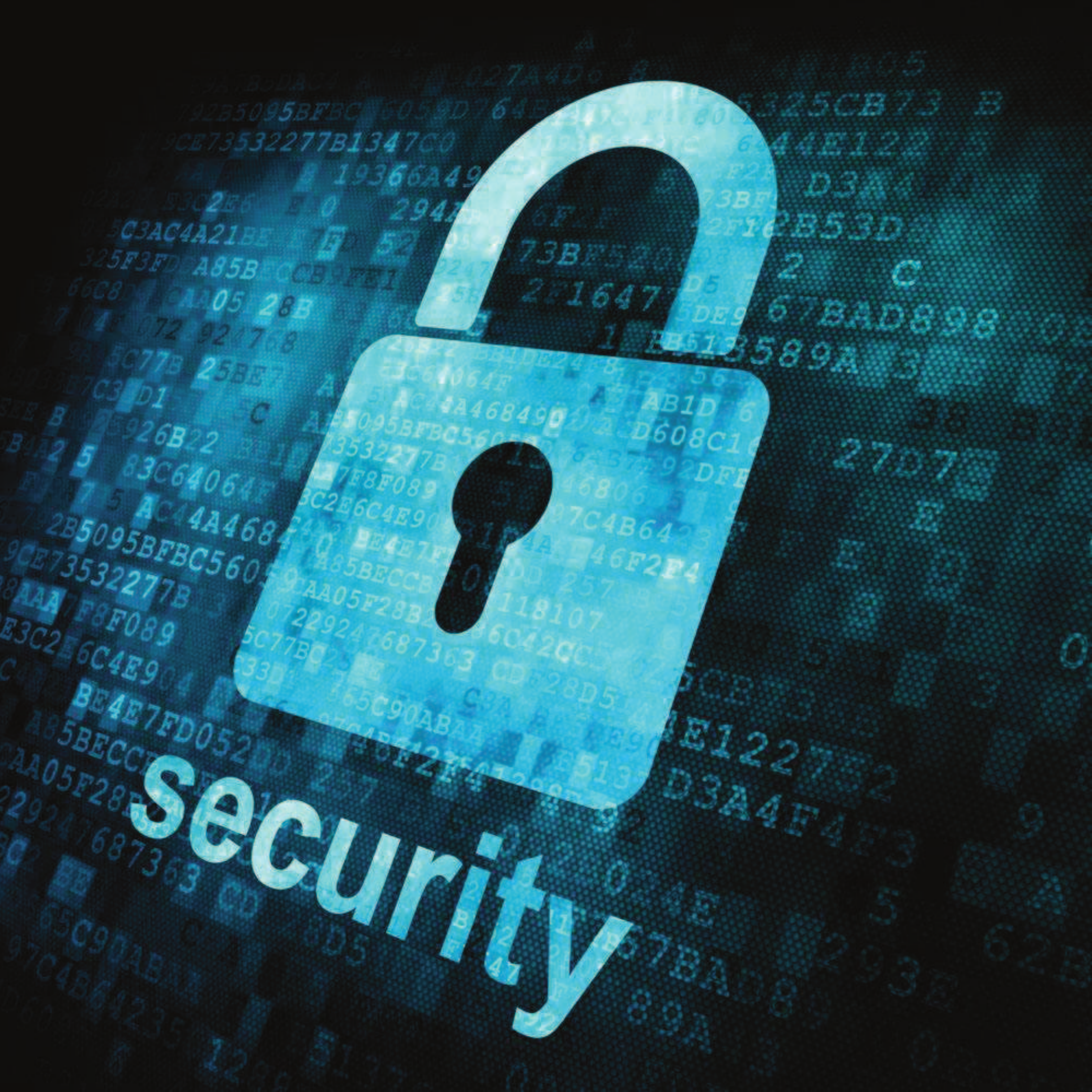}
\end{figure}

\end{center}
\end{titlepage}

\newpage 

\setcounter{page}{1}\pagenumbering{roman}

\tableofcontents

\newpage

\setcounter{page}{1}\pagenumbering{arabic}

        \section{Introduction}\label{intro}

A security protocol \cite{SP} includes some computational operations, some cryptographic operations (for examples, symmetric encryption/decryption, asymmetric encryption/decryption, hash
function, digital signatures, message authentication codes, random sequence generations, XOR operations, etc), some communication operations to exchanging data, and also the computational
logics among these operations.

Design a perfectly practical security protocol is a quite complex task, because of the open network environments and the complex security requirements against various known and unknown
attacks. How to design a security protocol usually heavenly depends on the experiences of security engineering. More for experiences, formal verifications can be used in the design of
security protocols to satisfy the main goal of the security protocol.

There are many formal verification tools to support the verifications of security protocols, such as BAN logic \cite{BAN} and those works based on process algebra. In the work based on process
algebra, there are works based on pi-calculus, such as spi-calculus \cite{SPI} and the applied pi-calculus \cite{AppPI}. The work based on process algebra has some advantages: they describe
the security protocols in a programming style, and have firmly theoretic foundations.

Based on our previous work on truly concurrent process algebras APTC \cite{APTC}, we use it to verify the security protocols. This work (called Secure APTC, abbreviated SAPTC) have the following
advantages in verifying security protocols:

\begin{enumerate}
  \item It has a firmly theoretic foundations, including equational logics, structured operational semantics, and axiomatizations between them;
  \item It has rich expressive powers to describe security protocols. Cryptographic operations are modeled as atomic actions and can be extended, explicit parallelism and communication
  mechanism to modeling communication operations and principals, rich computational properties to describing computational logics in the security protocols, including conditional guards,
  alternative composition, sequential composition, parallelism and communication, encapsulation and deadlock, recursion, abstraction.
  \item Especially by abstraction, it is convenient and obvious to observe the relations between the inputs and outputs of a security protocol, including the relations without any attack,
  the relations under each known attack, and the relations under unknown attacks if the unknown attacks can be described.
\end{enumerate}

This manuscript is organized as follows. In chapter \ref{aptc} and \ref{dmaptc}, we briefly introduce truly concurrent process algebra APTC and data manipulation in APTC. We extend APTC to
SAPTC to describe cryptographic properties in chapter \ref{saptc}. Then we introduce the cases of verifying security protocols, including key exchange related protocols in chapter
\ref{ke}, authentication protocols in chapter \ref{auth}, key exchange and authentication mixed protocols in chapter \ref{aopp}, other protocols in chapter \ref{aoop}, digital cash protocols in chapter \ref{aodcp},
and secure elections protocols in chapter \ref{aosep}.

\newpage\section{Truly Concurrent Process Algebra}\label{aptc}

In this chapter, we introduce the preliminaries on truly concurrent process algebra \cite{APTC}, which is based on truly concurrent operational semantics.

APTC eliminates the differences of structures of transition system, event structure, etc, and discusses their behavioral equivalences. It considers that there are two kinds of causality
relations: the chronological order modeled by the sequential composition and the causal order between different parallel branches modeled by the communication merge. It also considers
that there exist two kinds of confliction relations: the structural confliction modeled by the alternative composition and the conflictions in different parallel branches which should
be eliminated. Based on conservative extension, there are four modules in APTC: BATC (Basic Algebra for True Concurrency), APTC (Algebra for Parallelism in True Concurrency), recursion
and abstraction.

\subsection{Operational Semantics}\label{OS}

The semantics of $ACP$ is based on bisimulation/rooted branching bisimulation equivalences, and the modularity of $ACP$ relies on the concept of conservative extension, for the conveniences, we introduce some concepts and conclusions on them.

\begin{definition}[Bisimulation]
A bisimulation relation $R$ is a binary relation on processes such that: (1) if $p R q$ and $p\xrightarrow{a}p'$ then $q\xrightarrow{a}q'$ with $p' R q'$; (2) if $p R q$ and $q\xrightarrow{a}q'$ then $p\xrightarrow{a}p'$ with $p' R q'$; (3) if $p R q$ and $pP$, then $qP$; (4) if $p R q$ and $qP$, then $pP$. Two processes $p$ and $q$ are bisimilar, denoted by $p\sim_{HM} q$, if there is a bisimulation relation $R$ such that $p R q$.
\end{definition}

\begin{definition}[Congruence]
Let $\Sigma$ be a signature. An equivalence relation $R$ on $\mathcal{T}(\Sigma)$ is a congruence if for each $f\in\Sigma$, if $s_i R t_i$ for $i\in\{1,\cdots,ar(f)\}$, then $f(s_1,\cdots,s_{ar(f)}) R f(t_1,\cdots,t_{ar(f)})$.
\end{definition}

\begin{definition}[Branching bisimulation]
A branching bisimulation relation $R$ is a binary relation on the collection of processes such that: (1) if $p R q$ and $p\xrightarrow{a}p'$ then either $a\equiv \tau$ and $p' R q$ or there is a sequence of (zero or more) $\tau$-transitions $q\xrightarrow{\tau}\cdots\xrightarrow{\tau}q_0$ such that $p R q_0$ and $q_0\xrightarrow{a}q'$ with $p' R q'$; (2) if $p R q$ and $q\xrightarrow{a}q'$ then either $a\equiv \tau$ and $p R q'$ or there is a sequence of (zero or more) $\tau$-transitions $p\xrightarrow{\tau}\cdots\xrightarrow{\tau}p_0$ such that $p_0 R q$ and $p_0\xrightarrow{a}p'$ with $p' R q'$; (3) if $p R q$ and $pP$, then there is a sequence of (zero or more) $\tau$-transitions $q\xrightarrow{\tau}\cdots\xrightarrow{\tau}q_0$ such that $p R q_0$ and $q_0P$; (4) if $p R q$ and $qP$, then there is a sequence of (zero or more) $\tau$-transitions $p\xrightarrow{\tau}\cdots\xrightarrow{\tau}p_0$ such that $p_0 R q$ and $p_0P$. Two processes $p$ and $q$ are branching bisimilar, denoted by $p\approx_{bHM} q$, if there is a branching bisimulation relation $R$ such that $p R q$.
\end{definition}

\begin{definition}[Rooted branching bisimulation]
A rooted branching bisimulation relation $R$ is a binary relation on processes such that: (1) if $p R q$ and $p\xrightarrow{a}p'$ then $q\xrightarrow{a}q'$ with $p'\approx_{bHM} q'$; (2) if $p R q$ and $q\xrightarrow{a}q'$ then $p\xrightarrow{a}p'$ with $p'\approx_{bHM} q'$; (3) if $p R q$ and $pP$, then $qP$; (4) if $p R q$ and $qP$, then $pP$. Two processes $p$ and $q$ are rooted branching bisimilar, denoted by $p\approx_{rbHM} q$, if there is a rooted branching bisimulation relation $R$ such that $p R q$.
\end{definition}

\begin{definition}[Conservative extension]
Let $T_0$ and $T_1$ be TSSs (transition system specifications) over signatures $\Sigma_0$ and $\Sigma_1$, respectively. The TSS $T_0\oplus T_1$ is a conservative extension of $T_0$ if the LTSs (labeled transition systems) generated by $T_0$ and $T_0\oplus T_1$ contain exactly the same transitions $t\xrightarrow{a}t'$ and $tP$ with $t\in \mathcal{T}(\Sigma_0)$.
\end{definition}

\begin{definition}[Source-dependency]
The source-dependent variables in a transition rule of $\rho$ are defined inductively as follows: (1) all variables in the source of $\rho$ are source-dependent; (2) if $t\xrightarrow{a}t'$ is a premise of $\rho$ and all variables in $t$ are source-dependent, then all variables in $t'$ are source-dependent. A transition rule is source-dependent if all its variables are. A TSS is source-dependent if all its rules are.
\end{definition}

\begin{definition}[Freshness]
Let $T_0$ and $T_1$ be TSSs over signatures $\Sigma_0$ and $\Sigma_1$, respectively. A term in $\mathbb{T}(T_0\oplus T_1)$ is said to be fresh if it contains a function symbol from $\Sigma_1\setminus\Sigma_0$. Similarly, a transition label or predicate symbol in $T_1$ is fresh if it does not occur in $T_0$.
\end{definition}

\begin{theorem}[Conservative extension]\label{TCE}
Let $T_0$ and $T_1$ be TSSs over signatures $\Sigma_0$ and $\Sigma_1$, respectively, where $T_0$ and $T_0\oplus T_1$ are positive after reduction. Under the following conditions, $T_0\oplus T_1$ is a conservative extension of $T_0$. (1) $T_0$ is source-dependent. (2) For each $\rho\in T_1$, either the source of $\rho$ is fresh, or $\rho$ has a premise of the form $t\xrightarrow{a}t'$ or $tP$, where $t\in \mathbb{T}(\Sigma_0)$, all variables in $t$ occur in the source of $\rho$ and $t'$, $a$ or $P$ is fresh.
\end{theorem}

\subsection{Proof Techniques}\label{PT}

In this subsection, we introduce the concepts and conclusions about elimination, which is very important in the proof of completeness theorem.

\begin{definition}[Elimination property]
Let a process algebra with a defined set of basic terms as a subset of the set of closed terms over the process algebra. Then the process algebra has the elimination to basic terms property if for every closed term $s$ of the algebra, there exists a basic term $t$ of the algebra such that the algebra$\vdash s=t$.
\end{definition}

\begin{definition}[Strongly normalizing]
A term $s_0$ is called strongly normalizing if does not an infinite series of reductions beginning in $s_0$.
\end{definition}

\begin{definition}
We write $s>_{lpo} t$ if $s\rightarrow^+ t$ where $\rightarrow^+$ is the transitive closure of the reduction relation defined by the transition rules of an algebra.
\end{definition}

\begin{theorem}[Strong normalization]\label{SN}
Let a term rewriting system (TRS) with finitely many rewriting rules and let $>$ be a well-founded ordering on the signature of the corresponding algebra. If $s>_{lpo} t$ for each rewriting rule $s\rightarrow t$ in the TRS, then the term rewriting system is strongly normalizing.
\end{theorem}

\subsection{Basic Algebra for True Concurrency}

BATC has sequential composition $\cdot$ and alternative composition $+$ to capture the chronological ordered causality and the structural confliction. The constants are ranged over $A$,
the set of atomic actions. The algebraic laws on $\cdot$ and $+$ are sound and complete modulo truly concurrent bisimulation equivalences (including pomset bisimulation, step
bisimulation, hp-bisimulation and hhp-bisimulation).

\begin{definition}[Prime event structure with silent event]\label{PES}
Let $\Lambda$ be a fixed set of labels, ranged over $a,b,c,\cdots$ and $\tau$. A ($\Lambda$-labelled) prime event structure with silent event $\tau$ is a tuple
$\mathcal{E}=\langle \mathbb{E}, \leq, \sharp, \lambda\rangle$, where $\mathbb{E}$ is a denumerable set of events, including the silent event $\tau$. Let
$\hat{\mathbb{E}}=\mathbb{E}\backslash\{\tau\}$, exactly excluding $\tau$, it is obvious that $\hat{\tau^*}=\epsilon$, where $\epsilon$ is the empty event.
Let $\lambda:\mathbb{E}\rightarrow\Lambda$ be a labelling function and let $\lambda(\tau)=\tau$. And $\leq$, $\sharp$ are binary relations on $\mathbb{E}$,
called causality and conflict respectively, such that:

\begin{enumerate}
  \item $\leq$ is a partial order and $\lceil e \rceil = \{e'\in \mathbb{E}|e'\leq e\}$ is finite for all $e\in \mathbb{E}$. It is easy to see that
  $e\leq\tau^*\leq e'=e\leq\tau\leq\cdots\leq\tau\leq e'$, then $e\leq e'$.
  \item $\sharp$ is irreflexive, symmetric and hereditary with respect to $\leq$, that is, for all $e,e',e''\in \mathbb{E}$, if $e\sharp e'\leq e''$, then $e\sharp e''$.
\end{enumerate}

Then, the concepts of consistency and concurrency can be drawn from the above definition:

\begin{enumerate}
  \item $e,e'\in \mathbb{E}$ are consistent, denoted as $e\frown e'$, if $\neg(e\sharp e')$. A subset $X\subseteq \mathbb{E}$ is called consistent, if $e\frown e'$ for all
  $e,e'\in X$.
  \item $e,e'\in \mathbb{E}$ are concurrent, denoted as $e\parallel e'$, if $\neg(e\leq e')$, $\neg(e'\leq e)$, and $\neg(e\sharp e')$.
\end{enumerate}
\end{definition}

\begin{definition}[Configuration]
Let $\mathcal{E}$ be a PES. A (finite) configuration in $\mathcal{E}$ is a (finite) consistent subset of events $C\subseteq \mathcal{E}$, closed with respect to causality
(i.e. $\lceil C\rceil=C$). The set of finite configurations of $\mathcal{E}$ is denoted by $\mathcal{C}(\mathcal{E})$. We let $\hat{C}=C\backslash\{\tau\}$.
\end{definition}

A consistent subset of $X\subseteq \mathbb{E}$ of events can be seen as a pomset. Given $X, Y\subseteq \mathbb{E}$, $\hat{X}\sim \hat{Y}$ if $\hat{X}$ and $\hat{Y}$ are
isomorphic as pomsets. In the following of the paper, we say $C_1\sim C_2$, we mean $\hat{C_1}\sim\hat{C_2}$.

\begin{definition}[Pomset transitions and step]
Let $\mathcal{E}$ be a PES and let $C\in\mathcal{C}(\mathcal{E})$, and $\emptyset\neq X\subseteq \mathbb{E}$, if $C\cap X=\emptyset$ and $C'=C\cup X\in\mathcal{C}(\mathcal{E})$,
then $C\xrightarrow{X} C'$ is called a pomset transition from $C$ to $C'$. When the events in $X$ are pairwise concurrent, we say that $C\xrightarrow{X}C'$ is a step.
\end{definition}

\begin{definition}[Pomset, step bisimulation]\label{PSB}
Let $\mathcal{E}_1$, $\mathcal{E}_2$ be PESs. A pomset bisimulation is a relation $R\subseteq\mathcal{C}(\mathcal{E}_1)\times\mathcal{C}(\mathcal{E}_2)$, such that if
$(C_1,C_2)\in R$, and $C_1\xrightarrow{X_1}C_1'$ then $C_2\xrightarrow{X_2}C_2'$, with $X_1\subseteq \mathbb{E}_1$, $X_2\subseteq \mathbb{E}_2$, $X_1\sim X_2$ and $(C_1',C_2')\in R$,
and vice-versa. We say that $\mathcal{E}_1$, $\mathcal{E}_2$ are pomset bisimilar, written $\mathcal{E}_1\sim_p\mathcal{E}_2$, if there exists a pomset bisimulation $R$, such that
$(\emptyset,\emptyset)\in R$. By replacing pomset transitions with steps, we can get the definition of step bisimulation. When PESs $\mathcal{E}_1$ and $\mathcal{E}_2$ are step
bisimilar, we write $\mathcal{E}_1\sim_s\mathcal{E}_2$.
\end{definition}

\begin{definition}[Posetal product]
Given two PESs $\mathcal{E}_1$, $\mathcal{E}_2$, the posetal product of their configurations, denoted $\mathcal{C}(\mathcal{E}_1)\overline{\times}\mathcal{C}(\mathcal{E}_2)$,
is defined as

$$\{(C_1,f,C_2)|C_1\in\mathcal{C}(\mathcal{E}_1),C_2\in\mathcal{C}(\mathcal{E}_2),f:C_1\rightarrow C_2 \textrm{ isomorphism}\}.$$

A subset $R\subseteq\mathcal{C}(\mathcal{E}_1)\overline{\times}\mathcal{C}(\mathcal{E}_2)$ is called a posetal relation. We say that $R$ is downward closed when for any
$(C_1,f,C_2),(C_1',f',C_2')\in \mathcal{C}(\mathcal{E}_1)\overline{\times}\mathcal{C}(\mathcal{E}_2)$, if $(C_1,f,C_2)\subseteq (C_1',f',C_2')$ pointwise and $(C_1',f',C_2')\in R$,
then $(C_1,f,C_2)\in R$.

For $f:X_1\rightarrow X_2$, we define $f[x_1\mapsto x_2]:X_1\cup\{x_1\}\rightarrow X_2\cup\{x_2\}$, $z\in X_1\cup\{x_1\}$,(1)$f[x_1\mapsto x_2](z)=
x_2$,if $z=x_1$;(2)$f[x_1\mapsto x_2](z)=f(z)$, otherwise. Where $X_1\subseteq \mathbb{E}_1$, $X_2\subseteq \mathbb{E}_2$, $x_1\in \mathbb{E}_1$, $x_2\in \mathbb{E}_2$.
\end{definition}

\begin{definition}[(Hereditary) history-preserving bisimulation]\label{HHPB}
A history-preserving (hp-) bisimulation is a posetal relation $R\subseteq\mathcal{C}(\mathcal{E}_1)\overline{\times}\mathcal{C}(\mathcal{E}_2)$ such that if $(C_1,f,C_2)\in R$,
and $C_1\xrightarrow{e_1} C_1'$, then $C_2\xrightarrow{e_2} C_2'$, with $(C_1',f[e_1\mapsto e_2],C_2')\in R$, and vice-versa. $\mathcal{E}_1,\mathcal{E}_2$ are history-preserving
(hp-)bisimilar and are written $\mathcal{E}_1\sim_{hp}\mathcal{E}_2$ if there exists a hp-bisimulation $R$ such that $(\emptyset,\emptyset,\emptyset)\in R$.

A hereditary history-preserving (hhp-)bisimulation is a downward closed hp-bisimulation. $\mathcal{E}_1,\mathcal{E}_2$ are hereditary history-preserving (hhp-)bisimilar and are
written $\mathcal{E}_1\sim_{hhp}\mathcal{E}_2$.
\end{definition}

In the following, let $e_1, e_2, e_1', e_2'\in \mathbb{E}$, and let variables $x,y,z$ range over the set of terms for true concurrency, $p,q,s$ range over the set of closed terms.
The set of axioms of BATC consists of the laws given in Table \ref{AxiomsForBATC}.

\begin{center}
    \begin{table}
        \begin{tabular}{@{}ll@{}}
            \hline No. &Axiom\\
            $A1$ & $x+ y = y+ x$\\
            $A2$ & $(x+ y)+ z = x+ (y+ z)$\\
            $A3$ & $x+ x = x$\\
            $A4$ & $(x+ y)\cdot z = x\cdot z + y\cdot z$\\
            $A5$ & $(x\cdot y)\cdot z = x\cdot(y\cdot z)$\\
        \end{tabular}
        \caption{Axioms of BATC}
        \label{AxiomsForBATC}
    \end{table}
\end{center}

We give the operational transition rules of operators $\cdot$ and $+$ as Table \ref{TRForBATC} shows. And the predicate $\xrightarrow{e}\surd$ represents successful termination after
execution of the event $e$.

\begin{center}
    \begin{table}
        $$\frac{}{e\xrightarrow{e}\surd}$$
        $$\frac{x\xrightarrow{e}\surd}{x+ y\xrightarrow{e}\surd} \quad\frac{x\xrightarrow{e}x'}{x+ y\xrightarrow{e}x'} \quad\frac{y\xrightarrow{e}\surd}{x+ y\xrightarrow{e}\surd}
        \quad\frac{y\xrightarrow{e}y'}{x+ y\xrightarrow{e}y'}$$
        $$\frac{x\xrightarrow{e}\surd}{x\cdot y\xrightarrow{e} y} \quad\frac{x\xrightarrow{e}x'}{x\cdot y\xrightarrow{e}x'\cdot y}$$
        \caption{Transition rules of BATC}
        \label{TRForBATC}
    \end{table}
\end{center}

\begin{theorem}[Soundness of BATC modulo truly concurrent bisimulation equivalences]\label{SBATC}
The axiomatization of BATC is sound modulo truly concurrent bisimulation equivalences $\sim_{p}$, $\sim_{s}$, $\sim_{hp}$ and $\sim_{hhp}$. That is,

\begin{enumerate}
  \item let $x$ and $y$ be BATC terms. If BATC $\vdash x=y$, then $x\sim_{p} y$;
  \item let $x$ and $y$ be BATC terms. If BATC $\vdash x=y$, then $x\sim_{s} y$;
  \item let $x$ and $y$ be BATC terms. If BATC $\vdash x=y$, then $x\sim_{hp} y$;
  \item let $x$ and $y$ be BATC terms. If BATC $\vdash x=y$, then $x\sim_{hhp} y$.
\end{enumerate}

\end{theorem}

\begin{theorem}[Completeness of BATC modulo truly concurrent bisimulation equivalences]\label{CBATC}
The axiomatization of BATC is complete modulo truly concurrent bisimulation equivalences $\sim_{p}$, $\sim_{s}$, $\sim_{hp}$ and $\sim_{hhp}$. That is,

\begin{enumerate}
  \item let $p$ and $q$ be closed BATC terms, if $p\sim_{p} q$ then $p=q$;
  \item let $p$ and $q$ be closed BATC terms, if $p\sim_{s} q$ then $p=q$;
  \item let $p$ and $q$ be closed BATC terms, if $p\sim_{hp} q$ then $p=q$;
  \item let $p$ and $q$ be closed BATC terms, if $p\sim_{hhp} q$ then $p=q$.
\end{enumerate}

\end{theorem}

\subsection{Algebra for Parallelism in True Concurrency}

APTC uses the whole parallel operator $\between$, the auxiliary binary parallel $\parallel$ to model parallelism, and the communication merge $\mid$ to model communications among
different parallel branches, and also the unary conflict elimination operator $\Theta$ and the binary unless operator $\triangleleft$ to eliminate conflictions among different parallel
branches. Since a communication may be blocked, a new constant called deadlock $\delta$ is extended to $A$, and also a new unary encapsulation operator $\partial_H$ is introduced to
eliminate $\delta$, which may exist in the processes. The algebraic laws on these operators are also sound and complete modulo truly concurrent bisimulation equivalences (including
pomset bisimulation, step bisimulation, hp-bisimulation, but not hhp-bisimulation). Note that, the parallel operator $\parallel$ in a process cannot be eliminated by deductions on
the process using axioms of APTC, but other operators can eventually be steadied by $\cdot$, $+$ and $\parallel$, this is also why truly concurrent bisimulations are called an
\emph{truly concurrent} semantics.

We design the axioms of APTC in Table \ref{AxiomsForAPTC}, including algebraic laws of parallel operator $\parallel$, communication operator $\mid$, conflict elimination operator
$\Theta$ and unless operator $\triangleleft$, encapsulation operator $\partial_H$, the deadlock constant $\delta$, and also the whole parallel operator $\between$.

\begin{center}
    \begin{table}
        \begin{tabular}{@{}ll@{}}
            \hline No. &Axiom\\
            $A6$ & $x+ \delta = x$\\
            $A7$ & $\delta\cdot x =\delta$\\
            $P1$ & $x\between y = x\parallel y + x\mid y$\\
            $P2$ & $x\parallel y = y \parallel x$\\
            $P3$ & $(x\parallel y)\parallel z = x\parallel (y\parallel z)$\\
            $P4$ & $e_1\parallel (e_2\cdot y) = (e_1\parallel e_2)\cdot y$\\
            $P5$ & $(e_1\cdot x)\parallel e_2 = (e_1\parallel e_2)\cdot x$\\
            $P6$ & $(e_1\cdot x)\parallel (e_2\cdot y) = (e_1\parallel e_2)\cdot (x\between y)$\\
            $P7$ & $(x+ y)\parallel z = (x\parallel z)+ (y\parallel z)$\\
            $P8$ & $x\parallel (y+ z) = (x\parallel y)+ (x\parallel z)$\\
            $P9$ & $\delta\parallel x = \delta$\\
            $P10$ & $x\parallel \delta = \delta$\\
            $C11$ & $e_1\mid e_2 = \gamma(e_1,e_2)$\\
            $C12$ & $e_1\mid (e_2\cdot y) = \gamma(e_1,e_2)\cdot y$\\
            $C13$ & $(e_1\cdot x)\mid e_2 = \gamma(e_1,e_2)\cdot x$\\
            $C14$ & $(e_1\cdot x)\mid (e_2\cdot y) = \gamma(e_1,e_2)\cdot (x\between y)$\\
            $C15$ & $(x+ y)\mid z = (x\mid z) + (y\mid z)$\\
            $C16$ & $x\mid (y+ z) = (x\mid y)+ (x\mid z)$\\
            $C17$ & $\delta\mid x = \delta$\\
            $C18$ & $x\mid\delta = \delta$\\
            $CE19$ & $\Theta(e) = e$\\
            $CE20$ & $\Theta(\delta) = \delta$\\
            $CE21$ & $\Theta(x+ y) = \Theta(x)\triangleleft y + \Theta(y)\triangleleft x$\\
            $CE22$ & $\Theta(x\cdot y)=\Theta(x)\cdot\Theta(y)$\\
            $CE23$ & $\Theta(x\parallel y) = ((\Theta(x)\triangleleft y)\parallel y)+ ((\Theta(y)\triangleleft x)\parallel x)$\\
            $CE24$ & $\Theta(x\mid y) = ((\Theta(x)\triangleleft y)\mid y)+ ((\Theta(y)\triangleleft x)\mid x)$\\
            $U25$ & $(\sharp(e_1,e_2))\quad e_1\triangleleft e_2 = \tau$\\
            $U26$ & $(\sharp(e_1,e_2),e_2\leq e_3)\quad e_1\triangleleft e_3 = e_1$\\
            $U27$ & $(\sharp(e_1,e_2),e_2\leq e_3)\quad e3\triangleleft e_1 = \tau$\\
            $U28$ & $e\triangleleft \delta = e$\\
            $U29$ & $\delta \triangleleft e = \delta$\\
            $U30$ & $(x+ y)\triangleleft z = (x\triangleleft z)+ (y\triangleleft z)$\\
            $U31$ & $(x\cdot y)\triangleleft z = (x\triangleleft z)\cdot (y\triangleleft z)$\\
            $U32$ & $(x\parallel y)\triangleleft z = (x\triangleleft z)\parallel (y\triangleleft z)$\\
            $U33$ & $(x\mid y)\triangleleft z = (x\triangleleft z)\mid (y\triangleleft z)$\\
            $U34$ & $x\triangleleft (y+ z) = (x\triangleleft y)\triangleleft z$\\
            $U35$ & $x\triangleleft (y\cdot z)=(x\triangleleft y)\triangleleft z$\\
            $U36$ & $x\triangleleft (y\parallel z) = (x\triangleleft y)\triangleleft z$\\
            $U37$ & $x\triangleleft (y\mid z) = (x\triangleleft y)\triangleleft z$\\
            $D1$ & $e\notin H\quad\partial_H(e) = e$\\
            $D2$ & $e\in H\quad \partial_H(e) = \delta$\\
            $D3$ & $\partial_H(\delta) = \delta$\\
            $D4$ & $\partial_H(x+ y) = \partial_H(x)+\partial_H(y)$\\
            $D5$ & $\partial_H(x\cdot y) = \partial_H(x)\cdot\partial_H(y)$\\
            $D6$ & $\partial_H(x\parallel y) = \partial_H(x)\parallel\partial_H(y)$\\
        \end{tabular}
        \caption{Axioms of APTC}
        \label{AxiomsForAPTC}
    \end{table}
\end{center}

we give the transition rules of APTC in Table \ref{TRForAPTC}, it is suitable for all truly concurrent behavioral equivalence, including pomset bisimulation, step bisimulation,
hp-bisimulation and hhp-bisimulation.

\begin{center}
    \begin{table}
        $$\frac{x\xrightarrow{e_1}\surd\quad y\xrightarrow{e_2}\surd}{x\parallel y\xrightarrow{\{e_1,e_2\}}\surd} \quad\frac{x\xrightarrow{e_1}x'\quad y\xrightarrow{e_2}\surd}{x\parallel y\xrightarrow{\{e_1,e_2\}}x'}$$
        $$\frac{x\xrightarrow{e_1}\surd\quad y\xrightarrow{e_2}y'}{x\parallel y\xrightarrow{\{e_1,e_2\}}y'} \quad\frac{x\xrightarrow{e_1}x'\quad y\xrightarrow{e_2}y'}{x\parallel y\xrightarrow{\{e_1,e_2\}}x'\between y'}$$
        $$\frac{x\xrightarrow{e_1}\surd\quad y\xrightarrow{e_2}\surd}{x\mid y\xrightarrow{\gamma(e_1,e_2)}\surd} \quad\frac{x\xrightarrow{e_1}x'\quad y\xrightarrow{e_2}\surd}{x\mid y\xrightarrow{\gamma(e_1,e_2)}x'}$$
        $$\frac{x\xrightarrow{e_1}\surd\quad y\xrightarrow{e_2}y'}{x\mid y\xrightarrow{\gamma(e_1,e_2)}y'} \quad\frac{x\xrightarrow{e_1}x'\quad y\xrightarrow{e_2}y'}{x\mid y\xrightarrow{\gamma(e_1,e_2)}x'\between y'}$$
        $$\frac{x\xrightarrow{e_1}\surd\quad (\sharp(e_1,e_2))}{\Theta(x)\xrightarrow{e_1}\surd} \quad\frac{x\xrightarrow{e_2}\surd\quad (\sharp(e_1,e_2))}{\Theta(x)\xrightarrow{e_2}\surd}$$
        $$\frac{x\xrightarrow{e_1}x'\quad (\sharp(e_1,e_2))}{\Theta(x)\xrightarrow{e_1}\Theta(x')} \quad\frac{x\xrightarrow{e_2}x'\quad (\sharp(e_1,e_2))}{\Theta(x)\xrightarrow{e_2}\Theta(x')}$$
        $$\frac{x\xrightarrow{e_1}\surd \quad y\nrightarrow^{e_2}\quad (\sharp(e_1,e_2))}{x\triangleleft y\xrightarrow{\tau}\surd}
        \quad\frac{x\xrightarrow{e_1}x' \quad y\nrightarrow^{e_2}\quad (\sharp(e_1,e_2))}{x\triangleleft y\xrightarrow{\tau}x'}$$
        $$\frac{x\xrightarrow{e_1}\surd \quad y\nrightarrow^{e_3}\quad (\sharp(e_1,e_2),e_2\leq e_3)}{x\triangleleft y\xrightarrow{e_1}\surd}
        \quad\frac{x\xrightarrow{e_1}x' \quad y\nrightarrow^{e_3}\quad (\sharp(e_1,e_2),e_2\leq e_3)}{x\triangleleft y\xrightarrow{e_1}x'}$$
        $$\frac{x\xrightarrow{e_3}\surd \quad y\nrightarrow^{e_2}\quad (\sharp(e_1,e_2),e_1\leq e_3)}{x\triangleleft y\xrightarrow{\tau}\surd}
        \quad\frac{x\xrightarrow{e_3}x' \quad y\nrightarrow^{e_2}\quad (\sharp(e_1,e_2),e_1\leq e_3)}{x\triangleleft y\xrightarrow{\tau}x'}$$
        $$\frac{x\xrightarrow{e}\surd}{\partial_H(x)\xrightarrow{e}\surd}\quad (e\notin H)\quad\quad\frac{x\xrightarrow{e}x'}{\partial_H(x)\xrightarrow{e}\partial_H(x')}\quad(e\notin H)$$
        \caption{Transition rules of APTC}
        \label{TRForAPTC}
    \end{table}
\end{center}

\begin{theorem}[Soundness of APTC modulo truly concurrent bisimulation equivalences]\label{SAPTC}
The axiomatization of APTC is sound modulo truly concurrent bisimulation equivalences $\sim_{p}$, $\sim_{s}$, and $\sim_{hp}$. That is,

\begin{enumerate}
  \item let $x$ and $y$ be APTC terms. If APTC $\vdash x=y$, then $x\sim_{p} y$;
  \item let $x$ and $y$ be APTC terms. If APTC $\vdash x=y$, then $x\sim_{s} y$;
  \item let $x$ and $y$ be APTC terms. If APTC $\vdash x=y$, then $x\sim_{hp} y$.
\end{enumerate}

\end{theorem}

\begin{theorem}[Completeness of APTC modulo truly concurrent bisimulation equivalences]\label{CAPTC}
The axiomatization of APTC is complete modulo truly concurrent bisimulation equivalences $\sim_{p}$, $\sim_{s}$, and $\sim_{hp}$. That is,

\begin{enumerate}
  \item let $p$ and $q$ be closed APTC terms, if $p\sim_{p} q$ then $p=q$;
  \item let $p$ and $q$ be closed APTC terms, if $p\sim_{s} q$ then $p=q$;
  \item let $p$ and $q$ be closed APTC terms, if $p\sim_{hp} q$ then $p=q$.
\end{enumerate}

\end{theorem}

\subsection{Recursion}

To model infinite computation, recursion is introduced into APTC. In order to obtain a sound and complete theory, guarded recursion and linear recursion are needed. The corresponding
axioms are RSP (Recursive Specification Principle) and RDP (Recursive Definition Principle), RDP says the solutions of a recursive specification can represent the behaviors of the
specification, while RSP says that a guarded recursive specification has only one solution, they are sound with respect to APTC with guarded recursion modulo several truly concurrent
bisimulation equivalences (including pomset bisimulation, step bisimulation and hp-bisimulation), and they are complete with respect to APTC with linear recursion modulo several truly
concurrent bisimulation equivalences (including pomset bisimulation, step bisimulation and hp-bisimulation). In the following, $E,F,G$ are recursion specifications, $X,Y,Z$ are
recursive variables.

For a guarded recursive specifications $E$ with the form

$$X_1=t_1(X_1,\cdots,X_n)$$
$$\cdots$$
$$X_n=t_n(X_1,\cdots,X_n)$$

the behavior of the solution $\langle X_i|E\rangle$ for the recursion variable $X_i$ in $E$, where $i\in\{1,\cdots,n\}$, is exactly the behavior of their right-hand sides
$t_i(X_1,\cdots,X_n)$, which is captured by the two transition rules in Table \ref{TRForGR}.

\begin{center}
    \begin{table}
        $$\frac{t_i(\langle X_1|E\rangle,\cdots,\langle X_n|E\rangle)\xrightarrow{\{e_1,\cdots,e_k\}}\surd}{\langle X_i|E\rangle\xrightarrow{\{e_1,\cdots,e_k\}}\surd}$$
        $$\frac{t_i(\langle X_1|E\rangle,\cdots,\langle X_n|E\rangle)\xrightarrow{\{e_1,\cdots,e_k\}} y}{\langle X_i|E\rangle\xrightarrow{\{e_1,\cdots,e_k\}} y}$$
        \caption{Transition rules of guarded recursion}
        \label{TRForGR}
    \end{table}
\end{center}

The $RDP$ (Recursive Definition Principle) and the $RSP$ (Recursive Specification Principle) are shown in Table \ref{RDPRSP}.

\begin{center}
\begin{table}
  \begin{tabular}{@{}ll@{}}
\hline No. &Axiom\\
  $RDP$ & $\langle X_i|E\rangle = t_i(\langle X_1|E,\cdots,X_n|E\rangle)\quad (i\in\{1,\cdots,n\})$\\
  $RSP$ & if $y_i=t_i(y_1,\cdots,y_n)$ for $i\in\{1,\cdots,n\}$, then $y_i=\langle X_i|E\rangle \quad(i\in\{1,\cdots,n\})$\\
\end{tabular}
\caption{Recursive definition and specification principle}
\label{RDPRSP}
\end{table}
\end{center}

\begin{theorem}[Soundness of $APTC$ with guarded recursion]\label{SAPTCR}
Let $x$ and $y$ be $APTC$ with guarded recursion terms. If $APTC\textrm{ with guarded recursion}\vdash x=y$, then
\begin{enumerate}
  \item $x\sim_{s} y$;
  \item $x\sim_{p} y$;
  \item $x\sim_{hp} y$.
\end{enumerate}
\end{theorem}

\begin{theorem}[Completeness of $APTC$ with linear recursion]\label{CAPTCR}
Let $p$ and $q$ be closed $APTC$ with linear recursion terms, then,
\begin{enumerate}
  \item if $p\sim_{s} q$ then $p=q$;
  \item if $p\sim_{p} q$ then $p=q$;
  \item if $p\sim_{hp} q$ then $p=q$.
\end{enumerate}
\end{theorem}

\subsection{Abstraction}

To abstract away internal implementations from the external behaviors, a new constant $\tau$ called silent step is added to $A$, and also a new unary abstraction operator
$\tau_I$ is used to rename actions in $I$ into $\tau$ (the resulted APTC with silent step and abstraction operator is called $\textrm{APTC}_{\tau}$). The recursive specification
is adapted to guarded linear recursion to prevent infinite $\tau$-loops specifically. The axioms of $\tau$ and $\tau_I$ are sound modulo rooted branching truly concurrent bisimulation
 equivalences (several kinds of weakly truly concurrent bisimulation equivalences, including rooted branching pomset bisimulation, rooted branching step bisimulation and rooted branching hp-bisimulation). To eliminate infinite $\tau$-loops caused by $\tau_I$ and obtain the completeness, CFAR (Cluster Fair Abstraction Rule) is used to prevent infinite $\tau$-loops in a constructible way.

\begin{definition}[Weak pomset transitions and weak step]
Let $\mathcal{E}$ be a PES and let $C\in\mathcal{C}(\mathcal{E})$, and $\emptyset\neq X\subseteq \hat{\mathbb{E}}$, if $C\cap X=\emptyset$ and
$\hat{C'}=\hat{C}\cup X\in\mathcal{C}(\mathcal{E})$, then $C\xRightarrow{X} C'$ is called a weak pomset transition from $C$ to $C'$, where we define
$\xRightarrow{e}\triangleq\xrightarrow{\tau^*}\xrightarrow{e}\xrightarrow{\tau^*}$. And $\xRightarrow{X}\triangleq\xrightarrow{\tau^*}\xrightarrow{e}\xrightarrow{\tau^*}$,
for every $e\in X$. When the events in $X$ are pairwise concurrent, we say that $C\xRightarrow{X}C'$ is a weak step.
\end{definition}

\begin{definition}[Branching pomset, step bisimulation]\label{BPSB}
Assume a special termination predicate $\downarrow$, and let $\surd$ represent a state with $\surd\downarrow$. Let $\mathcal{E}_1$, $\mathcal{E}_2$ be PESs. A branching pomset
bisimulation is a relation $R\subseteq\mathcal{C}(\mathcal{E}_1)\times\mathcal{C}(\mathcal{E}_2)$, such that:
 \begin{enumerate}
   \item if $(C_1,C_2)\in R$, and $C_1\xrightarrow{X}C_1'$ then
   \begin{itemize}
     \item either $X\equiv \tau^*$, and $(C_1',C_2)\in R$;
     \item or there is a sequence of (zero or more) $\tau$-transitions $C_2\xrightarrow{\tau^*} C_2^0$, such that $(C_1,C_2^0)\in R$ and $C_2^0\xRightarrow{X}C_2'$ with
     $(C_1',C_2')\in R$;
   \end{itemize}
   \item if $(C_1,C_2)\in R$, and $C_2\xrightarrow{X}C_2'$ then
   \begin{itemize}
     \item either $X\equiv \tau^*$, and $(C_1,C_2')\in R$;
     \item or there is a sequence of (zero or more) $\tau$-transitions $C_1\xrightarrow{\tau^*} C_1^0$, such that $(C_1^0,C_2)\in R$ and $C_1^0\xRightarrow{X}C_1'$ with
     $(C_1',C_2')\in R$;
   \end{itemize}
   \item if $(C_1,C_2)\in R$ and $C_1\downarrow$, then there is a sequence of (zero or more) $\tau$-transitions $C_2\xrightarrow{\tau^*}C_2^0$ such that $(C_1,C_2^0)\in R$
   and $C_2^0\downarrow$;
   \item if $(C_1,C_2)\in R$ and $C_2\downarrow$, then there is a sequence of (zero or more) $\tau$-transitions $C_1\xrightarrow{\tau^*}C_1^0$ such that $(C_1^0,C_2)\in R$
   and $C_1^0\downarrow$.
 \end{enumerate}

We say that $\mathcal{E}_1$, $\mathcal{E}_2$ are branching pomset bisimilar, written $\mathcal{E}_1\approx_{bp}\mathcal{E}_2$, if there exists a branching pomset bisimulation $R$,
such that $(\emptyset,\emptyset)\in R$.

By replacing pomset transitions with steps, we can get the definition of branching step bisimulation. When PESs $\mathcal{E}_1$ and $\mathcal{E}_2$ are branching step bisimilar,
we write $\mathcal{E}_1\approx_{bs}\mathcal{E}_2$.
\end{definition}

\begin{definition}[Rooted branching pomset, step bisimulation]\label{RBPSB}
Assume a special termination predicate $\downarrow$, and let $\surd$ represent a state with $\surd\downarrow$. Let $\mathcal{E}_1$, $\mathcal{E}_2$ be PESs. A branching pomset
bisimulation is a relation $R\subseteq\mathcal{C}(\mathcal{E}_1)\times\mathcal{C}(\mathcal{E}_2)$, such that:
 \begin{enumerate}
   \item if $(C_1,C_2)\in R$, and $C_1\xrightarrow{X}C_1'$ then $C_2\xrightarrow{X}C_2'$ with $C_1'\approx_{bp}C_2'$;
   \item if $(C_1,C_2)\in R$, and $C_2\xrightarrow{X}C_2'$ then $C_1\xrightarrow{X}C_1'$ with $C_1'\approx_{bp}C_2'$;
   \item if $(C_1,C_2)\in R$ and $C_1\downarrow$, then $C_2\downarrow$;
   \item if $(C_1,C_2)\in R$ and $C_2\downarrow$, then $C_1\downarrow$.
 \end{enumerate}

We say that $\mathcal{E}_1$, $\mathcal{E}_2$ are rooted branching pomset bisimilar, written $\mathcal{E}_1\approx_{rbp}\mathcal{E}_2$, if there exists a rooted branching pomset
bisimulation $R$, such that $(\emptyset,\emptyset)\in R$.

By replacing pomset transitions with steps, we can get the definition of rooted branching step bisimulation. When PESs $\mathcal{E}_1$ and $\mathcal{E}_2$ are rooted branching step
bisimilar, we write $\mathcal{E}_1\approx_{rbs}\mathcal{E}_2$.
\end{definition}

\begin{definition}[Branching (hereditary) history-preserving bisimulation]\label{BHHPB}
Assume a special termination predicate $\downarrow$, and let $\surd$ represent a state with $\surd\downarrow$. A branching history-preserving (hp-) bisimulation is a weakly posetal
relation $R\subseteq\mathcal{C}(\mathcal{E}_1)\overline{\times}\mathcal{C}(\mathcal{E}_2)$ such that:

 \begin{enumerate}
   \item if $(C_1,f,C_2)\in R$, and $C_1\xrightarrow{e_1}C_1'$ then
   \begin{itemize}
     \item either $e_1\equiv \tau$, and $(C_1',f[e_1\mapsto \tau],C_2)\in R$;
     \item or there is a sequence of (zero or more) $\tau$-transitions $C_2\xrightarrow{\tau^*} C_2^0$, such that $(C_1,f,C_2^0)\in R$ and $C_2^0\xrightarrow{e_2}C_2'$ with
     $(C_1',f[e_1\mapsto e_2],C_2')\in R$;
   \end{itemize}
   \item if $(C_1,f,C_2)\in R$, and $C_2\xrightarrow{e_2}C_2'$ then
   \begin{itemize}
     \item either $X\equiv \tau$, and $(C_1,f[e_2\mapsto \tau],C_2')\in R$;
     \item or there is a sequence of (zero or more) $\tau$-transitions $C_1\xrightarrow{\tau^*} C_1^0$, such that $(C_1^0,f,C_2)\in R$ and $C_1^0\xrightarrow{e_1}C_1'$ with
     $(C_1',f[e_2\mapsto e_1],C_2')\in R$;
   \end{itemize}
   \item if $(C_1,f,C_2)\in R$ and $C_1\downarrow$, then there is a sequence of (zero or more) $\tau$-transitions $C_2\xrightarrow{\tau^*}C_2^0$ such that $(C_1,f,C_2^0)\in R$
   and $C_2^0\downarrow$;
   \item if $(C_1,f,C_2)\in R$ and $C_2\downarrow$, then there is a sequence of (zero or more) $\tau$-transitions $C_1\xrightarrow{\tau^*}C_1^0$ such that $(C_1^0,f,C_2)\in R$
   and $C_1^0\downarrow$.
 \end{enumerate}

$\mathcal{E}_1,\mathcal{E}_2$ are branching history-preserving (hp-)bisimilar and are written $\mathcal{E}_1\approx_{bhp}\mathcal{E}_2$ if there exists a branching hp-bisimulation
$R$ such that $(\emptyset,\emptyset,\emptyset)\in R$.

A branching hereditary history-preserving (hhp-)bisimulation is a downward closed branching hhp-bisimulation. $\mathcal{E}_1,\mathcal{E}_2$ are branching hereditary history-preserving
(hhp-)bisimilar and are written $\mathcal{E}_1\approx_{bhhp}\mathcal{E}_2$.
\end{definition}

\begin{definition}[Rooted branching (hereditary) history-preserving bisimulation]\label{RBHHPB}
Assume a special termination predicate $\downarrow$, and let $\surd$ represent a state with $\surd\downarrow$. A rooted branching history-preserving (hp-) bisimulation is a weakly
posetal relation $R\subseteq\mathcal{C}(\mathcal{E}_1)\overline{\times}\mathcal{C}(\mathcal{E}_2)$ such that:

 \begin{enumerate}
   \item if $(C_1,f,C_2)\in R$, and $C_1\xrightarrow{e_1}C_1'$, then $C_2\xrightarrow{e_2}C_2'$ with $C_1'\approx_{bhp}C_2'$;
   \item if $(C_1,f,C_2)\in R$, and $C_2\xrightarrow{e_2}C_1'$, then $C_1\xrightarrow{e_1}C_2'$ with $C_1'\approx_{bhp}C_2'$;
   \item if $(C_1,f,C_2)\in R$ and $C_1\downarrow$, then $C_2\downarrow$;
   \item if $(C_1,f,C_2)\in R$ and $C_2\downarrow$, then $C_1\downarrow$.
 \end{enumerate}

$\mathcal{E}_1,\mathcal{E}_2$ are rooted branching history-preserving (hp-)bisimilar and are written $\mathcal{E}_1\approx_{rbhp}\mathcal{E}_2$ if there exists rooted a branching
hp-bisimulation $R$ such that $(\emptyset,\emptyset,\emptyset)\in R$.

A rooted branching hereditary history-preserving (hhp-)bisimulation is a downward closed rooted branching hhp-bisimulation. $\mathcal{E}_1,\mathcal{E}_2$ are rooted branching
hereditary history-preserving (hhp-)bisimilar and are written $\mathcal{E}_1\approx_{rbhhp}\mathcal{E}_2$.
\end{definition}

The axioms and transition rules of $\textrm{APTC}_{\tau}$ are shown in Table \ref{AxiomsForTau} and Table \ref{TRForTau}.

\begin{center}
\begin{table}
  \begin{tabular}{@{}ll@{}}
\hline No. &Axiom\\
  $B1$ & $e\cdot\tau=e$\\
  $B2$ & $e\cdot(\tau\cdot(x+y)+x)=e\cdot(x+y)$\\
  $B3$ & $x\parallel\tau=x$\\
  $TI1$ & $e\notin I\quad \tau_I(e)=e$\\
  $TI2$ & $e\in I\quad \tau_I(e)=\tau$\\
  $TI3$ & $\tau_I(\delta)=\delta$\\
  $TI4$ & $\tau_I(x+y)=\tau_I(x)+\tau_I(y)$\\
  $TI5$ & $\tau_I(x\cdot y)=\tau_I(x)\cdot\tau_I(y)$\\
  $TI6$ & $\tau_I(x\parallel y)=\tau_I(x)\parallel\tau_I(y)$\\
  $CFAR$ & If $X$ is in a cluster for $I$ with exits \\
           & $\{(a_{11}\parallel\cdots\parallel a_{1i})Y_1,\cdots,(a_{m1}\parallel\cdots\parallel a_{mi})Y_m, b_{11}\parallel\cdots\parallel b_{1j},\cdots,b_{n1}\parallel\cdots\parallel b_{nj}\}$, \\
           & then $\tau\cdot\tau_I(\langle X|E\rangle)=$\\
           & $\tau\cdot\tau_I((a_{11}\parallel\cdots\parallel a_{1i})\langle Y_1|E\rangle+\cdots+(a_{m1}\parallel\cdots\parallel a_{mi})\langle Y_m|E\rangle+b_{11}\parallel\cdots\parallel b_{1j}+\cdots+b_{n1}\parallel\cdots\parallel b_{nj})$\\
\end{tabular}
\caption{Axioms of $\textrm{APTC}_{\tau}$}
\label{AxiomsForTau}
\end{table}
\end{center}

\begin{center}
    \begin{table}
        $$\frac{}{\tau\xrightarrow{\tau}\surd}$$
        $$\frac{x\xrightarrow{e}\surd}{\tau_I(x)\xrightarrow{e}\surd}\quad e\notin I
        \quad\quad\frac{x\xrightarrow{e}x'}{\tau_I(x)\xrightarrow{e}\tau_I(x')}\quad e\notin I$$

        $$\frac{x\xrightarrow{e}\surd}{\tau_I(x)\xrightarrow{\tau}\surd}\quad e\in I
        \quad\quad\frac{x\xrightarrow{e}x'}{\tau_I(x)\xrightarrow{\tau}\tau_I(x')}\quad e\in I$$
        \caption{Transition rule of $\textrm{APTC}_{\tau}$}
        \label{TRForTau}
    \end{table}
\end{center}

\begin{theorem}[Soundness of $APTC_{\tau}$ with guarded linear recursion]\label{SAPTCABS}
Let $x$ and $y$ be $APTC_{\tau}$ with guarded linear recursion terms. If $APTC_{\tau}$ with guarded linear recursion $\vdash x=y$, then
\begin{enumerate}
  \item $x\approx_{rbs} y$;
  \item $x\approx_{rbp} y$;
  \item $x\approx_{rbhp} y$.
\end{enumerate}
\end{theorem}

\begin{theorem}[Soundness of $CFAR$]\label{SCFAR}
$CFAR$ is sound modulo rooted branching truly concurrent bisimulation equivalences $\approx_{rbs}$, $\approx_{rbp}$ and $\approx_{rbhp}$.
\end{theorem}

\begin{theorem}[Completeness of $APTC_{\tau}$ with guarded linear recursion and $CFAR$]\label{CCFAR}
Let $p$ and $q$ be closed $APTC_{\tau}$ with guarded linear recursion and $CFAR$ terms, then,
\begin{enumerate}
  \item if $p\approx_{rbs} q$ then $p=q$;
  \item if $p\approx_{rbp} q$ then $p=q$;
  \item if $p\approx_{rbhp} q$ then $p=q$.
\end{enumerate}
\end{theorem}

\subsection{Placeholder}

We introduce a constant called shadow constant $\circledS$ to act for the placeholder that we ever used to deal entanglement in quantum process algebra. The transition rule of the shadow constant $\circledS$ is shown in Table \ref{TRForShadow}. The rule say that $\circledS$ can terminate successfully without executing any action.

\begin{center}
    \begin{table}
        $$\frac{}{\circledS\rightarrow\surd}$$
        \caption{Transition rule of the shadow constant}
        \label{TRForShadow}
    \end{table}
\end{center}

We need to adjust the definition of guarded linear recursive specification
to the following one.

\begin{definition}[Guarded linear recursive specification]\label{GLRSS}
A linear recursive specification $E$ is guarded if there does not exist an infinite sequence of $\tau$-transitions $\langle X|E\rangle\xrightarrow{\tau}\langle X'|E\rangle\xrightarrow{\tau}\langle X''|E\rangle\xrightarrow{\tau}\cdots$, and there does not exist an infinite sequence of $\circledS$-transitions $\langle X|E\rangle\rightarrow\langle X'|E\rangle\rightarrow\langle X''|E\rangle\rightarrow\cdots$.
\end{definition}

\begin{theorem}[Conservativity of $APTC$ with respect to the shadow constant]
$APTC_{\tau}$ with guarded linear recursion and shadow constant is a conservative extension of $APTC_{\tau}$ with guarded linear recursion.
\end{theorem}

We design the axioms for the shadow constant $\circledS$ in Table \ref{AxiomsForShadow}. And for $\circledS^e_i$, we add superscript $e$ to denote $\circledS$ is belonging to $e$ and subscript $i$ to denote that it is the $i$-th shadow of $e$. And we extend the set $\mathbb{E}$ to the set $\mathbb{E}\cup\{\tau\}\cup\{\delta\}\cup\{\circledS^{e}_i\}$.

\begin{center}
\begin{table}
  \begin{tabular}{@{}ll@{}}
\hline No. &Axiom\\
  $SC1$ & $\circledS\cdot x = x$\\
  $SC2$ & $x\cdot \circledS = x$\\
  $SC3$ & $\circledS^{e}\parallel e=e$\\
  $SC4$ & $e\parallel(\circledS^{e}\cdot y) = e\cdot y$\\
  $SC5$ & $\circledS^{e}\parallel(e\cdot y) = e\cdot y$\\
  $SC6$ & $(e\cdot x)\parallel\circledS^{e} = e\cdot x$\\
  $SC7$ & $(\circledS^{e}\cdot x)\parallel e = e\cdot x$\\
  $SC8$ & $(e\cdot x)\parallel(\circledS^{e}\cdot y) = e\cdot (x\between y)$\\
  $SC9$ & $(\circledS^{e}\cdot x)\parallel(e\cdot y) = e\cdot (x\between y)$\\
\end{tabular}
\caption{Axioms of shadow constant}
\label{AxiomsForShadow}
\end{table}
\end{center}

The mismatch of action and its shadows in parallelism will cause deadlock, that is, $e\parallel \circledS^{e'}=\delta$ with $e\neq e'$. We must make all shadows $\circledS^e_i$ are distinct, to ensure $f$ in hp-bisimulation is an isomorphism.

\begin{theorem}[Soundness of the shadow constant]\label{SShadow}
Let $x$ and $y$ be $APTC_{\tau}$ with guarded linear recursion and the shadow constant terms. If $APTC_{\tau}$ with guarded linear recursion and the shadow constant $\vdash x=y$, then
\begin{enumerate}
  \item $x\approx_{rbs} y$;
  \item $x\approx_{rbp} y$;
  \item $x\approx_{rbhp} y$.
\end{enumerate}
\end{theorem}

\begin{theorem}[Completeness of the shadow constant]\label{CRenaming}
Let $p$ and $q$ be closed $APTC_{\tau}$ with guarded linear recursion and $CFAR$ and the shadow constant terms, then,
\begin{enumerate}
  \item if $p\approx_{rbs} q$ then $p=q$;
  \item if $p\approx_{rbp} q$ then $p=q$;
  \item if $p\approx_{rbhp} q$ then $p=q$.
\end{enumerate}
\end{theorem}

With the shadow constant, we have

\begin{eqnarray}
\partial_H((a\cdot r_b)\between w_b)&=&\partial_H((a\cdot r_b) \between (\circledS^a_1\cdot w_b)) \nonumber\\
&=&a\cdot c_b\nonumber
\end{eqnarray}

with $H=\{r_b,w_b\}$ and $\gamma(r_b,w_b)\triangleq c_b$.

And we see the following example:

\begin{eqnarray}
a\between b&=&a\parallel b+a\mid b \nonumber\\
&=&a\parallel b + a\parallel b + a\parallel b +a\mid b \nonumber\\
&=&a\parallel (\circledS^a_1\cdot b) + (\circledS^b_1\cdot a)\parallel b+a\parallel b +a\mid b \nonumber\\
&=&(a\parallel\circledS^a_1)\cdot b + (\circledS^b_1\parallel b)\cdot a+a\parallel b +a\mid b \nonumber\\
&=&a\cdot b+b\cdot a+a\parallel b + a\mid b\nonumber
\end{eqnarray}

What do we see? Yes. The parallelism contains both interleaving and true concurrency. This may be why true concurrency is called \emph{\textbf{true} concurrency}.

\subsection{Axiomatization for Hhp-Bisimilarity}{\label{ahhpb}}

Since hhp-bisimilarity is a downward closed hp-bisimilarity and can be downward closed to single atomic event, which implies bisimilarity. As Moller \cite{ILM} proven, there is not a finite sound and complete axiomatization for parallelism $\parallel$ modulo bisimulation equivalence, so there is not a finite sound and complete axiomatization for parallelism $\parallel$ modulo hhp-bisimulation equivalence either. Inspired by the way of left merge to modeling the full merge for bisimilarity, we introduce a left parallel composition $\leftmerge$ to model the full parallelism $\parallel$ for hhp-bisimilarity.

In the following subsection, we add left parallel composition $\leftmerge$ to the whole theory. Because the resulting theory is similar to the former, we only list the significant differences, and all proofs of the conclusions are left to the reader.

\subsubsection{$APTC$ with Left Parallel Composition}

The transition rules of left parallel composition $\leftmerge$ are shown in Table \ref{TRForLeftParallel}. With a little abuse, we extend the causal order relation $\leq$ on $\mathbb{E}$ to include the original partial order (denoted by $<$) and concurrency (denoted by $=$).

\begin{center}
    \begin{table}
        $$\frac{x\xrightarrow{e_1}\surd\quad y\xrightarrow{e_2}\surd \quad(e_1\leq e_2)}{x\leftmerge y\xrightarrow{\{e_1,e_2\}}\surd} \quad\frac{x\xrightarrow{e_1}x'\quad y\xrightarrow{e_2}\surd \quad(e_1\leq e_2)}{x\leftmerge y\xrightarrow{\{e_1,e_2\}}x'}$$
        $$\frac{x\xrightarrow{e_1}\surd\quad y\xrightarrow{e_2}y' \quad(e_1\leq e_2)}{x\leftmerge y\xrightarrow{\{e_1,e_2\}}y'} \quad\frac{x\xrightarrow{e_1}x'\quad y\xrightarrow{e_2}y' \quad(e_1\leq e_2)}{x\leftmerge y\xrightarrow{\{e_1,e_2\}}x'\between y'}$$
        \caption{Transition rules of left parallel operator $\leftmerge$}
        \label{TRForLeftParallel}
    \end{table}
\end{center}

The new axioms for parallelism are listed in Table \ref{AxiomsForLeftParallelism}.

\begin{center}
    \begin{table}
        \begin{tabular}{@{}ll@{}}
            \hline No. &Axiom\\
            $A6$ & $x+ \delta = x$\\
            $A7$ & $\delta\cdot x =\delta$\\
            $P1$ & $x\between y = x\parallel y + x\mid y$\\
            $P2$ & $x\parallel y = y \parallel x$\\
            $P3$ & $(x\parallel y)\parallel z = x\parallel (y\parallel z)$\\
            $P4$ & $x\parallel y = x\leftmerge y + y\leftmerge x$\\
            $P5$ & $(e_1\leq e_2)\quad e_1\leftmerge (e_2\cdot y) = (e_1\leftmerge e_2)\cdot y$\\
            $P6$ & $(e_1\leq e_2)\quad (e_1\cdot x)\leftmerge e_2 = (e_1\leftmerge e_2)\cdot x$\\
            $P7$ & $(e_1\leq e_2)\quad (e_1\cdot x)\leftmerge (e_2\cdot y) = (e_1\leftmerge e_2)\cdot (x\between y)$\\
            $P8$ & $(x+ y)\leftmerge z = (x\leftmerge z)+ (y\leftmerge z)$\\
            $P9$ & $\delta\leftmerge x = \delta$\\
            $C10$ & $e_1\mid e_2 = \gamma(e_1,e_2)$\\
            $C11$ & $e_1\mid (e_2\cdot y) = \gamma(e_1,e_2)\cdot y$\\
            $C12$ & $(e_1\cdot x)\mid e_2 = \gamma(e_1,e_2)\cdot x$\\
            $C13$ & $(e_1\cdot x)\mid (e_2\cdot y) = \gamma(e_1,e_2)\cdot (x\between y)$\\
            $C14$ & $(x+ y)\mid z = (x\mid z) + (y\mid z)$\\
            $C15$ & $x\mid (y+ z) = (x\mid y)+ (x\mid z)$\\
            $C16$ & $\delta\mid x = \delta$\\
            $C17$ & $x\mid\delta = \delta$\\
            $CE18$ & $\Theta(e) = e$\\
            $CE19$ & $\Theta(\delta) = \delta$\\
            $CE20$ & $\Theta(x+ y) = \Theta(x)\triangleleft y + \Theta(y)\triangleleft x$\\
            $CE21$ & $\Theta(x\cdot y)=\Theta(x)\cdot\Theta(y)$\\
            $CE22$ & $\Theta(x\leftmerge y) = ((\Theta(x)\triangleleft y)\leftmerge y)+ ((\Theta(y)\triangleleft x)\leftmerge x)$\\
            $CE23$ & $\Theta(x\mid y) = ((\Theta(x)\triangleleft y)\mid y)+ ((\Theta(y)\triangleleft x)\mid x)$\\
            $U24$ & $(\sharp(e_1,e_2))\quad e_1\triangleleft e_2 = \tau$\\
            $U25$ & $(\sharp(e_1,e_2),e_2\leq e_3)\quad e_1\triangleleft e_3 = e_1$\\
            $U26$ & $(\sharp(e_1,e_2),e_2\leq e_3)\quad e3\triangleleft e_1 = \tau$\\
            $U27$ & $e\triangleleft \delta = e$\\
            $U28$ & $\delta \triangleleft e = \delta$\\
            $U29$ & $(x+ y)\triangleleft z = (x\triangleleft z)+ (y\triangleleft z)$\\
            $U30$ & $(x\cdot y)\triangleleft z = (x\triangleleft z)\cdot (y\triangleleft z)$\\
            $U31$ & $(x\leftmerge y)\triangleleft z = (x\triangleleft z)\leftmerge (y\triangleleft z)$\\
            $U32$ & $(x\mid y)\triangleleft z = (x\triangleleft z)\mid (y\triangleleft z)$\\
            $U33$ & $x\triangleleft (y+ z) = (x\triangleleft y)\triangleleft z$\\
            $U34$ & $x\triangleleft (y\cdot z)=(x\triangleleft y)\triangleleft z$\\
            $U35$ & $x\triangleleft (y\leftmerge z) = (x\triangleleft y)\triangleleft z$\\
            $U36$ & $x\triangleleft (y\mid z) = (x\triangleleft y)\triangleleft z$\\
        \end{tabular}
        \caption{Axioms of parallelism with left parallel composition}
        \label{AxiomsForLeftParallelism}
    \end{table}
\end{center}

\begin{definition}[Basic terms of $APTC$ with left parallel composition]
The set of basic terms of $APTC$, $\mathcal{B}(APTC)$, is inductively defined as follows:
\begin{enumerate}
  \item $\mathbb{E}\subset\mathcal{B}(APTC)$;
  \item if $e\in \mathbb{E}, t\in\mathcal{B}(APTC)$ then $e\cdot t\in\mathcal{B}(APTC)$;
  \item if $t,s\in\mathcal{B}(APTC)$ then $t+ s\in\mathcal{B}(APTC)$;
  \item if $t,s\in\mathcal{B}(APTC)$ then $t\leftmerge s\in\mathcal{B}(APTC)$.
\end{enumerate}
\end{definition}

\begin{theorem}[Generalization of the algebra for left parallelism with respect to $BATC$]
The algebra for left parallelism is a generalization of $BATC$.
\end{theorem}

\begin{theorem}[Congruence theorem of $APTC$ with left parallel composition]
Truly concurrent bisimulation equivalences $\sim_{p}$, $\sim_s$, $\sim_{hp}$ and $\sim_{hhp}$ are all congruences with respect to $APTC$ with left parallel composition.
\end{theorem}

\begin{theorem}[Elimination theorem of parallelism with left parallel composition]
Let $p$ be a closed $APTC$ with left parallel composition term. Then there is a basic $APTC$ term $q$ such that $APTC\vdash p=q$.
\end{theorem}

\begin{theorem}[Soundness of parallelism  with left parallel composition modulo truly concurrent bisimulation equivalences]
Let $x$ and $y$ be $APTC$ with left parallel composition terms. If $APTC\vdash x=y$, then

\begin{enumerate}
  \item $x\sim_{s} y$;
  \item $x\sim_{p} y$;
  \item $x\sim_{hp} y$;
  \item $x\sim_{hhp} y$.
\end{enumerate}
\end{theorem}

\begin{theorem}[Completeness of parallelism with left parallel composition modulo truly concurrent bisimulation equivalences]
Let $x$ and $y$ be $APTC$ terms.

\begin{enumerate}
  \item If $x\sim_{s} y$, then $APTC\vdash x=y$;
  \item if $x\sim_{p} y$, then $APTC\vdash x=y$;
  \item if $x\sim_{hp} y$, then $APTC\vdash x=y$;
  \item if $x\sim_{hhp} y$, then $APTC\vdash x=y$.
\end{enumerate}
\end{theorem}

The transition rules of encapsulation operator are the same, and the its axioms are shown in \ref{AxiomsForEncapsulationLeft}.

\begin{center}
    \begin{table}
        \begin{tabular}{@{}ll@{}}
            \hline No. &Axiom\\
            $D1$ & $e\notin H\quad\partial_H(e) = e$\\
            $D2$ & $e\in H\quad \partial_H(e) = \delta$\\
            $D3$ & $\partial_H(\delta) = \delta$\\
            $D4$ & $\partial_H(x+ y) = \partial_H(x)+\partial_H(y)$\\
            $D5$ & $\partial_H(x\cdot y) = \partial_H(x)\cdot\partial_H(y)$\\
            $D6$ & $\partial_H(x\leftmerge y) = \partial_H(x)\leftmerge\partial_H(y)$\\
        \end{tabular}
        \caption{Axioms of encapsulation operator with left parallel composition}
        \label{AxiomsForEncapsulationLeft}
    \end{table}
\end{center}

\begin{theorem}[Conservativity of $APTC$ with respect to the algebra for parallelism with left parallel composition]
$APTC$ is a conservative extension of the algebra for parallelism with left parallel composition.
\end{theorem}

\begin{theorem}[Congruence theorem of encapsulation operator $\partial_H$]
Truly concurrent bisimulation equivalences $\sim_{p}$, $\sim_s$, $\sim_{hp}$ and $\sim_{hhp}$ are all congruences with respect to encapsulation operator $\partial_H$.
\end{theorem}

\begin{theorem}[Elimination theorem of $APTC$]
Let $p$ be a closed $APTC$ term including the encapsulation operator $\partial_H$. Then there is a basic $APTC$ term $q$ such that $APTC\vdash p=q$.
\end{theorem}

\begin{theorem}[Soundness of $APTC$ modulo truly concurrent bisimulation equivalences]
Let $x$ and $y$ be $APTC$ terms including encapsulation operator $\partial_H$. If $APTC\vdash x=y$, then

\begin{enumerate}
  \item $x\sim_{s} y$;
  \item $x\sim_{p} y$;
  \item $x\sim_{hp} y$;
  \item $x\sim_{hhp} y$.
\end{enumerate}
\end{theorem}

\begin{theorem}[Completeness of $APTC$ modulo truly concurrent bisimulation equivalences]
Let $p$ and $q$ be closed $APTC$ terms including encapsulation operator $\partial_H$,

\begin{enumerate}
  \item if $p\sim_{s} q$ then $p=q$;
  \item if $p\sim_{p} q$ then $p=q$;
  \item if $p\sim_{hp} q$ then $p=q$;
  \item if $p\sim_{hhp} q$ then $p=q$.
\end{enumerate}
\end{theorem}

\subsubsection{Recursion}

\begin{definition}[Recursive specification]
A recursive specification is a finite set of recursive equations

$$X_1=t_1(X_1,\cdots,X_n)$$
$$\cdots$$
$$X_n=t_n(X_1,\cdots,X_n)$$

where the left-hand sides of $X_i$ are called recursion variables, and the right-hand sides $t_i(X_1,\cdots,X_n)$ are process terms in $APTC$ with possible occurrences of the recursion variables $X_1,\cdots,X_n$.
\end{definition}

\begin{definition}[Solution]
Processes $p_1,\cdots,p_n$ are a solution for a recursive specification $\{X_i=t_i(X_1,\cdots,X_n)|i\in\{1,\cdots,n\}\}$ (with respect to truly concurrent bisimulation equivalences $\sim_s$($\sim_p$, $\sim_{hp}$, $\sim_{hhp}$)) if $p_i\sim_s (\sim_p, \sim_{hp},\sim{hhp})t_i(p_1,\cdots,p_n)$ for $i\in\{1,\cdots,n\}$.
\end{definition}

\begin{definition}[Guarded recursive specification]
A recursive specification

$$X_1=t_1(X_1,\cdots,X_n)$$
$$...$$
$$X_n=t_n(X_1,\cdots,X_n)$$

is guarded if the right-hand sides of its recursive equations can be adapted to the form by applications of the axioms in $APTC$ and replacing recursion variables by the right-hand sides of their recursive equations,

$$(a_{11}\leftmerge\cdots\leftmerge a_{1i_1})\cdot s_1(X_1,\cdots,X_n)+\cdots+(a_{k1}\leftmerge\cdots\leftmerge a_{ki_k})\cdot s_k(X_1,\cdots,X_n)+(b_{11}\leftmerge\cdots\leftmerge b_{1j_1})+\cdots+(b_{1j_1}\leftmerge\cdots\leftmerge b_{lj_l})$$

where $a_{11},\cdots,a_{1i_1},a_{k1},\cdots,a_{ki_k},b_{11},\cdots,b_{1j_1},b_{1j_1},\cdots,b_{lj_l}\in \mathbb{E}$, and the sum above is allowed to be empty, in which case it represents the deadlock $\delta$.
\end{definition}

\begin{definition}[Linear recursive specification]
A recursive specification is linear if its recursive equations are of the form

$$(a_{11}\leftmerge\cdots\leftmerge a_{1i_1})X_1+\cdots+(a_{k1}\leftmerge\cdots\leftmerge a_{ki_k})X_k+(b_{11}\leftmerge\cdots\leftmerge b_{1j_1})+\cdots+(b_{1j_1}\leftmerge\cdots\leftmerge b_{lj_l})$$

where $a_{11},\cdots,a_{1i_1},a_{k1},\cdots,a_{ki_k},b_{11},\cdots,b_{1j_1},b_{1j_1},\cdots,b_{lj_l}\in \mathbb{E}$, and the sum above is allowed to be empty, in which case it represents the deadlock $\delta$.
\end{definition}

\begin{theorem}[Conservitivity of $APTC$ with guarded recursion]
$APTC$ with guarded recursion is a conservative extension of $APTC$.
\end{theorem}

\begin{theorem}[Congruence theorem of $APTC$ with guarded recursion]
Truly concurrent bisimulation equivalences $\sim_{p}$, $\sim_s$, $\sim_{hp}$, $\sim_{hhp}$ are all congruences with respect to $APTC$ with guarded recursion.
\end{theorem}

\begin{theorem}[Elimination theorem of $APTC$ with linear recursion]
Each process term in $APTC$ with linear recursion is equal to a process term $\langle X_1|E\rangle$ with $E$ a linear recursive specification.
\end{theorem}

\begin{theorem}[Soundness of $APTC$ with guarded recursion]
Let $x$ and $y$ be $APTC$ with guarded recursion terms. If $APTC\textrm{ with guarded recursion}\vdash x=y$, then
\begin{enumerate}
  \item $x\sim_{s} y$;
  \item $x\sim_{p} y$;
  \item $x\sim_{hp} y$;
  \item $x\sim_{hhp} y$.
\end{enumerate}
\end{theorem}

\begin{theorem}[Completeness of $APTC$ with linear recursion]
Let $p$ and $q$ be closed $APTC$ with linear recursion terms, then,
\begin{enumerate}
  \item if $p\sim_{s} q$ then $p=q$;
  \item if $p\sim_{p} q$ then $p=q$;
  \item if $p\sim_{hp} q$ then $p=q$;
  \item if $p\sim_{hhp} q$ then $p=q$.
\end{enumerate}
\end{theorem}

\subsubsection{Abstraction}

\begin{definition}[Guarded linear recursive specification]
A recursive specification is linear if its recursive equations are of the form

$$(a_{11}\leftmerge\cdots\leftmerge a_{1i_1})X_1+\cdots+(a_{k1}\leftmerge\cdots\leftmerge a_{ki_k})X_k+(b_{11}\leftmerge\cdots\leftmerge b_{1j_1})+\cdots+(b_{1j_1}\leftmerge\cdots\leftmerge b_{lj_l})$$

where $a_{11},\cdots,a_{1i_1},a_{k1},\cdots,a_{ki_k},b_{11},\cdots,b_{1j_1},b_{1j_1},\cdots,b_{lj_l}\in \mathbb{E}\cup\{\tau\}$, and the sum above is allowed to be empty, in which case it represents the deadlock $\delta$.

A linear recursive specification $E$ is guarded if there does not exist an infinite sequence of $\tau$-transitions $\langle X|E\rangle\xrightarrow{\tau}\langle X'|E\rangle\xrightarrow{\tau}\langle X''|E\rangle\xrightarrow{\tau}\cdots$.
\end{definition}

The transition rules of $\tau$ are the same, and axioms of $\tau$ are as Table \ref{AxiomsForTauLeft} shows.

\begin{theorem}[Conservitivity of $APTC$ with silent step and guarded linear recursion]
$APTC$ with silent step and guarded linear recursion is a conservative extension of $APTC$ with linear recursion.
\end{theorem}

\begin{theorem}[Congruence theorem of $APTC$ with silent step and guarded linear recursion]
Rooted branching truly concurrent bisimulation equivalences $\approx_{rbp}$, $\approx_{rbs}$, $\approx_{rbhp}$, and $\approx_{rbhhp}$ are all congruences with respect to $APTC$ with silent step and guarded linear recursion.
\end{theorem}

\begin{center}
\begin{table}
  \begin{tabular}{@{}ll@{}}
\hline No. &Axiom\\
  $B1$ & $e\cdot\tau=e$\\
  $B2$ & $e\cdot(\tau\cdot(x+y)+x)=e\cdot(x+y)$\\
  $B3$ & $x\leftmerge\tau=x$\\
\end{tabular}
\caption{Axioms of silent step}
\label{AxiomsForTauLeft}
\end{table}
\end{center}

\begin{theorem}[Elimination theorem of $APTC$ with silent step and guarded linear recursion]
Each process term in $APTC$ with silent step and guarded linear recursion is equal to a process term $\langle X_1|E\rangle$ with $E$ a guarded linear recursive specification.
\end{theorem}

\begin{theorem}[Soundness of $APTC$ with silent step and guarded linear recursion]
Let $x$ and $y$ be $APTC$ with silent step and guarded linear recursion terms. If $APTC$ with silent step and guarded linear recursion $\vdash x=y$, then
\begin{enumerate}
  \item $x\approx_{rbs} y$;
  \item $x\approx_{rbp} y$;
  \item $x\approx_{rbhp} y$;
  \item $x\approx_{rbhhp} y$.
\end{enumerate}
\end{theorem}

\begin{theorem}[Completeness of $APTC$ with silent step and guarded linear recursion]
Let $p$ and $q$ be closed $APTC$ with silent step and guarded linear recursion terms, then,
\begin{enumerate}
  \item if $p\approx_{rbs} q$ then $p=q$;
  \item if $p\approx_{rbp} q$ then $p=q$;
  \item if $p\approx_{rbhp} q$ then $p=q$;
  \item if $p\approx_{rbhhp} q$ then $p=q$.
\end{enumerate}
\end{theorem}

The transition rules of $\tau_I$ are the same, and the axioms are shown in Table \ref{AxiomsForAbstractionLeft}.

\begin{theorem}[Conservitivity of $APTC_{\tau}$ with guarded linear recursion]
$APTC_{\tau}$ with guarded linear recursion is a conservative extension of $APTC$ with silent step and guarded linear recursion.
\end{theorem}

\begin{theorem}[Congruence theorem of $APTC_{\tau}$ with guarded linear recursion]
Rooted branching truly concurrent bisimulation equivalences $\approx_{rbp}$, $\approx_{rbs}$, $\approx_{rbhp}$ and $\approx_{rbhhp}$ are all congruences with respect to $APTC_{\tau}$ with guarded linear recursion.
\end{theorem}

\begin{center}
\begin{table}
  \begin{tabular}{@{}ll@{}}
\hline No. &Axiom\\
  $TI1$ & $e\notin I\quad \tau_I(e)=e$\\
  $TI2$ & $e\in I\quad \tau_I(e)=\tau$\\
  $TI3$ & $\tau_I(\delta)=\delta$\\
  $TI4$ & $\tau_I(x+y)=\tau_I(x)+\tau_I(y)$\\
  $TI5$ & $\tau_I(x\cdot y)=\tau_I(x)\cdot\tau_I(y)$\\
  $TI6$ & $\tau_I(x\leftmerge y)=\tau_I(x)\leftmerge\tau_I(y)$\\
\end{tabular}
\caption{Axioms of abstraction operator}
\label{AxiomsForAbstractionLeft}
\end{table}
\end{center}

\begin{theorem}[Soundness of $APTC_{\tau}$ with guarded linear recursion]
Let $x$ and $y$ be $APTC_{\tau}$ with guarded linear recursion terms. If $APTC_{\tau}$ with guarded linear recursion $\vdash x=y$, then
\begin{enumerate}
  \item $x\approx_{rbs} y$;
  \item $x\approx_{rbp} y$;
  \item $x\approx_{rbhp} y$;
  \item $x\approx_{rbhhp} y$.
\end{enumerate}
\end{theorem}

\begin{definition}[Cluster]
Let $E$ be a guarded linear recursive specification, and $I\subseteq \mathbb{E}$. Two recursion variable $X$ and $Y$ in $E$ are in the same cluster for $I$ iff there exist sequences of transitions $\langle X|E\rangle\xrightarrow{\{b_{11},\cdots, b_{1i}\}}\cdots\xrightarrow{\{b_{m1},\cdots, b_{mi}\}}\langle Y|E\rangle$ and $\langle Y|E\rangle\xrightarrow{\{c_{11},\cdots, c_{1j}\}}\cdots\xrightarrow{\{c_{n1},\cdots, c_{nj}\}}\langle X|E\rangle$, where $b_{11},\cdots,b_{mi},c_{11},\cdots,c_{nj}\in I\cup\{\tau\}$.

$a_1\leftmerge\cdots\leftmerge a_k$ or $(a_1\leftmerge\cdots\leftmerge a_k) X$ is an exit for the cluster $C$ iff: (1) $a_1\leftmerge\cdots\leftmerge a_k$ or $(a_1\leftmerge\cdots\leftmerge a_k) X$ is a summand at the right-hand side of the recursive equation for a recursion variable in $C$, and (2) in the case of $(a_1\leftmerge\cdots\leftmerge a_k) X$, either $a_l\notin I\cup\{\tau\}(l\in\{1,2,\cdots,k\})$ or $X\notin C$.
\end{definition}

\begin{center}
\begin{table}
  \begin{tabular}{@{}ll@{}}
\hline No. &Axiom\\
  $CFAR$ & If $X$ is in a cluster for $I$ with exits \\
           & $\{(a_{11}\leftmerge\cdots\leftmerge a_{1i})Y_1,\cdots,(a_{m1}\leftmerge\cdots\leftmerge a_{mi})Y_m, b_{11}\leftmerge\cdots\leftmerge b_{1j},\cdots,b_{n1}\leftmerge\cdots\leftmerge b_{nj}\}$, \\
           & then $\tau\cdot\tau_I(\langle X|E\rangle)=$\\
           & $\tau\cdot\tau_I((a_{11}\leftmerge\cdots\leftmerge a_{1i})\langle Y_1|E\rangle+\cdots+(a_{m1}\leftmerge\cdots\leftmerge a_{mi})\langle Y_m|E\rangle+b_{11}\leftmerge\cdots\leftmerge b_{1j}+\cdots+b_{n1}\leftmerge\cdots\leftmerge b_{nj})$\\
\end{tabular}
\caption{Cluster fair abstraction rule}
\label{CFARLeft}
\end{table}
\end{center}

\begin{theorem}[Soundness of $CFAR$]
$CFAR$ is sound modulo rooted branching truly concurrent bisimulation equivalences $\approx_{rbs}$, $\approx_{rbp}$, $\approx_{rbhp}$ and $\approx_{rbhhp}$.
\end{theorem}

\begin{theorem}[Completeness of $APTC_{\tau}$ with guarded linear recursion and $CFAR$]
Let $p$ and $q$ be closed $APTC_{\tau}$ with guarded linear recursion and $CFAR$ terms, then,
\begin{enumerate}
  \item if $p\approx_{rbs} q$ then $p=q$;
  \item if $p\approx_{rbp} q$ then $p=q$;
  \item if $p\approx_{rbhp} q$ then $p=q$;
  \item if $p\approx_{rbhhp} q$ then $p=q$.
\end{enumerate}
\end{theorem}

\subsection{Applications}\label{app}

$APTC$ provides a formal framework based on truly concurrent behavioral semantics, which can be used to verify the correctness of system behaviors. In this subsection,
we tend to choose alternating bit protocol (ABP) \cite{ABP}.

The ABP protocol is used to ensure successful transmission of data through a corrupted channel. This success is based on the assumption that data can be resent an unlimited number of times, which is illustrated in Figure \ref{ABP}, we alter it into the true concurrency situation.

\begin{enumerate}
  \item Data elements $d_1,d_2,d_3,\cdots$ from a finite set $\Delta$ are communicated between a Sender and a Receiver.
  \item If the Sender reads a datum from channel $A_1$, then this datum is sent to the Receiver in parallel through channel $A_2$.
  \item The Sender processes the data in $\Delta$, formes new data, and sends them to the Receiver through channel $B$.
  \item And the Receiver sends the datum into channel $C_2$.
  \item If channel $B$ is corrupted, the message communicated through $B$ can be turn into an error message $\bot$.
  \item Every time the Receiver receives a message via channel $B$, it sends an acknowledgement to the Sender via channel $D$, which is also corrupted.
  \item Finally, then Sender and the Receiver send out their outputs in parallel through channels $C_1$ and $C_2$.
\end{enumerate}

\begin{figure}
    \centering
    \includegraphics{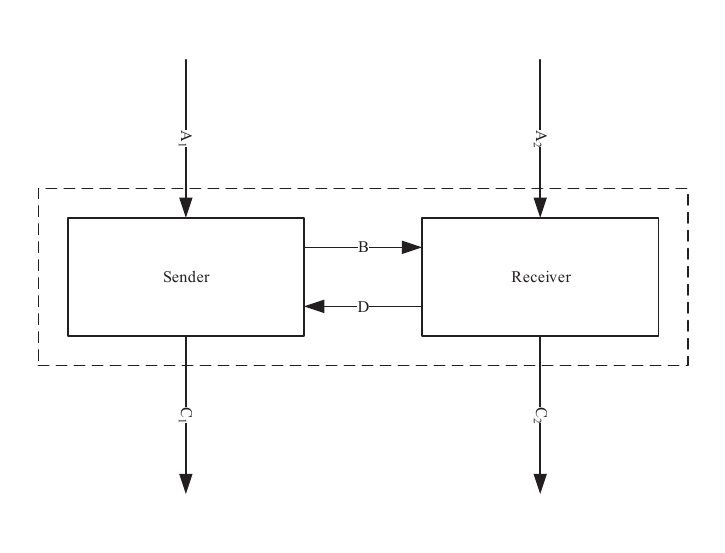}
    \caption{Alternating bit protocol}
    \label{ABP}
\end{figure}

In the truly concurrent ABP, the Sender sends its data to the Receiver; and the Receiver can also send its data to the Sender, for simplicity and without loss of generality, we assume that only the Sender sends its data and the Receiver only receives the data from the Sender. The Sender attaches a bit 0 to data elements $d_{2k-1}$ and a bit 1 to data elements $d_{2k}$, when they are sent into channel $B$. When the Receiver reads a datum, it sends back the attached bit via channel $D$. If the Receiver receives a corrupted message, then it sends back the previous acknowledgement to the Sender.

Then the state transition of the Sender can be described by $APTC$ as follows.

\begin{eqnarray}
&&S_b=\sum_{d\in\Delta}r_{A_1}(d)\cdot T_{db}\nonumber\\
&&T_{db}=(\sum_{d'\in\Delta}(s_B(d',b)\cdot s_{C_1}(d'))+s_B(\bot))\cdot U_{db}\nonumber\\
&&U_{db}=r_D(b)\cdot S_{1-b}+(r_D(1-b)+r_D(\bot))\cdot T_{db}\nonumber
\end{eqnarray}

where $s_B$ denotes sending data through channel $B$, $r_D$ denotes receiving data through channel $D$, similarly, $r_{A_1}$ means receiving data via channel $A_1$, $s_{C_1}$ denotes sending data via channel $C_1$, and $b\in\{0,1\}$.

And the state transition of the Receiver can be described by $APTC$ as follows.

\begin{eqnarray}
&&R_b=\sum_{d\in\Delta}r_{A_2}(d)\cdot R_b'\nonumber\\
&&R_b'=\sum_{d'\in\Delta}\{r_B(d',b)\cdot s_{C_2}(d')\cdot Q_b+r_B(d',1-b)\cdot Q_{1-b}\}+r_B(\bot)\cdot Q_{1-b}\nonumber\\
&&Q_b=(s_D(b)+s_D(\bot))\cdot R_{1-b}\nonumber
\end{eqnarray}

where $r_{A_2}$ denotes receiving data via channel $A_2$, $r_B$ denotes receiving data via channel $B$, $s_{C_2}$ denotes sending data via channel $C_2$, $s_D$ denotes sending data via channel $D$, and $b\in\{0,1\}$.

The send action and receive action of the same data through the same channel can communicate each other, otherwise, a deadlock $\delta$ will be caused. We define the following communication functions.

\begin{eqnarray}
&&\gamma(s_B(d',b),r_B(d',b))\triangleq c_B(d',b)\nonumber\\
&&\gamma(s_B(\bot),r_B(\bot))\triangleq c_B(\bot)\nonumber\\
&&\gamma(s_D(b),r_D(b))\triangleq c_D(b)\nonumber\\
&&\gamma(s_D(\bot),r_D(\bot))\triangleq c_D(\bot)\nonumber
\end{eqnarray}

Let $R_0$ and $S_0$ be in parallel, then the system $R_0S_0$ can be represented by the following process term.

$$\tau_I(\partial_H(\Theta(R_0\between S_0)))=\tau_I(\partial_H(R_0\between S_0))$$

where $H=\{s_B(d',b),r_B(d',b),s_D(b),r_D(b)|d'\in\Delta,b\in\{0,1\}\}\\
\{s_B(\bot),r_B(\bot),s_D(\bot),r_D(\bot)\}$

$I=\{c_B(d',b),c_D(b)|d'\in\Delta,b\in\{0,1\}\}\cup\{c_B(\bot),c_D(\bot)\}$.

Then we get the following conclusion.

\begin{theorem}[Correctness of the ABP protocol]
The ABP protocol $\tau_I(\partial_H(R_0\between S_0))$ can exhibit desired external behaviors.
\end{theorem}

\begin{proof}

By use of the algebraic laws of $APTC$, we have the following expansions.

\begin{eqnarray}
R_0\between S_0&\overset{\text{P1}}{=}&R_0\parallel S_0+R_0\mid S_0\nonumber\\
&\overset{\text{RDP}}{=}&(\sum_{d\in\Delta}r_{A_2}(d)\cdot R_0')\parallel(\sum_{d\in\Delta}r_{A_1}(d)T_{d0})\nonumber\\
&&+(\sum_{d\in\Delta}r_{A_2}(d)\cdot R_0')\mid(\sum_{d\in\Delta}r_{A_1}(d)T_{d0})\nonumber\\
&\overset{\text{P6,C14}}{=}&\sum_{d\in\Delta}(r_{A_2}(d)\parallel r_{A_1}(d))R_0'\between T_{d0} + \delta\cdot R_0'\between T_{d0}\nonumber\\
&\overset{\text{A6,A7}}{=}&\sum_{d\in\Delta}(r_{A_2}(d)\parallel r_{A_1}(d))R_0'\between T_{d0}\nonumber
\end{eqnarray}

\begin{eqnarray}
\partial_H(R_0\between S_0)&=&\partial_H(\sum_{d\in\Delta}(r_{A_2}(d)\parallel r_{A_1}(d))R_0'\between T_{d0})\nonumber\\
&&=\sum_{d\in\Delta}(r_{A_2}(d)\parallel r_{A_1}(d))\partial_H(R_0'\between T_{d0})\nonumber
\end{eqnarray}

Similarly, we can get the following equations.

\begin{eqnarray}
\partial_H(R_0\between S_0)&=&\sum_{d\in\Delta}(r_{A_2}(d)\parallel r_{A_1}(d))\cdot\partial_H(T_{d0}\between R_0')\nonumber\\
\partial_H(T_{d0}\between R_0')&=&c_B(d',0)\cdot(s_{C_1}(d')\parallel s_{C_2}(d'))\cdot\partial_H(U_{d0}\between Q_0)+c_B(\bot)\cdot\partial_H(U_{d0}\between Q_1)\nonumber\\
\partial_H(U_{d0}\between Q_1)&=&(c_D(1)+c_D(\bot))\cdot\partial_H(T_{d0}\between R_0')\nonumber\\
\partial_H(Q_0\between U_{d0})&=&c_D(0)\cdot\partial_H(R_1\between S_1)+c_D(\bot)\cdot\partial_H(R_1'\between T_{d0})\nonumber\\
\partial_H(R_1'\between T_{d0})&=&(c_B(d',0)+c_B(\bot))\cdot\partial_H(Q_0\between U_{d0})\nonumber\\
\partial_H(R_1\between S_1)&=&\sum_{d\in\Delta}(r_{A_2}(d)\parallel r_{A_1}(d))\cdot\partial_H(T_{d1}\between R_1')\nonumber\\
\partial_H(T_{d1}\between R_1')&=&c_B(d',1)\cdot(s_{C_1}(d')\parallel s_{C_2}(d'))\cdot\partial_H(U_{d1}\between Q_1)+c_B(\bot)\cdot\partial_H(U_{d1}\between Q_0')\nonumber\\
\partial_H(U_{d1}\between Q_0')&=&(c_D(0)+c_D(\bot))\cdot\partial_H(T_{d1}\between R_1')\nonumber\\
\partial_H(Q_1\between U_{d1})&=&c_D(1)\cdot\partial_H(R_0\between S_0)+c_D(\bot)\cdot\partial_H(R_0'\between T_{d1})\nonumber\\
\partial_H(R_0'\between T_{d1})&=&(c_B(d',1)+c_B(\bot))\cdot\partial_H(Q_1\between U_{d1})\nonumber
\end{eqnarray}

Let $\partial_H(R_0\between S_0)=\langle X_1|E\rangle$, where E is the following guarded linear recursion specification:

\begin{eqnarray}
&&\{X_1=\sum_{d\in \Delta}(r_{A_2}(d)\parallel r_{A_1}(d))\cdot X_{2d},Y_1=\sum_{d\in\Delta}(r_{A_2}(d)\parallel r_{A_1}(d))\cdot Y_{2d},\nonumber\\
&&X_{2d}=c_B(d',0)\cdot X_{4d}+c_B(\bot)\cdot X_{3d}, Y_{2d}=c_B(d',1)\cdot Y_{4d}+c_B(\bot)\cdot Y_{3d},\nonumber\\
&&X_{3d}=(c_D(1)+c_D(\bot))\cdot X_{2d}, Y_{3d}=(c_D(0)+c_D(\bot))\cdot Y_{2d},\nonumber\\
&&X_{4d}=(s_{C_1}(d')\parallel s_{C_2}(d'))\cdot X_{5d}, Y_{4d}=(s_{C_1}(d')\parallel s_{C_2}(d'))\cdot Y_{5d},\nonumber\\
&&X_{5d}=c_D(0)\cdot Y_1+c_D(\bot)\cdot X_{6d}, Y_{5d}=c_D(1)\cdot X_1+c_D(\bot)\cdot Y_{6d},\nonumber\\
&&X_{6d}=(c_B(d,0)+c_B(\bot))\cdot X_{5d}, Y_{6d}=(c_B(d,1)+c_B(\bot))\cdot Y_{5d}\nonumber\\
&&|d,d'\in\Delta\}\nonumber
\end{eqnarray}

Then we apply abstraction operator $\tau_I$ into $\langle X_1|E\rangle$.

\begin{eqnarray}
\tau_I(\langle X_1|E\rangle)
&=&\sum_{d\in\Delta}(r_{A_1}(d)\parallel r_{A_2}(d))\cdot\tau_I(\langle X_{2d}|E\rangle)\nonumber\\
&=&\sum_{d\in\Delta}(r_{A_1}(d)\parallel r_{A_2}(d))\cdot\tau_I(\langle X_{4d}|E\rangle)\nonumber\\
&=&\sum_{d,d'\in\Delta}(r_{A_1}(d)\parallel r_{A_2}(d))\cdot (s_{C_1}(d')\parallel s_{C_2}(d'))\cdot\tau_I(\langle X_{5d}|E\rangle)\nonumber\\
&=&\sum_{d,d'\in\Delta}(r_{A_1}(d)\parallel r_{A_2}(d))\cdot (s_{C_1}(d')\parallel s_{C_2}(d'))\cdot\tau_I(\langle Y_1|E\rangle)\nonumber
\end{eqnarray}

Similarly, we can get $\tau_I(\langle Y_1|E\rangle)=\sum_{d,d'\in\Delta}(r_{A_1}(d)\parallel r_{A_2}(d))\cdot (s_{C_1}(d')\parallel s_{C_2}(d'))\cdot\tau_I(\langle X_1|E\rangle)
$.

We get $\tau_I(\partial_H(R_0\between S_0))=\sum_{d,d'\in \Delta}(r_{A_1}(d)\parallel r_{A_2}(d))\cdot (s_{C_1}(d')\parallel s_{C_2}(d'))\cdot \tau_I(\partial_H(R_0\between S_0))$. So, the ABP protocol $\tau_I(\partial_H(R_0\between S_0))$ can exhibit desired external behaviors.
\end{proof}

With the help of shadow constant, now we can verify the traditional alternating bit protocol (ABP) \cite{ABP}.

The ABP protocol is used to ensure successful transmission of data through a corrupted channel. This success is based on the assumption that data can be resent an unlimited number of times, which is illustrated in Figure \ref{ABP2}, we alter it into the true concurrency situation.

\begin{enumerate}
  \item Data elements $d_1,d_2,d_3,\cdots$ from a finite set $\Delta$ are communicated between a Sender and a Receiver.
  \item If the Sender reads a datum from channel $A$.
  \item The Sender processes the data in $\Delta$, formes new data, and sends them to the Receiver through channel $B$.
  \item And the Receiver sends the datum into channel $C$.
  \item If channel $B$ is corrupted, the message communicated through $B$ can be turn into an error message $\bot$.
  \item Every time the Receiver receives a message via channel $B$, it sends an acknowledgement to the Sender via channel $D$, which is also corrupted.
\end{enumerate}

\begin{figure}
    \centering
    \includegraphics{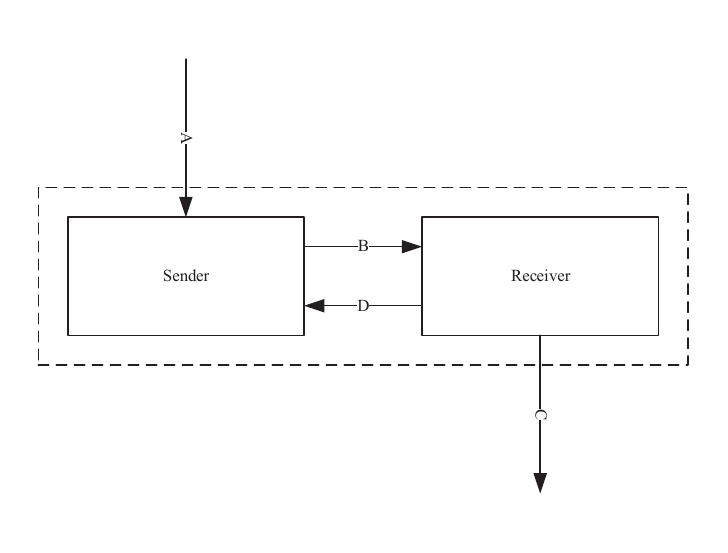}
    \caption{Alternating bit protocol}
    \label{ABP2}
\end{figure}

The Sender attaches a bit 0 to data elements $d_{2k-1}$ and a bit 1 to data elements $d_{2k}$, when they are sent into channel $B$. When the Receiver reads a datum, it sends back the attached bit via channel $D$. If the Receiver receives a corrupted message, then it sends back the previous acknowledgement to the Sender.

Then the state transition of the Sender can be described by $APTC$ as follows.

\begin{eqnarray}
&&S_b=\sum_{d\in\Delta}r_{A}(d)\cdot T_{db}\nonumber\\
&&T_{db}=(\sum_{d'\in\Delta}(s_B(d',b)\cdot \circledS^{s_{C}(d')})+s_B(\bot))\cdot U_{db}\nonumber\\
&&U_{db}=r_D(b)\cdot S_{1-b}+(r_D(1-b)+r_D(\bot))\cdot T_{db}\nonumber
\end{eqnarray}

where $s_B$ denotes sending data through channel $B$, $r_D$ denotes receiving data through channel $D$, similarly, $r_{A}$ means receiving data via channel $A$, $\circledS^{s_{C}(d')}$ denotes the shadow of $s_{C}(d')$.

And the state transition of the Receiver can be described by $APTC$ as follows.

\begin{eqnarray}
&&R_b=\sum_{d\in\Delta}\circledS^{r_{A}(d)}\cdot R_b'\nonumber\\
&&R_b'=\sum_{d'\in\Delta}\{r_B(d',b)\cdot s_{C}(d')\cdot Q_b+r_B(d',1-b)\cdot Q_{1-b}\}+r_B(\bot)\cdot Q_{1-b}\nonumber\\
&&Q_b=(s_D(b)+s_D(\bot))\cdot R_{1-b}\nonumber
\end{eqnarray}

where $\circledS^{r_{A}(d)}$ denotes the shadow of $r_{A}(d)$, $r_B$ denotes receiving data via channel $B$, $s_{C}$ denotes sending data via channel $C$, $s_D$ denotes sending data via channel $D$, and $b\in\{0,1\}$.

The send action and receive action of the same data through the same channel can communicate each other, otherwise, a deadlock $\delta$ will be caused. We define the following communication functions.

\begin{eqnarray}
&&\gamma(s_B(d',b),r_B(d',b))\triangleq c_B(d',b)\nonumber\\
&&\gamma(s_B(\bot),r_B(\bot))\triangleq c_B(\bot)\nonumber\\
&&\gamma(s_D(b),r_D(b))\triangleq c_D(b)\nonumber\\
&&\gamma(s_D(\bot),r_D(\bot))\triangleq c_D(\bot)\nonumber
\end{eqnarray}

Let $R_0$ and $S_0$ be in parallel, then the system $R_0S_0$ can be represented by the following process term.

$$\tau_I(\partial_H(\Theta(R_0\between S_0)))=\tau_I(\partial_H(R_0\between S_0))$$

where $H=\{s_B(d',b),r_B(d',b),s_D(b),r_D(b)|d'\in\Delta,b\in\{0,1\}\}\\
\{s_B(\bot),r_B(\bot),s_D(\bot),r_D(\bot)\}$

$I=\{c_B(d',b),c_D(b)|d'\in\Delta,b\in\{0,1\}\}\cup\{c_B(\bot),c_D(\bot)\}$.

Then we get the following conclusion.

\begin{theorem}[Correctness of the ABP protocol]
The ABP protocol $\tau_I(\partial_H(R_0\between S_0))$ can exhibit desired external behaviors.
\end{theorem}

\begin{proof}

Similarly, we can get $\tau_I(\langle X_1|E\rangle)=\sum_{d,d'\in\Delta}r_{A}(d)\cdot s_{C}(d')\cdot\tau_I(\langle Y_1|E\rangle)$ and $\tau_I(\langle Y_1|E\rangle)=\sum_{d,d'\in\Delta}r_{A}(d)\cdot s_{C}(d')\cdot\tau_I(\langle X_1|E\rangle)$.

So, the ABP protocol $\tau_I(\partial_H(R_0\between S_0))$ can exhibit desired external behaviors.
\end{proof}

\newpage\section{Data Manipulation in APTC}\label{dmaptc}

This chapter is organized as follows. We introduce the operational semantics of guards in section \ref{os2}, $BATC$ with Guards in section \ref{batcg}, $APTC$ with Guards \ref{aptcg}, recursion in section \ref{recg}, abstraction in section \ref{absg}.

\subsection{Operational Semantics}{\label{os2}}

In this section, we extend truly concurrent bisimilarities to the ones containing data states.

\begin{definition}[Prime event structure with silent event and empty event]\label{PESG}
Let $\Lambda$ be a fixed set of labels, ranged over $a,b,c,\cdots$ and $\tau,\epsilon$. A ($\Lambda$-labelled) prime event structure with silent event $\tau$ and empty event $\epsilon$ is a tuple $\mathcal{E}=\langle \mathbb{E}, \leq, \sharp, \lambda\rangle$, where $\mathbb{E}$ is a denumerable set of events, including the silent event $\tau$ and empty event $\epsilon$. Let $\hat{\mathbb{E}}=\mathbb{E}\backslash\{\tau,\epsilon\}$, exactly excluding $\tau$ and $\epsilon$, it is obvious that $\hat{\tau^*}=\epsilon$. Let $\lambda:\mathbb{E}\rightarrow\Lambda$ be a labelling function and let $\lambda(\tau)=\tau$ and $\lambda(\epsilon)=\epsilon$. And $\leq$, $\sharp$ are binary relations on $\mathbb{E}$, called causality and conflict respectively, such that:

\begin{enumerate}
  \item $\leq$ is a partial order and $\lceil e \rceil = \{e'\in \mathbb{E}|e'\leq e\}$ is finite for all $e\in \mathbb{E}$. It is easy to see that $e\leq\tau^*\leq e'=e\leq\tau\leq\cdots\leq\tau\leq e'$, then $e\leq e'$.
  \item $\sharp$ is irreflexive, symmetric and hereditary with respect to $\leq$, that is, for all $e,e',e''\in \mathbb{E}$, if $e\sharp e'\leq e''$, then $e\sharp e''$.
\end{enumerate}

Then, the concepts of consistency and concurrency can be drawn from the above definition:

\begin{enumerate}
  \item $e,e'\in \mathbb{E}$ are consistent, denoted as $e\frown e'$, if $\neg(e\sharp e')$. A subset $X\subseteq \mathbb{E}$ is called consistent, if $e\frown e'$ for all $e,e'\in X$.
  \item $e,e'\in \mathbb{E}$ are concurrent, denoted as $e\parallel e'$, if $\neg(e\leq e')$, $\neg(e'\leq e)$, and $\neg(e\sharp e')$.
\end{enumerate}
\end{definition}

\begin{definition}[Configuration]
Let $\mathcal{E}$ be a PES. A (finite) configuration in $\mathcal{E}$ is a (finite) consistent subset of events $C\subseteq \mathcal{E}$, closed with respect to causality (i.e. $\lceil C\rceil=C$), and a data state $s\in S$ with $S$ the set of all data states, denoted $\langle C, s\rangle$. The set of finite configurations of $\mathcal{E}$ is denoted by $\langle\mathcal{C}(\mathcal{E}), S\rangle$. We let $\hat{C}=C\backslash\{\tau\}\cup\{\epsilon\}$.
\end{definition}

A consistent subset of $X\subseteq \mathbb{E}$ of events can be seen as a pomset. Given $X, Y\subseteq \mathbb{E}$, $\hat{X}\sim \hat{Y}$ if $\hat{X}$ and $\hat{Y}$ are isomorphic as pomsets. In the following of the paper, we say $C_1\sim C_2$, we mean $\hat{C_1}\sim\hat{C_2}$.

\begin{definition}[Pomset transitions and step]
Let $\mathcal{E}$ be a PES and let $C\in\mathcal{C}(\mathcal{E})$, and $\emptyset\neq X\subseteq \mathbb{E}$, if $C\cap X=\emptyset$ and $C'=C\cup X\in\mathcal{C}(\mathcal{E})$, then $\langle C,s\rangle\xrightarrow{X} \langle C',s'\rangle$ is called a pomset transition from $\langle C,s\rangle$ to $\langle C',s'\rangle$. When the events in $X$ are pairwise concurrent, we say that $\langle C,s\rangle\xrightarrow{X}\langle C',s'\rangle$ is a step. It is obvious that $\rightarrow^*\xrightarrow{X}\rightarrow^*=\xrightarrow{X}$ and $\rightarrow^*\xrightarrow{e}\rightarrow^*=\xrightarrow{e}$ for any $e\in\mathbb{E}$ and $X\subseteq\mathbb{E}$.
\end{definition}

\begin{definition}[Weak pomset transitions and weak step]
Let $\mathcal{E}$ be a PES and let $C\in\mathcal{C}(\mathcal{E})$, and $\emptyset\neq X\subseteq \hat{\mathbb{E}}$, if $C\cap X=\emptyset$ and $\hat{C'}=\hat{C}\cup X\in\mathcal{C}(\mathcal{E})$, then $\langle C,s\rangle\xRightarrow{X} \langle C',s'\rangle$ is called a weak pomset transition from $\langle C,s\rangle$ to $\langle C',s'\rangle$, where we define $\xRightarrow{e}\triangleq\xrightarrow{\tau^*}\xrightarrow{e}\xrightarrow{\tau^*}$. And $\xRightarrow{X}\triangleq\xrightarrow{\tau^*}\xrightarrow{e}\xrightarrow{\tau^*}$, for every $e\in X$. When the events in $X$ are pairwise concurrent, we say that $\langle C,s\rangle\xRightarrow{X}\langle C',s'\rangle$ is a weak step.
\end{definition}

We will also suppose that all the PESs in this paper are image finite, that is, for any PES $\mathcal{E}$ and $C\in \mathcal{C}(\mathcal{E})$ and $a\in \Lambda$, $\{e\in \mathbb{E}|\langle C,s\rangle\xrightarrow{e} \langle C',s'\rangle\wedge \lambda(e)=a\}$ and $\{e\in\hat{\mathbb{E}}|\langle C,s\rangle\xRightarrow{e} \langle C',s'\rangle\wedge \lambda(e)=a\}$ is finite.

\begin{definition}[Pomset, step bisimulation]\label{PSBG}
Let $\mathcal{E}_1$, $\mathcal{E}_2$ be PESs. A pomset bisimulation is a relation $R\subseteq\langle\mathcal{C}(\mathcal{E}_1),S\rangle\times\langle\mathcal{C}(\mathcal{E}_2),S\rangle$, such that if $(\langle C_1,s\rangle,\langle C_2,s\rangle)\in R$, and $\langle C_1,s\rangle\xrightarrow{X_1}\langle C_1',s'\rangle$ then $\langle C_2,s\rangle\xrightarrow{X_2}\langle C_2',s'\rangle$, with $X_1\subseteq \mathbb{E}_1$, $X_2\subseteq \mathbb{E}_2$, $X_1\sim X_2$ and $(\langle C_1',s'\rangle,\langle C_2',s'\rangle)\in R$ for all $s,s'\in S$, and vice-versa. We say that $\mathcal{E}_1$, $\mathcal{E}_2$ are pomset bisimilar, written $\mathcal{E}_1\sim_p\mathcal{E}_2$, if there exists a pomset bisimulation $R$, such that $(\langle\emptyset,\emptyset\rangle,\langle\emptyset,\emptyset\rangle)\in R$. By replacing pomset transitions with steps, we can get the definition of step bisimulation. When PESs $\mathcal{E}_1$ and $\mathcal{E}_2$ are step bisimilar, we write $\mathcal{E}_1\sim_s\mathcal{E}_2$.
\end{definition}

\begin{definition}[Weak pomset, step bisimulation]\label{WPSBG}
Let $\mathcal{E}_1$, $\mathcal{E}_2$ be PESs. A weak pomset bisimulation is a relation $R\subseteq\langle\mathcal{C}(\mathcal{E}_1),S\rangle\times\langle\mathcal{C}(\mathcal{E}_2),S\rangle$, such that if $(\langle C_1,s\rangle,\langle C_2,s\rangle)\in R$, and $\langle C_1,s\rangle\xRightarrow{X_1}\langle C_1',s'\rangle$ then $\langle C_2,s\rangle\xRightarrow{X_2}\langle C_2',s'\rangle$, with $X_1\subseteq \hat{\mathbb{E}_1}$, $X_2\subseteq \hat{\mathbb{E}_2}$, $X_1\sim X_2$ and $(\langle C_1',s'\rangle,\langle C_2',s'\rangle)\in R$ for all $s,s'\in S$, and vice-versa. We say that $\mathcal{E}_1$, $\mathcal{E}_2$ are weak pomset bisimilar, written $\mathcal{E}_1\approx_p\mathcal{E}_2$, if there exists a weak pomset bisimulation $R$, such that $(\langle\emptyset,\emptyset\rangle,\langle\emptyset,\emptyset\rangle)\in R$. By replacing weak pomset transitions with weak steps, we can get the definition of weak step bisimulation. When PESs $\mathcal{E}_1$ and $\mathcal{E}_2$ are weak step bisimilar, we write $\mathcal{E}_1\approx_s\mathcal{E}_2$.
\end{definition}

\begin{definition}[Posetal product]
Given two PESs $\mathcal{E}_1$, $\mathcal{E}_2$, the posetal product of their configurations, denoted $\langle\mathcal{C}(\mathcal{E}_1),S\rangle\overline{\times}\langle\mathcal{C}(\mathcal{E}_2),S\rangle$, is defined as

$$\{(\langle C_1,s\rangle,f,\langle C_2,s\rangle)|C_1\in\mathcal{C}(\mathcal{E}_1),C_2\in\mathcal{C}(\mathcal{E}_2),f:C_1\rightarrow C_2 \textrm{ isomorphism}\}.$$

A subset $R\subseteq\langle\mathcal{C}(\mathcal{E}_1),S\rangle\overline{\times}\langle\mathcal{C}(\mathcal{E}_2),S\rangle$ is called a posetal relation. We say that $R$ is downward closed when for any $(\langle C_1,s\rangle,f,\langle C_2,s\rangle),(\langle C_1',s'\rangle,f',\langle C_2',s'\rangle)\in \langle\mathcal{C}(\mathcal{E}_1),S\rangle\overline{\times}\langle\mathcal{C}(\mathcal{E}_2),S\rangle$, if $(\langle C_1,s\rangle,f,\langle C_2,s\rangle)\subseteq (\langle C_1',s'\rangle,f',\langle C_2',s'\rangle)$ pointwise and $(\langle C_1',s'\rangle,f',\langle C_2',s'\rangle)\in R$, then $(\langle C_1,s\rangle,f,\langle C_2,s\rangle)\in R$.

For $f:X_1\rightarrow X_2$, we define $f[x_1\mapsto x_2]:X_1\cup\{x_1\}\rightarrow X_2\cup\{x_2\}$, $z\in X_1\cup\{x_1\}$,(1)$f[x_1\mapsto x_2](z)=
x_2$,if $z=x_1$;(2)$f[x_1\mapsto x_2](z)=f(z)$, otherwise. Where $X_1\subseteq \mathbb{E}_1$, $X_2\subseteq \mathbb{E}_2$, $x_1\in \mathbb{E}_1$, $x_2\in \mathbb{E}_2$.
\end{definition}

\begin{definition}[Weakly posetal product]
Given two PESs $\mathcal{E}_1$, $\mathcal{E}_2$, the weakly posetal product of their configurations, denoted $\langle\mathcal{C}(\mathcal{E}_1),S\rangle\overline{\times}\langle\mathcal{C}(\mathcal{E}_2),S\rangle$, is defined as

$$\{(\langle C_1,s\rangle,f,\langle C_2,s\rangle)|C_1\in\mathcal{C}(\mathcal{E}_1),C_2\in\mathcal{C}(\mathcal{E}_2),f:\hat{C_1}\rightarrow \hat{C_2} \textrm{ isomorphism}\}.$$

A subset $R\subseteq\langle\mathcal{C}(\mathcal{E}_1),S\rangle\overline{\times}\langle\mathcal{C}(\mathcal{E}_2),S\rangle$ is called a weakly posetal relation. We say that $R$ is downward closed when for any $(\langle C_1,s\rangle,f,\langle C_2,s\rangle),(\langle C_1',s'\rangle,f,\langle C_2',s'\rangle)\in \langle\mathcal{C}(\mathcal{E}_1),S\rangle\overline{\times}\langle\mathcal{C}(\mathcal{E}_2),S\rangle$, if $(\langle C_1,s\rangle,f,\langle C_2,s\rangle)\subseteq (\langle C_1',s'\rangle,f',\langle C_2',s'\rangle)$ pointwise and $(\langle C_1',s'\rangle,f',\langle C_2',s'\rangle)\in R$, then $(\langle C_1,s\rangle,f,\langle C_2,s\rangle)\in R$.

For $f:X_1\rightarrow X_2$, we define $f[x_1\mapsto x_2]:X_1\cup\{x_1\}\rightarrow X_2\cup\{x_2\}$, $z\in X_1\cup\{x_1\}$,(1)$f[x_1\mapsto x_2](z)=
x_2$,if $z=x_1$;(2)$f[x_1\mapsto x_2](z)=f(z)$, otherwise. Where $X_1\subseteq \hat{\mathbb{E}_1}$, $X_2\subseteq \hat{\mathbb{E}_2}$, $x_1\in \hat{\mathbb{E}}_1$, $x_2\in \hat{\mathbb{E}}_2$. Also, we define $f(\tau^*)=f(\tau^*)$.
\end{definition}

\begin{definition}[(Hereditary) history-preserving bisimulation]\label{HHPBG}
A history-preserving (hp-) bisimulation is a posetal relation $R\subseteq\langle\mathcal{C}(\mathcal{E}_1),S\rangle\overline{\times}\langle\mathcal{C}(\mathcal{E}_2),S\rangle$ such that if $(\langle C_1,s\rangle,f,\langle C_2,s\rangle)\in R$, and $\langle C_1,s\rangle\xrightarrow{e_1} \langle C_1',s'\rangle$, then $\langle C_2,s\rangle\xrightarrow{e_2} \langle C_2',s'\rangle$, with $(\langle C_1',s'\rangle,f[e_1\mapsto e_2],\langle C_2',s'\rangle)\in R$ for all $s,s'\in S$, and vice-versa. $\mathcal{E}_1,\mathcal{E}_2$ are history-preserving (hp-)bisimilar and are written $\mathcal{E}_1\sim_{hp}\mathcal{E}_2$ if there exists a hp-bisimulation $R$ such that $(\langle\emptyset,\emptyset\rangle,\emptyset,\langle\emptyset,\emptyset\rangle)\in R$.

A hereditary history-preserving (hhp-)bisimulation is a downward closed hp-bisimulation. $\mathcal{E}_1,\mathcal{E}_2$ are hereditary history-preserving (hhp-)bisimilar and are written $\mathcal{E}_1\sim_{hhp}\mathcal{E}_2$.
\end{definition}

\begin{definition}[Weak (hereditary) history-preserving bisimulation]\label{WHHPBG}
A weak history-preserving (hp-) bisimulation is a weakly posetal relation $R\subseteq\langle\mathcal{C}(\mathcal{E}_1),S\rangle\overline{\times}\langle\mathcal{C}(\mathcal{E}_2),S\rangle$ such that if $(\langle C_1,s\rangle,f,\langle C_2,s\rangle)\in R$, and $\langle C_1,s\rangle\xRightarrow{e_1} \langle C_1',s'\rangle$, then $\langle C_2,s\rangle\xRightarrow{e_2} \langle C_2',s'\rangle$, with $(\langle C_1',s'\rangle,f[e_1\mapsto e_2],\langle C_2',s'\rangle)\in R$ for all $s,s'\in S$, and vice-versa. $\mathcal{E}_1,\mathcal{E}_2$ are weak history-preserving (hp-)bisimilar and are written $\mathcal{E}_1\approx_{hp}\mathcal{E}_2$ if there exists a weak hp-bisimulation $R$ such that $(\langle\emptyset,\emptyset\rangle,\emptyset,\langle\emptyset,\emptyset\rangle)\in R$.

A weakly hereditary history-preserving (hhp-)bisimulation is a downward closed weak hp-bisimulation. $\mathcal{E}_1,\mathcal{E}_2$ are weakly hereditary history-preserving (hhp-)bisimilar and are written $\mathcal{E}_1\approx_{hhp}\mathcal{E}_2$.
\end{definition}

\subsection{$BATC$ with Guards}{\label{batcg}}

In this subsection, we will discuss the guards for $BATC$, which is denoted as $BATC_G$. Let $\mathbb{E}$ be the set of atomic events (actions), and we assume that there is a data set $\Delta$ and data $D_1,\cdots,D_n\in\Delta$, the data variable $d_1,\cdots,d_n$ range over $\Delta$, and $d_i$ has the same data type as $D_i$ and can have a substitution $D_i/d_i$, for process $x$, $x[D_i/d_i]$ denotes that all occurrences of $d_i$ in $x$ are replaced by $D_i$. And also the atomic action $e$ may manipulate on data and has the form $e(d_1,\cdots,d_n)$ or $e(D_1,\cdots,D_n)$. $G_{at}$ be the set of atomic guards, $\delta$ be the deadlock constant, and $\epsilon$ be the empty event. We extend $G_{at}$ to the set of basic guards $G$ with element $\phi,\psi,\cdots$, which is generated by the following formation rules:

$$\phi::=\delta|\epsilon|\neg\phi|\psi\in G_{at}|\phi+\psi|\phi\cdot\psi$$

In the following, let $e_1, e_2, e_1', e_2'\in \mathbb{E}$, $\phi,\psi\in G$ and let variables $x,y,z$ range over the set of terms for true concurrency, $p,q,s$ range over the set of closed terms. The predicate $test(\phi,s)$ represents that $\phi$ holds in the state $s$, and $test(\epsilon,s)$ holds and $test(\delta,s)$ does not hold. $effect(e,s)\in S$ denotes $s'$ in $s\xrightarrow{e}s'$. The predicate weakest precondition $wp(e,\phi)$ denotes that $\forall s\in S, test(\phi,effect(e,s))$ holds.

The set of axioms of $BATC_G$ consists of the laws given in Table \ref{AxiomsForBATCG}.

\begin{center}
    \begin{table}
        \begin{tabular}{@{}ll@{}}
            \hline No. &Axiom\\
            $A1$ & $x+ y = y+ x$\\
            $A2$ & $(x+ y)+ z = x+ (y+ z)$\\
            $A3$ & $x+ x = x$\\
            $A4$ & $(x+ y)\cdot z = x\cdot z + y\cdot z$\\
            $A5$ & $(x\cdot y)\cdot z = x\cdot(y\cdot z)$\\
            $A6$ & $x+\delta = x$\\
            $A7$ & $\delta\cdot x = \delta$\\
            $A8$ & $\epsilon\cdot x = x$\\
            $A9$ & $x\cdot\epsilon = x$\\
            $G1$ & $\phi\cdot\neg\phi = \delta$\\
            $G2$ & $\phi+\neg\phi = \epsilon$\\
            $G3$ & $\phi\delta = \delta$\\
            $G4$ & $\phi(x+y)=\phi x+\phi y$\\
            $G5$ & $\phi(x\cdot y)= \phi x\cdot y$\\
            $G6$ & $(\phi+\psi)x = \phi x + \psi x$\\
            $G7$ & $(\phi\cdot \psi)\cdot x = \phi\cdot(\psi\cdot x)$\\
            $G8$ & $\phi=\epsilon$ if $\forall s\in S.test(\phi,s)$\\
            $G9$ & $\phi_0\cdot\cdots\cdot\phi_n = \delta$ if $\forall s\in S,\exists i\leq n.test(\neg\phi_i,s)$\\
            $G10$ & $wp(e,\phi)e\phi=wp(e,\phi)e$\\
            $G11$ & $\neg wp(e,\phi)e\neg\phi=\neg wp(e,\phi)e$\\
        \end{tabular}
        \caption{Axioms of $BATC_G$}
        \label{AxiomsForBATCG}
    \end{table}
\end{center}

Note that, by eliminating atomic event from the process terms, the axioms in Table \ref{AxiomsForBATCG} will lead to a Boolean Algebra. And $G9$ is a precondition of $e$ and $\phi$, $G10$ is the weakest precondition of $e$ and $\phi$. A data environment with $effect$ function is sufficiently deterministic, and it is obvious that if the weakest precondition is expressible and $G9$, $G10$ are sound, then the related data environment is sufficiently deterministic.

\begin{definition}[Basic terms of $BATC_G$]\label{BTBATCG}
The set of basic terms of $BATC_G$, $\mathcal{B}(BATC_G)$, is inductively defined as follows:
\begin{enumerate}
  \item $\mathbb{E}\subset\mathcal{B}(BATC_G)$;
  \item $G\subset\mathcal{B}(BATC_G)$;
  \item if $e\in \mathbb{E}, t\in\mathcal{B}(BATC_G)$ then $e\cdot t\in\mathcal{B}(BATC_G)$;
  \item if $\phi\in G, t\in\mathcal{B}(BATC_G)$ then $\phi\cdot t\in\mathcal{B}(BATC_G)$;
  \item if $t,s\in\mathcal{B}(BATC_G)$ then $t+ s\in\mathcal{B}(BATC_G)$.
\end{enumerate}
\end{definition}

\begin{theorem}[Elimination theorem of $BATC_G$]\label{ETBATCG}
Let $p$ be a closed $BATC_G$ term. Then there is a basic $BATC_G$ term $q$ such that $BATC_G\vdash p=q$.
\end{theorem}

We will define a term-deduction system which gives the operational semantics of $BATC_G$. We give the operational transition rules for $\epsilon$, atomic guard $\phi\in G_{at}$, atomic event $e\in\mathbb{E}$, operators $\cdot$ and $+$ as Table \ref{SETRForBATCG} shows. And the predicate $\xrightarrow{e}\surd$ represents successful termination after execution of the event $e$.

\begin{center}
    \begin{table}
        $$\frac{}{\langle\epsilon,s\rangle\rightarrow\langle\surd,s\rangle}$$
        $$\frac{}{\langle e,s\rangle\xrightarrow{e}\langle\surd,s'\rangle}\textrm{ if }s'\in effect(e,s)$$
        $$\frac{}{\langle\phi,s\rangle\rightarrow\langle\surd,s\rangle}\textrm{ if }test(\phi,s)$$
        $$\frac{\langle x,s\rangle\xrightarrow{e}\langle\surd,s'\rangle}{\langle x+ y,s\rangle\xrightarrow{e}\langle\surd,s'\rangle} \quad\frac{\langle x,s\rangle\xrightarrow{e}\langle x',s'\rangle}{\langle x+ y,s\rangle\xrightarrow{e}\langle x',s'\rangle}$$
        $$\frac{\langle y,s\rangle\xrightarrow{e}\langle\surd,s'\rangle}{\langle x+ y,s\rangle\xrightarrow{e}\langle\surd,s'\rangle} \quad\frac{\langle y,s\rangle\xrightarrow{e}\langle y',s'\rangle}{\langle x+ y,s\rangle\xrightarrow{e}\langle y',s'\rangle}$$
        $$\frac{\langle x,s\rangle\xrightarrow{e}\langle\surd,s'\rangle}{\langle x\cdot y,s\rangle\xrightarrow{e} \langle y,s'\rangle} \quad\frac{\langle x,s\rangle\xrightarrow{e}\langle x',s'\rangle}{\langle x\cdot y,s\rangle\xrightarrow{e}\langle x'\cdot y,s'\rangle}$$
        \caption{Single event transition rules of $BATC_G$}
        \label{SETRForBATCG}
    \end{table}
\end{center}

Note that, we replace the single atomic event $e\in\mathbb{E}$ by $X\subseteq\mathbb{E}$, we can obtain the pomset transition rules of $BATC_G$, and omit them.

\begin{theorem}[Congruence of $BATC_G$ with respect to truly concurrent bisimulation equivalences]\label{CBATCG}
(1) Pomset bisimulation equivalence $\sim_{p}$ is a congruence with respect to $BATC_G$.

(2) Step bisimulation equivalence $\sim_{s}$ is a congruence with respect to $BATC_G$.

(3) Hp-bisimulation equivalence $\sim_{hp}$ is a congruence with respect to $BATC_G$.

(4) Hhp-bisimulation equivalence $\sim_{hhp}$ is a congruence with respect to $BATC_G$.
\end{theorem}

\begin{theorem}[Soundness of $BATC_G$ modulo truly concurrent bisimulation equivalences]\label{SBATCG}
(1) Let $x$ and $y$ be $BATC_G$ terms. If $BATC\vdash x=y$, then $x\sim_{p} y$.

(2) Let $x$ and $y$ be $BATC_G$ terms. If $BATC\vdash x=y$, then $x\sim_{s} y$.

(3) Let $x$ and $y$ be $BATC_G$ terms. If $BATC\vdash x=y$, then $x\sim_{hp} y$.

(4) Let $x$ and $y$ be $BATC_G$ terms. If $BATC\vdash x=y$, then $x\sim_{hhp} y$.
\end{theorem}

\begin{theorem}[Completeness of $BATC_G$ modulo truly concurrent bisimulation equivalences]\label{CBATCG}
(1) Let $p$ and $q$ be closed $BATC_G$ terms, if $p\sim_{p} q$ then $p=q$.

(2) Let $p$ and $q$ be closed $BATC_G$ terms, if $p\sim_{s} q$ then $p=q$.

(3) Let $p$ and $q$ be closed $BATC_G$ terms, if $p\sim_{hp} q$ then $p=q$.

(4) Let $p$ and $q$ be closed $BATC_G$ terms, if $p\sim_{hhp} q$ then $p=q$.
\end{theorem}

\begin{theorem}[Sufficient determinacy]
All related data environments with respect to $BATC_G$ can be sufficiently deterministic.
\end{theorem}

\subsection{$APTC$ with Guards}{\label{aptcg}}

In this subsection, we will extend $APTC$ with guards, which is abbreviated $APTC_G$. The set of basic guards $G$ with element $\phi,\psi,\cdots$, which is extended by the following formation rules:

$$\phi::=\delta|\epsilon|\neg\phi|\psi\in G_{at}|\phi+\psi|\phi\cdot\psi|\phi\parallel\psi$$

The set of axioms of $APTC_G$ including axioms of $BATC_G$ in Table \ref{AxiomsForBATCG} and the axioms are shown in Table \ref{AxiomsForAPTCG}.

\begin{center}
    \begin{table}
        \begin{tabular}{@{}ll@{}}
            \hline No. &Axiom\\
            $P1$ & $x\between y = x\parallel y + x\mid y$\\
            $P2$ & $e_1\parallel (e_2\cdot y) = (e_1\parallel e_2)\cdot y$\\
            $P3$ & $(e_1\cdot x)\parallel e_2 = (e_1\parallel e_2)\cdot x$\\
            $P4$ & $(e_1\cdot x)\parallel (e_2\cdot y) = (e_1\parallel e_2)\cdot (x\between y)$\\
            $P5$ & $(x+ y)\parallel z = (x\parallel z)+ (y\parallel z)$\\
            $P6$ & $x\parallel (y+ z) = (x\parallel y)+ (x\parallel z)$\\
            $P7$ & $\delta\parallel x = \delta$\\
            $P8$ & $x\parallel \delta = \delta$\\
            $P9$ & $\epsilon\parallel x = x$\\
            $P10$ & $x\parallel \epsilon = x$\\
            $C1$ & $e_1\mid e_2 = \gamma(e_1,e_2)$\\
            $C2$ & $e_1\mid (e_2\cdot y) = \gamma(e_1,e_2)\cdot y$\\
            $C3$ & $(e_1\cdot x)\mid e_2 = \gamma(e_1,e_2)\cdot x$\\
            $C4$ & $(e_1\cdot x)\mid (e_2\cdot y) = \gamma(e_1,e_2)\cdot (x\between y)$\\
            $C5$ & $(x+ y)\mid z = (x\mid z) + (y\mid z)$\\
            $C6$ & $x\mid (y+ z) = (x\mid y)+ (x\mid z)$\\
            $C7$ & $\delta\mid x = \delta$\\
            $C8$ & $x\mid\delta = \delta$\\
            $C9$ & $\epsilon\mid x = \delta$\\
            $C10$ & $x\mid\epsilon = \delta$\\
            $CE1$ & $\Theta(e) = e$\\
            $CE2$ & $\Theta(\delta) = \delta$\\
            $CE3$ & $\Theta(\epsilon) = \epsilon$\\
            $CE4$ & $\Theta(x+ y) = \Theta(x)\triangleleft y + \Theta(y)\triangleleft x$\\
            $CE5$ & $\Theta(x\cdot y)=\Theta(x)\cdot\Theta(y)$\\
            $CE6$ & $\Theta(x\parallel y) = ((\Theta(x)\triangleleft y)\parallel y)+ ((\Theta(y)\triangleleft x)\parallel x)$\\
            $CE7$ & $\Theta(x\mid y) = ((\Theta(x)\triangleleft y)\mid y)+ ((\Theta(y)\triangleleft x)\mid x)$\\
            $U1$ & $(\sharp(e_1,e_2))\quad e_1\triangleleft e_2 = \tau$\\
            $U2$ & $(\sharp(e_1,e_2),e_2\leq e_3)\quad e_1\triangleleft e_3 = e_1$\\
            $U3$ & $(\sharp(e_1,e_2),e_2\leq e_3)\quad e3\triangleleft e_1 = \tau$\\
            $U4$ & $e\triangleleft \delta = e$\\
            $U5$ & $\delta \triangleleft e = \delta$\\
            $U6$ & $e\triangleleft \epsilon = e$\\
            $U7$ & $\epsilon \triangleleft e = e$\\
            $U8$ & $(x+ y)\triangleleft z = (x\triangleleft z)+ (y\triangleleft z)$\\
            $U9$ & $(x\cdot y)\triangleleft z = (x\triangleleft z)\cdot (y\triangleleft z)$\\
            $U10$ & $(x\parallel y)\triangleleft z = (x\triangleleft z)\parallel (y\triangleleft z)$\\
            $U11$ & $(x\mid y)\triangleleft z = (x\triangleleft z)\mid (y\triangleleft z)$\\
            $U12$ & $x\triangleleft (y+ z) = (x\triangleleft y)\triangleleft z$\\
            $U13$ & $x\triangleleft (y\cdot z)=(x\triangleleft y)\triangleleft z$\\
            $U14$ & $x\triangleleft (y\parallel z) = (x\triangleleft y)\triangleleft z$\\
            $U15$ & $x\triangleleft (y\mid z) = (x\triangleleft y)\triangleleft z$\\
        \end{tabular}
        \caption{Axioms of $APTC_G$}
        \label{AxiomsForAPTCG}
    \end{table}
\end{center}

\begin{center}
    \begin{table}
        \begin{tabular}{@{}ll@{}}
            \hline No. &Axiom\\
            $D1$ & $e\notin H\quad\partial_H(e) = e$\\
            $D2$ & $e\in H\quad \partial_H(e) = \delta$\\
            $D3$ & $\partial_H(\delta) = \delta$\\
            $D4$ & $\partial_H(x+ y) = \partial_H(x)+\partial_H(y)$\\
            $D5$ & $\partial_H(x\cdot y) = \partial_H(x)\cdot\partial_H(y)$\\
            $D6$ & $\partial_H(x\parallel y) = \partial_H(x)\parallel\partial_H(y)$\\
            $G12$ & $\phi(x\parallel y) =\phi x\parallel \phi y$\\
            $G13$ & $\phi(x\mid y) =\phi x\mid \phi y$\\
            $G14$ & $\phi\parallel \delta = \delta$\\
            $G15$ & $\delta\parallel \phi = \delta$\\
            $G16$ & $\phi\mid \delta = \delta$\\
            $G17$ & $\delta\mid \phi = \delta$\\
            $G18$ & $\phi\parallel \epsilon = \phi$\\
            $G19$ & $\epsilon\parallel \phi = \phi$\\
            $G20$ & $\phi\mid \epsilon = \delta$\\
            $G21$ & $\epsilon\mid \phi = \delta$\\
            $G22$ & $\phi\parallel\neg\phi = \delta$\\
            $G23$ & $\Theta(\phi) = \phi$\\
            $G24$ & $\partial_H(\phi) = \phi$\\
            $G25$ & $\phi_0\parallel\cdots\parallel\phi_n = \delta$ if $\forall s_0,\cdots,s_n\in S,\exists i\leq n.test(\neg\phi_i,s_0\cup\cdots\cup s_n)$\\
        \end{tabular}
        \caption{Axioms of $APTC_G$(continuing)}
        \label{AxiomsForAPTCG2}
    \end{table}
\end{center}

\begin{definition}[Basic terms of $APTC_G$]\label{BTAPTCG}
The set of basic terms of $APTC_G$, $\mathcal{B}(APTC_G)$, is inductively defined as follows:
\begin{enumerate}
    \item $\mathbb{E}\subset\mathcal{B}(APTC_G)$;
    \item $G\subset\mathcal{B}(APTC_G)$;
    \item if $e\in \mathbb{E}, t\in\mathcal{B}(APTC_G)$ then $e\cdot t\in\mathcal{B}(APTC_G)$;
    \item if $\phi\in G, t\in\mathcal{B}(APTC_G)$ then $\phi\cdot t\in\mathcal{B}(APTC_G)$;
    \item if $t,s\in\mathcal{B}(APTC_G)$ then $t+ s\in\mathcal{B}(APTC_G)$.
    \item if $t,s\in\mathcal{B}(APTC_G)$ then $t\parallel s\in\mathcal{B}(APTC_G)$.
\end{enumerate}
\end{definition}

Based on the definition of basic terms for $APTC_G$ (see Definition \ref{BTAPTCG}) and axioms of $APTC_G$, we can prove the elimination theorem of $APTC_G$.

\begin{theorem}[Elimination theorem of $APTC_G$]\label{ETAPTCG}
Let $p$ be a closed $APTC_G$ term. Then there is a basic $APTC_G$ term $q$ such that $APTC_G\vdash p=q$.
\end{theorem}

We will define a term-deduction system which gives the operational semantics of $APTC_G$. Two atomic events $e_1$ and $e_2$ are in race condition, which are denoted $e_1\% e_2$.

\begin{center}
    \begin{table}
        $$\frac{}{\langle e_1\parallel\cdots \parallel e_n,s\rangle\xrightarrow{\{e_1,\cdots,e_n\}}\langle\surd,s'\rangle}\textrm{ if }s'\in effect(e_1,s)\cup\cdots\cup effect(e_n,s)$$
        $$\frac{}{\langle\phi_1\parallel\cdots\parallel \phi_n,s\rangle\rightarrow\langle\surd,s\rangle}\textrm{ if }test(\phi_1,s),\cdots,test(\phi_n,s)$$

        $$\frac{\langle x,s\rangle\xrightarrow{e_1}\langle\surd,s'\rangle\quad \langle y,s\rangle\xrightarrow{e_2}\langle\surd,s''\rangle}{\langle x\parallel y,s\rangle\xrightarrow{\{e_1,e_2\}}\langle\surd,s'\cup s''\rangle} \quad\frac{\langle x,s\rangle\xrightarrow{e_1}\langle x',s'\rangle\quad \langle y,s\rangle\xrightarrow{e_2}\langle\surd,s''\rangle}{\langle x\parallel y,s\rangle\xrightarrow{\{e_1,e_2\}}\langle x',s'\cup s''\rangle}$$

        $$\frac{\langle x,s\rangle\xrightarrow{e_1}\langle\surd,s'\rangle\quad \langle y,s\rangle\xrightarrow{e_2}\langle y',s''\rangle}{\langle x\parallel y,s\rangle\xrightarrow{\{e_1,e_2\}}\langle y',s'\cup s''\rangle} \quad\frac{\langle x,s\rangle\xrightarrow{e_1}\langle x',s'\rangle\quad \langle y,s\rangle\xrightarrow{e_2}\langle y',s''\rangle}{\langle x\parallel y,s\rangle\xrightarrow{\{e_1,e_2\}}\langle x'\between y',s'\cup s''\rangle}$$

        $$\frac{\langle x,s\rangle\xrightarrow{e_1}\langle\surd,s'\rangle\quad \langle y,s\rangle\xnrightarrow{e_2}\quad(e_1\%e_2)}{\langle x\parallel y,s\rangle\xrightarrow{e_1}\langle y,s'\rangle} \quad\frac{\langle x,s\rangle\xrightarrow{e_1}\langle x',s'\rangle\quad \langle y,s\rangle\xnrightarrow{e_2}\quad(e_1\%e_2)}{\langle x\parallel y,s\rangle\xrightarrow{e_1}\langle x'\between y,s'\rangle}$$

        $$\frac{\langle x,s\rangle\xnrightarrow{e_1}\quad \langle y,s\rangle\xrightarrow{e_2}\langle\surd,s''\rangle\quad(e_1\%e_2)}{\langle x\parallel y,s\rangle\xrightarrow{e_2}\langle x,s''\rangle} \quad\frac{\langle x,s\rangle\xnrightarrow{e_1}\quad \langle y,s\rangle\xrightarrow{e_2}\langle y',s''\rangle\quad(e_1\%e_2)}{\langle x\parallel y,s\rangle\xrightarrow{e_2}\langle x\between y',s''\rangle}$$

        $$\frac{\langle x,s\rangle\xrightarrow{e_1}\langle\surd,s'\rangle\quad \langle y,s\rangle\xrightarrow{e_2}\langle\surd,s''\rangle}{\langle x\mid y,s\rangle\xrightarrow{\gamma(e_1,e_2)}\langle\surd,effect(\gamma(e_1,e_2),s)\rangle} \quad\frac{\langle x,s\rangle\xrightarrow{e_1}\langle x',s'\rangle\quad \langle y,s\rangle\xrightarrow{e_2}\langle\surd,s''\rangle}{\langle x\mid y,s\rangle\xrightarrow{\gamma(e_1,e_2)}\langle x',effect(\gamma(e_1,e_2),s)\rangle}$$

        $$\frac{\langle x,s\rangle\xrightarrow{e_1}\langle\surd,s'\rangle\quad \langle y,s\rangle\xrightarrow{e_2}\langle y',s''\rangle}{\langle x\mid y,s\rangle\xrightarrow{\gamma(e_1,e_2)}\langle y',effect(\gamma(e_1,e_2),s)\rangle} \quad\frac{\langle x,s\rangle\xrightarrow{e_1}\langle x',s'\rangle\quad \langle y,s\rangle\xrightarrow{e_2}\langle y',s''\rangle}{\langle x\mid y,s\rangle\xrightarrow{\gamma(e_1,e_2)}\langle x'\between y',effect(\gamma(e_1,e_2),s)\rangle}$$

        $$\frac{\langle x,s\rangle\xrightarrow{e_1}\langle\surd,s'\rangle\quad (\sharp(e_1,e_2))}{\langle \Theta(x),s\rangle\xrightarrow{e_1}\langle\surd,s'\rangle} \quad\frac{\langle x,s\rangle\xrightarrow{e_2}\langle\surd,s''\rangle\quad (\sharp(e_1,e_2))}{\langle\Theta(x),s\rangle\xrightarrow{e_2}\langle\surd,s''\rangle}$$

        $$\frac{\langle x,s\rangle\xrightarrow{e_1}\langle x',s'\rangle\quad (\sharp(e_1,e_2))}{\langle\Theta(x),s\rangle\xrightarrow{e_1}\langle\Theta(x'),s'\rangle} \quad\frac{\langle x,s\rangle\xrightarrow{e_2}\langle x'',s''\rangle\quad (\sharp(e_1,e_2))}{\langle\Theta(x),s\rangle\xrightarrow{e_2}\langle\Theta(x''),s''\rangle}$$

        $$\frac{\langle x,s\rangle\xrightarrow{e_1}\langle\surd,s'\rangle \quad \langle y,s\rangle\nrightarrow^{e_2}\quad (\sharp(e_1,e_2))}{\langle x\triangleleft y,s\rangle\xrightarrow{\tau}\langle\surd,s'\rangle}
        \quad\frac{\langle x,s\rangle\xrightarrow{e_1}\langle x',s'\rangle \quad \langle y,s\rangle\nrightarrow^{e_2}\quad (\sharp(e_1,e_2))}{\langle x\triangleleft y,s\rangle\xrightarrow{\tau}\langle x',s'\rangle}$$

        $$\frac{\langle x,s\rangle\xrightarrow{e_1}\langle\surd,s\rangle \quad \langle y,s\rangle\nrightarrow^{e_3}\quad (\sharp(e_1,e_2),e_2\leq e_3)}{\langle x\triangleleft y,s\rangle\xrightarrow{e_1}\langle\surd,s'\rangle}
        \quad\frac{\langle x,s\rangle\xrightarrow{e_1}\langle x',s'\rangle \quad \langle y,s\rangle\nrightarrow^{e_3}\quad (\sharp(e_1,e_2),e_2\leq e_3)}{\langle x\triangleleft y,s\rangle\xrightarrow{e_1}\langle x',s'\rangle}$$

        $$\frac{\langle x,s\rangle\xrightarrow{e_3}\langle\surd,s'\rangle \quad \langle y,s\rangle\nrightarrow^{e_2}\quad (\sharp(e_1,e_2),e_1\leq e_3)}{\langle x\triangleleft y,s\rangle\xrightarrow{\tau}\langle\surd,s'\rangle}
        \quad\frac{\langle x,s\rangle\xrightarrow{e_3}\langle x',s'\rangle \quad \langle y,s\rangle\nrightarrow^{e_2}\quad (\sharp(e_1,e_2),e_1\leq e_3)}{\langle x\triangleleft y,s\rangle\xrightarrow{\tau}\langle x',s'\rangle}$$

        $$\frac{\langle x,s\rangle\xrightarrow{e}\langle\surd,s'\rangle}{\langle\partial_H(x),s\rangle\xrightarrow{e}\langle\surd,s'\rangle}\quad (e\notin H)\quad\frac{\langle x,s\rangle\xrightarrow{e}\langle x',s'\rangle}{\langle\partial_H(x),s\rangle\xrightarrow{e}\langle\partial_H(x'),s'\rangle}\quad(e\notin H)$$
        \caption{Transition rules of $APTC_G$}
        \label{TRForAPTCG}
    \end{table}
\end{center}

\begin{theorem}[Generalization of $APTC_G$ with respect to $BATC_G$]
$APTC_G$ is a generalization of $BATC_G$.
\end{theorem}

\begin{theorem}[Congruence of $APTC_G$ with respect to truly concurrent bisimulation equivalences]\label{CAPTCG}
(1) Pomset bisimulation equivalence $\sim_{p}$ is a congruence with respect to $APTC_G$.

(2) Step bisimulation equivalence $\sim_{s}$ is a congruence with respect to $APTC_G$.

(3) Hp-bisimulation equivalence $\sim_{hp}$ is a congruence with respect to $APTC_G$.

(4) Hhp-bisimulation equivalence $\sim_{hhp}$ is a congruence with respect to $APTC_G$.
\end{theorem}

\begin{theorem}[Soundness of $APTC_G$ modulo truly concurrent bisimulation equivalences]\label{SAPTCG}
(1) Let $x$ and $y$ be $APTC_G$ terms. If $APTC\vdash x=y$, then $x\sim_{p} y$.

(2) Let $x$ and $y$ be $APTC_G$ terms. If $APTC\vdash x=y$, then $x\sim_{s} y$.

(3) Let $x$ and $y$ be $APTC_G$ terms. If $APTC\vdash x=y$, then $x\sim_{hp} y$.
\end{theorem}

\begin{theorem}[Completeness of $APTC_G$ modulo truly concurrent bisimulation equivalences]\label{CAPTCG}
(1) Let $p$ and $q$ be closed $APTC_G$ terms, if $p\sim_{p} q$ then $p=q$.

(2) Let $p$ and $q$ be closed $APTC_G$ terms, if $p\sim_{s} q$ then $p=q$.

(3) Let $p$ and $q$ be closed $APTC_G$ terms, if $p\sim_{hp} q$ then $p=q$.
\end{theorem}

\begin{theorem}[Sufficient determinacy]
All related data environments with respect to $APTC_G$ can be sufficiently deterministic.
\end{theorem}

\subsection{Recursion}{\label{recg}}

In this subsection, we introduce recursion to capture infinite processes based on $APTC_G$. In the following, $E,F,G$ are recursion specifications, $X,Y,Z$ are recursive variables.

\begin{definition}[Guarded recursive specification]
A recursive specification

$$X_1=t_1(X_1,\cdots,X_n)$$
$$...$$
$$X_n=t_n(X_1,\cdots,X_n)$$

is guarded if the right-hand sides of its recursive equations can be adapted to the form by applications of the axioms in $APTC$ and replacing recursion variables by the right-hand sides of their recursive equations,

$$(a_{11}\parallel\cdots\parallel a_{1i_1})\cdot s_1(X_1,\cdots,X_n)+\cdots+(a_{k1}\parallel\cdots\parallel a_{ki_k})\cdot s_k(X_1,\cdots,X_n)+(b_{11}\parallel\cdots\parallel b_{1j_1})+\cdots+(b_{1j_1}\parallel\cdots\parallel b_{lj_l})$$

where $a_{11},\cdots,a_{1i_1},a_{k1},\cdots,a_{ki_k},b_{11},\cdots,b_{1j_1},b_{1j_1},\cdots,b_{lj_l}\in \mathbb{E}$, and the sum above is allowed to be empty, in which case it represents the deadlock $\delta$. And there does not exist an infinite sequence of $\epsilon$-transitions $\langle X|E\rangle\rightarrow\langle X'|E\rangle\rightarrow\langle X''|E\rangle\rightarrow\cdots$.
\end{definition}

\begin{center}
    \begin{table}
        $$\frac{\langle t_i(\langle X_1|E\rangle,\cdots,\langle X_n|E\rangle),s\rangle\xrightarrow{\{e_1,\cdots,e_k\}}\langle\surd,s'\rangle}{\langle\langle X_i|E\rangle,s\rangle\xrightarrow{\{e_1,\cdots,e_k\}}\langle\surd,s'\rangle}$$
        $$\frac{\langle t_i(\langle X_1|E\rangle,\cdots,\langle X_n|E\rangle),s\rangle\xrightarrow{\{e_1,\cdots,e_k\}} \langle y,s'\rangle}{\langle\langle X_i|E\rangle,s\rangle\xrightarrow{\{e_1,\cdots,e_k\}} \langle y,s'\rangle}$$
        \caption{Transition rules of guarded recursion}
        \label{TRForGRG}
    \end{table}
\end{center}

\begin{theorem}[Conservitivity of $APTC_G$ with guarded recursion]
$APTC_G$ with guarded recursion is a conservative extension of $APTC_G$.
\end{theorem}

\begin{theorem}[Congruence theorem of $APTC_G$ with guarded recursion]
Truly concurrent bisimulation equivalences $\sim_{p}$, $\sim_s$ and $\sim_{hp}$ are all congruences with respect to $APTC_G$ with guarded recursion.
\end{theorem}

\begin{theorem}[Elimination theorem of $APTC_G$ with linear recursion]\label{ETRecursionG}
Each process term in $APTC_G$ with linear recursion is equal to a process term $\langle X_1|E\rangle$ with $E$ a linear recursive specification.
\end{theorem}

\begin{theorem}[Soundness of $APTC_G$ with guarded recursion]\label{SAPTC_GRG}
Let $x$ and $y$ be $APTC_G$ with guarded recursion terms. If $APTC_G\textrm{ with guarded recursion}\vdash x=y$, then

(1) $x\sim_{s} y$.

(2) $x\sim_{p} y$.

(3) $x\sim_{hp} y$.
\end{theorem}

\begin{theorem}[Completeness of $APTC_G$ with linear recursion]\label{CAPTC_GRG}
Let $p$ and $q$ be closed $APTC_G$ with linear recursion terms, then,

(1) if $p\sim_{s} q$ then $p=q$.

(2) if $p\sim_{p} q$ then $p=q$.

(3) if $p\sim_{hp} q$ then $p=q$.
\end{theorem}

\subsection{Abstraction}{\label{absg}}

To abstract away from the internal implementations of a program, and verify that the program exhibits the desired external behaviors, the silent step $\tau$ and abstraction operator $\tau_I$ are introduced, where $I\subseteq \mathbb{E}\cup G_{at}$ denotes the internal events or guards. The silent step $\tau$ represents the internal events or guards, when we consider the external behaviors of a process, $\tau$ steps can be removed, that is, $\tau$ steps must keep silent. The transition rule of $\tau$ is shown in Table \ref{TRForTauG}. In the following, let the atomic event $e$ range over $\mathbb{E}\cup\{\epsilon\}\cup\{\delta\}\cup\{\tau\}$, and $\phi$ range over $G\cup \{\tau\}$, and let the communication function $\gamma:\mathbb{E}\cup\{\tau\}\times \mathbb{E}\cup\{\tau\}\rightarrow \mathbb{E}\cup\{\delta\}$, with each communication involved $\tau$ resulting in $\delta$. We use $\tau(s)$ to denote $effect(\tau,s)$, for the fact that $\tau$ only change the state of internal data environment, that is, for the external data environments, $s=\tau(s)$.

\begin{center}
    \begin{table}
        $$\frac{}{\langle\tau,s\rangle\rightarrow\langle\surd,s\rangle}\textrm{ if }test(\tau,s)$$
        $$\frac{}{\langle\tau,s\rangle\xrightarrow{\tau}\langle\surd,\tau(s)\rangle}$$
        \caption{Transition rule of the silent step}
        \label{TRForTauG}
    \end{table}
\end{center}

In section \ref{os2}, we introduce $\tau$ into event structure, and also give the concept of weakly true concurrency. In this subsection, we give the concepts of rooted branching truly concurrent bisimulation equivalences, based on these concepts, we can design the axiom system of the silent step $\tau$ and the abstraction operator $\tau_I$.

\begin{definition}[Branching pomset, step bisimulation]\label{BPSBG}
Assume a special termination predicate $\downarrow$, and let $\surd$ represent a state with $\surd\downarrow$. Let $\mathcal{E}_1$, $\mathcal{E}_2$ be PESs. A branching pomset bisimulation is a relation $R\subseteq\langle\mathcal{C}(\mathcal{E}_1),S\rangle\times\langle\mathcal{C}(\mathcal{E}_2),S\rangle$, such that:
 \begin{enumerate}
   \item if $(\langle C_1,s\rangle,\langle C_2,s\rangle)\in R$, and $\langle C_1,s\rangle\xrightarrow{X}\langle C_1',s'\rangle$ then
   \begin{itemize}
     \item either $X\equiv \tau^*$, and $(\langle C_1',s'\rangle,\langle C_2,s\rangle)\in R$ with $s'\in \tau(s)$;
     \item or there is a sequence of (zero or more) $\tau$-transitions $\langle C_2,s\rangle\xrightarrow{\tau^*} \langle C_2^0,s^0\rangle$, such that $(\langle C_1,s\rangle,\langle C_2^0,s^0\rangle)\in R$ and $\langle C_2^0,s^0\rangle\xRightarrow{X}\langle C_2',s'\rangle$ with $(\langle C_1',s'\rangle,\langle C_2',s'\rangle)\in R$;
   \end{itemize}
   \item if $(\langle C_1,s\rangle,\langle C_2,s\rangle)\in R$, and $\langle C_2,s\rangle\xrightarrow{X}\langle C_2',s'\rangle$ then
   \begin{itemize}
     \item either $X\equiv \tau^*$, and $(\langle C_1,s\rangle,\langle C_2',s'\rangle)\in R$;
     \item or there is a sequence of (zero or more) $\tau$-transitions $\langle C_1,s\rangle\xrightarrow{\tau^*} \langle C_1^0,s^0\rangle$, such that $(\langle C_1^0,s^0\rangle,\langle C_2,s\rangle)\in R$ and $\langle C_1^0,s^0\rangle\xRightarrow{X}\langle C_1',s'\rangle$ with $(\langle C_1',s'\rangle,\langle C_2',s'\rangle)\in R$;
   \end{itemize}
   \item if $(\langle C_1,s\rangle,\langle C_2,s\rangle)\in R$ and $\langle C_1,s\rangle\downarrow$, then there is a sequence of (zero or more) $\tau$-transitions $\langle C_2,s\rangle\xrightarrow{\tau^*}\langle C_2^0,s^0\rangle$ such that $(\langle C_1,s\rangle,\langle C_2^0,s^0\rangle)\in R$ and $\langle C_2^0,s^0\rangle\downarrow$;
   \item if $(\langle C_1,s\rangle,\langle C_2,s\rangle)\in R$ and $\langle C_2,s\rangle\downarrow$, then there is a sequence of (zero or more) $\tau$-transitions $\langle C_1,s\rangle\xrightarrow{\tau^*}\langle C_1^0,s^0\rangle$ such that $(\langle C_1^0,s^0\rangle,\langle C_2,s\rangle)\in R$ and $\langle C_1^0,s^0\rangle\downarrow$.
 \end{enumerate}

We say that $\mathcal{E}_1$, $\mathcal{E}_2$ are branching pomset bisimilar, written $\mathcal{E}_1\approx_{bp}\mathcal{E}_2$, if there exists a branching pomset bisimulation $R$, such that $(\langle\emptyset,\emptyset\rangle,\langle\emptyset,\emptyset\rangle)\in R$.

By replacing pomset transitions with steps, we can get the definition of branching step bisimulation. When PESs $\mathcal{E}_1$ and $\mathcal{E}_2$ are branching step bisimilar, we write $\mathcal{E}_1\approx_{bs}\mathcal{E}_2$.
\end{definition}

\begin{definition}[Rooted branching pomset, step bisimulation]\label{RBPSBG}
Assume a special termination predicate $\downarrow$, and let $\surd$ represent a state with $\surd\downarrow$. Let $\mathcal{E}_1$, $\mathcal{E}_2$ be PESs. A rooted branching pomset bisimulation is a relation $R\subseteq\langle\mathcal{C}(\mathcal{E}_1),S\rangle\times\langle\mathcal{C}(\mathcal{E}_2),S\rangle$, such that:
 \begin{enumerate}
   \item if $(\langle C_1,s\rangle,\langle C_2,s\rangle)\in R$, and $\langle C_1,s\rangle\xrightarrow{X}\langle C_1',s'\rangle$ then $\langle C_2,s\rangle\xrightarrow{X}\langle C_2',s'\rangle$ with $\langle C_1',s'\rangle\approx_{bp}\langle C_2',s'\rangle$;
   \item if $(\langle C_1,s\rangle,\langle C_2,s\rangle)\in R$, and $\langle C_2,s\rangle\xrightarrow{X}\langle C_2',s'\rangle$ then $\langle C_1,s\rangle\xrightarrow{X}\langle C_1',s'\rangle$ with $\langle C_1',s'\rangle\approx_{bp}\langle C_2',s'\rangle$;
   \item if $(\langle C_1,s\rangle,\langle C_2,s\rangle)\in R$ and $\langle C_1,s\rangle\downarrow$, then $\langle C_2,s\rangle\downarrow$;
   \item if $(\langle C_1,s\rangle,\langle C_2,s\rangle)\in R$ and $\langle C_2,s\rangle\downarrow$, then $\langle C_1,s\rangle\downarrow$.
 \end{enumerate}

We say that $\mathcal{E}_1$, $\mathcal{E}_2$ are rooted branching pomset bisimilar, written $\mathcal{E}_1\approx_{rbp}\mathcal{E}_2$, if there exists a rooted branching pomset bisimulation $R$, such that $(\langle\emptyset,\emptyset\rangle,\langle\emptyset,\emptyset\rangle)\in R$.

By replacing pomset transitions with steps, we can get the definition of rooted branching step bisimulation. When PESs $\mathcal{E}_1$ and $\mathcal{E}_2$ are rooted branching step bisimilar, we write $\mathcal{E}_1\approx_{rbs}\mathcal{E}_2$.
\end{definition}

\begin{definition}[Branching (hereditary) history-preserving bisimulation]\label{BHHPBG}
Assume a special termination predicate $\downarrow$, and let $\surd$ represent a state with $\surd\downarrow$. A branching history-preserving (hp-) bisimulation is a weakly posetal relation $R\subseteq\langle\mathcal{C}(\mathcal{E}_1),S\rangle\overline{\times}\langle\mathcal{C}(\mathcal{E}_2),S\rangle$ such that:

 \begin{enumerate}
   \item if $(\langle C_1,s\rangle,f,\langle C_2,s\rangle)\in R$, and $\langle C_1,s\rangle\xrightarrow{e_1}\langle C_1',s'\rangle$ then
   \begin{itemize}
     \item either $e_1\equiv \tau$, and $(\langle C_1',s'\rangle,f[e_1\mapsto \tau],\langle C_2,s\rangle)\in R$;
     \item or there is a sequence of (zero or more) $\tau$-transitions $\langle C_2,s\rangle\xrightarrow{\tau^*} \langle C_2^0,s^0\rangle$, such that $(\langle C_1,s\rangle,f,\langle C_2^0,s^0\rangle)\in R$ and $\langle C_2^0,s^0\rangle\xrightarrow{e_2}\langle C_2',s'\rangle$ with $(\langle C_1',s'\rangle,f[e_1\mapsto e_2],\langle C_2',s'\rangle)\in R$;
   \end{itemize}
   \item if $(\langle C_1,s\rangle,f,\langle C_2,s\rangle)\in R$, and $\langle C_2,s\rangle\xrightarrow{e_2}\langle C_2',s'\rangle$ then
   \begin{itemize}
     \item either $e_2\equiv \tau$, and $(\langle C_1,s\rangle,f[e_2\mapsto \tau],\langle C_2',s'\rangle)\in R$;
     \item or there is a sequence of (zero or more) $\tau$-transitions $\langle C_1,s\rangle\xrightarrow{\tau^*} \langle C_1^0,s^0\rangle$, such that $(\langle C_1^0,s^0\rangle,f,\langle C_2,s\rangle)\in R$ and $\langle C_1^0,s^0\rangle\xrightarrow{e_1}\langle C_1',s'\rangle$ with $(\langle C_1',s'\rangle,f[e_2\mapsto e_1],\langle C_2',s'\rangle)\in R$;
   \end{itemize}
   \item if $(\langle C_1,s\rangle,f,\langle C_2,s\rangle)\in R$ and $\langle C_1,s\rangle\downarrow$, then there is a sequence of (zero or more) $\tau$-transitions $\langle C_2,s\rangle\xrightarrow{\tau^*}\langle C_2^0,s^0\rangle$ such that $(\langle C_1,s\rangle,f,\langle C_2^0,s^0\rangle)\in R$ and $\langle C_2^0,s^0\rangle\downarrow$;
   \item if $(\langle C_1,s\rangle,f,\langle C_2,s\rangle)\in R$ and $\langle C_2,s\rangle\downarrow$, then there is a sequence of (zero or more) $\tau$-transitions $\langle C_1,s\rangle\xrightarrow{\tau^*}\langle C_1^0,s^0\rangle$ such that $(\langle C_1^0,s^0\rangle,f,\langle C_2,s\rangle)\in R$ and $\langle C_1^0,s^0\rangle\downarrow$.
 \end{enumerate}

$\mathcal{E}_1,\mathcal{E}_2$ are branching history-preserving (hp-)bisimilar and are written $\mathcal{E}_1\approx_{bhp}\mathcal{E}_2$ if there exists a branching hp-bisimulation $R$ such that $(\langle\emptyset,\emptyset\rangle,\emptyset,\langle\emptyset,\emptyset\rangle)\in R$.

A branching hereditary history-preserving (hhp-)bisimulation is a downward closed branching hp-bisimulation. $\mathcal{E}_1,\mathcal{E}_2$ are branching hereditary history-preserving (hhp-)bisimilar and are written $\mathcal{E}_1\approx_{bhhp}\mathcal{E}_2$.
\end{definition}

\begin{definition}[Rooted branching (hereditary) history-preserving bisimulation]\label{RBHHPBG}
Assume a special termination predicate $\downarrow$, and let $\surd$ represent a state with $\surd\downarrow$. A rooted branching history-preserving (hp-) bisimulation is a weakly posetal relation $R\subseteq\langle\mathcal{C}(\mathcal{E}_1),S\rangle\overline{\times}\langle\mathcal{C}(\mathcal{E}_2),S\rangle$ such that:

 \begin{enumerate}
   \item if $(\langle C_1,s\rangle,f,\langle C_2,s\rangle)\in R$, and $\langle C_1,s\rangle\xrightarrow{e_1}\langle C_1',s'\rangle$, then $\langle C_2,s\rangle\xrightarrow{e_2}\langle C_2',s'\rangle$ with $\langle C_1',s'\rangle\approx_{bhp}\langle C_2',s'\rangle$;
   \item if $(\langle C_1,s\rangle,f,\langle C_2,s\rangle)\in R$, and $\langle C_2,s\rangle\xrightarrow{e_2}\langle C_2',s'\rangle$, then $\langle C_1,s\rangle\xrightarrow{e_1}\langle C_1',s'\rangle$ with $\langle C_1',s'\rangle\approx_{bhp}\langle C_2',s'\rangle$;
   \item if $(\langle C_1,s\rangle,f,\langle C_2,s\rangle)\in R$ and $\langle C_1,s\rangle\downarrow$, then $\langle C_2,s\rangle\downarrow$;
   \item if $(\langle C_1,s\rangle,f,\langle C_2,s\rangle)\in R$ and $\langle C_2,s\rangle\downarrow$, then $\langle C_1,s\rangle\downarrow$.
 \end{enumerate}

$\mathcal{E}_1,\mathcal{E}_2$ are rooted branching history-preserving (hp-)bisimilar and are written $\mathcal{E}_1\approx_{rbhp}\mathcal{E}_2$ if there exists a rooted branching hp-bisimulation $R$ such that $(\langle\emptyset,\emptyset\rangle,\emptyset,\langle\emptyset,\emptyset\rangle)\in R$.

A rooted branching hereditary history-preserving (hhp-)bisimulation is a downward closed rooted branching hp-bisimulation. $\mathcal{E}_1,\mathcal{E}_2$ are rooted branching hereditary history-preserving (hhp-)bisimilar and are written $\mathcal{E}_1\approx_{rbhhp}\mathcal{E}_2$.
\end{definition}

\begin{definition}[Guarded linear recursive specification]\label{GLRSG}
A linear recursive specification $E$ is guarded if there does not exist an infinite sequence of $\tau$-transitions $\langle X|E\rangle\xrightarrow{\tau}\langle X'|E\rangle\xrightarrow{\tau}\langle X''|E\rangle\xrightarrow{\tau}\cdots$, and there does not exist an infinite sequence of $\epsilon$-transitions $\langle X|E\rangle\rightarrow\langle X'|E\rangle\rightarrow\langle X''|E\rangle\rightarrow\cdots$.
\end{definition}

\begin{theorem}[Conservitivity of $APTC_G$ with silent step and guarded linear recursion]
$APTC_G$ with silent step and guarded linear recursion is a conservative extension of $APTC_G$ with linear recursion.
\end{theorem}

\begin{theorem}[Congruence theorem of $APTC_G$ with silent step and guarded linear recursion]
Rooted branching truly concurrent bisimulation equivalences $\approx_{rbp}$, $\approx_{rbs}$ and $\approx_{rbhp}$ are all congruences with respect to $APTC_G$ with silent step and guarded linear recursion.
\end{theorem}

We design the axioms for the silent step $\tau$ in Table \ref{AxiomsForTauG}.

\begin{center}
\begin{table}
  \begin{tabular}{@{}ll@{}}
\hline No. &Axiom\\
  $B1$ & $e\cdot\tau=e$\\
  $B2$ & $e\cdot(\tau\cdot(x+y)+x)=e\cdot(x+y)$\\
  $B3$ & $x\parallel\tau=x$\\
  $G26$ & $\phi\cdot\tau=\phi$\\
  $G27$ & $\phi\cdot(\tau\cdot(x+y)+x)=\phi\cdot(x+y)$\\
\end{tabular}
\caption{Axioms of silent step}
\label{AxiomsForTauG}
\end{table}
\end{center}

\begin{theorem}[Elimination theorem of $APTC_G$ with silent step and guarded linear recursion]\label{ETTauG}
Each process term in $APTC_G$ with silent step and guarded linear recursion is equal to a process term $\langle X_1|E\rangle$ with $E$ a guarded linear recursive specification.
\end{theorem}

\begin{theorem}[Soundness of $APTC_G$ with silent step and guarded linear recursion]\label{SAPTC_GTAUG}
Let $x$ and $y$ be $APTC_G$ with silent step and guarded linear recursion terms. If $APTC_G$ with silent step and guarded linear recursion $\vdash x=y$, then

(1) $x\approx_{rbs} y$.

(2) $x\approx_{rbp} y$.

(3) $x\approx_{rbhp} y$.
\end{theorem}

\begin{theorem}[Completeness of $APTC_G$ with silent step and guarded linear recursion]\label{CAPTC_GTAUG}
Let $p$ and $q$ be closed $APTC_G$ with silent step and guarded linear recursion terms, then,

(1) if $p\approx_{rbs} q$ then $p=q$.

(2) if $p\approx_{rbp} q$ then $p=q$.

(3) if $p\approx_{rbhp} q$ then $p=q$.
\end{theorem}

The unary abstraction operator $\tau_I$ ($I\subseteq \mathbb{E}\cup G_{at}$) renames all atomic events or atomic guards in $I$ into $\tau$. $APTC_G$ with silent step and abstraction operator is called $APTC_{G_{\tau}}$. The transition rules of operator $\tau_I$ are shown in Table \ref{TRForAbstractionG}.

\begin{center}
    \begin{table}
        $$\frac{\langle x,s\rangle\xrightarrow{e}\langle\surd,s'\rangle}{\langle\tau_I(x),s\rangle\xrightarrow{e}\langle\surd,s'\rangle}\quad e\notin I
        \quad\quad\frac{\langle x,s\rangle\xrightarrow{e}\langle x',s'\rangle}{\langle\tau_I(x),s\rangle\xrightarrow{e}\langle\tau_I(x'),s'\rangle}\quad e\notin I$$

        $$\frac{\langle x,s\rangle\xrightarrow{e}\langle\surd,s'\rangle}{\langle\tau_I(x),s\rangle\xrightarrow{\tau}\langle\surd,\tau(s)\rangle}\quad e\in I
        \quad\quad\frac{\langle x,s\rangle\xrightarrow{e}\langle x',s'\rangle}{\langle\tau_I(x),s\rangle\xrightarrow{\tau}\langle\tau_I(x'),\tau(s)\rangle}\quad e\in I$$
        \caption{Transition rule of the abstraction operator}
        \label{TRForAbstractionG}
    \end{table}
\end{center}

\begin{theorem}[Conservitivity of $APTC_{G_{\tau}}$ with guarded linear recursion]
$APTC_{G_{\tau}}$ with guarded linear recursion is a conservative extension of $APTC_G$ with silent step and guarded linear recursion.
\end{theorem}

\begin{theorem}[Congruence theorem of $APTC_{G_{\tau}}$ with guarded linear recursion]
Rooted branching truly concurrent bisimulation equivalences $\approx_{rbp}$, $\approx_{rbs}$ and $\approx_{rbhp}$ are all congruences with respect to $APTC_{G_{\tau}}$ with guarded linear recursion.
\end{theorem}

We design the axioms for the abstraction operator $\tau_I$ in Table \ref{AxiomsForAbstractionG}.

\begin{center}
\begin{table}
  \begin{tabular}{@{}ll@{}}
\hline No. &Axiom\\
  $TI1$ & $e\notin I\quad \tau_I(e)=e$\\
  $TI2$ & $e\in I\quad \tau_I(e)=\tau$\\
  $TI3$ & $\tau_I(\delta)=\delta$\\
  $TI4$ & $\tau_I(x+y)=\tau_I(x)+\tau_I(y)$\\
  $TI5$ & $\tau_I(x\cdot y)=\tau_I(x)\cdot\tau_I(y)$\\
  $TI6$ & $\tau_I(x\parallel y)=\tau_I(x)\parallel\tau_I(y)$\\
  $G28$ & $\phi\notin I\quad \tau_I(\phi)=\phi$\\
  $G29$ & $\phi\in I\quad \tau_I(\phi)=\tau$\\
\end{tabular}
\caption{Axioms of abstraction operator}
\label{AxiomsForAbstractionG}
\end{table}
\end{center}

\begin{theorem}[Soundness of $APTC_{G_{\tau}}$ with guarded linear recursion]\label{SAPTC_GABSG}
Let $x$ and $y$ be $APTC_{G_{\tau}}$ with guarded linear recursion terms. If $APTC_{G_{\tau}}$ with guarded linear recursion $\vdash x=y$, then

(1) $x\approx_{rbs} y$.

(2) $x\approx_{rbp} y$.

(3) $x\approx_{rbhp} y$.
\end{theorem}

Though $\tau$-loops are prohibited in guarded linear recursive specifications (see Definition \ref{GLRSG}) in a specifiable way, they can be constructed using the abstraction operator, for example, there exist $\tau$-loops in the process term $\tau_{\{a\}}(\langle X|X=aX\rangle)$. To avoid $\tau$-loops caused by $\tau_I$ and ensure fairness, the concept of cluster and $CFAR$ (Cluster Fair Abstraction Rule) \cite{CFAR} are still needed.

\begin{theorem}[Completeness of $APTC_{G_{\tau}}$ with guarded linear recursion and $CFAR$]\label{CCFARG}
Let $p$ and $q$ be closed $APTC_{G_{\tau}}$ with guarded linear recursion and $CFAR$ terms, then,

(1) if $p\approx_{rbs} q$ then $p=q$.

(2) if $p\approx_{rbp} q$ then $p=q$.

(3) if $p\approx_{rbhp} q$ then $p=q$.
\end{theorem}

\newpage\section{Secure APTC}\label{saptc}

Cryptography mainly includes two aspects: the cryptographic operations and security protocols. The former includes symmetric and asymmetric encryption/decryption, hash, digital signatures,
message authentication codes, random sequence generation, and XOR, etc. The latter includes the computational logic driven by the security application logics among the cryptographic
operations.

In this chapter, we model the above two cryptographic properties by APTC ($APTC_G$). In section \ref{se}, we model symmetric encryption/decryption by APTC. And we model asymmetric
encryption/decryption, hash, digital signatures, message authentication codes, random sequence generation, blind signatures, and XOR in section \ref{ae}, \ref{hash}, \ref{ds}, \ref{mac}, \ref{rsg},
\ref{bs},\ref{xor}. In section \ref{ec}, we extended the communication merge to support data substitution. Finally, in section \ref{asp}, we show that how to analyze the security protocols
by use of APTC ($APTC_G$).

\subsection{Symmetric Encryption}\label{se}

In the symmetric encryption and decryption, there uses only one key $k$. The inputs of symmetric encryption are the key $k$ and the plaintext $D$ and the output is the ciphertext, so
we treat the symmetric encryption as an atomic action denoted $enc_k(D)$. We also use $ENC_k(D)$ to denote the ciphertext output. The inputs of symmetric
decryption are the same key $k$ and the ciphertext $ENC_k(D)$ and output is the plaintext $D$, we also treat the symmetric decryption as an atomic action $dec_k(ENC_k(D))$. And we also
use $DEC_k(ENC_k(D))$ to denote the output of the corresponding decryption.

For $D$ is plaintext, it is obvious that $DEC_k(ENC_k(D))=D$ and $enc_k(D)\leq dec_k(ENC_k(D))$, where $\leq$ is the causal relation; and for $D$ is the ciphertext, $ENC_k(DEC_k(D))=D$
and $dec_k(D)\leq enc_k(DEC_k(D))$ hold.

\subsection{Asymmetric Encryption}\label{ae}

In the asymmetric encryption and decryption, there uses two keys: the public key $pk_s$ and the private key $sk_s$ generated from the same seed $s$. The inputs of asymmetric encryption are
the key $pk_s$ or $sk_s$ and the plaintext $D$ and the output is the ciphertext, so
we treat the asymmetric encryption as an atomic action denoted $enc_{pk_s}(D)$ or $enc_{sk_s}(D)$. We also use
$ENC_{pk_s}(D)$ and $ENC_{sk_s}(D)$ to denote the ciphertext outputs. The inputs of asymmetric
decryption are the corresponding key $sk_s$ or $pk_s$ and the ciphertext $ENC_{pk_s}(D)$ or $ENC_{sk_s}(D)$, and output is the plaintext $D$, we also treat the asymmetric decryption
as an atomic action $dec_{sk_s}(ENC_{pk_s}(D))$ and $dec_{pk_s}(ENC_{sk_s}(D))$. And we also
use $DEC_{sk_s}(ENC_{pk_s}(D))$ and $DEC_{pk_s}(ENC_{sk_s}(D))$ to denote the corresponding decryption outputs.

For $D$ is plaintext, it is obvious that $DEC_{sk_s}(ENC_{pk_s}(D))=D$ and $DEC_{pk_s}(ENC_{sk_s}(D))=D$, and $enc_{pk_s}(D)\leq dec_{sk_s}(ENC_{pk_s}(D))$ and
$enc_{sk_s}(D)\leq dec_{pk_s}(ENC_{sk_s}(D))$, where $\leq$ is the causal relation; and for $D$ is the ciphertext, $ENC_{sk_s}(DEC_{pk_s}(D))=D$ and $ENC_{pk_s}(DEC_{sk_s}(D))=D$,
and $dec_{pk_s}(D)\leq enc_{sk_s}(DEC_{pk_s}(D))$ and $dec_{sk_s}(D)\leq enc_{pk_s}(DEC_{sk_s}(D))$.

\subsection{Hash}\label{hash}

The hash function is used to generate the digest of the data. The input of the hash function $hash$ is the data $D$ and the output is the digest of the data. We treat the hash function
as an atomic action denoted $hash(D)$, and we also use $HASH(D)$ to denote the output digest.

For $D_1=D_2$, it is obvious that $HASH(D_1)=HASH(D_2)$.

\subsection{Digital Signatures}\label{ds}

Digital signature uses the private key $sk_s$ to encrypt some data and the public key $pk_s$ to decrypt the encrypted data to implement the so-called non-repudiation. The inputs of sigh function
are some data $D$ and the private key $sk_s$ and the output is the signature. We treat the signing function as an atomic action $sign_{sk_s}(D)$, and also use $SIGN_{sk_s}(D)$ to denote
the signature. The inputs of the de-sign function are the public key $pk_s$ and the signature $SIGN_{sk_s}(D)$, and the output is the original data $D$. We also treat the de-sign function
as an atomic action $de\textrm{-}sign_{pk_s}(SIGN_{sk_s}(D))$, and also we use $DE\textrm{-}SIGN_{pk_s}(SIGN_{sk_s}(D))$ to denote the output of the de-sign action.

It is obvious that $DE\textrm{-}SIGN_{pk_s}(SIGN_{sk_s}(D))=D$.

\subsection{Message Authentication Codes}\label{mac}

MAC (Message Authentication Code) is used to authenticate data by symmetric keys $k$ and often assumed that $k$ is privately shared only between two principals $A$ and $B$. The inputs of the MAC function
are the key $k$ and some data $D$, and the output is the MACs. We treat the MAC function as an atomic action $mac_k(D)$, and use $MAC_k(D)$ to denote the output MACs.

The MACs $MAC_k(D)$ are generated by one principal $A$ and with $D$ together sent to the other principal $B$. The other principal $B$ regenerate the MACs $MAC_k(D)'$, if $MAC_k(D)=MAC_k(D)'$, then
the data $D$ are from $A$.

\subsection{Random Sequence Generation}\label{rsg}

Random sequence generation is used to generate a random sequence, which may be a symmetric key $k$, a pair of public key $pk_s$ and $sk_s$, or a nonce $nonce$ (usually used
to resist replay attacks). We treat the random sequence generation function as an atomic action $rsg_k$ for symmetric key generation, $rsg_{pk_s,sk_s}$ for asymmetric key pair generation,
and $rsg_N$ for nonce generation, and the corresponding outputs are $k$, $pk_s$ and $sk_s$, $N$ respectively.

\subsection{Blind Signatures}\label{bs}

In the blind signatures, there uses only one key $k$. The inputs of blind function are the key $k$ and the plaintext $D$ and the output is the ciphertext, so
we treat the blind function as an atomic action denoted $blind_k(D)$. We also use $BLIND_k(D)$ to denote the ciphertext output. The inputs of unblind function
are the same key $k$ and the ciphertext $BLIND_k(D)$ and output is the plaintext $D$, we also treat the unblind function as an atomic action $unblind_k(BLIND_k(D))$. And we also
use $UNBLIND_k(BLIND_k(D))$ to denote the output of the corresponding unblind function.

For $D$ is plaintext, it is obvious that $UNBLIND_k(BLIND_k(D))=D$ and \\$blind_k(D)\leq unblind_k(BLIND_k(D))$, where $\leq$ is the causal relation; and for $D$ is the ciphertext.
And also $UNBLIND_k(SIGN_{sk}(BLIND_k(D)))=SIGN_{sk}(D)$.

\subsection{XOR}\label{xor}

The inputs of the XOR function are two data $D_1$ and $D_2$, and the output is the XOR result. We treat the XOR function as an atomic action $xor(D_1,D_2)$, and we also use $XOR(D_1,D_2)$
to denoted the XOR result.

It is obvious that the following equations hold:

\begin{enumerate}
  \item $XOR(XOR(D_1,D_2),D_3)=XOR(D_1,XOR(D_2,D_3))$.
  \item $XOR(D_1,D_2)=XOR(D_2,D_1)$.
  \item $XOR(D,0)=D$.
  \item $XOR(D,D)=0$.
  \item $XOR(D_2,XOR(D_1,D_2))=D_1$
\end{enumerate}

\subsection{Extended Communications}\label{ec}

In APTC ($APTC_G$), the communication between two parallel processes is modeled as the communication merge of two communicating actions. One communicating action is the sending data
($D_1,\cdots,D_n \in\Delta$) action through certain channel $A$ which is denoted $s_A(D_1,\cdots,D_n)$, the other communicating action is the receiving data ($d_1,\cdots,d_n$ range over
$\Delta$) action through the corresponding channel $A$ which is denoted $r_A(d_1,\cdots,d_n)$, note that $d_i$ and $D_i$ for $1\leq i\leq n$ have the same data type.

We extend communication merge to this situation. The axioms of the extended communication merge are shown in Table \ref{AOECM}, and the transition rules are shown in Table \ref{TOECM}.

\begin{center}
    \begin{table}
        \begin{tabular}{@{}ll@{}}
            \hline No. &Axiom\\
            $C1$ & $e_1(D_1,\cdots,D_n)\mid e_2(d_1,\cdots,d_n) = \gamma(e_1(D_1,\cdots,D_n),e_2(d_1,\cdots,d_n))$\\
            $C2$ & $e_1(D_1,\cdots,D_n)\mid (e_2(d_1,\cdots,d_n)\cdot y) = \gamma(e_1(D_1,\cdots,D_n),e_2(d_1,\cdots,d_n))\cdot y[D_1/d_1,\cdots,D_n/d_n]$\\
            $C3$ & $(e_1(D_1,\cdots,D_n)\cdot x)\mid e_2(d_1,\cdots,d_n) = \gamma(e_1(D_1,\cdots,D_n),e_2(d_1,\cdots,d_n))\cdot x$\\
            $C4$ & $(e_1(D_1,\cdots,D_n)\cdot x)\mid (e_2(d_1,\cdots,d_n)\cdot y) = $\\
            &$\gamma(e_1(D_1,\cdots,D_n),e_2(d_1,\cdots,d_n))\cdot (x\between y[D_1/d_1,\cdots,D_n/d_n])$\\
         \end{tabular}
        \caption{Axioms of the Extended Communication Merge}
        \label{AOECM}
    \end{table}
\end{center}
\begin{center}
    \begin{table}
            $$\frac{\langle x,s\rangle\xrightarrow{e_1(D_1,\cdots,D_n)}\langle\surd,s'\rangle\quad \langle y,s\rangle\xrightarrow{e_2(d_1,\cdots,d_n)}\langle\surd,s''\rangle}
            {\langle x\mid y,s\rangle\xrightarrow{\gamma(e_1(D_1,\cdots,D_n),e_2(d_1,\cdots,d_n))}\langle\surd,effect(\gamma(e_1(D_1,\cdots,D_n),e_2(d_1,\cdots,d_n)),s)\rangle}$$
            $$\frac{\langle x,s\rangle\xrightarrow{e_1(D_1,\cdots,D_n)}\langle x',s'\rangle\quad \langle y,s\rangle\xrightarrow{e_2(d_1,\cdots,d_n)}\langle\surd,s''\rangle}
            {\langle x\mid y,s\rangle\xrightarrow{\gamma(e_1(D_1,\cdots,D_n),e_2(d_1,\cdots,d_n))}\langle x',effect(\gamma(e_1(D_1,\cdots,D_n),e_2(d_1,\cdots,d_n)),s)\rangle}$$
            $$\frac{\langle x,s\rangle\xrightarrow{e_1(D_1,\cdots,D_n)}\langle\surd,s'\rangle\quad \langle y,s\rangle\xrightarrow{e_2(d_1,\cdots,d_n)}\langle y',s''\rangle}
            {\langle x\mid y,s\rangle\xrightarrow{\gamma(e_1(D_1,\cdots,D_n),e_2(d_1,\cdots,d_n))}\langle y'[D_1/d_1,\cdots,D_n/d_n],effect(\gamma(e_1(D_1,\cdots,D_n),e_2(d_1,\cdots,d_n)),s)\rangle}$$
            $$\frac{\langle x,s\rangle\xrightarrow{e_1(D_1,\cdots,D_n)}\langle x',s'\rangle\quad \langle y,s\rangle\xrightarrow{e_2(d_1,\cdots,d_n)}\langle y',s''\rangle}
            {\langle x\mid y,s\rangle\xrightarrow{\gamma(e_1(D_1,\cdots,D_n),e_2(d_1,\cdots,d_n))}\langle x'\between y'[D_1/d_1,\cdots,D_n/d_n],effect(\gamma(e_1(D_1,\cdots,D_n),e_2(d_1,\cdots,d_n)),s)\rangle}$$
        \caption{Transition Rules of the Extended Communication Merge}
        \label{TOECM}
    \end{table}
\end{center}

Obviously, the conclusions of the theories of $APTC$ and $APTC_G$ still hold without any alternation.

\subsection{Analyses of Security Protocols}\label{asp}

In this section, we will show the application of analyzing security protocols by APTC ($APTC_G$) via several examples.

\subsubsection{A Protocol Using Private Channels}
The protocol shown in Figure \ref{PC4} uses private channels, that is, the channel $C_{AB}$ between Alice and Bob is private to Alice and Bob, there is no one can use this channel.

\begin{figure}
    \centering
    \includegraphics{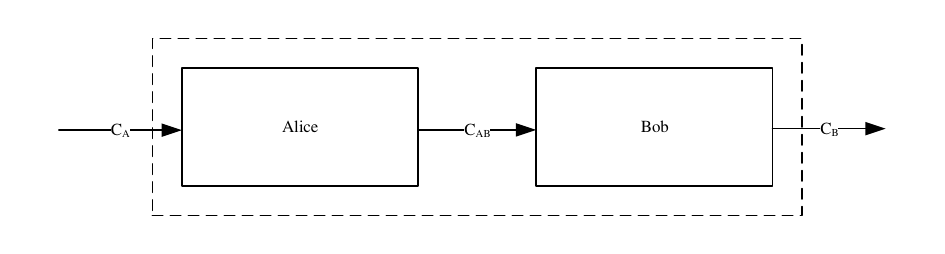}
    \caption{A protocol using private channels}
    \label{PC4}
\end{figure}

The process of the protocol is as follows.

\begin{enumerate}
  \item Alice receives some messages $D$ from the outside through the channel $C_{A}$ (the corresponding reading action is denoted $r_{C_A}(D)$), after an internal processing $af$,
  she sends $D$ to Bob through the private channel $C_{AB}$ (the corresponding sending action is denoted $s_{C_{AB}}(D)$);
  \item Bob receives the message $D$ through the private channel $C_{AB}$ (the corresponding reading action is denoted $r_{C_{AB}}(D)$), after and internal processing $bf$,
  he sends $D$ to the outside through the channel $C_B$ (the corresponding sending action is denoted $s_{C_B}(D)$).
\end{enumerate}

Where $D\in\Delta$, $\Delta$ is the set of data.

Alice's state transitions described by $APTC_G$ are as follows.

$A=\sum_{D\in\Delta}r_{C_A}(D)\cdot A_2$

$A_2=af\cdot A_3$

$A_3=s_{C_{AB}}(D)\cdot A$

Bob's state transitions described by $APTC_G$ are as follows.

$B=r_{C_{AB}}(D)\cdot B_2$

$B_2=bf\cdot B_3$

$B_3=s_{C_{B}}(D)\cdot B$

The sending action and the reading action of the same type data through the same channel can communicate with each other, otherwise, will cause a deadlock $\delta$. We define the following
communication functions.

$\gamma(r_{C_{AB}}(D),s_{C_{AB}}(D)\triangleq c_{C_{AB}}(D)$

Let all modules be in parallel, then the protocol $A\quad B$ can be presented by the following process term.

$$\tau_I(\partial_H(\Theta(A\between B)))=\tau_I(\partial_H(A\between B))$$

where $H=\{r_{C_{AB}}(D),s_{C_{AB}}(D)|D\in\Delta\}$, $I=\{c_{C_{AB}}(D),af,bf|D\in\Delta\}$.

Then we get the following conclusion on the protocol.

\begin{theorem}
The protocol using private channels in Figure \ref{PC4} is secure.
\end{theorem}

\begin{proof}
Based on the above state transitions of the above modules, by use of the algebraic laws of $APTC_G$, we can prove that

$\tau_I(\partial_H(A\between B))=\sum_{D\in\Delta}(r_{C_A}(D)\cdot s_{C_{B}}(D))\cdot
\tau_I(\partial_H(A\between B))$.

For the details of proof, please refer to section \ref{app}, and we omit it.

That is, the protocol in Figure \ref{PC4} $\tau_I(\partial_H(A\between B))$ can exhibit desired external behaviors, and because the channel $C_{AB}$ is private, there is no any attack.

So, The protocol using private channels in Figure \ref{PC4} is secure.
\end{proof}

\subsubsection{Secure Communication Protocols Using Symmetric Keys}\label{confi}

The protocol shown in Figure \ref{CS4} uses symmetric keys for secure communication, that is, the key $k_{AB}$ between Alice and Bob is privately shared to Alice and Bob,
there is no one can use this key. For secure communication, the main challenge is the information leakage to against the confidentiality. Since all channels in Figure \ref{CS4} are
public, so there may be an Eve to intercept the messages sent from Alice to Bob, and try to crack the secrets.

\begin{figure}
    \centering
    \includegraphics{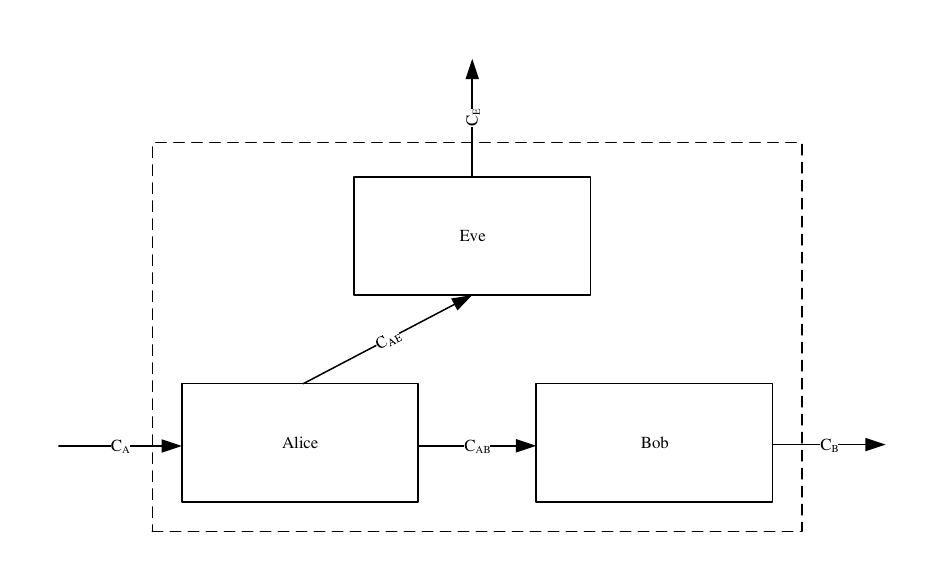}
    \caption{Secure communication protocol using symmetric keys}
    \label{CS4}
\end{figure}

The process of the protocol is as follows.

\begin{enumerate}
  \item Alice receives some messages $D$ from the outside through the channel $C_{A}$ (the corresponding reading action is denoted $r_{C_A}(D)$), after an encryption processing $enc_{k_{AB}}(D)$,
  she sends $ENC_{k_{AB}}(D)$ to Bob through the channel $C_{AB}$ (the corresponding sending action is denoted $s_{C_{AB}}(ENC_{k_{AB}}(D))$). She also
  sends $ENC_{k_{AB}}(D)$ to Eve through the channel $C_{AE}$ (the corresponding sending action is denoted $s_{C_{AE}}(ENC_{k_{AB}}(D))$);
  \item Bob receives the message $ENC_{k_{AB}}(D)$ through the channel $C_{AB}$ (the corresponding reading action is denoted $r_{C_{AB}}(ENC_{k_{AB}}(D))$), after a decryption processing $dec_{k_{AB}}(ENC_{k_{AB}}(D))$,
  he sends $D$ to the outside through the channel $C_B$ (the corresponding sending action is denoted $s_{C_B}(D)$);
  \item Eve receives the message $ENC_{k_{AB}}(D)$ through the channel $C_{AE}$ (the corresponding reading action is denoted $r_{C_{AE}}(ENC_{k_{AB}}(D))$), after a decryption processing $dec_{k_{E}}(ENC_{k_{AB}}(D))$,
  he sends $DEC_{k_{E}}(ENC_{k_{AB}}(D))$ to the outside through the channel $C_E$ (the corresponding sending action is denoted $s_{C_E}(DEC_{k_{E}}(ENC_{k_{AB}}(D)))$).
\end{enumerate}

Where $D\in\Delta$, $\Delta$ is the set of data.

Alice's state transitions described by $APTC_G$ are as follows.

$A=\sum_{D\in\Delta}r_{C_A}(D)\cdot A_2$

$A_2=enc_{k_{AB}}(D)\cdot A_3$

$A_3=(s_{C_{AB}}(ENC_{k_{AB}}(D))\parallel s_{C_{AE}}(ENC_{k_{AB}}(D)))\cdot A$

Bob's state transitions described by $APTC_G$ are as follows.

$B=r_{C_{AB}}(ENC_{k_{AB}}(D))\cdot B_2$

$B_2=dec_{k_{AB}}(ENC_{k_{AB}}(D))\cdot B_3$

$B_3=s_{C_{B}}(D)\cdot B$

Eve's state transitions described by $APTC_G$ are as follows.

$E=r_{C_{AE}}(ENC_{k_{AB}}(D))\cdot E_2$

$E_2=dec_{k_{E}}(ENC_{k_{AB}}(D))\cdot E_3$

$E_3=(\{k_E\neq k_{AB}\}\cdot s_{C_E}(DEC_{k_{E}}(ENC_{k_{AB}}(D)))+\{k_E=k_{AB}\}\cdot s_{C_E}(D))\cdot E$

The sending action and the reading action of the same type data through the same channel can communicate with each other, otherwise, will cause a deadlock $\delta$. We define the following
communication functions.

$\gamma(r_{C_{AB}}(ENC_{k_{AB}}(D)),s_{C_{AB}}(ENC_{k_{AB}}(D))\triangleq c_{C_{AB}}(ENC_{k_{AB}}(D))$

$\gamma(r_{C_{AE}}(ENC_{k_{AB}}(D)),s_{C_{AE}}(ENC_{k_{AB}}(D))\triangleq c_{C_{AE}}(ENC_{k_{AB}}(D))$

Let all modules be in parallel, then the protocol $A\quad B\quad E$ can be presented by the following process term.

$$\tau_I(\partial_H(\Theta(A\between B\between E)))=\tau_I(\partial_H(A\between B\between E))$$

where $H=\{r_{C_{AB}}(ENC_{k_{AB}}(D)),s_{C_{AB}}(ENC_{k_{AB}}(D)),r_{C_{AE}}(ENC_{k_{AB}}(D)),
s_{C_{AE}}(ENC_{k_{AB}}(D))|D\in\Delta\}$,

$I=\{c_{C_{AB}}(ENC_{k_{AB}}(D)),c_{C_{AE}}(ENC_{k_{AB}}(D)),enc_{k_{AB}}(D),dec_{k_{AB}}(ENC_{k_{AB}}(D)),dec_{k_{E}}(ENC_{k_{AB}}(D)),\\
\{k_E\neq k_{AB}\},\{k_E=k_{AB}\}|D\in\Delta\}$.

Then we get the following conclusion on the protocol.

\begin{theorem}
The protocol using symmetric keys for secure communication in Figure \ref{CS4} is confidential.
\end{theorem}

\begin{proof}
Based on the above state transitions of the above modules, by use of the algebraic laws of $APTC_G$, we can prove that

$\tau_I(\partial_H(A\between B\between E))=\sum_{D\in\Delta}(r_{C_A}(D)\cdot (s_{C_{B}}(D)\parallel s_{C_E}(DEC_{k_{E}}(ENC_{k_{AB}}(D)))))\cdot
\tau_I(\partial_H(A\between B\between E))$.

For the details of proof, please refer to section \ref{app}, and we omit it.

That is, the protocol in Figure \ref{CS4} $\tau_I(\partial_H(A\between B\between E))$ can exhibit desired external behaviors, and because the key $k_{AB}$ is private,
$DEC_{k_{E}}(ENC_{k_{AB}}(D))\neq D$ (for $k_E\neq k_{AB}$).

So, The protocol using symmetric keys in Figure \ref{CS4} is confidential.
\end{proof}

\subsubsection{Discussion}

Through the above subsection, we can see the process of analysis of security protocols, that is, through abstract away the internal series of cryptographic operations, we can see the
relation between the inputs and the outputs of the whole protocol, then we can get the conclusions of if or not the protocol being secure.

A security protocol is designed for one or several goals. For example, the secure communication protocol using symmetric keys in Figure \ref{CS4} is designed for the confidentiality of
the communication. So, we only verify if the protocol is confidential. In fact, the protocol in Figure \ref{CS4} can not resist other attacks, for example, the replay attack, as Figure \ref{CS42}
shows.

\begin{figure}
    \centering
    \includegraphics{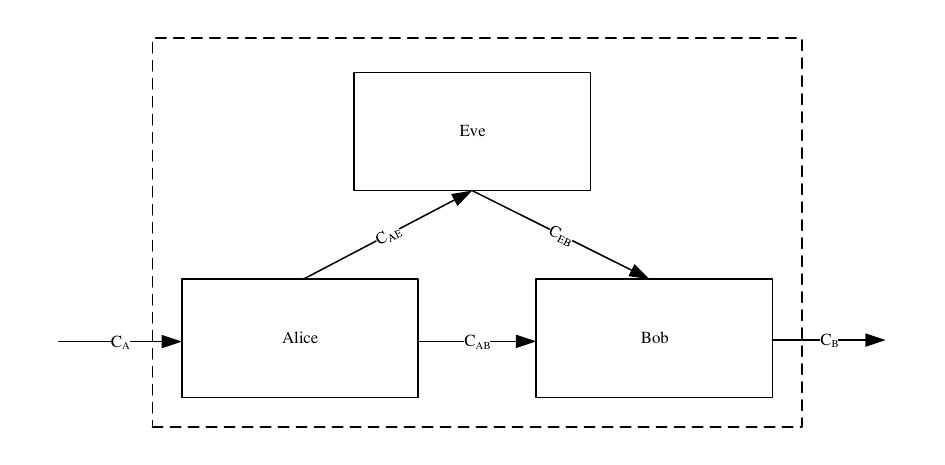}
    \caption{Secure communication protocol using symmetric keys with Replay Attack}
    \label{CS42}
\end{figure}

The process of the protocol is as follows.

\begin{enumerate}
  \item Alice receives some messages $D$ from the outside through the channel $C_{A}$ (the corresponding reading action is denoted $r_{C_A}(D)$), after an encryption processing $enc_{k_{AB}}(D)$,
  she sends $ENC_{k_{AB}}(D)$ to Bob through the channel $C_{AB}$ (the corresponding sending action is denoted $s_{C_{AB}}(ENC_{k_{AB}}(D))$). She also
  sends $ENC_{k_{AB}}(D)$ to Eve through the channel $C_{AE}$ (the corresponding sending action is denoted $s_{C_{AE}}(ENC_{k_{AB}}(D))$);
  \item Eve receives the message $ENC_{k_{AB}}(D)$ through the channel $C_{AE}$ (the corresponding reading action is denoted $r_{C_{AE}}(ENC_{k_{AB}}(D))$), without an internal processing,
  he sends $ENC_{k_{AB}}(D)$ to the outside through the channel $C_{EB}$ (the corresponding sending action is denoted $s_{C_{EB}}(ENC_{k_{AB}}(D))$);
  \item Bob receives the message $ENC_{k_{AB}}(D)$ through the channel $C_{AB}$ (the corresponding reading action is denoted $r_{C_{AB}}(ENC_{k_{AB}}(D))$), after a decryption processing $dec_{k_{AB}}(ENC_{k_{AB}}(D))$,
  he sends $D$ to the outside through the channel $C_B$ (the corresponding sending action is denoted $s_{C_B}(D)$); Bob receives the message $ENC_{k_{AB}}(D)$ through the channel $C_{EB}$
  (the corresponding reading action is denoted $r_{C_{EB}}(ENC_{k_{AB}}(D))$), after a decryption processing $dec_{k_{AB}}(ENC_{k_{AB}}(D))$,
  he sends $D$ to the outside through the channel $C_B$ (the corresponding sending action is denoted $s_{C_B}(D)$).
\end{enumerate}

Where $D\in\Delta$, $\Delta$ is the set of data.

Alice's state transitions described by $APTC_G$ are as follows.

$A=\sum_{D\in\Delta}r_{C_A}(D)\cdot A_2$

$A_2=enc_{k_{AB}}(D)\cdot A_3$

$A_3=(s_{C_{AB}}(ENC_{k_{AB}}(D))\parallel s_{C_{AE}}(ENC_{k_{AB}}(D)))\cdot A$

Bob's state transitions described by $APTC_G$ are as follows.

$B=(r_{C_{AB}}(ENC_{k_{AB}}(D))\parallel r_{C_{AB}}(ENC_{k_{EB}}(D)))\cdot B_2$

$B_2=(dec_{k_{AB}}(ENC_{k_{AB}}(D))\parallel dec_{k_{AB}}(ENC_{k_{AB}}(D)))\cdot B_3$

$B_3=(s_{C_{B}}(D)\parallel s_{C_{B}}(D))\cdot B$

Eve's state transitions described by $APTC_G$ are as follows.

$E=r_{C_{AE}}(ENC_{k_{AB}}(D))\cdot E_2$

$E_2=dec_{k_{E}}(ENC_{k_{AB}}(D))\cdot E_3$

$E_3=s_{C_{EB}}(ENC_{k_{AB}}(D))\cdot E$

The sending action and the reading action of the same type data through the same channel can communicate with each other, otherwise, will cause a deadlock $\delta$. We define the following
communication functions.

$\gamma(r_{C_{AB}}(ENC_{k_{AB}}(D)),s_{C_{AB}}(ENC_{k_{AB}}(D))\triangleq c_{C_{AB}}(ENC_{k_{AB}}(D))$

$\gamma(r_{C_{AE}}(ENC_{k_{AB}}(D)),s_{C_{AE}}(ENC_{k_{AB}}(D))\triangleq c_{C_{AE}}(ENC_{k_{AB}}(D))$

$\gamma(r_{C_{BE}}(ENC_{k_{AB}}(D)),s_{C_{BE}}(ENC_{k_{AB}}(D))\triangleq c_{C_{BE}}(ENC_{k_{AB}}(D))$

Let all modules be in parallel, then the protocol $A\quad B\quad E$ can be presented by the following process term.

$$\tau_I(\partial_H(\Theta(A\between B\between E)))=\tau_I(\partial_H(A\between B\between E))$$

where $H=\{r_{C_{AB}}(ENC_{k_{AB}}(D)),s_{C_{AB}}(ENC_{k_{AB}}(D)),r_{C_{AE}}(ENC_{k_{AB}}(D)),s_{C_{AE}}(ENC_{k_{AB}}(D)),\\
r_{C_{BE}}(ENC_{k_{AB}}(D)),s_{C_{BE}}(ENC_{k_{AB}}(D))|D\in\Delta\}$,

$I=\{c_{C_{AB}}(ENC_{k_{AB}}(D)),c_{C_{AE}}(ENC_{k_{AB}}(D)),c_{C_{BE}}(ENC_{k_{AB}}(D)),enc_{k_{AB}}(D),dec_{k_{AB}}(ENC_{k_{AB}}(D))|D\in\Delta\}$.

Then we get the following conclusion on the protocol.

\begin{theorem}
The protocol using symmetric keys for secure communication in Figure \ref{CS4} is not secure for replay attack.
\end{theorem}

\begin{proof}
Based on the above state transitions of the above modules, by use of the algebraic laws of $APTC_G$, we can prove that

$\tau_I(\partial_H(A\between B\between E))=\sum_{D\in\Delta}(r_{C_A}(D)\cdot (s_{C_{B}}(D)\parallel s_{C_{B}}(D)))\cdot
\tau_I(\partial_H(A\between B\between E))$.

For the details of proof, please refer to section \ref{app}, and we omit it.

That is, the protocol in Figure \ref{CS4} $\tau_I(\partial_H(A\between B\between E))$ can exhibit undesired external behaviors ($D$ is outputted twice times).

So, The protocol using symmetric keys in Figure \ref{CS4} is not secure for replay attack.
\end{proof}

Generally, in the following chapters, when we introduce the analysis of a security protocol, we will mainly analyze the secure properties related to its design goal.

\newpage\section{Analyses of Key Exchange Protocols}\label{ke}

In this chapter, we will introduce several key exchange protocols, including key exchange protocols with symmetric cryptography in section \ref{kesc} and public key cryptography
in section \ref{kepc}, interlock protocol against
man-in-the-middle attack in section \ref{kepc2}, key exchange protocol with digital signature in section \ref{MIM}, key and message transmission protocol in section \ref{kmt},
and key and message broadcast protocol in section \ref{kmb}.

\subsection{Key Exchange with Symmetric Cryptography}\label{kesc}

The protocol shown in Figure \ref{KESC5} uses symmetric keys for secure communication, that is, the key $k_{AB}$ between Alice and Bob is privately shared to Alice and Bob, and $k_{AB}$
is generated by the Trent, Alice, Bob have shared keys $k_{AT}$ and $k_{BT}$ already. For secure communication, the main challenge is the information leakage to against the confidentiality.

\begin{figure}
    \centering
    \includegraphics{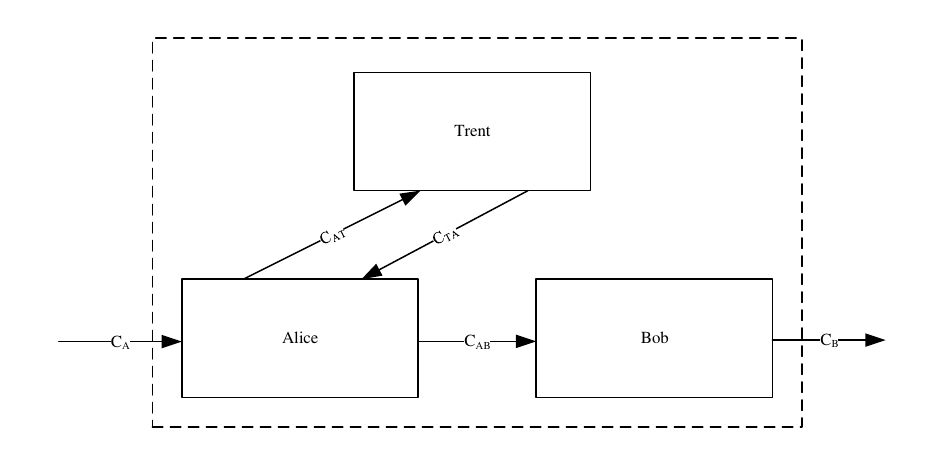}
    \caption{Key exchange protocol with symmetric cryptography}
    \label{KESC5}
\end{figure}

The process of the protocol is as follows.

\begin{enumerate}
  \item Alice receives some messages $D$ from the outside through the channel $C_{A}$ (the corresponding reading action is denoted $r_{C_A}(D)$), if $k_{AB}$ is not established,
  she sends a key request message $M$ to Trent through the channel $C_{AT}$ (the corresponding sending action is denoted $s_{C_{AT}}(M)$);
  \item Trent receives the message $M$ through the channel $C_{AT}$ (the corresponding reading action is denoted $r_{C_{AT}}(M)$), generates a session key $k_{AB}$ through an action
  $rsg_{k_{AB}}$, and encrypts it for Alice and Bob through an action $enc_{k_{AT}}(k_{AB})$ and action $enc_{k_{BT}}(k_{AB})$ respectively, he sends $ENC_{k_{AT}}(k_{AB}),ENC_{k_{BT}}(k_{AB})$
  to the Alice through the channel $C_{TA}$ (the corresponding sending action is denoted $s_{C_{TA}}(ENC_{k_{AT}}(k_{AB}),ENC_{k_{BT}}(k_{AB}))$);
  \item Alice receives $ENC_{k_{AT}}(k_{AB}),ENC_{k_{BT}}(k_{AB})$ from Trent through the channel $C_{TA}$ (the corresponding reading action is denoted $r_{C_{TA}}(ENC_{k_{AT}}(k_{AB}),ENC_{k_{BT}}(k_{AB}))$),
  she decrypts $ENC_{k_{AT}}(k_{AB})$ through an action $dec_{k_{AT}}(ENC_{k_{AT}}(k_{AB}))$ and gets $k_{AB}$, and sends $ENC_{k_{BT}}(k_{AB})$ to Bob through the channel $C_{AB}$
  (the corresponding sending action is denoted $s_{C_{AB}}(ENC_{k_{BT}}(k_{AB}))$);
  \item Bob receives $ENC_{k_{BT}}(k_{AB})$ from Alice through the channel $C_{AB}$ (the corresponding reading action is denoted $r_{C_{AB}}(ENC_{k_{BT}}(k_{AB}))$),
  he decrypts $ENC_{k_{AT}}(k_{AB})$ through an action $dec_{k_{AT}}(ENC_{k_{AT}}(k_{AB}))$ and gets $k_{AB}$, then $k_{AB}$ is established;
  \item If $k_{AB}$ is established, after an encryption processing $enc_{k_{AB}}(D)$, Alice sends $ENC_{k_{AB}}(D)$ to Bob through the channel $C_{AB}$
  (the corresponding sending action is denoted $s_{C_{AB}}(ENC_{k_{AB}}(D))$);
  \item Bob receives the message $ENC_{k_{AB}}(D)$ through the channel $C_{AB}$ (the corresponding reading action is denoted $r_{C_{AB}}(ENC_{k_{AB}}(D))$), after a decryption processing $dec_{k_{AB}}(ENC_{k_{AB}}(D))$,
  he sends $D$ to the outside through the channel $C_B$ (the corresponding sending action is denoted $s_{C_B}(D)$).
\end{enumerate}

Where $D\in\Delta$, $\Delta$ is the set of data.

Alice's state transitions described by $APTC_G$ are as follows.

$A=\sum_{D\in\Delta}r_{C_A}(D)\cdot A_2$

$A_2=\{k_{AB}=NULL\}\cdot s_{C_{AT}}(M)\cdot A_3+ \{k_{AB}\neq NULL\}\cdot A_6$

$A_3=r_{C_{TA}}(ENC_{k_{AT}}(k_{AB}),ENC_{k_{BT}}(k_{AB}))\cdot A_4$

$A_4=dec_{k_{AT}}(ENC_{k_{AT}}(k_{AB}))\cdot A_5$

$A_5=s_{C_{AB}}(ENC_{k_{BT}}(k_{AB}))\cdot A_6$

$A_6=enc_{k_{AB}}(D)\cdot A_7$

$A_7=s_{C_{AB}}(ENC_{k_{AB}}(D))\cdot A$

Bob's state transitions described by $APTC_G$ are as follows.

$B=\{k_{AB}=NULL\}\cdot B_1+ \{k_{AB}\neq NULL\}\cdot B_3$

$B_1=r_{C_{AB}}(ENC_{k_{BT}}(k_{AB}))\cdot B_2$

$B_2=dec_{k_{BT}}(ENC_{k_{AB}}(k_{AB}))\cdot B_3$

$B_3=r_{C_{AB}}(ENC_{k_{AB}}(D))\cdot B_4$

$B_4=dec_{k_{AB}}(ENC_{k_{AB}}(D))\cdot B_5$

$B_5=s_{C_{B}}(D)\cdot B$

Trent's state transitions described by $APTC_G$ are as follows.

$T=r_{C_{AT}}(M)\cdot T_2$

$T_2=rsg_{k_{AB}}\cdot T_3$

$T_3=(enc_{k_{AT}}(k_{AB})\parallel enc_{k_{BT}}(k_{AB}))\cdot T_4$

$T_4=s_{C_{TA}}(ENC_{k_{AT}}(k_{AB}),ENC_{k_{BT}}(k_{AB}))\cdot T$

The sending action and the reading action of the same type data through the same channel can communicate with each other, otherwise, will cause a deadlock $\delta$. We define the following
communication functions.

$\gamma(r_{C_{AT}}(M),s_{C_{AT}}(M))\triangleq c_{C_{AT}}(M)$

$\gamma(r_{C_{TA}}(ENC_{k_{AT}}(k_{AB}),ENC_{k_{BT}}(k_{AB})),s_{C_{TA}}(ENC_{k_{AT}}(k_{AB}),ENC_{k_{BT}}(k_{AB})))\\
\triangleq c_{C_{TA}}(ENC_{k_{AT}}(k_{AB}),ENC_{k_{BT}}(k_{AB}))$

$\gamma(r_{C_{AB}}(ENC_{k_{BT}}(k_{AB})),s_{C_{AB}}(ENC_{k_{BT}}(k_{AB})))\triangleq c_{C_{AB}}(ENC_{k_{BT}}(k_{AB}))$

$\gamma(r_{C_{AB}}(ENC_{k_{AB}}(D)),s_{C_{AB}}(ENC_{k_{AB}}(D))\triangleq c_{C_{AB}}(ENC_{k_{AB}}(D))$

Let all modules be in parallel, then the protocol $A\quad B\quad T$ can be presented by the following process term.

$$\tau_I(\partial_H(\Theta(A\between B\between T)))=\tau_I(\partial_H(A\between B\between T))$$

where $H=\{r_{C_{AT}}(M),s_{C_{AT}}(M),r_{C_{TA}}(ENC_{k_{AT}}(k_{AB}),ENC_{k_{BT}}(k_{AB})),\\
s_{C_{TA}}(ENC_{k_{AT}}(k_{AB}),ENC_{k_{BT}}(k_{AB})),r_{C_{AB}}(ENC_{k_{BT}}(k_{AB})),s_{C_{AB}}(ENC_{k_{BT}}(k_{AB})),\\
r_{C_{AB}}(ENC_{k_{AB}}(D)),s_{C_{AB}}(ENC_{k_{AB}}(D)|D\in\Delta\}$,

$I=\{c_{C_{AT}}(M),c_{C_{TA}}(ENC_{k_{AT}}(k_{AB}),ENC_{k_{BT}}(k_{AB})),c_{C_{AB}}(ENC_{k_{BT}}(k_{AB})),c_{C_{AB}}(ENC_{k_{AB}}(D)),\\
\{k_{AB}=NULL\},\{k_{AB}\neq NULL\},dec_{k_{AT}}(ENC_{k_{AT}}(k_{AB})),enc_{k_{AB}}(D),\\
dec_{k_{BT}}(ENC_{k_{AB}}(k_{AB})),dec_{k_{AB}}(ENC_{k_{AB}}(D)),rsg_{k_{AB}},enc_{k_{AT}}(k_{AB}),enc_{k_{BT}}(k_{AB})|D\in\Delta\}$.

Then we get the following conclusion on the protocol.

\begin{theorem}
The key exchange protocol with symmetric cryptography in Figure \ref{KESC5} is confidential.
\end{theorem}

\begin{proof}
Based on the above state transitions of the above modules, by use of the algebraic laws of $APTC_G$, we can prove that

$\tau_I(\partial_H(A\between B\between T))=\sum_{D\in\Delta}(r_{C_A}(D)\cdot (s_{C_{B}}(D)\parallel s_{C_E}(DEC_{k_{E}}(ENC_{k_{AB}}(D)))))\cdot
\tau_I(\partial_H(A\between B\between T))$.

For the details of proof, please refer to section \ref{app}, and we omit it.

That is, the protocol in Figure \ref{KESC5} $\tau_I(\partial_H(A\between B\between T))$ can exhibit desired external behaviors, and because the key $k_{AB}$ is private,
The protocol using symmetric keys in Figure \ref{KESC5} is confidential and similar to the protocol in section \ref{confi}, and we do not model the information leakage attack.
\end{proof}

\subsection{Key Exchange with Public-Key Cryptography}\label{kepc}

The protocol shown in Figure \ref{KEPC5} uses public keys for secure communication with man-in-the-middle attack, that is, Alice, Bob have shared their public keys $pk_{A}$ and $pk_{B}$ already.
For secure communication, the main challenge is the information leakage to against the confidentiality.

The process of key exchange protocol with public-key cryptography is:

\begin{enumerate}
  \item Alice gets Bob's public key from Trent;
  \item Alice generates a random session key, encrypts it using Bob's public key, and sends to Bob;
  \item Bob receives the encrypted session key, decrypted by his private key, and gets the session key;
  \item Alice and Bob can communicate by use of the session key.
\end{enumerate}

We do not verify the above protocols, and verify the above protocols with man-in-the-middle attack as Figure \ref{KEPC5} shows.

\begin{figure}
    \centering
    \includegraphics{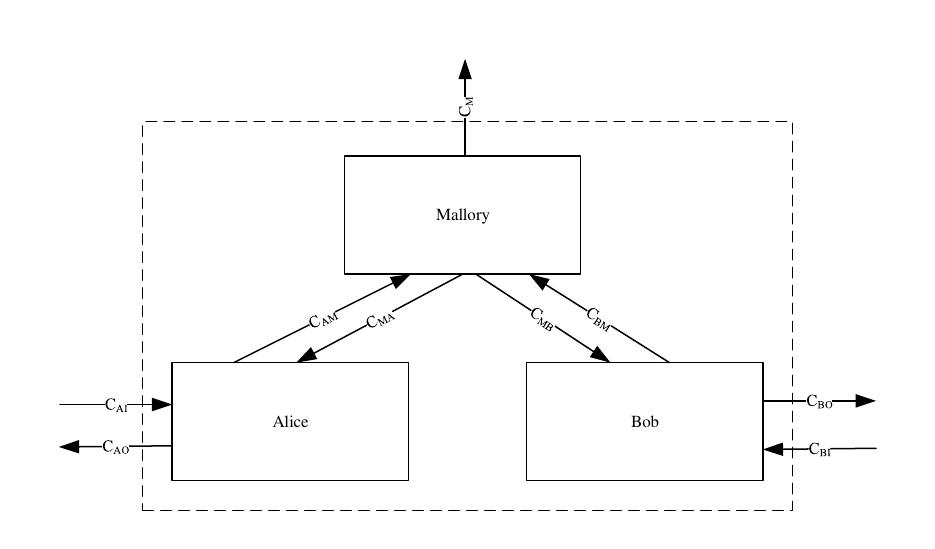}
    \caption{Key exchange protocol with public-key cryptography and man-in-the-middle attack}
    \label{KEPC5}
\end{figure}

The process of the protocol with man-in-the-middle attack is as follows, and we only consider the message in one direction: from Alice to Bob.

\begin{enumerate}
  \item Alice receives some messages $D$ from the outside through the channel $C_{AI}$ (the corresponding reading action is denoted $r_{C_{AI}}(D)$), she sends a key request message $Me$
  to Mallory through the channel $C_{AM}$ (the corresponding sending action is denoted $s_{C_{AM}}(Me)$);
  \item Mallory receives the message $Me$ through the channel $C_{AM}$ (the corresponding reading action is denoted $r_{C_{AM}}(Me)$), he sends $Me$
  to the Bob through the channel $C_{MB}$ (the corresponding sending action is denoted $s_{C_{MB}}(Me)$);
  \item Bob receives the message $Me$ from Mallory through the channel $C_{MB}$ (the corresponding reading action is denoted $r_{C_{MB}}(Me)$), and sends his public key $pk_B$ to Mallory
  through the channel $C_{BM}$ (the corresponding sending action is denoted $s_{C_{BM}}(pk_B)$);
  \item Mallory receives $pk_B$ from Bob through the channel $C_{BM}$ (the corresponding reading action is denoted $r_{C_{BM}}(pk_B)$), then he stores $pk_B$, and sends his public key
  $pk_M$ to Alice through the channel $C_{MA}$ (the corresponding sending action is denoted $s_{C_{MA}}(pk_M)$);
  \item Alice receives $pk_M$ from Mallory through the channel $C_{MA}$ (the corresponding reading action is denoted $r_{C_{MA}}(pk_M)$), she encrypts the message $D$ with Mallory's
  public key $pk_M$ through the action $enc_{pk_M}(D)$, then Alice sends $ENC_{pk_{M}}(D)$ to Mallory through the channel $C_{AM}$
  (the corresponding sending action is denoted $s_{C_{AM}}(ENC_{pk_{M}}(D))$);
  \item Mallory receives $ENC_{pk_M}(D)$ from Alice through the channel $C_{AM}$ (the corresponding reading action is denoted $r_{C_{AM}}(ENC_{pk_M}(D))$), he decrypts the message with
  his private key $sk_M$ through the action $dec_{sk_M}(ENC_{pk_M}(D))$ to get the message $D$, and sends $D$ to the outside through the channel $C_M$ (the corresponding sending action
  is denoted $s_{C_M}(D)$), then he encrypts $D$ with Bob's public key $pk_B$ through the action $enc_{pk_B}(D)$ and sends $ENC_{pk_B}(D)$ to Bob through the channel $C_{MB}$ (the corresponding
  sending action is denoted $s_{C_{MB}}(ENC_{pk_B}(D))$);
  \item Bob receives the message $ENC_{pk_B}(D)$ through the channel $C_{MB}$ (the corresponding reading action is denoted $r_{C_{MB}}(ENC_{pk_B}(D))$), after a decryption processing
  $dec_{sk_{B}}(ENC_{pk_B}(D))$ to get the message $D$, then he sends $D$ to the outside through the channel $C_{BO}$ (the corresponding sending action is denoted $s_{C_{BO}}(D)$).
\end{enumerate}

Where $D\in\Delta$, $\Delta$ is the set of data.

Alice's state transitions described by $APTC_G$ are as follows.

$A=\sum_{D\in\Delta}r_{C_{AI}}(D)\cdot A_2$

$A_2=s_{C_{AM}}(Me)\cdot A_3$

$A_3=r_{C_{MA}}(pk_M)\cdot A_4$

$A_4=enc_{pk_M}(D)\cdot A_5$

$A_5=s_{C_{AM}}(ENC_{pk_{M}}(D))\cdot A$

Bob's state transitions described by $APTC_G$ are as follows.

$B=r_{C_{MB}}(Me)\cdot B_2$

$B_2=s_{C_{BM}}(pk_B)\cdot B_3$

$B_3=r_{C_{MB}}(ENC_{pk_B}(D))\cdot B_4$

$B_4=dec_{sk_{B}}(ENC_{pk_B}(D))\cdot B_5$

$B_5=s_{C_{BO}}(D)\cdot B$

Mallory's state transitions described by $APTC_G$ are as follows.

$Ma=r_{C_{AM}}(Me)\cdot Ma_2$

$Ma_2=s_{C_{MB}}(Me)\cdot Ma_3$

$Ma_3=r_{C_{BM}}(pk_B)\cdot Ma_4$

$Ma_4=s_{C_{MA}}(pk_M)\cdot Ma_5$

$Ma_5=r_{C_{AM}}(ENC_{pk_M}(D))\cdot Ma_6$

$Ma_6=dec_{sk_M}(ENC_{pk_M}(D))\cdot Ma_7$

$Ma_7=s_{C_M}(D)\cdot Ma_8$

$Ma_8=enc_{pk_B}(D)\cdot Ma_9$

$Ma_9=s_{C_{MB}}(ENC_{pk_B}(D))\cdot Ma$

The sending action and the reading action of the same type data through the same channel can communicate with each other, otherwise, will cause a deadlock $\delta$. We define the following
communication functions.

$\gamma(r_{C_{AM}}(Me),s_{C_{AM}}(Me))\triangleq c_{C_{AM}}(Me)$

$\gamma(r_{C_{MB}}(Me),s_{C_{MB}}(Me))\triangleq c_{C_{MB}}(Me)$

$\gamma(r_{C_{BM}}(pk_B),s_{C_{BM}}(pk_B))\triangleq c_{C_{BM}}(pk_B)$

$\gamma(r_{C_{MA}}(pk_M),s_{C_{MA}}(pk_M))\triangleq c_{C_{MA}}(pk_M)$

$\gamma(r_{C_{AM}}(ENC_{pk_M}(D)),s_{C_{AM}}(ENC_{pk_M}(D)))\triangleq c_{C_{AM}}(ENC_{pk_M}(D))$

$\gamma(r_{C_{MB}}(ENC_{pk_B}(D)),s_{C_{MB}}(ENC_{pk_B}(D)))\triangleq c_{C_{MB}}(ENC_{pk_B}(D))$

Let all modules be in parallel, then the protocol $A\quad B\quad Ma$ can be presented by the following process term.

$$\tau_I(\partial_H(\Theta(A\between B\between Ma)))=\tau_I(\partial_H(A\between B\between Ma))$$

where $H=\{r_{C_{AM}}(Me),s_{C_{AM}}(Me),r_{C_{MB}}(Me),s_{C_{MB}}(Me),r_{C_{BM}}(pk_B),s_{C_{BM}}(pk_B),\\
r_{C_{MA}}(pk_M),s_{C_{MA}}(pk_M),r_{C_{AM}}(ENC_{pk_M}(D)),s_{C_{AM}}(ENC_{pk_M}(D)),\\
r_{C_{MB}}(ENC_{pk_B}(D)),s_{C_{MB}}(ENC_{pk_B}(D))|D\in\Delta\}$,

$I=\{c_{C_{AM}}(Me),c_{C_{MB}}(Me),c_{C_{BM}}(pk_B),c_{C_{MA}}(pk_M),c_{C_{AM}}(ENC_{pk_M}(D)),\\
c_{C_{MB}}(ENC_{pk_B}(D)),enc_{pk_M}(D),dec_{sk_{B}}(ENC_{pk_B}(D)),dec_{sk_M}(ENC_{pk_M}(D)),enc_{pk_B}(D)|D\in\Delta\}$.

Then we get the following conclusion on the protocol.

\begin{theorem}
The key exchange protocol with public key cryptography in Figure \ref{KEPC5} is insecure.
\end{theorem}

\begin{proof}
Based on the above state transitions of the above modules, by use of the algebraic laws of $APTC_G$, we can prove that

$\tau_I(\partial_H(A\between B\between Ma))=\sum_{D\in\Delta}(r_{C_{AI}}(D)\cdot s_{C_{M}}(D)\cdot s_{C_{BO}}(D))\cdot
\tau_I(\partial_H(A\between B\between Ma))$.

For the details of proof, please refer to section \ref{app}, and we omit it.

That is, the protocol in Figure \ref{KEPC5} $\tau_I(\partial_H(A\between B\between Ma))$ can exhibit undesired external behaviors, that is, there is an external action $s_{C_{M}}(D)$ while
Alice and Bob do not aware.
\end{proof}

\subsection{Interlock Protocol}\label{kepc2}

The interlock protocol shown in Figure \ref{KEPC25} also uses public keys for secure communication with man-in-the-middle attack, that is, Alice, Bob have shared their public keys $pk_{A}$ and $pk_{B}$ already.
But, the interlock protocol can resist man-in-the-middle attack, that is, Alice and Bob can aware of the existence of the man in the middle.

\begin{figure}
    \centering
    \includegraphics{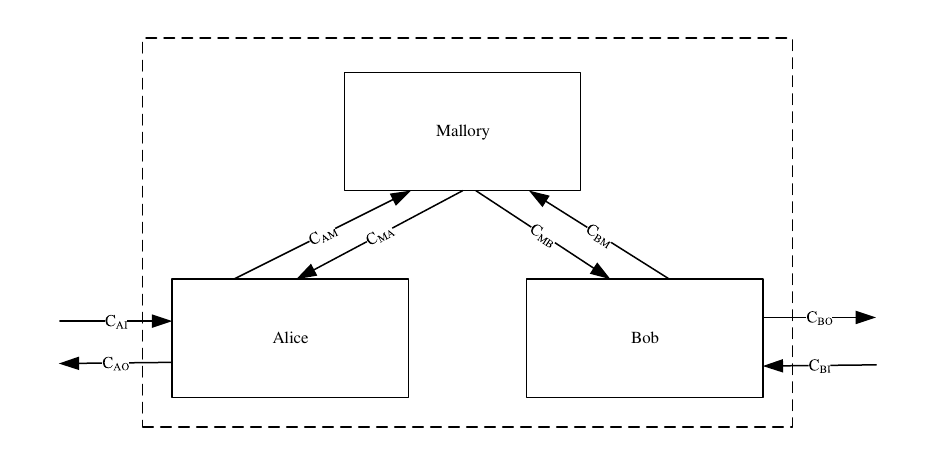}
    \caption{Interlock protocol with man-in-the-middle attack}
    \label{KEPC25}
\end{figure}

The process of the interlock protocol with man-in-the-middle attack is as follows, we assume that Alice has "Bob's" public key $pk_M$, Bob has "Alice's" public key $pk_M$, and Mallory
has Alice's public key $pk_A$ and Bob's public key $pk_B$.

\begin{enumerate}
  \item Alice receives some messages $D_A$ from the outside through the channel $C_{AI}$ (the corresponding reading action is denoted $r_{C_{AI}}(D_A)$), she encrypts the message $D_A$ with Mallory's
  public key $pk_M$ through the action $enc_{pk_M}(D_A)$, then Alice sends the half of $ENC_{pk_{M}}(D_A)$ to Mallory through the channel $C_{AM}$
  (the corresponding sending action is denoted $s_{C_{AM}}(ENC_{pk_{M}}(D_A)/2)$);
  \item Mallory receives $ENC_{pk_M}(D_A)/2$ from Alice through the channel $C_{AM}$ (the corresponding reading action is denoted $r_{C_{AM}}(ENC_{pk_M}(D_A)/2)$), he can not decrypt the message with
  his private key $sk_M$, and has to make another message $D_A'$ and encrypt $D_A'$ with Bob's public key $pk_B$ through the action $enc_{pk_B}(D_A')$, and sends the half of $ENC_{pk_B}(D_A')$
  to Bob through the channel $C_{MB}$ (the corresponding sending action is denoted $s_{C_{MB}}(ENC_{pk_B}(D_A')/2)$);
  \item Bob receives the message $ENC_{pk_B}(D_A')/2$ through the channel $C_{MB}$ (the corresponding reading action is denoted $r_{C_{MB}}(ENC_{pk_B}(D_A')/2)$), and receives some message
  $D_B$ from the outside through the channel $C_{BI}$ (the corresponding reading action is denoted $r_{BI}(D_B)$), after an encryption processing
  $enc_{pk_{M}}(D_B)$ to get the message $ENC_{pk_M}(D_B)$, then he sends the half of $ENC_{pk_M}(D_B)$ to Mallory through the channel $C_{BM}$
  (the corresponding sending action is denoted $s_{C_{BM}}(ENC_{pk_{M}}(D_B)/2)$);
  \item Mallory receives $ENC_{pk_M}(D_B)/2$ from Bob through the channel $C_{BM}$ (the corresponding reading action is denoted $r_{C_{BM}}(ENC_{pk_M}(D_B)/2)$), he can not decrypt the message with
  his private key $sk_M$, and has to make another message $D_B'$ and encrypt $D_B'$ with Alice's public key $pk_A$ through the action $enc_{pk_A}(D_B')$, and sends the half of $ENC_{pk_A}(D_B')$
  to Alice through the channel $C_{MA}$ (the corresponding sending action is denoted $s_{C_{MA}}(ENC_{pk_A}(D_B')/2)$);
  \item Alice receives the message $ENC_{pk_A}(D_B')/2$ through the channel $C_{MA}$ (the corresponding reading action is denoted $r_{C_{MA}}(ENC_{pk_A}(D_B')/2)$), and sends the other half of $ENC_{pk_{M}}(D_A)$ to Mallory through the channel $C_{AM}$
  (the corresponding sending action is denoted \\$s_{C_{AM}}(ENC_{pk_{M}}(D_A)/2)$);
  \item Mallory receives $ENC_{pk_M}(D_A)/2$ from Alice through the channel $C_{AM}$ (the corresponding reading action is denoted $r_{C_{AM}}(ENC_{pk_M}(D_A)/2)$), he can combine the two half
  of $ENC_{pk_M}(D_A)/2$ and decrypt the message with
  his private key $sk_M$, and but he has to send the other half of $ENC_{pk_B}(D_A')$
  to Bob through the channel $C_{MB}$ (the corresponding sending action is denoted $s_{C_{MB}}(ENC_{pk_B}(D_A')/2)$);
  \item Bob receives the message $ENC_{pk_B}(D_A')/2$ through the channel $C_{MB}$ (the corresponding reading action is denoted $r_{C_{MB}}(ENC_{pk_B}(D_A')/2)$), after a combination of
  two half of $ENC_{pk_B}(D_A')$ and a decryption processing
  $dec_{sk_B}(ENC_{pk_{B}}(D_A'))$ to get the message $D_A'$, then he sends it to the outside through the channel $C_{BO}$ (the corresponding sending action is denoted $s_{C_{BO}}(D_A')$).
  Then he sends the other half of $ENC_{pk_M}(D_B)$ to Mallory through the channel $C_{BM}$
  (the corresponding sending action is denoted $s_{C_{BM}}(ENC_{pk_{M}}(D_B)/2)$);
  \item Mallory receives $ENC_{pk_M}(D_B)/2$ from Bob through the channel $C_{BM}$ (the corresponding reading action is denoted $r_{C_{BM}}(ENC_{pk_M}(D_B)/2)$), he can combine the two half
  of $ENC_{pk_M}(D_B)/2$ and decrypt the message with
  his private key $sk_M$, and but he has to send the other half of $ENC_{pk_A}(D_B')$
  to Alice through the channel $C_{MA}$ (the corresponding sending action is denoted $s_{C_{MA}}(ENC_{pk_A}(D_B')/2)$);
  \item Alice receives the message $ENC_{pk_A}(D_B')/2$ through the channel $C_{MA}$ (the corresponding reading action is denoted $r_{C_{MA}}(ENC_{pk_A}(D_B')/2)$), after a combination of
  two half of $ENC_{pk_A}(D_B')$ and a decryption processing
  $dec_{sk_A}(ENC_{pk_{A}}(D_B'))$ to get the message $D_B'$, then she sends it to the outside through the channel $C_{AO}$ (the corresponding sending action is denoted $s_{C_{AO}}(D_B')$).
\end{enumerate}

Where $D\in\Delta$, $\Delta$ is the set of data.

Alice's state transitions described by $APTC_G$ are as follows.

$A=\sum_{D_A\in\Delta}r_{C_{AI}}(D_A)\cdot A_2$

$A_2=enc_{pk_M}(D_A)\cdot A_3$

$A_3=s_{C_{AM}}(ENC_{pk_{M}}(D_A)/2)\cdot A_4$

$A_4=r_{C_{MA}}(ENC_{pk_A}(D_B')/2)\cdot A_5$

$A_5=s_{C_{AM}}(ENC_{pk_{M}}(D_A)/2)\cdot A_6$

$A_6=r_{C_{MA}}(ENC_{pk_A}(D_B')/2)\cdot A_7$

$A_7=dec_{sk_A}(ENC_{pk_{A}}(D_B'))\cdot A_8$

$A_8=s_{C_{AO}}(D_B')\cdot A$

Bob's state transitions described by $APTC_G$ are as follows.

$B=r_{C_{MB}}(ENC_{pk_B}(D_A')/2)\cdot B_2$

$B_2=\sum_{D_B\in\Delta}r_{BI}(D_B)\cdot B_3$

$B_3=enc_{pk_{M}}(D_B)\cdot B_4$

$B_4=s_{C_{BM}}(ENC_{pk_{M}}(D_B)/2)\cdot B_5$

$B_5=r_{C_{MB}}(ENC_{pk_B}(D_A')/2)\cdot B_6$

$B_6=dec_{sk_B}(ENC_{pk_{B}}(D_A'))\cdot B_7$

$B_7=s_{C_{BO}}(D_A')\cdot B_8$

$B_8=s_{C_{BM}}(ENC_{pk_{M}}(D_B)/2)\cdot B$

Mallory's state transitions described by $APTC_G$ are as follows.

$Ma=r_{C_{AM}}(ENC_{pk_M}(D_A)/2)\cdot Ma_2$

$Ma_2=enc_{pk_B}(D_A')\cdot Ma_3$

$Ma_3=s_{C_{MB}}(ENC_{pk_B}(D_A')/2)\cdot Ma_4$

$Ma_4=r_{C_{BM}}(ENC_{pk_M}(D_B)/2)\cdot Ma_5$

$Ma_5=enc_{pk_A}(D_B')\cdot Ma_6$

$Ma_6=s_{C_{MA}}(ENC_{pk_A}(D_B')/2)\cdot Ma_7$

$Ma_7=r_{C_{AM}}(ENC_{pk_M}(D_A)/2)\cdot Ma_8$

$Ma_8=s_{C_{MB}}(ENC_{pk_B}(D_A')/2)\cdot Ma_9$

$Ma_9=r_{C_{BM}}(ENC_{pk_M}(D_B)/2)\cdot Ma_{10}$

$Ma_{10}=s_{C_{MA}}(ENC_{pk_A}(D_B')/2)\cdot Ma$

The sending action and the reading action of the same type data through the same channel can communicate with each other, otherwise, will cause a deadlock $\delta$. We define the following
communication functions.

$\gamma(r_{C_{AM}}(ENC_{pk_M}(D_A)/2),s_{C_{AM}}(ENC_{pk_M}(D_A)/2))\triangleq c_{C_{AM}}(ENC_{pk_M}(D_A)/2)$

$\gamma(r_{C_{MB}}(ENC_{pk_B}(D_A')/2),s_{C_{MB}}(ENC_{pk_B}(D_A')/2))\triangleq c_{C_{MB}}(ENC_{pk_B}(D_A')/2)$

$\gamma(r_{C_{BM}}(ENC_{pk_M}(D_B)/2),s_{C_{BM}}(ENC_{pk_M}(D_B)/2))\triangleq c_{C_{BM}}(ENC_{pk_M}(D_B)/2)$

$\gamma(r_{C_{MA}}(ENC_{pk_A}(D_B')/2),s_{C_{MA}}(ENC_{pk_A}(D_B')/2))\triangleq c_{C_{MA}}(ENC_{pk_A}(D_B')/2)$

$\gamma(r_{C_{AM}}(ENC_{pk_M}(D_A)/2),s_{C_{AM}}(ENC_{pk_M}(D_A)/2))\triangleq c_{C_{AM}}(ENC_{pk_M}(D_A)/2)$

$\gamma(r_{C_{MB}}(ENC_{pk_B}(D_A')/2),s_{C_{MB}}(ENC_{pk_B}(D_A')/2))\triangleq c_{C_{MB}}(ENC_{pk_B}(D_A')/2)$

$\gamma(r_{C_{BM}}(ENC_{pk_M}(D_B)/2),s_{C_{BM}}(ENC_{pk_M}(D_B)/2))\triangleq c_{C_{BM}}(ENC_{pk_M}(D_B)/2)$

$\gamma(r_{C_{MA}}(ENC_{pk_A}(D_B')/2),s_{C_{MA}}(ENC_{pk_A}(D_B')/2))\triangleq c_{C_{MA}}(ENC_{pk_A}(D_B')/2)$

Let all modules be in parallel, then the protocol $A\quad B\quad Ma$ can be presented by the following process term.

$$\tau_I(\partial_H(\Theta(A\between B\between Ma)))=\tau_I(\partial_H(A\between B\between Ma))$$

where $H=\{r_{C_{AM}}(ENC_{pk_M}(D_A)/2),s_{C_{AM}}(ENC_{pk_M}(D_A)/2),r_{C_{MB}}(ENC_{pk_B}(D_A')/2),\\
s_{C_{MB}}(ENC_{pk_B}(D_A')/2),r_{C_{BM}}(ENC_{pk_M}(D_B)/2),s_{C_{BM}}(ENC_{pk_M}(D_B)/2),\\
r_{C_{MA}}(ENC_{pk_A}(D_B')/2),s_{C_{MA}}(ENC_{pk_A}(D_B')/2),r_{C_{AM}}(ENC_{pk_M}(D_A)/2),\\
s_{C_{AM}}(ENC_{pk_M}(D_A)/2),r_{C_{MB}}(ENC_{pk_B}(D_A')/2),s_{C_{MB}}(ENC_{pk_B}(D_A')/2),\\
r_{C_{BM}}(ENC_{pk_M}(D_B)/2),s_{C_{BM}}(ENC_{pk_M}(D_B)/2),r_{C_{MA}}(ENC_{pk_A}(D_B')/2),\\
s_{C_{MA}}(ENC_{pk_A}(D_B')/2)|D_A,D_B,D_A',D_B'\in\Delta\}$,

$I=\{c_{C_{AM}}(ENC_{pk_M}(D_A)/2),c_{C_{MB}}(ENC_{pk_B}(D_A')/2),c_{C_{BM}}(ENC_{pk_M}(D_B)/2),c_{C_{MA}}(ENC_{pk_A}(D_B')/2),\\
c_{C_{AM}}(ENC_{pk_M}(D_A)/2),c_{C_{MB}}(ENC_{pk_B}(D_A')/2),c_{C_{BM}}(ENC_{pk_M}(D_B)/2),c_{C_{MA}}(ENC_{pk_A}(D_B')/2),\\
enc_{pk_M}(D_A),dec_{sk_A}(ENC_{pk_{A}}(D_B')),enc_{pk_{M}}(D_B),dec_{sk_B}(ENC_{pk_{B}}(D_A')),enc_{pk_B}(D_A'),\\
enc_{pk_A}(D_B')|D_A,D_B,D_A',D_B'\in\Delta\}$.

Then we get the following conclusion on the protocol.

\begin{theorem}
The interlock protocol with public key cryptography in Figure \ref{KEPC25} is secure.
\end{theorem}

\begin{proof}
Based on the above state transitions of the above modules, by use of the algebraic laws of $APTC_G$, we can prove that

$\tau_I(\partial_H(A\between B\between Ma))=\sum_{D_A,D_B,D_A',D_B'\in\Delta}(r_{C_{AI}}(D_A)\cdot r_{C_{BI}}(D_B)\cdot s_{C_{BO}}(D_A')\cdot s_{C_{AO}}(D_B'))\cdot
\tau_I(\partial_H(A\between B\between Ma))$.

For the details of proof, please refer to section \ref{app}, and we omit it.

That is, the interlock protocol in Figure \ref{KEPC25} $\tau_I(\partial_H(A\between B\between Ma))$ can exhibit desired external behaviors, that is, Alice and Bob can aware the existence of the
man in the middle.
\end{proof}

\subsection{Key Exchange with Digital Signatures}\label{MIM}

The protocol shown in Figure \ref{KEDS5} uses digital signature for secure communication with man-in-the-middle attack, that is, Alice, Bob have shared their public keys $pk_{A}$ and $pk_{B}$,
and the public keys are signed by the Trent: $SIGN_{sk_T}(A,pk_A)$, $SIGN_{sk_T}(B,pk_B)$ and $SIGN_{sk_T}(M,pk_M)$. Note that, Trent's public key $pk_T$ is well-known.
And also, the key exchange protocol with digital signature can resist man-in-the-middle attack, that is, Alice and Bob can aware of the existence of the man in the middle.

\begin{figure}
    \centering
    \includegraphics{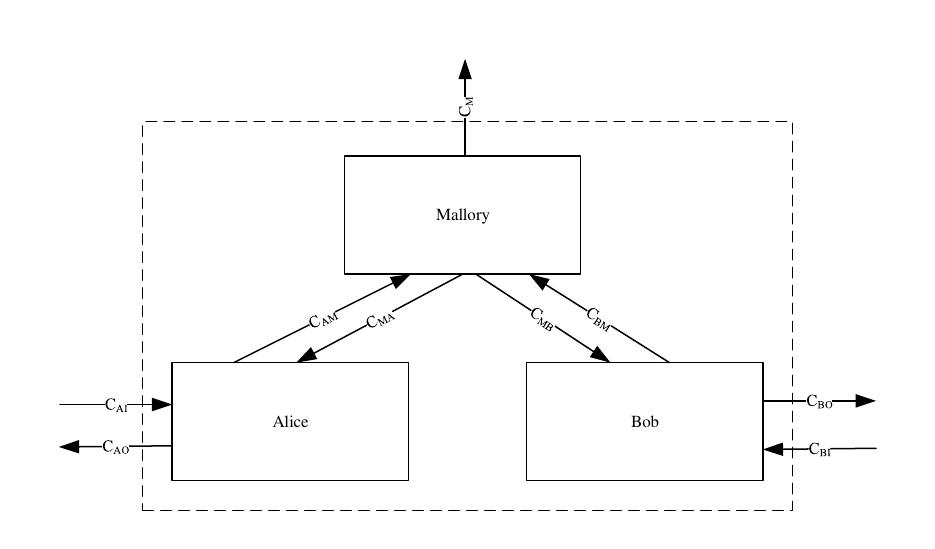}
    \caption{Key exchange protocol with digital signature and man-in-the-middle attack}
    \label{KEDS5}
\end{figure}

The process of the protocol with man-in-the-middle attack is as follows, and we only consider the message in one direction: from Alice to Bob.

\begin{enumerate}
  \item Alice receives some messages $D$ from the outside through the channel $C_{AI}$ (the corresponding reading action is denoted $r_{C_{AI}}(D)$), she sends a key request message $Me$
  to Mallory through the channel $C_{AM}$ (the corresponding sending action is denoted $s_{C_{AM}}(Me)$);
  \item Mallory receives the message $Me$ through the channel $C_{AM}$ (the corresponding reading action is denoted $r_{C_{AM}}(Me)$), he sends $Me$
  to the Bob through the channel $C_{MB}$ (the corresponding sending action is denoted $s_{C_{MB}}(Me)$);
  \item Bob receives the message $Me$ from Mallory through the channel $C_{MB}$ (the corresponding reading action is denoted $r_{C_{MB}}(Me)$), and sends his signed public key $SIGN_{sk_T}(B,pk_B)$ to Mallory
  through the channel $C_{BM}$ (the corresponding sending action is denoted \\$s_{C_{BM}}(SIGN_{sk_T}(B,pk_B))$);
  \item Mallory receives $SIGN_{sk_T}(B,pk_B)$ from Bob through the channel $C_{BM}$ (the corresponding reading action is denoted $r_{C_{BM}}(SIGN_{sk_T}(B,pk_B))$), he can get $pk_B$, then he sends his signed public key
  $SIGN_{sk_T}(M,pk_M)$ or $SIGN_{sk_T}(B,pk_B)$ to Alice through the channel $C_{MA}$ (the corresponding sending action is denoted $s_{C_{MA}}(SIGN_{sk_T}(M,pk_M))$);
  \item Alice receives $SIGN_{sk_T}(M,pk_M)$ or $SIGN_{sk_T}(B,pk_B)$ from Mallory through the channel $C_{MA}$ (the corresponding reading action is denoted $r_{C_{MA}}(SIGN_{sk_T}(d_1,d_2))$), she de-sign this
  message using Trent's public key $pk_T$ through the action $de\textrm{-}sign_{pk_T}(SIGN_{sk_T}(d_1,d_2))$, if $d_1=B$: she encrypts the message $D$ with Bob's public key $pk_B$ through the action $enc_{pk_B}(D)$,
  then Alice sends $ENC_{pk_{B}}(D)$ to Mallory through the channel $C_{AM}$ (the corresponding sending action is denoted $s_{C_{AM}}(ENC_{pk_{B}}(D))$); if $d_1\neq B$, she encrypts
  the message $\bot$ (a special meaningless message) with Mallory's public key $pk_M$ through the action $enc_{pk_M}(\bot)$, then Alice sends $ENC_{pk_{M}}(\bot)$ to Mallory through the channel $C_{AM}$
  (the corresponding sending action is denoted $s_{C_{AM}}(ENC_{pk_{M}}(\bot))$);
  \item Mallory receives $ENC_{pk_M}(d_3)$ from Alice through the channel $C_{AM}$ (the corresponding reading action is denoted $r_{C_{AM}}(ENC_{pk_M}(d_3))$), he decrypts the message with
  his private key $sk_M$ through the action $dec_{sk_M}(ENC_{pk_M}(d_3))$ to get the message $d_3$ (maybe $\bot$ or another meaningless data, all denoted $\bot$), and sends $\bot$ to
  the outside through the channel $C_M$ (the corresponding sending action is denoted $s_{C_M}(\bot)$), then he encrypts $\bot$ with Bob's public key $pk_B$ through the action $enc_{pk_B}(\bot)$ and
  sends $ENC_{pk_B}(\bot)$ to Bob through the channel $C_{MB}$ (the corresponding sending action is denoted $s_{C_{MB}}(ENC_{pk_B}(\bot))$);
  \item Bob receives the message $ENC_{pk_B}(\bot)$ through the channel $C_{MB}$ (the corresponding reading action is denoted $r_{C_{MB}}(ENC_{pk_B}(\bot))$), after a decryption processing
  $dec_{sk_{B}}(ENC_{pk_B}(\bot))$ to get the message $\bot$, then he sends $\bot$ to the outside through the channel $C_{BO}$ (the corresponding sending action is denoted $s_{C_{BO}}(\bot)$).
\end{enumerate}

Where $D\in\Delta$, $\Delta$ is the set of data.

Alice's state transitions described by $APTC_G$ are as follows.

$A=\sum_{D\in\Delta}r_{C_{AI}}(D)\cdot A_2$

$A_2=s_{C_{AM}}(Me)\cdot A_3$

$A_3=r_{C_{MA}}(SIGN_{sk_T}(d_1,d_2))\cdot A_4$

$A_4=de\textrm{-}sign_{pk_T}(SIGN_{sk_T}(d_1,d_2))\cdot A_5$

$A_5=(\{d_1=B\}\cdot enc_{pk_{B}}(D)\cdot s_{C_{AM}}(ENC_{pk_{B}}(D))+\{d_1\neq B\}\cdot enc_{pk_M}(\bot)\cdot s_{C_{AM}}(ENC_{pk_{M}}(\bot)))\cdot A$

Bob's state transitions described by $APTC_G$ are as follows.

$B=r_{C_{MB}}(Me)\cdot B_2$

$B_2=s_{C_{BM}}(SIGN_{sk_T}(B,pk_B))\cdot B_3$

$B_3=r_{C_{MB}}(ENC_{pk_B}(\bot))\cdot B_4$

$B_4=dec_{sk_{B}}(ENC_{pk_B}(\bot))\cdot B_5$

$B_5=s_{C_{BO}}(\bot)\cdot B$

Mallory's state transitions described by $APTC_G$ are as follows.

$Ma=r_{C_{AM}}(Me)\cdot Ma_2$

$Ma_2=s_{C_{MB}}(Me)\cdot Ma_3$

$Ma_3=r_{C_{BM}}(SIGN_{sk_T}(B,pk_B))\cdot Ma_4$

$Ma_4=s_{C_{MA}}(SIGN_{sk_T}(M,pk_M))\cdot Ma_5$

$Ma_5=r_{C_{AM}}(ENC_{pk_M}(d_3))\cdot Ma_6$

$Ma_6=dec_{sk_M}(ENC_{pk_M}(d_3))\cdot Ma_7$

$Ma_7=s_{C_M}(\bot)\cdot Ma_8$

$Ma_8=enc_{pk_B}(\bot)\cdot Ma_9$

$Ma_9=s_{C_{MB}}(ENC_{pk_B}(\bot))\cdot Ma$

The sending action and the reading action of the same type data through the same channel can communicate with each other, otherwise, will cause a deadlock $\delta$. We define the following
communication functions.

$\gamma(r_{C_{AM}}(Me),s_{C_{AM}}(Me))\triangleq c_{C_{AM}}(Me)$

$\gamma(r_{C_{MB}}(Me),s_{C_{MB}}(Me))\triangleq c_{C_{MB}}(Me)$

$\gamma(r_{C_{BM}}(SIGN_{sk_T}(B,pk_B)),s_{C_{BM}}(SIGN_{sk_T}(B,pk_B)))\triangleq c_{C_{BM}}(SIGN_{sk_T}(B,pk_B))$

$\gamma(r_{C_{MA}}(SIGN_{sk_T}(M,pk_M)),s_{C_{MA}}(SIGN_{sk_T}(M,pk_M)))\triangleq c_{C_{MA}}(SIGN_{sk_T}(M,pk_M))$

$\gamma(r_{C_{AM}}(ENC_{pk_M}(d_3)),s_{C_{AM}}(ENC_{pk_M}(d_3)))\triangleq c_{C_{AM}}(ENC_{pk_M}(d_3))$

$\gamma(r_{C_{MB}}(ENC_{pk_B}(\bot)),s_{C_{MB}}(ENC_{pk_B}(\bot)))\triangleq c_{C_{MB}}(ENC_{pk_B}(\bot))$

Let all modules be in parallel, then the protocol $A\quad B\quad Ma$ can be presented by the following process term.

$$\tau_I(\partial_H(\Theta(A\between B\between Ma)))=\tau_I(\partial_H(A\between B\between Ma))$$

where $H=\{r_{C_{AM}}(Me),s_{C_{AM}}(Me),r_{C_{MB}}(Me),s_{C_{MB}}(Me),r_{C_{BM}}(SIGN_{sk_T}(B,pk_B)),\\
s_{C_{BM}}(SIGN_{sk_T}(B,pk_B)),r_{C_{MA}}(SIGN_{sk_T}(M,pk_M)),s_{C_{MA}}(SIGN_{sk_T}(M,pk_M)),\\
r_{C_{AM}}(ENC_{pk_M}(d_3)),s_{C_{AM}}(ENC_{pk_M}(d_3)),r_{C_{MB}}(ENC_{pk_B}(\bot)),s_{C_{MB}}(ENC_{pk_B}(\bot))|D\in\Delta\}$,

$I=\{c_{C_{AM}}(Me),c_{C_{MB}}(Me),c_{C_{BM}}(SIGN_{sk_T}(B,pk_B)),c_{C_{MA}}(SIGN_{sk_T}(M,pk_M)),\\
c_{C_{AM}}(ENC_{pk_M}(d_3)),c_{C_{MB}}(ENC_{pk_B}(\bot)),de\textrm{-}sign_{pk_T}(SIGN_{sk_T}(d_1,d_2)),\\
\{d_1=B\}, ENC_{pk_{B}}(D),\{d_1\neq B\}, enc_{pk_M}(\bot),dec_{sk_{B}}(ENC_{pk_B}(\bot)),\\
dec_{sk_M}(ENC_{pk_M}(d_3)),enc_{pk_B}(\bot)|D\in\Delta\}$.

Then we get the following conclusion on the protocol.

\begin{theorem}
The key exchange protocol with digital signature in Figure \ref{KEDS5} is secure.
\end{theorem}

\begin{proof}
Based on the above state transitions of the above modules, by use of the algebraic laws of $APTC_G$, we can prove that

$\tau_I(\partial_H(A\between B\between Ma))=\sum_{D\in\Delta}(r_{C_{AI}}(D)\cdot s_{C_{M}}(\bot)\cdot s_{C_{BO}}(\bot))\cdot
\tau_I(\partial_H(A\between B\between Ma))$.

For the details of proof, please refer to section \ref{app}, and we omit it.

That is, the protocol in Figure \ref{KEDS5} $\tau_I(\partial_H(A\between B\between Ma))$ can exhibit desired external behaviors, that is, Alice and Bob can aware the existence of the
man in the middle.
\end{proof}

\subsection{Key and Message Transmission}\label{kmt}

The protocol shown in Figure \ref{KMT5} uses digital signature for secure communication, that is, Alice, Bob have shared their public keys $pk_{A}$ and $pk_{B}$,
and the public keys are signed by the Trent: $SIGN_{sk_T}(A,pk_A)$, $SIGN_{sk_T}(B,pk_B)$. Note that, Trent's public key $pk_T$ is well-known.
There is not a session key exchange process before the message is transferred.

\begin{figure}
    \centering
    \includegraphics{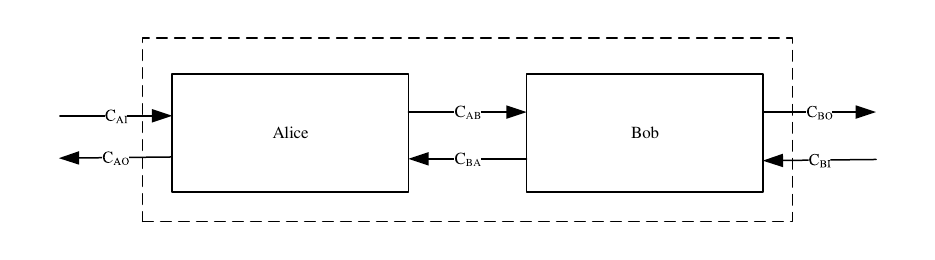}
    \caption{Key and message transmission protocol}
    \label{KMT5}
\end{figure}

The process of the protocol is as follows, and we only consider the message in one direction: from Alice to Bob.

\begin{enumerate}
  \item Alice receives some messages $D$ from the outside through the channel $C_{AI}$ (the corresponding reading action is denoted $r_{C_{AI}}(D)$), she has $SIGN_{sk_T}(B,pk_B)$, she \\
  $de\textrm{-}sign_{pk_T}(SIGN_{sk_T}(B,pk_B))$ and gets $pk_B$, then generate a session key $k_{AB}$ through an action $rsg_{k_AB}$, and she encrypts the message $D$ with $k_{AB}$
  through an action $enc_{k_{AB}}(D)$ and encrypts $k_{AB}$ with Bob's public key $pk_B$ through the action $enc_{pk_B}(k_{AB})$,
  then Alice sends $ENC_{pk_{B}}(k_{AB}),ENC_{k_{AB}}(D)$ to Bob through the channel $C_{AB}$ (the corresponding sending action is denoted $s_{C_{AB}}(ENC_{pk_{B}}(k_{AB}),ENC_{k_{AB}}(D))$);
  \item Bob receives the message $ENC_{pk_{B}}(k_{AB}),ENC_{k_{AB}}(D)$ through the channel $C_{AB}$ (the corresponding reading action is denoted
  $r_{C_{AB}}(ENC_{pk_{B}}(k_{AB}),ENC_{k_{AB}}(D))$), after a decryption processing $dec_{sk_{B}}(ENC_{pk_B}(k_{AB}))$ to get the message $k_{AB}$ and a decryption processing
  $dec_{k_{AB}}(ENC_{k_{AB}}(D))$ to get $D$, then he sends $D$ to the outside through the channel $C_{BO}$ (the corresponding sending action is denoted $s_{C_{BO}}(D)$).
\end{enumerate}

Where $D\in\Delta$, $\Delta$ is the set of data.

Alice's state transitions described by $APTC_G$ are as follows.

$A=\sum_{D\in\Delta}r_{C_{AI}}(D)\cdot A_2$

$A_2=de\textrm{-}sign_{pk_T}(SIGN_{sk_T}(B,pk_B))\cdot A_3$

$A_3=rsg_{k_{AB}}\cdot A_4$

$A_4=enc_{k_{AB}}(D)\cdot A_5$

$A_5=enc_{pk_B}(k_{AB})\cdot A_6$

$A_6=s_{C_{AB}}(ENC_{pk_{B}}(k_{AB}),ENC_{k_{AB}}(D))\cdot A$

Bob's state transitions described by $APTC_G$ are as follows.

$B=r_{C_{AB}}(ENC_{pk_{B}}(k_{AB}),ENC_{k_{AB}}(D))\cdot B_2$

$B_2=dec_{sk_B}(ENC_{pk_{B}}(k_{AB}))\cdot B_3$

$B_3=dec_{k_{AB}}(ENC_{k_{AB}}(D))\cdot B_4$

$B_4=s_{C_{BO}}(D)\cdot B$

The sending action and the reading action of the same type data through the same channel can communicate with each other, otherwise, will cause a deadlock $\delta$. We define the following
communication functions.

$\gamma(r_{C_{AB}}(ENC_{pk_{B}}(k_{AB}),ENC_{k_{AB}}(D)),s_{C_{AB}}(ENC_{pk_{B}}(k_{AB}),ENC_{k_{AB}}(D)))\\\triangleq c_{C_{AB}}(ENC_{pk_{B}}(k_{AB}),ENC_{k_{AB}}(D))$

Let all modules be in parallel, then the protocol $A\quad B$ can be presented by the following process term.

$$\tau_I(\partial_H(\Theta(A\between B)))=\tau_I(\partial_H(A\between B))$$

where $H=\{r_{C_{AB}}(ENC_{pk_{B}}(k_{AB}),ENC_{k_{AB}}(D)),s_{C_{AB}}(ENC_{pk_{B}}(k_{AB}),ENC_{k_{AB}}(D))|D\in\Delta\}$,

$I=\{c_{C_{AB}}(ENC_{pk_{B}}(k_{AB}),ENC_{k_{AB}}(D)),de\textrm{-}sign_{pk_T}(SIGN_{sk_T}(B,pk_B)),rsg_{k_{AB}},\\
enc_{k_{AB}}(D),enc_{pk_B}(k_{AB}),dec_{sk_B}(ENC_{pk_{B}}(k_{AB})),dec_{k_{AB}}(ENC_{k_{AB}}(D))|D\in\Delta\}$.

Then we get the following conclusion on the protocol.

\begin{theorem}
The key and message transmission protocol with digital signature in Figure \ref{KMT5} is secure.
\end{theorem}

\begin{proof}
Based on the above state transitions of the above modules, by use of the algebraic laws of $APTC_G$, we can prove that

$\tau_I(\partial_H(A\between B))=\sum_{D\in\Delta}(r_{C_{AI}}(D)\cdot s_{C_{BO}}(D))\cdot
\tau_I(\partial_H(A\between B))$.

For the details of proof, please refer to section \ref{app}, and we omit it.

That is, the protocol in Figure \ref{KMT5} $\tau_I(\partial_H(A\between B))$ can exhibit desired external behaviors, and similarly to the protocol in subsection \ref{MIM}, this protocol can
resist the man-in-the-middle attack.
\end{proof}

\subsection{Key and Message Broadcast}\label{kmb}

The protocol shown in Figure \ref{KMB5} uses digital signature for secure broadcast communication, that is, Alice, Bob, Carol, and Dave have shared their public keys $pk_{A}$, $pk_{B}$, $pk_C$ and $pk_{Da}$
and the public keys are signed by the Trent: $SIGN_{sk_T}(A,pk_A)$, $SIGN_{sk_T}(B,pk_B)$, $SIGN_{sk_T}(C,pk_C)$ and $SIGN_{sk_T}(Da,pk_{Da})$. Note that, Trent's public key $pk_T$ is well-known.
There is not a session key exchange process before the message is transferred.

\begin{figure}
    \centering
    \includegraphics{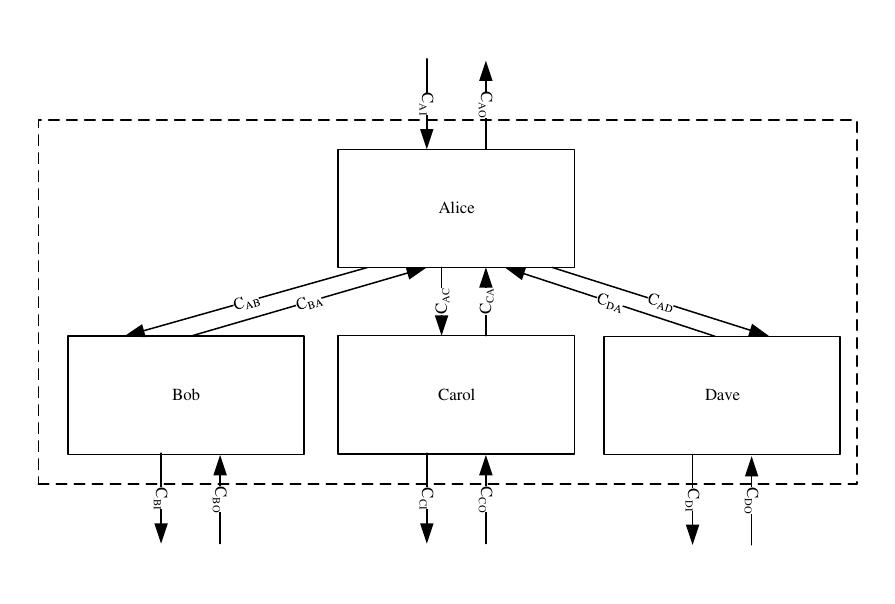}
    \caption{Key and message broadcast protocol}
    \label{KMB5}
\end{figure}

The process of the protocol is as follows, and we only consider the message in one direction: from Alice to Bob, Carol and Dave.

\begin{enumerate}
  \item Alice receives some messages $D$ from the outside through the channel $C_{AI}$ (the corresponding reading action is denoted $r_{C_{AI}}(D)$), she has $SIGN_{sk_T}(B,pk_B)$,
  $SIGN_{sk_T}(C,pk_C)$, and $SIGN_{sk_T}(Da,pk_{Da})$, she $de\textrm{-}sign_{pk_T}(SIGN_{sk_T}(B,pk_B))$ and gets $pk_B$, \\
  $de\textrm{-}sign_{pk_T}(SIGN_{sk_T}(C,pk_C))$ and gets $pk_C$,
  $de\textrm{-}sign_{pk_T}(SIGN_{sk_T}(Da,pk_{Da}))$ and gets $pk_{Da}$, then generate a session key $k$ through an action $rsg_{k}$, and she encrypts the message $D$ with $k$
  through an action $enc_{k}(D)$ and encrypts $k$ with Bob's public key $pk_B$ through the action $enc_{pk_B}(k_{AB})$, Carol's public key $pk_C$ through the action $enc_{pk_C}(k)$,
  Dave's public key $pk_{Da}$ through the action $enc_{pk_{Da}}(k)$,
  then Alice sends $ENC_{pk_{B}}(k),ENC_{k}(D)$ to Bob through the channel $C_{AB}$ (the corresponding sending action is denoted \\
  $s_{C_{AB}}(ENC_{pk_{B}}(k),ENC_{k}(D))$),
  sends $ENC_{pk_{C}}(k),ENC_{k}(D)$ to Bob through the channel $C_{AC}$ (the corresponding sending action is denoted $s_{C_{AC}}(ENC_{pk_{C}}(k),ENC_{k}(D))$),
  sends $ENC_{pk_{Da}}(k),ENC_{k}(D)$ to Bob through the channel $C_{AD}$ (the corresponding sending action is denoted $s_{C_{AD}}(ENC_{pk_{Da}}(k),ENC_{k}(D))$);
  \item Bob receives the message $ENC_{pk_{B}}(k),ENC_{k}(D)$ through the channel $C_{AB}$ (the corresponding reading action is denoted
  $r_{C_{AB}}(ENC_{pk_{B}}(k),ENC_{k}(D))$), after a decryption processing $dec_{sk_{B}}(ENC_{pk_B}(k))$ to get the key $k$ and a decryption processing
  $dec_{k}(ENC_{k}(D))$ to get $D$, then he sends $D$ to the outside through the channel $C_{BO}$ (the corresponding sending action is denoted $s_{C_{BO}}(D)$);
  \item Carol receives the message $ENC_{pk_{C}}(k),ENC_{k}(D)$ through the channel $C_{AC}$ (the corresponding reading action is denoted
  $r_{C_{AC}}(ENC_{pk_{C}}(k),ENC_{k}(D))$), after a decryption processing $dec_{sk_{C}}(ENC_{pk_C}(k))$ to get the message $k$ and a decryption processing
  $dec_{k}(ENC_{k}(D))$ to get $D$, then he sends $D$ to the outside through the channel $C_{CO}$ (the corresponding sending action is denoted $s_{C_{CO}}(D)$);
  \item Dave receives the message $ENC_{pk_{Da}}(k),ENC_{k}(D)$ through the channel $C_{AD}$ (the corresponding reading action is denoted
  $r_{C_{AD}}(ENC_{pk_{Da}}(k),ENC_{k}(D))$), after a decryption processing $dec_{sk_{Da}}(ENC_{pk_{Da}}(k))$ to get the message $k$ and a decryption processing
  $dec_{k}(ENC_{k}(D))$ to get $D$, then he sends $D$ to the outside through the channel $C_{DO}$ (the corresponding sending action is denoted $s_{C_{DO}}(D)$).
\end{enumerate}

Where $D\in\Delta$, $\Delta$ is the set of data.

Alice's state transitions described by $APTC_G$ are as follows.

$A=\sum_{D\in\Delta}r_{C_{AI}}(D)\cdot A_2$

$A_2=(de\textrm{-}sign_{pk_T}(SIGN_{sk_T}(B,pk_B))\parallel de\textrm{-}sign_{pk_T}(SIGN_{sk_T}(C,pk_C))\\\parallel de\textrm{-}sign_{pk_T}(SIGN_{sk_T}(D,pk_D)))\cdot A_3$

$A_3=rsg_{k}\cdot A_4$

$A_4=enc_{k}(D)\cdot A_5$

$A_5=(enc_{pk_B}(k)\parallel enc_{pk_C}(k)\parallel enc_{pk_D}(k))\cdot A_6$

$A_6=(s_{C_{AB}}(ENC_{pk_{B}}(k),ENC_{k}(D))\parallel s_{C_{AC}}(ENC_{pk_{C}}(k),ENC_{k}(D))\\\parallel s_{C_{AD}}(ENC_{pk_{D}}(k),ENC_{k}(D)))\cdot A$

Bob's state transitions described by $APTC_G$ are as follows.

$B=r_{C_{AB}}(ENC_{pk_{B}}(k),ENC_{k}(D))\cdot B_2$

$B_2=dec_{sk_B}(ENC_{pk_{B}}(k))\cdot B_3$

$B_3=dec_{k}(ENC_{k}(D))\cdot B_4$

$B_4=s_{C_{BO}}(D)\cdot B$

Carol's state transitions described by $APTC_G$ are as follows.

$C=r_{C_{AC}}(ENC_{pk_{C}}(k),ENC_{k}(D))\cdot C_2$

$C_2=dec_{sk_C}(ENC_{pk_{C}}(k))\cdot C_3$

$C_3=dec_{k}(ENC_{k}(D))\cdot C_4$

$C_4=s_{C_{CO}}(D)\cdot C$

Dave's state transitions described by $APTC_G$ are as follows.

$Da=r_{C_{AD}}(ENC_{pk_{Da}}(k),ENC_{k}(D))\cdot Da_2$

$Da_2=dec_{sk_{Da}}(ENC_{pk_{Da}}(k))\cdot Da_3$

$Da_3=dec_{k}(ENC_{k}(D))\cdot Da_4$

$Da_4=s_{C_{DO}}(D)\cdot Da$

The sending action and the reading action of the same type data through the same channel can communicate with each other, otherwise, will cause a deadlock $\delta$. We define the following
communication functions.

$\gamma(r_{C_{AB}}(ENC_{pk_{B}}(k),ENC_{k}(D)),s_{C_{AB}}(ENC_{pk_{B}}(k),ENC_{k}(D)))\triangleq c_{C_{AB}}(ENC_{pk_{B}}(k),ENC_{k}(D))$

$\gamma(r_{C_{AC}}(ENC_{pk_{C}}(k),ENC_{k}(D)),s_{C_{AC}}(ENC_{pk_{C}}(k),ENC_{k}(D)))\triangleq c_{C_{AC}}(ENC_{pk_{C}}(k),ENC_{k}(D))$

$\gamma(r_{C_{AD}}(ENC_{pk_{Da}}(k),ENC_{k}(D)),s_{C_{AD}}(ENC_{pk_{Da}}(k),ENC_{k}(D)))\triangleq c_{C_{AD}}(ENC_{pk_{Da}}(k),ENC_{k}(D))$

Let all modules be in parallel, then the protocol $A\quad \quad B\quad C\quad Da$ can be presented by the following process term.

$$\tau_I(\partial_H(\Theta(A\between B\between C\between Da)))=\tau_I(\partial_H(A\between B\between C\between Da))$$

where $H=\{r_{C_{AB}}(ENC_{pk_{B}}(k),ENC_{k}(D)),s_{C_{AB}}(ENC_{pk_{B}}(k),ENC_{k}(D)),r_{C_{AC}}(ENC_{pk_{C}}(k),ENC_{k}(D)),\\
s_{C_{AC}}(ENC_{pk_{C}}(k),ENC_{k}(D)),r_{C_{AD}}(ENC_{pk_{Da}}(k),ENC_{k}(D)),s_{C_{AD}}(ENC_{pk_{Da}}(k),ENC_{k}(D))|D\in\Delta\}$,

$I=\{c_{C_{AB}}(ENC_{pk_{B}}(k),ENC_{k}(D)),c_{C_{AC}}(ENC_{pk_{C}}(k),ENC_{k}(D)),c_{C_{AD}}(ENC_{pk_{Da}}(k),ENC_{k}(D)),\\
de\textrm{-}sign_{pk_T}(SIGN_{sk_T}(C,pk_C)),de\textrm{-}sign_{pk_T}(SIGN_{sk_T}(Da,pk_{Da})),de\textrm{-}sign_{pk_T}(SIGN_{sk_T}(B,pk_B)),\\
rsg_{k},enc_{k}(D),enc_{pk_B}(k),enc_{pk_C}(k),enc_{pk_{Da}}(k),dec_{k}(ENC_{k}(D)),\\
dec_{sk_B}(ENC_{pk_{B}}(k)),dec_{sk_C}(ENC_{pk_{C}}(k)),dec_{sk_{Da}}(ENC_{pk_{Da}}(k))|D\in\Delta\}$.

Then we get the following conclusion on the protocol.

\begin{theorem}
The key and message broadcast protocol with digital signature in Figure \ref{KMB5} is secure.
\end{theorem}

\begin{proof}
Based on the above state transitions of the above modules, by use of the algebraic laws of $APTC_G$, we can prove that

$\tau_I(\partial_H(A\between B\between C\between Da))=\sum_{D\in\Delta}(r_{C_{AI}}(D)\cdot (s_{C_{BO}}(D)\parallel s_{C_{CO}}(D)\parallel s_{C_{DO}}(D)))\cdot
\tau_I(\partial_H(A\between B\between C\between Da))$.

For the details of proof, please refer to section \ref{app}, and we omit it.

That is, the protocol in Figure \ref{KMB5} $\tau_I(\partial_H(A\between B\between C\between Da))$ can exhibit desired external behaviors, and similarly to the protocol in subsection \ref{MIM}, this protocol can
resist the man-in-the-middle attack.
\end{proof}

\newpage\section{Analyses of Authentication Protocols}\label{auth}

An authentication protocol is used to verify the principal's identity, including verification of one principal's identity and mutual verifications of more that two principals' identities.
We omit some quite simple authentication protocols, including authentication using one-way functions, etc. We will analyze mutual authentication using the interlock protocol against man-in-the-middle
attack in section \ref{maip}, and SKID in section \ref{skid}.

\subsection{Mutual Authentication Using the Interlock Protocol}\label{maip}

The mutual authentication using the interlock protocol shown in Figure \ref{MAIP6} also uses public keys for secure communication with man-in-the-middle attack, that is, Alice, Bob have shared their public keys $pk_{A}$ and $pk_{B}$ already.
But, the interlock protocol can resist man-in-the-middle attack, that is, Alice and Bob can aware of the existence of the man in the middle.

\begin{figure}
    \centering
    \includegraphics{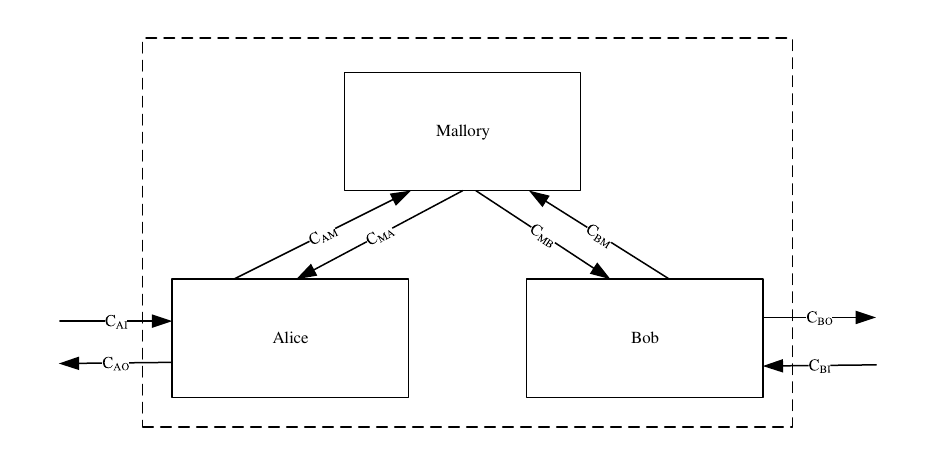}
    \caption{Mutual authentication using the interlock protocol with man-in-the-middle attack}
    \label{MAIP6}
\end{figure}

The process of the mutual authentication using the interlock protocol with man-in-the-middle attack is as follows, we assume that Alice has "Bob's" public key $pk_M$, Bob has "Alice's" public key $pk_M$, and Mallory
has Alice's public key $pk_A$ and Bob's public key $pk_B$.

\begin{enumerate}
  \item Alice receives some password $P_A$ from the outside through the channel $C_{AI}$ (the corresponding reading action is denoted $r_{C_{AI}}(P_A)$), she encrypts the password $P_A$ with Mallory's
  public key $pk_M$ through the action $enc_{pk_M}(P_A)$, then Alice sends the half of $ENC_{pk_{M}}(P_A)$ to Mallory through the channel $C_{AM}$
  (the corresponding sending action is denoted $s_{C_{AM}}(ENC_{pk_{M}}(P_A)/2)$);
  \item Mallory receives $ENC_{pk_M}(P_A)/2$ from Alice through the channel $C_{AM}$ (the corresponding reading action is denoted $r_{C_{AM}}(ENC_{pk_M}(P_A)/2)$), he can not decrypt the password with
  his private key $sk_M$, and has to make another password $P_A'$ and encrypt $P_A'$ with Bob's public key $pk_B$ through the action $enc_{pk_B}(P_A')$, and sends the half of $ENC_{pk_B}(P_A')$
  to Bob through the channel $C_{MB}$ (the corresponding sending action is denoted $s_{C_{MB}}(ENC_{pk_B}(P_A')/2)$);
  \item Bob receives the password $ENC_{pk_B}(P_A')/2$ through the channel $C_{MB}$ (the corresponding reading action is denoted $r_{C_{MB}}(ENC_{pk_B}(P_A')/2)$), and receives some password
  $P_B$ from the outside through the channel $C_{BI}$ (the corresponding reading action is denoted $r_{BI}(P_B)$), after an encryption processing
  $enc_{pk_{M}}(P_B)$ to get the password $ENC_{pk_M}(P_B)$, then he sends the half of $ENC_{pk_M}(P_B)$ to Mallory through the channel $C_{BM}$
  (the corresponding sending action is denoted $s_{C_{BM}}(ENC_{pk_{M}}(P_B)/2)$);
  \item Mallory receives $ENC_{pk_M}(P_B)/2$ from Bob through the channel $C_{BM}$ (the corresponding reading action is denoted $r_{C_{BM}}(ENC_{pk_M}(P_B)/2)$), he can not decrypt the password with
  his private key $sk_M$, and has to make another password $P_B'$ and encrypt $P_B'$ with Alice's public key $pk_A$ through the action $enc_{pk_A}(P_B')$, and sends the half of $ENC_{pk_A}(P_B')$
  to Alice through the channel $C_{MA}$ (the corresponding sending action is denoted $s_{C_{MA}}(ENC_{pk_A}(P_B')/2)$);
  \item Alice receives the password $ENC_{pk_A}(P_B')/2$ through the channel $C_{MA}$ (the corresponding reading action is denoted $r_{C_{MA}}(ENC_{pk_A}(P_B')/2)$), and sends the other half of $ENC_{pk_{M}}(P_A)$ to Mallory through the channel $C_{AM}$
  (the corresponding sending action is denoted \\$s_{C_{AM}}(ENC_{pk_{M}}(P_A)/2)$);
  \item Mallory receives $ENC_{pk_M}(P_A)/2$ from Alice through the channel $C_{AM}$ (the corresponding reading action is denoted $r_{C_{AM}}(ENC_{pk_M}(P_A)/2)$), he can combine the two half
  of $ENC_{pk_M}(P_A)/2$ and decrypt the password with
  his private key $sk_M$, and but he has to send the other half of $ENC_{pk_B}(P_A')$
  to Bob through the channel $C_{MB}$ (the corresponding sending action is denoted $s_{C_{MB}}(ENC_{pk_B}(P_A')/2)$);
  \item Bob receives the password $ENC_{pk_B}(P_A')/2$ through the channel $C_{MB}$ (the corresponding reading action is denoted $r_{C_{MB}}(ENC_{pk_B}(P_A')/2)$), after a combination of
  two half of $ENC_{pk_B}(P_A')$ and a decryption processing
  $dec_{sk_B}(ENC_{pk_{B}}(P_A'))$ to get the password $P_A'$, then he sends it to the outside through the channel $C_{BO}$ (the corresponding sending action is denoted $s_{C_{BO}}(P_A')$).
  Then he sends the other half of $ENC_{pk_M}(P_B)$ to Mallory through the channel $C_{BM}$
  (the corresponding sending action is denoted $s_{C_{BM}}(ENC_{pk_{M}}(P_B)/2)$);
  \item Mallory receives $ENC_{pk_M}(P_B)/2$ from Bob through the channel $C_{BM}$ (the corresponding reading action is denoted $r_{C_{BM}}(ENC_{pk_M}(P_B)/2)$), he can combine the two half
  of $ENC_{pk_M}(P_B)/2$ and decrypt the password with
  his private key $sk_M$, and but he has to send the other half of $ENC_{pk_A}(P_B')$
  to Alice through the channel $C_{MA}$ (the corresponding sending action is denoted $s_{C_{MA}}(ENC_{pk_A}(P_B')/2)$);
  \item Alice receives the password $ENC_{pk_A}(P_B')/2$ through the channel $C_{MA}$ (the corresponding reading action is denoted $r_{C_{MA}}(ENC_{pk_A}(P_B')/2)$), after a combination of
  two half of $ENC_{pk_A}(P_B')$ and a decryption processing
  $dec_{sk_A}(ENC_{pk_{A}}(P_B'))$ to get the password $P_B'$, then she sends it to the outside through the channel $C_{AO}$ (the corresponding sending action is denoted $s_{C_{AO}}(P_B')$).
\end{enumerate}

Where $P_A,P_B,P_A',P_B'\in\Delta$, $\Delta$ is the set of data.

Alice's state transitions described by $APTC_G$ are as follows.

$A=\sum_{P_A\in\Delta}r_{C_{AI}}(P_A)\cdot A_2$

$A_2=enc_{pk_M}(P_A)\cdot A_3$

$A_3=s_{C_{AM}}(ENC_{pk_{M}}(P_A)/2)\cdot A_4$

$A_4=r_{C_{MA}}(ENC_{pk_A}(P_B')/2)\cdot A_5$

$A_5=s_{C_{AM}}(ENC_{pk_{M}}(P_A)/2)\cdot A_6$

$A_6=r_{C_{MA}}(ENC_{pk_A}(P_B')/2)\cdot A_7$

$A_7=dec_{sk_A}(ENC_{pk_{A}}(P_B'))\cdot A_8$

$A_8=s_{C_{AO}}(P_B')\cdot A$

Bob's state transitions described by $APTC_G$ are as follows.

$B=r_{C_{MB}}(ENC_{pk_B}(P_A')/2)\cdot B_2$

$B_2=\sum_{P_B\in\Delta}r_{BI}(P_B)\cdot B_3$

$B_3=enc_{pk_{M}}(P_B)\cdot B_4$

$B_4=s_{C_{BM}}(ENC_{pk_{M}}(P_B)/2)\cdot B_5$

$B_5=r_{C_{MB}}(ENC_{pk_B}(P_A')/2)\cdot B_6$

$B_6=dec_{sk_B}(ENC_{pk_{B}}(P_A'))\cdot B_7$

$B_7=s_{C_{BO}}(P_A')\cdot B_8$

$B_8=s_{C_{BM}}(ENC_{pk_{M}}(P_B)/2)\cdot B$

Mallory's state transitions described by $APTC_G$ are as follows.

$Ma=r_{C_{AM}}(ENC_{pk_M}(P_A)/2)\cdot Ma_2$

$Ma_2=enc_{pk_B}(P_A')\cdot Ma_3$

$Ma_3=s_{C_{MB}}(ENC_{pk_B}(P_A')/2)\cdot Ma_4$

$Ma_4=r_{C_{BM}}(ENC_{pk_M}(P_B)/2)\cdot Ma_5$

$Ma_5=enc_{pk_A}(P_B')\cdot Ma_6$

$Ma_6=s_{C_{MA}}(ENC_{pk_A}(P_B')/2)\cdot Ma_7$

$Ma_7=r_{C_{AM}}(ENC_{pk_M}(P_A)/2)\cdot Ma_8$

$Ma_8=s_{C_{MB}}(ENC_{pk_B}(P_A')/2)\cdot Ma_9$

$Ma_9=r_{C_{BM}}(ENC_{pk_M}(P_B)/2)\cdot Ma_{10}$

$Ma_{10}=s_{C_{MA}}(ENC_{pk_A}(P_B')/2)\cdot Ma$

The sending action and the reading action of the same type data through the same channel can communicate with each other, otherwise, will cause a deadlock $\delta$. We define the following
communication functions.

$\gamma(r_{C_{AM}}(ENC_{pk_M}(P_A)/2),s_{C_{AM}}(ENC_{pk_M}(P_A)/2))\triangleq c_{C_{AM}}(ENC_{pk_M}(P_A)/2)$

$\gamma(r_{C_{MB}}(ENC_{pk_B}(P_A')/2),s_{C_{MB}}(ENC_{pk_B}(P_A')/2))\triangleq c_{C_{MB}}(ENC_{pk_B}(P_A')/2)$

$\gamma(r_{C_{BM}}(ENC_{pk_M}(P_B)/2),s_{C_{BM}}(ENC_{pk_M}(P_B)/2))\triangleq c_{C_{BM}}(ENC_{pk_M}(P_B)/2)$

$\gamma(r_{C_{MA}}(ENC_{pk_A}(P_B')/2),s_{C_{MA}}(ENC_{pk_A}(P_B')/2))\triangleq c_{C_{MA}}(ENC_{pk_A}(P_B')/2)$

$\gamma(r_{C_{AM}}(ENC_{pk_M}(P_A)/2),s_{C_{AM}}(ENC_{pk_M}(P_A)/2))\triangleq c_{C_{AM}}(ENC_{pk_M}(P_A)/2)$

$\gamma(r_{C_{MB}}(ENC_{pk_B}(P_A')/2),s_{C_{MB}}(ENC_{pk_B}(P_A')/2))\triangleq c_{C_{MB}}(ENC_{pk_B}(P_A')/2)$

$\gamma(r_{C_{BM}}(ENC_{pk_M}(P_B)/2),s_{C_{BM}}(ENC_{pk_M}(P_B)/2))\triangleq c_{C_{BM}}(ENC_{pk_M}(P_B)/2)$

$\gamma(r_{C_{MA}}(ENC_{pk_A}(P_B')/2),s_{C_{MA}}(ENC_{pk_A}(P_B')/2))\triangleq c_{C_{MA}}(ENC_{pk_A}(P_B')/2)$

Let all modules be in parallel, then the protocol $A\quad B\quad Ma$ can be presented by the following process term.

$$\tau_I(\partial_H(\Theta(A\between B\between Ma)))=\tau_I(\partial_H(A\between B\between Ma))$$

where $H=\{r_{C_{AM}}(ENC_{pk_M}(P_A)/2),s_{C_{AM}}(ENC_{pk_M}(P_A)/2),r_{C_{MB}}(ENC_{pk_B}(P_A')/2),\\
s_{C_{MB}}(ENC_{pk_B}(P_A')/2),r_{C_{BM}}(ENC_{pk_M}(P_B)/2),s_{C_{BM}}(ENC_{pk_M}(P_B)/2),\\
r_{C_{MA}}(ENC_{pk_A}(P_B')/2),s_{C_{MA}}(ENC_{pk_A}(P_B')/2),r_{C_{AM}}(ENC_{pk_M}(P_A)/2),\\
s_{C_{AM}}(ENC_{pk_M}(P_A)/2),r_{C_{MB}}(ENC_{pk_B}(P_A')/2),s_{C_{MB}}(ENC_{pk_B}(P_A')/2),\\
r_{C_{BM}}(ENC_{pk_M}(P_B)/2),s_{C_{BM}}(ENC_{pk_M}(P_B)/2),r_{C_{MA}}(ENC_{pk_A}(P_B')/2),\\
s_{C_{MA}}(ENC_{pk_A}(P_B')/2)|P_A,P_B,P_A',P_B'\in\Delta\}$,

$I=\{c_{C_{AM}}(ENC_{pk_M}(P_A)/2),c_{C_{MB}}(ENC_{pk_B}(P_A')/2),c_{C_{BM}}(ENC_{pk_M}(P_B)/2),c_{C_{MA}}(ENC_{pk_A}(P_B')/2),\\
c_{C_{AM}}(ENC_{pk_M}(P_A)/2),c_{C_{MB}}(ENC_{pk_B}(P_A')/2),c_{C_{BM}}(ENC_{pk_M}(P_B)/2),c_{C_{MA}}(ENC_{pk_A}(P_B')/2),\\
enc_{pk_M}(P_A),dec_{sk_A}(ENC_{pk_{A}}(P_B')),enc_{pk_{M}}(P_B),dec_{sk_B}(ENC_{pk_{B}}(P_A')),enc_{pk_B}(P_A'),\\
enc_{pk_A}(P_B')|P_A,P_B,P_A',P_B'\in\Delta\}$.

Then we get the following conclusion on the protocol.

\begin{theorem}
The mutual authentication using the interlock protocol in Figure \ref{MAIP6} is secure.
\end{theorem}

\begin{proof}
Based on the above state transitions of the above modules, by use of the algebraic laws of $APTC_G$, we can prove that

$\tau_I(\partial_H(A\between B\between Ma))=\sum_{P_A,P_B,P_A',P_B'\in\Delta}(r_{C_{AI}}(P_A)\cdot r_{C_{BI}}(P_B)\cdot s_{C_{BO}}(P_A')\cdot s_{C_{AO}}(P_B'))\cdot
\tau_I(\partial_H(A\between B\between Ma))$.

For the details of proof, please refer to section \ref{app}, and we omit it.

That is, the mutual authentication using the interlock protocol in Figure \ref{MAIP6} $\tau_I(\partial_H(A\between B\between Ma))$ can exhibit desired external behaviors,
that is, Alice and Bob can aware the existence of the man in the middle.
\end{proof}

\subsection{SKID}\label{skid}

The SKID protocol shown in Figure \ref{SKID6} uses symmetric cryptography to authenticate each other, that is, Alice, Bob have shared their key $k_{AB}$.

\begin{figure}
    \centering
    \includegraphics{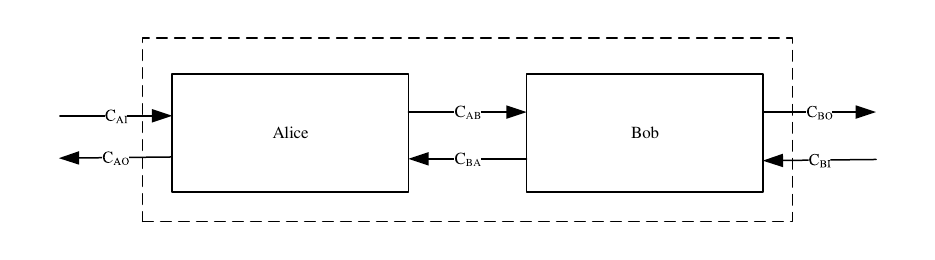}
    \caption{SKID protocol}
    \label{SKID6}
\end{figure}

The process of the protocol is as follows, and we only consider the message in one direction: from Alice to Bob.

\begin{enumerate}
  \item Alice receives some messages $D$ from the outside through the channel $C_{AI}$ (the corresponding reading action is denoted $r_{C_{AI}}(D)$), she generates a random number $R_A$
  through an action $rsg_{R_A}$, she sends $R_A$ to Bob through the channel $C_{AB}$ (the corresponding sending action is denoted $s_{C_{AB}}(R_A)$);
  \item Bob receives the number $R_A$ through the channel $C_{AB}$ (the corresponding reading action is denoted
  $r_{C_{AB}}(R_A)$), he generates a random number $R_B$ through an action $rsg_{R_B}$, and generates a MAC (Message Authentication Code) through an action $mac_{k_{AB}}(R_A, R_B, B)$,
  then he sends $B,R_B,MAC_{k_{AB}}(R_A,R_B,B)$ to Alice through the channel $C_{BA}$ (the corresponding sending action is denoted $s_{C_{BA}}(B,R_B,MAC_{k_{AB}}(R_A,R_B,B))$);
  \item Alice receives $d_B,d_{R_B},d_{MAC_{k_{AB}}(R_A,R_B,B)}$ from Bob through the channel $C_{BA}$ (the corresponding reading action is denoted $r_{C_{BA}}(d_B,d_{R_B},d_{MAC_{k_{AB}}(R_A,R_B,B)})$),
  she generates a MAC through an action $mac_{k_{AB}}(R_A, d_{R_B}, d_B)$, if $MAC_{k_{AB}}(R_A, d_{R_B}, d_B)=d_{MAC_{k_{AB}}(R_A,R_B,B)}$, she generates a MAC
  through an action $mac_{k_{AB}}(R_B, A)$ and encrypts $D$ by $k_{AB}$ through an action $enc_{k_{AB}}(D)$, then she sends $A,MAC_{k_{AB}}(R_B,A),ENC_{k_{AB}}(D)$
  to Bob through the channel $C_{AB}$ (the corresponding sending action is denoted \\
  $s_{C_{AB}}(A,MAC_{k_{AB}}(R_B,A),ENC_{k_{AB}}(D))$);
  \item Bob receives the data $d_A,d_{MAC_{k_{AB}}(R_B,A)},ENC_{k_{AB}}(D)$ from Alice through the channel $C_{AB}$ (the corresponding reading action is denoted
  $r_{C_{AB}}(d_A,d_{MAC_{k_{AB}}(R_B,A)},ENC_{k_{AB}}(D))$), he generates a MAC through an action $mac_{k_{AB}}(R_B, d_A)$, if $MAC_{k_{AB}}(R_B, d_A)=d_{MAC_{k_{AB}}(R_B,A)}$, he
  decrypts $ENC_{k_{AB}}(D)$ by $k_{AB}$ through an action $dec_{k_{AB}}(ENC_{k_{AB}}(D))$ to get $D$, then she sends $D$
  to the outside through the channel $C_{BO}$ (the corresponding sending action is denoted $s_{C_{BO}}(D)$).
\end{enumerate}

Where $D\in\Delta$, $\Delta$ is the set of data.

Alice's state transitions described by $APTC_G$ are as follows.

$A=\sum_{D\in\Delta}r_{C_{AI}}(D)\cdot A_2$

$A_2=rsg_{R_A}\cdot A_3$

$A_3=s_{C_{AB}}(R_A)\cdot A_4$

$A_4=r_{C_{BA}}(d_B,d_{R_B},d_{MAC_{k_{AB}}(R_A,R_B,B)})\cdot A_5$

$A_5=mac_{k_{AB}}(R_A, d_{R_B}, d_B)\cdot A_6$

$A_6=\{MAC_{k_{AB}}(R_A, d_{R_B}, d_B)=d_{MAC_{k_{AB}}(R_A,R_B,B)}\}\cdot mac_{k_{AB}}(R_B, A)\cdot enc_{k_{AB}}(D) \\
\cdot s_{C_{AB}}(A,MAC_{k_{AB}}(R_B,A),ENC_{k_{AB}}(D))\cdot A$

Bob's state transitions described by $APTC_G$ are as follows.

$B=r_{C_{AB}}(R_A)\cdot B_2$

$B_2=rsg_{R_B}\cdot B_3$

$B_3=mac_{k_{AB}}(R_A, R_B, B)\cdot B_4$

$B_4=s_{C_{BA}}(B,R_B,MAC_{k_{AB}}(R_A,R_B,B))\cdot B_5$

$B_5=r_{C_{AB}}(d_A,d_{MAC_{k_{AB}}(R_B,A)},ENC_{k_{AB}}(D))\cdot B_6$

$B_6=mac_{k_{AB}}(R_B, d_A)\cdot B_7$

$B_7=\{MAC_{k_{AB}}(R_B, d_A)=d_{MAC_{k_{AB}}(R_B,A)}\}\cdot dec_{k_{AB}}(ENC_{k_{AB}}(D))\cdot s_{C_{BO}}(D)\cdot B$

The sending action and the reading action of the same type data through the same channel can communicate with each other, otherwise, will cause a deadlock $\delta$. We define the following
communication functions.

$\gamma(r_{C_{AB}}(R_A),s_{C_{AB}}(R_A))\triangleq c_{C_{AB}}(R_A)$

$\gamma(r_{C_{BA}}(B,R_B,MAC_{k_{AB}}(R_A,R_B,B)),s_{C_{BA}}(B,R_B,MAC_{k_{AB}}(R_A,R_B,B)))\\
\triangleq c_{C_{BA}}(B,R_B,MAC_{k_{AB}}(R_A,R_B,B))$

$\gamma(r_{C_{AB}}(d_A,d_{MAC_{k_{AB}}(R_B,A)},ENC_{k_{AB}}(D)),s_{C_{AB}}(d_A,d_{MAC_{k_{AB}}(R_B,A)},ENC_{k_{AB}}(D)))\\
\triangleq c_{C_{AB}}(d_A,d_{MAC_{k_{AB}}(R_B,A)},ENC_{k_{AB}}(D))$

Let all modules be in parallel, then the protocol $A\quad B$ can be presented by the following process term.

$$\tau_I(\partial_H(\Theta(A\between B)))=\tau_I(\partial_H(A\between B))$$

where $H=\{r_{C_{AB}}(R_A),s_{C_{AB}}(R_A),r_{C_{BA}}(B,R_B,MAC_{k_{AB}}(R_A,R_B,B)),\\
s_{C_{BA}}(B,R_B,MAC_{k_{AB}}(R_A,R_B,B)),r_{C_{AB}}(d_A,d_{MAC_{k_{AB}}(R_B,A)},ENC_{k_{AB}}(D)),\\
s_{C_{AB}}(d_A,d_{MAC_{k_{AB}}(R_B,A)},ENC_{k_{AB}}(D))|D\in\Delta\}$,

$I=\{c_{C_{AB}}(R_A),c_{C_{BA}}(B,R_B,MAC_{k_{AB}}(R_A,R_B,B)),c_{C_{AB}}(d_A,d_{MAC_{k_{AB}}(R_B,A)},ENC_{k_{AB}}(D)),\\
rsg_{R_A},mac_{k_{AB}}(R_A, d_{R_B}, d_B),\{MAC_{k_{AB}}(R_A, d_{R_B}, d_B)=d_{MAC_{k_{AB}}(R_A,R_B,B)}\}, mac_{k_{AB}}(R_B, A),\\
enc_{k_{AB}}(D),rsg_{R_B},mac_{k_{AB}}(R_A, R_B, B),mac_{k_{AB}}(R_B, d_A),\{MAC_{k_{AB}}(R_B, d_A)=d_{MAC_{k_{AB}}(R_B,A)}\},\\
dec_{k_{AB}}(ENC_{k_{AB}}(D))|D\in\Delta\}$.

Then we get the following conclusion on the protocol.

\begin{theorem}
The key and message transmission protocol with digital signature in Figure \ref{SKID6} is secure.
\end{theorem}

\begin{proof}
Based on the above state transitions of the above modules, by use of the algebraic laws of $APTC_G$, we can prove that

$\tau_I(\partial_H(A\between B))=\sum_{D\in\Delta}(r_{C_{AI}}(D)\cdot s_{C_{BO}}(D))\cdot
\tau_I(\partial_H(A\between B))$.

For the details of proof, please refer to section \ref{app}, and we omit it.

That is, the protocol in Figure \ref{SKID6} $\tau_I(\partial_H(A\between B))$ can exhibit desired external behaviors, and similarly to the protocol in subsection \ref{MIM}, without leasing
of $k_{AB}$, this protocol can resist the man-in-the-middle attack.
\end{proof}

\newpage\section{Analyses of Practical Protocols}\label{aopp}

In this chapter, we will introduce analyses of some practical authentication and key exchange protocols. For a perfectly practical security protocol, it should can resist any kind of
attack. There are many kinds of attacks, it is difficult to model all known attacks, for simplicity, we only analyses the protocols with several kinds of main attacks.

We introduce analyses of Wide-Mouth Frog protocol in section \ref{wmf}, Yahalom protocol in section \ref{yp}, Needham-Schroeder protocol in section \ref{nsp}, Otway-Rees protocol
in section \ref{orp}, Kerberos protocol in section \ref{kp}, Neuman-Stubblebine protocol in section \ref{nsp2}, Denning-Sacco protocol in section
\ref{dsp}, DASS protocol in section \ref{dass} and Woo-Lam protocol in section \ref{wlp}.

\subsection{Wide-Mouth Frog Protocol}\label{wmf}

The Wide-Mouth Frog protocol shown in Figure \ref{WMF7} uses symmetric keys for secure communication, that is, the key $k_{AB}$ between Alice and Bob is privately shared to Alice and Bob,
Alice, Bob have shared keys with Trent $k_{AT}$ and $k_{BT}$ already.

\begin{figure}
    \centering
    \includegraphics{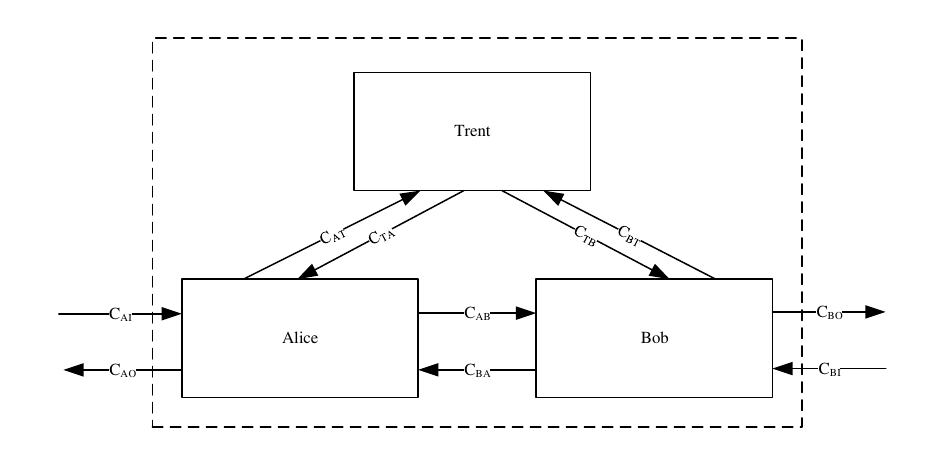}
    \caption{Wide-Mouth Frog protocol}
    \label{WMF7}
\end{figure}

The process of the protocol is as follows.

\begin{enumerate}
  \item Alice receives some messages $D$ from the outside through the channel $C_{AI}$ (the corresponding reading action is denoted $r_{C_{AI}}(D)$), if $k_{AB}$ is not established,
  she generates a random session key $k_{AB}$ through an action $rsg_{k_{AB}}$, encrypts the key request message $T_A,B,k_{AB}$ with $k_{AT}$ through an action $enc_{k_{AT}}(T_A,B,k_{AB})$ where
  $T_A$ Alice's time stamp, and sends $A,ENC_{k_{AT}}(T_A,B,k_{AB})$ to Trent through the channel $C_{AT}$ (the corresponding sending action is denoted $s_{C_{AT}}(A,ENC_{k_{AT}}(T_A,B,k_{AB}))$);
  \item Trent receives the message $A,ENC_{k_{AT}}(T_A,B,k_{AB})$ through the channel $C_{AT}$ (the corresponding reading action is denoted $r_{C_{AT}}(A,ENC_{k_{AT}}(T_A,B,k_{AB}))$),
  he decrypts the message through an action $dec_{k_{AT}}(ENC_{k_{AT}}(T_A,B,k_{AB}))$. If $isFresh(T_A)=TRUE$ where $isFresh$ is a function to deciding whether a time stamp is fresh,
  he encrypts $T_B,A,K_{AB}$ with $k_{BT}$ through an action $enc_{k_{BT}}(T_B,A,K_{AB})$, sends $\top$ to Alice through the channel $C_{TA}$ (the corresponding sending action is denoted
  $s_{C_{TA}}(\top)$) and $ENC_{k_{BT}}(T_B,A,K_{AB})$ to Bob through the channel $C_{TB}$ (the corresponding sending action is denoted \\
  $s_{C_{TB}}(ENC_{k_{BT}}(T_B,A,K_{AB}))$); else
  if $isFresh(T_A)=FLASE$, he sends $\bot$ to Alice and Bob (the corresponding sending actions are denoted $s_{C_{TA}}(\bot)$ and $s_{C_{TB}}(\bot)$ respectively);
  \item Bob receives $d_{TB}$ from Trent through the channel $C_{TB}$ (the corresponding reading action is denoted $r_{C_{TB}}(d_{TB})$). If $d_{TB}=\bot$, he sends $\bot$ to Alice through
  the channel $C_{BA}$ (the corresponding sending action is denoted $s_{C_{BA}}(\bot)$); if $d_{TB}\neq \bot$, he decrypts $ENC_{k_{BT}}(T_B,A,K_{AB})$ through an action
  $dec_{k_{BT}}(ENC_{k_{BT}}(T_B,A,K_{AB}))$. If $isFresh(T_B)=TRUE$, he gets $k_{AB}$, and sends $\top$ to Alice (the corresponding sending action is denoted $s_{C_{BA}}(\top)$);
  if $isFresh(T_B)=FALSE$, he sends $\bot$ to Alice through the channel $C_{BA}$ (the corresponding sending action is denoted $s_{C_{BA}}(\bot)$);
  \item Alice receives $d_{TA}$ from Trent through the channel $C_{TA}$ (the corresponding reading action is denoted $r_{C_{TA}}(d_{TA})$), receives $d_{BA}$ from Bob through the channel
  $C_{BA}$ (the corresponding reading action is denoted $r_{C_{BA}}(d_{BA})$). If $d_{TA}=\top\cdot d_{BA}=\top$, after an encryption processing $enc_{k_{AB}}(T_{A_D},D)$, Alice sends
  $ENC_{k_{AB}}(T_{A_D},D)$ to Bob through the channel $C_{AB}$ (the corresponding sending action is denoted $s_{C_{AB}}(T_{A_D},ENC_{k_{AB}}(D))$); else if $d_{TA}=\bot+ d_{BA}=\bot$,
  Alice sends $\bot$ to the outside through the channel $C_{AO}$ (the corresponding sending action is denoted $s_{C_{AO}}(\bot)$);
  \item Bob receives the message $ENC_{k_{AB}}(T_{A_D},D)$ through the channel $C_{AB}$ (the corresponding reading action is denoted $r_{C_{AB}}(T_{A_D},ENC_{k_{AB}}(D))$),
  after a decryption processing $dec_{k_{AB}}(ENC_{k_{AB}}(T_{A_D},D))$, if $isFresh(T_{A_D})=TRUE$, he sends $D$ to the outside through the channel $C_{BO}$
  (the corresponding sending action is denoted $s_{C_{BO}}(D)$), if $isFresh(T_{A_D})=FALSE$, he sends $\bot$ to the outside through the channel $C_{BO}$
  (the corresponding sending action is denoted $s_{C_{BO}}(\bot)$).
\end{enumerate}

Where $D\in\Delta$, $\Delta$ is the set of data.

Alice's state transitions described by $APTC_G$ are as follows.

$A=\sum_{D\in\Delta}r_{C_{AI}}(D)\cdot A_2$

$A_2=\{k_{AB}=NULL\}\cdot rsg_{k_{AB}}\cdot A_3+ \{k_{AB}\neq NULL\}\cdot A_7$

$A_3=enc_{k_{AT}}(T_A,B,k_{AB})\cdot A_4$

$A_4=s_{C_{AT}}(A,ENC_{k_{AT}}(T_A,B,k_{AB}))\cdot A_5$

$A_5=(r_{C_{TA}}(d_{TA})\parallel r_{C_{BA}}(d_{BA}))\cdot A_6$

$A_6=\{d_{TA}=\top\cdot d_{BA}=\top\}\cdot A_7+\{d_{TA}=\bot+ d_{BA}=\bot\}\cdot A_{9}$

$A_7=enc_{k_{AB}}(T_{A_D},D)\cdot A_8$

$A_8=s_{C_{AB}}(T_{A_D},ENC_{k_{AB}}(D))\cdot A$

$A_{9}=s_{C_{AO}}(\bot)\cdot A$

Bob's state transitions described by $APTC_G$ are as follows.

$B=\{k_{AB}=NULL\}\cdot B_1+ \{k_{AB}\neq NULL\}\cdot B_5$

$B_1=r_{C_{TB}}(d_{TB})\cdot B_2$

$B_2=\{d_{TB}\neq\bot\}\cdot B_3 + \{d_{TB}= \bot\}\cdot s_{C_{BA}}(\bot)\cdot B$

$B_3=dec_{k_{BT}}(ENC_{k_{BT}}(T_B,A,K_{AB}))\cdot B_4$

$B_4=\{isFresh(T_B)=TRUE\}\cdot s_{C_{BA}}(\top)\cdot B_5+\{isFresh(T_B)=FALSE\}\cdot s_{C_{BA}}(\bot)\cdot B$

$B_5=r_{C_{AB}}(T_{A_D},ENC_{k_{AB}}(D))\cdot B_6$

$B_6=dec_{k_{AB}}(ENC_{k_{AB}}(T_{A_D},D))\cdot B_7$

$B_7=\{isFresh(T_{A_D})=TRUE\}\cdot s_{C_{BO}}(D)\cdot B+\{isFresh(T_{A_D})=FALSE\}\cdot s_{C_{BO}}(\bot)\cdot B$

Trent's state transitions described by $APTC_G$ are as follows.

$T=r_{C_{AT}}(A,ENC_{k_{AT}}(T_A,B,k_{AB}))\cdot T_2$

$T_2=dec_{k_{AT}}(ENC_{k_{AT}}(T_A,B,k_{AB}))\cdot T_3$

$T_3=\{isFresh(T_A)=TRUE\}\cdot enc_{k_{BT}}(T_B,A,K_{AB})\cdot (s_{C_{TA}}(\top)\parallel s_{C_{TB}}(ENC_{k_{BT}}(T_B,A,K_{AB})))T\\
+\{isFresh(T_A)=FALSE\}\cdot (s_{C_{TA}}(\bot)\parallel s_{C_{TB}}(\bot))\cdot T$

The sending action and the reading action of the same type data through the same channel can communicate with each other, otherwise, will cause a deadlock $\delta$. We define the following
communication functions.

$\gamma(r_{C_{AT}}(A,ENC_{k_{AT}}(T_A,B,k_{AB})),s_{C_{AT}}(A,ENC_{k_{AT}}(T_A,B,k_{AB})))\triangleq c_{C_{AT}}(A,ENC_{k_{AT}}(T_A,B,k_{AB}))$

$\gamma(r_{C_{TA}}(d_{TA}),s_{C_{TA}}(d_{TA}))\triangleq c_{C_{TA}}(d_{TA})$

$\gamma(r_{C_{BA}}(d_{BA}),s_{C_{BA}}(d_{BA}))\triangleq c_{C_{BA}}(d_{BA})$

$\gamma(r_{C_{AB}}(T_{A_D},ENC_{k_{AB}}(D)),s_{C_{AB}}(T_{A_D},ENC_{k_{AB}}(D)))\triangleq c_{C_{AB}}(T_{A_D},ENC_{k_{AB}}(D))$

$\gamma(r_{C_{TB}}(d_{TB}),s_{C_{TB}}(d_{TB}))\triangleq c_{C_{TB}}(d_{TB})$

Let all modules be in parallel, then the protocol $A\quad B\quad T$ can be presented by the following process term.

$$\tau_I(\partial_H(\Theta(A\between B\between T)))=\tau_I(\partial_H(A\between B\between T))$$

where $H=\{r_{C_{AT}}(A,ENC_{k_{AT}}(T_A,B,k_{AB})),s_{C_{AT}}(A,ENC_{k_{AT}}(T_A,B,k_{AB})),\\
r_{C_{TA}}(d_{TA}),s_{C_{TA}}(d_{TA}),r_{C_{BA}}(d_{BA}),s_{C_{BA}}(d_{BA}),\\
r_{C_{AB}}(T_{A_D},ENC_{k_{AB}}(D)),s_{C_{AB}}(T_{A_D},ENC_{k_{AB}}(D)),r_{C_{TB}}(d_{TB}),s_{C_{TB}}(d_{TB})|D\in\Delta\}$,

$I=\{c_{C_{AT}}(A,ENC_{k_{AT}}(T_A,B,k_{AB})),c_{C_{TA}}(d_{TA}),c_{C_{BA}}(d_{BA}),\\
c_{C_{AB}}(T_{A_D},ENC_{k_{AB}}(D)),c_{C_{TB}}(d_{TB}),\{k_{AB}=NULL\}, rsg_{k_{AB}},\\
\{k_{AB}\neq NULL\},enc_{k_{AT}}(T_A,B,k_{AB}),\{d_{TA}=\top\cdot d_{BA}=\top\},\{d_{TA}=\bot+ d_{BA}=\bot\},\\
enc_{k_{AB}}(T_{A_D},D),\{d_{TB}\neq\bot\},\{d_{TB}=\bot\},dec_{k_{BT}}(ENC_{k_{BT}}(T_B,A,K_{AB})),\\
\{isFresh(T_B)=TRUE\},\{isFresh(T_B)=FALSE\},dec_{k_{AB}}(ENC_{k_{AB}}(T_{A_D},D)),\\
\{isFresh(T_{A_D})=TRUE\},\{isFresh(T_{A_D})=FALSE\},dec_{k_{AT}}(ENC_{k_{AT}}(T_A,B,k_{AB})),\\
\{isFresh(T_A)=TRUE\}, enc_{k_{BT}}(T_B,A,K_{AB}),\{isFresh(T_A)=FALSE\}|D\in\Delta\}$.

Then we get the following conclusion on the protocol.

\begin{theorem}
The Wide-Mouth Frog protocol in Figure \ref{WMF7} is secure.
\end{theorem}

\begin{proof}
Based on the above state transitions of the above modules, by use of the algebraic laws of $APTC_G$, we can prove that

$\tau_I(\partial_H(A\between B\between T))=\sum_{D\in\Delta}(r_{C_{AI}}(D)\cdot ((s_{C_{AO}}(\bot)\parallel s_{C_{BO}}(\bot))+s_{C_{BO}}(D)))\cdot
\tau_I(\partial_H(A\between B\between T))$.

For the details of proof, please refer to section \ref{app}, and we omit it.

That is, the Wide-Mouth Frog protocol in Figure \ref{WMF7} $\tau_I(\partial_H(A\between B\between T))$ can exhibit desired external behaviors:

\begin{enumerate}
  \item For information leakage, because $k_{AT}$ is privately shared only between Alice and Trent, $k_{BT}$ is privately shared only between Bob and Trent, $k_{AB}$ is privately shared
  only among Trent, Alice and Bob. For the modeling of confidentiality, it is similar to the protocol in section \ref{confi}, the Wide-Mouth Frog protocol is confidential;
  \item For replay attack, the using of time stamps $T_A$, $T_B$, and $T_{A_D}$, makes that $\tau_I(\partial_H(A\between B\between T))=\sum_{D\in\Delta}(r_{C_{AI}}(D)\cdot (s_{C_{AO}}(\bot)\parallel s_{C_{BO}}(\bot)))\cdot
\tau_I(\partial_H(A\between B\between T))$, it is desired;
  \item Without replay attack, the protocol would be $\tau_I(\partial_H(A\between B\between T))=\sum_{D\in\Delta}(r_{C_{AI}}(D)\cdot s_{C_{BO}}(D))\cdot
\tau_I(\partial_H(A\between B\between T))$, it is desired;
  \item For the man-in-the-middle attack, because $k_{AT}$ is privately shared only between Alice and Trent, $k_{BT}$ is privately shared only between Bob and Trent, $k_{AB}$ is privately shared
  only among Trent, Alice and Bob. For the modeling of the man-in-the-middle attack, it is similar to the protocol in section \ref{MIM}, the Wide-Mouth Frog protocol can be against the
  man-in-the-middle attack;
  \item For the unexpected and non-technical leaking of $k_{AT}$, $k_{BT}$, $k_{AB}$, or they being not strong enough, or Trent being dishonest, they are out of the scope of analyses of security protocols;
  \item For malicious tampering and transmission errors, they are out of the scope of analyses of security protocols.
\end{enumerate}
\end{proof}

\subsection{Yahalom Protocol}\label{yp}

The Yahalom protocol shown in Figure \ref{YP7} uses symmetric keys for secure communication, that is, the key $k_{AB}$ between Alice and Bob is privately shared to Alice and Bob,
Alice, Bob have shared keys with Trent $k_{AT}$ and $k_{BT}$ already.

\begin{figure}
    \centering
    \includegraphics{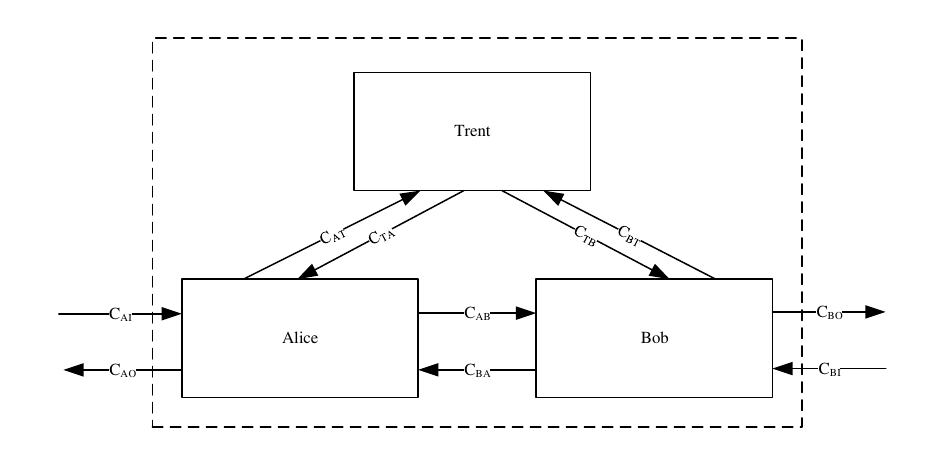}
    \caption{Yahalom protocol}
    \label{YP7}
\end{figure}

The process of the protocol is as follows.

\begin{enumerate}
  \item Alice receives some messages $D$ from the outside through the channel $C_{AI}$ (the corresponding reading action is denoted $r_{C_{AI}}(D)$), if $k_{AB}$ is not established,
  she generates a random number $R_A$ through an action $rsg_{R_A}$, and sends $A, R_A$ to Bob through the channel $C_{AB}$ (the corresponding sending action is denoted
  $s_{C_{AB}}(A,R_A)$);
  \item Bob receives $A,R_A$ from Alice through the channel $C_{AB}$ (the corresponding reading action is denoted $r_{C_{AB}}(A,R_A)$), he generates a random number $R_B$ through
  an action $rsg_{R_B}$, encrypts $A,R_A,R_B$ by $k_{BT}$ through an action $enc_{k_{BT}}(A,R_A,R_B)$, and sends $ENC_{k_{BT}}(A,R_A,R_B)$ to Trent through the channel $C_{BT}$
  (the corresponding sending action is denoted $s_{C_{BT}}(ENC_{k_{BT}}(A,R_A,R_B))$);
  \item Trent receives $ENC_{k_{BT}}(A,R_A,R_B)$ through the channel $C_{BT}$ (the corresponding reading action is denoted $r_{C_{BT}}(ENC_{k_{BT}}(A,R_A,R_B))$), he decrypts the message
  through an action $dec_{k_{BT}}(ENC_{k_{BT}}(A,R_A,R_B))$, generates a random session key $k_{AB}$ through an action $rsg_{k_{AB}}$, then he encrypts $B,k_{AB},R_A,R_B$ by $k_{AT}$ through an action
  $enc_{k_{AT}}(B,k_{AB},R_A,R_B)$, encrypts $A,k_{AB}$ by $k_{BT}$ through an action $enc_{k_{BT}}(A,k_{AB})$, and sends them to Alice through the channel $C_{TA}$ (the corresponding sending
  action is denoted \\
  $s_{C_{TA}}(ENC_{k_{AT}}(B,k_{AB},R_A,R_B),ENC_{k_{BT}}(A,k_{AB}))$);
  \item Alice receives the message from Trent through the channel $C_{TA}$ (the corresponding reading action is denoted $r_{C_{TA}}(ENC_{k_{AT}}(B,k_{AB},d_{R_A},R_B),ENC_{k_{BT}}(A,k_{AB}))$),
  she decrypts $ENC_{k_{AT}}(B,k_{AB},d_{R_A},R_B)$ by $k_{AT}$ through an action $dec_{k_{AT}}(ENC_{k_{AT}}(B,k_{AB},d_{R_A},R_B))$, if $d_{R_A}=R_A$, she encrypts $R_B,D$ by $k_{AB}$ through
  an action $enc_{k_{AB}}(R_B,D)$, and sends $ENC_{k_{BT}}(A,k_{AB}),ENC_{k_{AB}}(R_B,D)$
  to Bob through the channel $C_{AB}$ (the corresponding sending action is denoted $s_{C_{AB}}(ENC_{k_{BT}}(A,k_{AB}),ENC_{k_{AB}}(R_B,D))$); else if $d_{R_A}\neq R_A$, she sends $\bot$ to Bob
  through the channel $C_{AB}$ (the corresponding sending action is denoted $s_{C_{AB}}(\bot)$);
  \item Bob receives $d_{AB}$ from Alice (the corresponding reading action is denoted $r_{C_{AB}}(d_{AB})$), if $d_{AB}=\bot$, he sends $\bot$ to the outside through the channel $C_{BO}$
  (the corresponding sending action is denoted $s_{C_{BO}}(\bot)$); else if $d_{AB}\neq \bot$, $d_{AB}$ must be the form of \\
  $ENC_{k_{BT}}(A,k_{AB}),ENC_{k_{AB}}(d_{R_B},D)$ (without considering
  the malicious tampering and transmission errors), he decrypts $ENC_{k_{BT}}(A,k_{AB})$ by $k_{BT}$ through an action \\
  $dec_{k_{BT}}(ENC_{k_{BT}}(A,k_{AB}))$ to ensure the message
  is from Alice and get $k_{AB}$, then he decrypts
  $ENC_{k_{AB}}(d_{R_B},D)$ by $k_{AB}$ through an action $dec_{k_{AB}}(ENC_{k_{AB}}(d_{R_B},D))$, if $d_{R_B}=R_B$, he sends $D$ to the outside through the channel $C_{BO}$
  (the corresponding sending action is denoted $s_{C_{BO}}(D)$), else if
  $d_{R_B}\neq R_B$, he sends $\bot$ to the outside through the channel $C_{BO}$ (the corresponding sending action is denoted $s_{C_{BO}}(\bot)$).
\end{enumerate}

Where $D\in\Delta$, $\Delta$ is the set of data.

Alice's state transitions described by $APTC_G$ are as follows.

$A=\sum_{D\in\Delta}r_{C_{AI}}(D)\cdot A_2$

$A_2=\{k_{AB}=NULL\}\cdot rsg_{R_A}\cdot A_3+ \{k_{AB}\neq NULL\}\cdot A_7$

$A_3=s_{C_{AB}}(A,R_A)\cdot A_4$

$A_4=r_{C_{TA}}(ENC_{k_{AT}}(B,k_{AB},d_{R_A},R_B),ENC_{k_{BT}}(A,k_{AB}))\cdot A_5$

$A_5=dec_{k_{AT}}(ENC_{k_{AT}}(B,k_{AB},d_{R_A},R_B))\cdot A_6$

$A_6=\{d_{R_A}=R_A\}\cdot A_7+\{d_{R_A}\neq R_A\}\cdot A_{9}$

$A_7=enc_{k_{AB}}(R_B,D)\cdot A_8$

$A_8=s_{C_{AB}}(ENC_{k_{BT}}(A,k_{AB}),ENC_{k_{AB}}(R_B,D))\cdot A$

$A_{9}=s_{C_{AB}}(\bot)\cdot A$

Bob's state transitions described by $APTC_G$ are as follows.

$B=\{k_{AB}=NULL\}\cdot B_1+ \{k_{AB}\neq NULL\}\cdot B_5$

$B_1=r_{C_{AB}}(A,R_A)\cdot B_2$

$B_2=rsg_{R_B}\cdot B_3$

$B_3=enc_{k_{BT}}(A,R_A,R_B)\cdot B_4$

$B_4=s_{C_{BT}}(ENC_{k_{BT}}(A,R_A,R_B))\cdot B_5$

$B_5=r_{C_{AB}}(d_{AB})\cdot B_6$

$B_6=\{d_{AB}=\bot\}\cdot s_{C_{BO}}(\bot)\cdot B+ \{d_{AB}\neq\bot\}\cdot B_7$

$B_7=dec_{k_{BT}}(ENC_{k_{BT}}(A,k_{AB}))\cdot B_8$

$B_8=dec_{k_{AB}}(ENC_{k_{AB}}(d_{R_B},D))\cdot B_9$

$B_9=\{d_{R_B}=R_B\}\cdot s_{C_{BO}}(D)\cdot B+\{d_{R_B}\neq R_B\}\cdot s_{C_{BO}}(\bot)\cdot B$

Trent's state transitions described by $APTC_G$ are as follows.

$T=r_{C_{BT}}(ENC_{k_{BT}}(A,R_A,R_B))\cdot T_2$

$T_2=dec_{k_{BT}}(ENC_{k_{BT}}(A,R_A,R_B))\cdot T_3$

$T_3=rsg_{k_{AB}}\cdot T_4$

$T_4=enc_{k_{AT}}(B,k_{AB},R_A,R_B)\cdot T_5$

$T_5=enc_{k_{BT}}(A,k_{AB})\cdot T_6$

$T_6=s_{C_{TA}}(ENC_{k_{AT}}(B,k_{AB},R_A,R_B),ENC_{k_{BT}}(A,k_{AB}))\cdot T$

The sending action and the reading action of the same type data through the same channel can communicate with each other, otherwise, will cause a deadlock $\delta$. We define the following
communication functions.

$\gamma(r_{C_{AB}}(A,R_A),s_{C_{AB}}(A,R_A))\triangleq c_{C_{AB}}(A,R_A)$

$\gamma(r_{C_{TA}}(ENC_{k_{AT}}(B,k_{AB},d_{R_A},R_B),ENC_{k_{BT}}(A,k_{AB})),\\
s_{C_{TA}}(ENC_{k_{AT}}(B,k_{AB},d_{R_A},R_B),ENC_{k_{BT}}(A,k_{AB})))\\
\triangleq c_{C_{TA}}(ENC_{k_{AT}}(B,k_{AB},d_{R_A},R_B),ENC_{k_{BT}}(A,k_{AB}))$

$\gamma(r_{C_{BT}}(ENC_{k_{BT}}(A,R_A,R_B)),s_{C_{BT}}(ENC_{k_{BT}}(A,R_A,R_B)))\triangleq c_{C_{BT}}(ENC_{k_{BT}}(A,R_A,R_B))$

$\gamma(r_{C_{AB}}(d_{AB}),s_{C_{AB}}(d_{AB}))\triangleq c_{C_{AB}}(d_{AB})$

Let all modules be in parallel, then the protocol $A\quad B\quad T$ can be presented by the following process term.

$$\tau_I(\partial_H(\Theta(A\between B\between T)))=\tau_I(\partial_H(A\between B\between T))$$

where $H=\{r_{C_{AB}}(A,R_A),s_{C_{AB}}(A,R_A),r_{C_{AB}}(d_{AB}),s_{C_{AB}}(d_{AB}),\\
r_{C_{TA}}(ENC_{k_{AT}}(B,k_{AB},d_{R_A},R_B),ENC_{k_{BT}}(A,k_{AB})),s_{C_{TA}}(ENC_{k_{AT}}(B,k_{AB},d_{R_A},R_B),\\
ENC_{k_{BT}}(A,k_{AB})),r_{C_{BT}}(ENC_{k_{BT}}(A,R_A,R_B)),s_{C_{BT}}(ENC_{k_{BT}}(A,R_A,R_B))|D\in\Delta\}$,

$I=\{c_{C_{AB}}(A,R_A),c_{C_{TA}}(ENC_{k_{AT}}(B,k_{AB},d_{R_A},R_B),ENC_{k_{BT}}(A,k_{AB})),\\
c_{C_{BT}}(ENC_{k_{BT}}(A,R_A,R_B)),c_{C_{AB}}(d_{AB}),\{k_{AB}=NULL\}, rsg_{R_A},\{k_{AB}\neq NULL\},\\
dec_{k_{AT}}(ENC_{k_{AT}}(B,k_{AB},d_{R_A},R_B)),\{d_{R_A}=R_A\},\{d_{R_A}\neq R_A\},\\
enc_{k_{AB}}(R_B,D),rsg_{R_B},enc_{k_{BT}}(A,R_A,R_B),\{d_{AB}=\bot\},\{d_{AB}\neq\bot\},\\
dec_{k_{BT}}(ENC_{k_{BT}}(A,k_{AB})),dec_{k_{AB}}(ENC_{k_{AB}}(d_{R_B},D)),\\
\{d_{R_B}=R_B\},\{d_{R_B}\neq R_B\},dec_{k_{BT}}(ENC_{k_{BT}}(A,R_A,R_B)),rsg_{k_{AB}},\\
enc_{k_{AT}}(B,k_{AB},R_A,R_B),enc_{k_{BT}}(A,k_{AB})|D\in\Delta\}$.

Then we get the following conclusion on the protocol.

\begin{theorem}
The Yahalom protocol in Figure \ref{YP7} is secure.
\end{theorem}

\begin{proof}
Based on the above state transitions of the above modules, by use of the algebraic laws of $APTC_G$, we can prove that

$\tau_I(\partial_H(A\between B\between T))=\sum_{D\in\Delta}(r_{C_{AI}}(D)\cdot (s_{C_{BO}}(\bot)+s_{C_{BO}}(D)))\cdot
\tau_I(\partial_H(A\between B\between T))$.

For the details of proof, please refer to section \ref{app}, and we omit it.

That is, the Yahalom protocol in Figure \ref{YP7} $\tau_I(\partial_H(A\between B\between T))$ can exhibit desired external behaviors:

\begin{enumerate}
  \item For information leakage, because $k_{AT}$ is privately shared only between Alice and Trent, $k_{BT}$ is privately shared only between Bob and Trent, $k_{AB}$ is privately shared
  only among Trent, Alice and Bob. For the modeling of confidentiality, it is similar to the protocol in section \ref{confi}, the Yahalom protocol is confidential;
  \item For the man-in-the-middle attack, because $k_{AT}$ is privately shared only between Alice and Trent, $k_{BT}$ is privately shared only between Bob and Trent, $k_{AB}$ is privately shared
  only among Trent, Alice and Bob, and the use of the random numbers $R_A$ and $R_B$, the protocol would be $\tau_I(\partial_H(A\between B\between T))=\sum_{D\in\Delta}(r_{C_{AI}}(D)\cdot s_{C_{BO}}(\bot))\cdot
  \tau_I(\partial_H(A\between B\between T))$, it is desired, the Yahalom protocol can be against the
  man-in-the-middle attack;
  \item Without man-in-the-middle attack, the protocol would be $\tau_I(\partial_H(A\between B\between T))=\sum_{D\in\Delta}(r_{C_{AI}}(D)\cdot s_{C_{BO}}(D))\cdot
  \tau_I(\partial_H(A\between B\between T))$, it is desired;
  \item For the unexpected and non-technical leaking of $k_{AT}$, $k_{BT}$, $k_{AB}$, or they being not strong enough, or Trent being dishonest, they are out of the scope of analyses of security protocols;
  \item For malicious tampering and transmission errors, they are out of the scope of analyses of security protocols.
\end{enumerate}
\end{proof}

\subsection{Needham-Schroeder Protocol}\label{nsp}

The Needham-Schroeder protocol shown in Figure \ref{NSP7} uses symmetric keys for secure communication, that is, the key $k_{AB}$ between Alice and Bob is privately shared to Alice and Bob,
Alice, Bob have shared keys with Trent $k_{AT}$ and $k_{BT}$ already.

\begin{figure}
    \centering
    \includegraphics{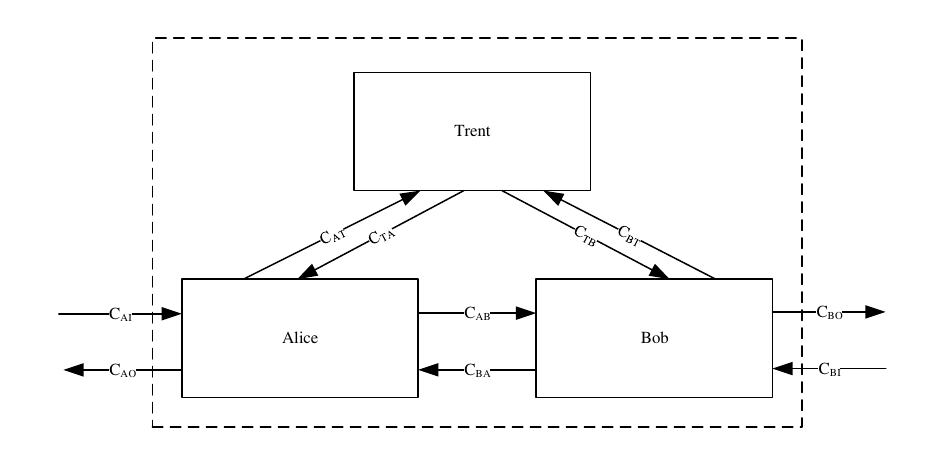}
    \caption{Needham-Schroeder protocol}
    \label{NSP7}
\end{figure}

The process of the protocol is as follows.

\begin{enumerate}
  \item Alice receives some messages $D$ from the outside through the channel $C_{AI}$ (the corresponding reading action is denoted $r_{C_{AI}}(D)$), if $k_{AB}$ is not established,
  she generates a random number $R_A$ through an action $rsg_{R_A}$, and sends $A, B, R_A$ to trent through the channel $C_{AT}$ (the corresponding sending action is denoted
  $s_{C_{AT}}(A,B,R_A)$);
  \item Trent receives $A,B,R_A$ from Alice through the channel $C_{AT}$ (the corresponding reading action is denoted $r_{C_{AT}}(A,B,R_A)$), he generates a random session key $k_{AB}$ through
  an action $rsg_{k_{AB}}$, then he encrypts $A,k_{AB}$ by $k_{BT}$ through an action
  $enc_{k_{BT}}(A,k_{AB})$, encrypts $R_A,B,k_{AB},ENC_{k_{BT}}(A,k_{AB})$ by $k_{AT}$ through an action $enc_{k_{AT}}(R_A,B,k_{AB},ENC_{k_{BT}}(A,k_{AB}))$, and sends them to Alice through the channel $C_{TA}$ (the corresponding sending
  action is denoted
  $s_{C_{TA}}(ENC_{k_{AT}}(R_A,B,k_{AB},ENC_{k_{BT}}(A,k_{AB})))$);
  \item Alice receives the message from Trent through the channel $C_{TA}$ (the corresponding reading action is denoted $r_{C_{TA}}(ENC_{k_{AT}}(d_{R_A},B,k_{AB},ENC_{k_{BT}}(A,k_{AB})))$),
  she decrypts \\
  $ENC_{k_{AT}}(d_{R_A},B,k_{AB},ENC_{k_{BT}}(A,k_{AB}))$ by $k_{AT}$ through an action \\
  $dec_{k_{AT}}(ENC_{k_{AT}}(d_{R_A},B,k_{AB},ENC_{k_{BT}}(A,k_{AB})))$, if $d_{R_A}=R_A$,
  she sends $ENC_{k_{BT}}(A,k_{AB})$
  to Bob through the channel $C_{AB}$ (the corresponding sending action is denoted \\
  $s_{C_{AB}}(ENC_{k_{BT}}(A,k_{AB}))$); else if $d_{R_A}\neq R_A$, she sends $\bot$ to Bob
  through the channel $C_{AB}$ (the corresponding sending action is denoted $s_{C_{AB}}(\bot)$);
  \item Bob receives $d_{AB}$ from Alice (the corresponding reading action is denoted $r_{C_{AB}}(d_{AB})$), if $d_{AB}=\bot$, he sends $\bot$ to the outside through the channel $C_{BO}$
  (the corresponding sending action is denoted $s_{C_{BO}}(\bot)$), and sends $\bot$ to Alice through the channel $C_{AB}$ (the corresponding sending action is denoted $s_{C_{AB}}(\bot)$);
  else if $d_{AB}\neq \bot$, $d_{AB}$ must be the form of
  $ENC_{k_{BT}}(A,k_{AB})$ (without considering
  the malicious tampering and transmission errors), he decrypts $ENC_{k_{BT}}(A,k_{AB})$ by $k_{BT}$ through an action
  $dec_{k_{BT}}(ENC_{k_{BT}}(A,k_{AB}))$ to ensure the message
  is from Alice and get $k_{AB}$, then he generates a random number $R_B$ through an action $rsg_{R_B}$, encrypts $R_B$ by $k_{AB}$ through an action $enc_{k_{AB}}(R_B)$, and sends
  $ENC_{k_{AB}}(R_B)$ to Alice through the channel $C_{BA}$ (the corresponding sending action is denoted $s_{C_{BA}}(ENC_{k_{AB}}(R_B))$);
  \item Alice receives $d_{BA}$ from Bob through the channel $C_{BA}$ (the corresponding reading action is denoted $r_{C_{BA}}(d_{BA})$), if $d_{BA}\neq \bot$, she decrypts $ENC_{k_{AB}}(R_B)$ to
  get $R_B$ by $k_{AB}$ through an action $dec_{k_{AB}}(ENC_{k_{AB}}(R_B))$, encrypts $R_B-1,D$ through an action $enc_{k_{AB}}(R_B-1,D)$, and sends $ENC_{k_{AB}}(R_B-1,D)$ to Bob through the channel $C_{AB}$ (the corresponding sending action is denoted
  $s_{C_{AB}}(ENC_{k_{AB}}(R_B-1,D))$); else if $d_{BA}=\bot$, she sends $\bot$ to Bob through the channel $C_{AB}$ (the corresponding sending action is denoted $s_{C_{AB}}(\bot)$);
  \item Bob receives $d_{AB}'$ from Alice through the channel $C_{AB}$ (the corresponding reading action is denoted $r_{C_{AB}}(d_{AB}')$), if $d_{AB}'\neq \bot$, he
  decrypts $ENC_{k_{AB}}(d_{R_B-1},D)$ by $k_{AB}$ through an action $dec_{k_{AB}}(ENC_{k_{AB}}(d_{R_B-1},D))$, if $d_{R_B-1}=R_B-1$, he sends $D$ to the outside through the channel $C_{BO}$
  (the corresponding sending action is denoted $s_{C_{BO}}(D)$), else if
  $d_{R_B-1}\neq R_B-1$, he sends $\bot$ to the outside through the channel $C_{BO}$ (the corresponding sending action is denoted $s_{C_{BO}}(\bot)$);else if
  $d_{AB}'=\bot$, he sends $\bot$ to the outside through the channel $C_{BO}$ (the corresponding sending action is denoted $s_{C_{BO}}(\bot)$).
\end{enumerate}

Where $D\in\Delta$, $\Delta$ is the set of data.

Alice's state transitions described by $APTC_G$ are as follows.

$A=\sum_{D\in\Delta}r_{C_{AI}}(D)\cdot A_2$

$A_2=\{k_{AB}=NULL\}\cdot rsg_{R_A}\cdot A_3+ \{k_{AB}\neq NULL\}\cdot A_{11}$

$A_3=s_{C_{AT}}(A,B,R_A)\cdot A_4$

$A_4=r_{C_{TA}}(ENC_{k_{AT}}(d_{R_A},B,k_{AB},ENC_{k_{BT}}(A,k_{AB})))\cdot A_5$

$A_5=dec_{k_{AT}}(enc_{k_{AT}}(d_{R_A},B,k_{AB},ENC_{k_{BT}}(A,k_{AB})))\cdot A_6$

$A_6=\{d_{R_A}=R_A\}\cdot A_7+\{d_{R_A}\neq R_A\}\cdot s_{C_{AB}}(\bot)\cdot A_{8}$

$A_7=s_{C_{AB}}(ENC_{k_{BT}}(A,k_{AB}))\cdot A_8$

$A_8=r_{C_{BA}}(d_{BA})\cdot A_9$

$A_9=\{d_{BA}\neq \bot\}\cdot A_{10}+\{d_{BA}= \bot\}\cdot A_{13}$

$A_{10}=dec_{k_{AB}}(ENC_{k_{AB}}(R_B))\cdot A_{11}$

$A_{11}=enc_{k_{AB}}(R_B-1,D)\cdot A_{12}$

$A_{12}=s_{C_{AB}}(ENC_{k_{AB}}(R_B-1,D))\cdot A$

$A_{13}=s_{C_{AB}}(\bot)\cdot A$

Bob's state transitions described by $APTC_G$ are as follows.

$B=\{k_{AB}=NULL\}\cdot B_1+ \{k_{AB}\neq NULL\}\cdot B_7$

$B_1=r_{C_{AB}}(d_{AB})\cdot B_2$

$B_2=\{d_{AB}=\bot\}\cdot(s_{C_{BO}}(\bot)\parallel s_{C_{AB}}(\bot))\cdot B_7+\{d_{AB}\neq \bot\}\cdot B_3$

$B_3=dec_{k_{BT}}(ENC_{k_{BT}}(A,k_{AB}))\cdot B_4$

$B_4=rsg_{R_B}\cdot B_5$

$B_5=enc_{k_{AB}}(R_B)\cdot B_6$

$B_6=s_{C_{BA}}(ENC_{k_{AB}}(R_B))\cdot B_7$

$B_7=r_{C_{AB}}(d_{AB}')\cdot B_8$

$B_8=\{d_{AB}'\neq \bot\}\cdot B_9+\{d_{AB}'= \bot\}\cdot s_{C_{BO}}(\bot)\cdot B$

$B_9=dec_{k_{AB}}(ENC_{k_{AB}}(d_{R_B-1},D))\cdot B_{10}$

$B_{10}=\{d_{R_B-1}=R_B-1\}\cdot s_{C_{BO}}(D)\cdot B+\{d_{R_B-1}\neq R_B-1\}\cdot s_{C_{BO}}(\bot)\cdot B$

Trent's state transitions described by $APTC_G$ are as follows.

$T=r_{C_{AT}}(A,B,R_A)\cdot T_2$

$T_2=rsg_{k_{AB}}\cdot T_3$

$T_3=enc_{k_{BT}}(A,k_{AB})\cdot T_4$

$T_4=enc_{k_{AT}}(R_A,B,k_{AB},ENC_{k_{BT}}(A,k_{AB}))\cdot T_5$

$T_5=s_{C_{TA}}(ENC_{k_{AT}}(R_A,B,k_{AB},ENC_{k_{BT}}(A,k_{AB})))\cdot T$

The sending action and the reading action of the same type data through the same channel can communicate with each other, otherwise, will cause a deadlock $\delta$. We define the following
communication functions.

$\gamma(r_{C_{AT}}(A,B,R_A),s_{C_{AT}}(A,B,R_A))\triangleq c_{C_{AT}}(A,B,R_A)$

$\gamma(r_{C_{TA}}(ENC_{k_{AT}}(d_{R_A},B,k_{AB},ENC_{k_{BT}}(A,k_{AB}))),\\
s_{C_{TA}}(ENC_{k_{AT}}(d_{R_A},B,k_{AB},ENC_{k_{BT}}(A,k_{AB}))))\\
\triangleq c_{C_{TA}}(ENC_{k_{AT}}(d_{R_A},B,k_{AB},ENC_{k_{BT}}(A,k_{AB})))$

$\gamma(r_{C_{AB}}(d_{AB}),s_{C_{AB}}(d_{AB}))\triangleq c_{C_{AB}}(d_{AB})$

$\gamma(r_{C_{BA}}(d_{BA}),s_{C_{BA}}(d_{BA}))\triangleq c_{C_{BA}}(d_{BA})$

$\gamma(r_{C_{AB}}(d_{AB}'),s_{C_{AB}}(d_{AB}'))\triangleq c_{C_{AB}}(d_{AB}')$

Let all modules be in parallel, then the protocol $A\quad B\quad T$ can be presented by the following process term.

$$\tau_I(\partial_H(\Theta(A\between B\between T)))=\tau_I(\partial_H(A\between B\between T))$$

where $H=\{r_{C_{AT}}(A,B,R_A),s_{C_{AT}}(A,B,R_A),r_{C_{AB}}(d_{AB}),s_{C_{AB}}(d_{AB}),\\
r_{C_{BA}}(d_{BA}),s_{C_{BA}}(d_{BA}),r_{C_{AB}}(d_{AB}'),s_{C_{AB}}(d_{AB}'),\\
r_{C_{TA}}(ENC_{k_{AT}}(d_{R_A},B,k_{AB},ENC_{k_{BT}}(A,k_{AB}))),\\
s_{C_{TA}}(ENC_{k_{AT}}(d_{R_A},B,k_{AB},ENC_{k_{BT}}(A,k_{AB})))|D\in\Delta\}$,

$I=\{c_{C_{AT}}(A,B,R_A),c_{C_{AB}}(d_{AB}),c_{C_{BA}}(d_{BA}),c_{C_{AB}}(d_{AB}'),\\
c_{C_{TA}}(ENC_{k_{AT}}(d_{R_A},B,k_{AB},ENC_{k_{BT}}(A,k_{AB}))),\\
\{k_{AB}=NULL\}, rsg_{R_A},\{k_{AB}\neq NULL\},\\
dec_{k_{AT}}(enc_{k_{AT}}(d_{R_A},B,k_{AB},ENC_{k_{BT}}(A,k_{AB}))),\\
\{d_{R_A}=R_A\},\{d_{R_A}\neq R_A\},\{d_{BA}\neq \bot\},\{d_{BA}= \bot\},\\
dec_{k_{AB}}(ENC_{k_{AB}}(R_B)),enc_{k_{AB}}(R_B-1,D),\{d_{AB}=\bot\},\{d_{AB}\neq\bot\},\\
dec_{k_{BT}}(ENC_{k_{BT}}(A,k_{AB})),rsg_{R_B},enc_{k_{AB}}(R_B),\\
\{d_{AB}'= \bot\},\{d_{AB}'\neq \bot\},dec_{k_{AB}}(ENC_{k_{AB}}(d_{R_B-1},D)),\\
\{d_{R_B-1}=R_B-1\},\{d_{R_B-1}\neq R_B-1\},rsg_{k_{AB}},enc_{k_{BT}}(A,k_{AB}),\\
enc_{k_{AT}}(R_A,B,k_{AB},ENC_{k_{BT}}(A,k_{AB}))|D\in\Delta\}$.

Then we get the following conclusion on the protocol.

\begin{theorem}
The Needham-Schroeder protocol in Figure \ref{NSP7} is secure.
\end{theorem}

\begin{proof}
Based on the above state transitions of the above modules, by use of the algebraic laws of $APTC_G$, we can prove that

$\tau_I(\partial_H(A\between B\between T))=\sum_{D\in\Delta}(r_{C_{AI}}(D)\cdot (s_{C_{BO}}(\bot)+s_{C_{BO}}(D)))\cdot
\tau_I(\partial_H(A\between B\between T))$.

For the details of proof, please refer to section \ref{app}, and we omit it.

That is, the Needham-Schroeder protocol in Figure \ref{NSP7} $\tau_I(\partial_H(A\between B\between T))$ can exhibit desired external behaviors:

\begin{enumerate}
  \item For information leakage, because $k_{AT}$ is privately shared only between Alice and Trent, $k_{BT}$ is privately shared only between Bob and Trent, $k_{AB}$ is privately shared
  only among Trent, Alice and Bob. For the modeling of confidentiality, it is similar to the protocol in section \ref{confi}, the Needham-Schroeder protocol is confidential;
  \item For replay attack, the using of random numbers $R_A$, $R_B$, makes that $\tau_I(\partial_H(A\between B\between T))=\sum_{D\in\Delta}(r_{C_{AI}}(D)\cdot s_{C_{BO}}(\bot))\cdot
\tau_I(\partial_H(A\between B\between T))$, it is desired;
  \item Without replay attack, the protocol would be $\tau_I(\partial_H(A\between B\between T))=\sum_{D\in\Delta}(r_{C_{AI}}(D)\cdot s_{C_{BO}}(D))\cdot
\tau_I(\partial_H(A\between B\between T))$, it is desired;
  \item For the man-in-the-middle attack, because $k_{AT}$ is privately shared only between Alice and Trent, $k_{BT}$ is privately shared only between Bob and Trent, $k_{AB}$ is privately shared
  only among Trent, Alice and Bob. For the modeling of the man-in-the-middle attack, it is similar to the protocol in section \ref{MIM}, the Needham-Schroeder protocol can be against the
  man-in-the-middle attack;
  \item For the unexpected and non-technical leaking of $k_{AT}$, $k_{BT}$, $k_{AB}$, or they being not strong enough, or Trent being dishonest, they are out of the scope of analyses of security protocols;
  \item For malicious tampering and transmission errors, they are out of the scope of analyses of security protocols.
\end{enumerate}
\end{proof}

\subsection{Otway-Rees Protocol}\label{orp}

The Otway-Rees protocol shown in Figure \ref{ORP7} uses symmetric keys for secure communication, that is, the key $k_{AB}$ between Alice and Bob is privately shared to Alice and Bob,
Alice, Bob have shared keys with Trent $k_{AT}$ and $k_{BT}$ already.

\begin{figure}
    \centering
    \includegraphics{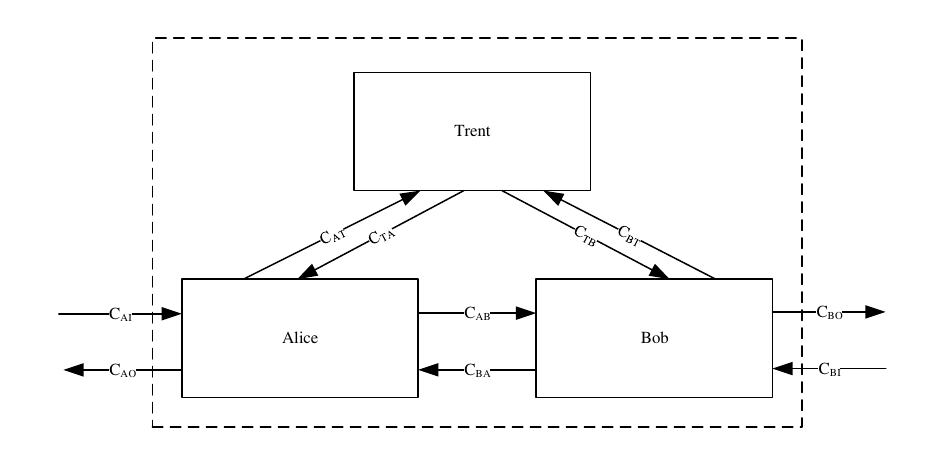}
    \caption{Otway-Rees protocol}
    \label{ORP7}
\end{figure}

The process of the protocol is as follows.

\begin{enumerate}
  \item Alice receives some messages $D$ from the outside through the channel $C_{AI}$ (the corresponding reading action is denoted $r_{C_{AI}}(D)$), if $k_{AB}$ is not established,
  she generates the random numbers $I$, $R_A$ through the actions $rsg_{I}$ and $rsg_{R_A}$, encrypts $R_A, I, A, B$ by $k_{AT}$ through an action $enc_{k_{AT}}(R_A,I,A,B)$,
  and sends $I,A,B,ENC_{k_{AT}}(R_A,I,A,B)$ to Bob through the channel $C_{AB}$ (the corresponding sending action is denoted\\
  $s_{C_{AB}}(I,A,B,ENC_{k_{AT}}(R_A,I,A,B))$);
  \item Bob receives $I,A,B,ENC_{k_{AT}}(R_A,I,A,B)$ from Alice through the channel $C_{AB}$ (the corresponding reading action is denoted $r_{C_{AB}}(I,A,B,ENC_{k_{AT}}(R_A,I,A,B))$),
  he generates a random number $R_B$ through
  an action $rsg_{R_B}$, encrypts $R_B,I,A,B$ by $k_{BT}$ through an action $enc_{k_{BT}}(R_B,I,A,B)$, and sends $I,A,B,ENC_{k_{AT}}(R_A,I,A,B),ENC_{k_{BT}}(R_B,I,A,B)$ to Trent
  through the channel $C_{BT}$
  (the corresponding sending action is denoted \\
  $s_{C_{BT}}(I,A,B,ENC_{k_{AT}}(R_A,I,A,B),ENC_{k_{BT}}(R_B,I,A,B))$);
  \item Trent receives $I,A,B,ENC_{k_{AT}}(R_A,I,A,B),ENC_{k_{BT}}(R_B,I,A,B)$ through the channel $C_{BT}$ (the corresponding reading action is denoted\\
  $r_{C_{BT}}(I,A,B,ENC_{k_{AT}}(R_A,I,A,B),ENC_{k_{BT}}(R_B,I,A,B))$), he decrypts the message \\
  $ENC_{k_{AT}}(R_A,I,A,B)$
  through an action $dec_{k_{AT}}(ENC_{k_{AT}}(R_A,I,A,B))$ and the message $ENC_{k_{BT}}(R_B,I,A,B)$ through an action $dec_{k_{BT}}(ENC_{k_{BT}}(R_B,I,A,B))$, generates a random session
  key $k_{AB}$ through an action $rsg_{k_{AB}}$, then he encrypts $R_A,k_{AB}$ by $k_{AT}$ through an action
  $enc_{k_{AT}}(R_A,k_{AB})$, encrypts $R_B,k_{AB}$ by $k_{BT}$ through an action $enc_{k_{BT}}(R_B,k_{AB})$, and sends them to Bob through the channel $C_{TB}$ (the corresponding
  sending action is denoted $s_{C_{TB}}(I,ENC_{k_{AT}}(R_A,k_{AB}),ENC_{k_{BT}}(R_B,k_{AB}))$);
  \item Bob receives the message from Trent through the channel $C_{TB}$ (the corresponding reading action is denoted $r_{C_{TB}}(d_I,ENC_{k_{AT}}(R_A,k_{AB}),ENC_{k_{BT}}(d_{R_B},k_{AB}))$),
  he decrypts \\
  $ENC_{k_{BT}}(d_{R_B},k_{AB})$ by $k_{BT}$ through an action $dec_{k_{BT}}(ENC_{k_{BT}}(d_{R_B},k_{AB}))$, if $d_{R_B}=R_B$ and $d_I=I$, he sends
  $I,ENC_{k_{AT}}(R_A,k_{AB})$
  to Alice through the channel $C_{BA}$ (the corresponding sending action is denoted $s_{C_{BA}}(I,ENC_{k_{AB}}(R_A,k_{AB}))$); else if $d_{R_B}\neq R_B$ or $D_I\neq I$, he sends $\bot$ to Alice
  through the channel $C_{BA}$ (the corresponding sending action is denoted $s_{C_{BA}}(\bot)$);
  \item Alice receives $d_{BA}$ from Bob (the corresponding reading action is denoted $r_{C_{BA}}(d_{BA})$), if $d_{BA}=\bot$, she sends $\bot$ to Bob through the channel $C_{AB}$
  (the corresponding sending action is denoted $s_{C_{AB}}(\bot)$); else if $d_{BA}\neq \bot$, she decrypts $ENC_{k_{AT}}(R_A,k_{AB})$ by $k_{AT}$ through an action
  $dec_{k_{AT}}(ENC_{k_{AT}}(R_A,k_{AB}))$, if $d_{R_A}=R_A$ and $d_I=I$, she generates a random number $R_D$ through an action $rsg_{R_D}$, encrypts $R_D,D$ by $k_{AB}$ through an
  action $enc_{k_{AB}}(R_D,D)$, and sends it to Bob through the channel $C_{AB}$ (the corresponding sending action is denoted $s_{C_{AB}}(ENC_{k_{AB}}(R_D,D))$), else if
  $d_{R_A}\neq R_A$ or $d_I\neq I$, she sends $\bot$ to Bob through the channel $C_{AB}$ (the corresponding sending action is denoted $s_{C_{AB}}(\bot)$);
  \item Bob receives $d_{AB}$ from Alice (the corresponding reading action is denoted $r_{C_{AB}}(d_{AB})$), if $d_{AB}=\bot$, he sends $\bot$ to the outside through the channel $C_{BO}$
  (the corresponding sending action is denoted $s_{C_{BO}}(\bot)$); else if $d_{AB}\neq \bot$, she decrypts $ENC_{k_{AB}}(R_D,D)$ by $k_{AB}$ through an action
  $dec_{k_{AB}}(ENC_{k_{AB}}(R_D,D))$, if $isFresh(R_D)=TRUE$, she sends $D$ to the outside through the channel $C_{BO}$ (the corresponding sending action is denoted $s_{C_{BO}}(D)$), else if
  $isFresh(d_{R_D})=FALSE$, he sends $\bot$ to the outside through the channel $C_{BO}$ (the corresponding sending action is denoted $s_{C_{BO}}(\bot)$).
\end{enumerate}

Where $D\in\Delta$, $\Delta$ is the set of data.

Alice's state transitions described by $APTC_G$ are as follows.

$A=\sum_{D\in\Delta}r_{C_{AI}}(D)\cdot A_2$

$A_2=\{k_{AB}=NULL\}\cdot rsg_I \cdot rsg_{R_A}\cdot A_3+ \{k_{AB}\neq NULL\}\cdot A_9$

$A_3=enc_{k_{AT}}(R_A,I,A,B)\cdot A_4$

$A_4=s_{C_{AB}}(I,A,B,ENC_{k_{AT}}(R_A,I,A,B))\cdot A_5$

$A_5=r_{C_{BA}}(d_{BA})\cdot A_6$

$A_6=\{d_{BA}\neq \bot\}\cdot A_7+\{d_{BA}=\bot\}\cdot s_{C_{AB}}(\bot)\cdot A$

$A_7=dec_{k_{AT}}(ENC_{k_{AT}}(R_A,k_{AB}))\cdot A_8$

$A_8=\{d_{R_A}=R_A\cdot d_I=I\}\cdot A_9+\{d_{R_A}\neq R_A+ d_I\neq I\}\cdot A_{12}$

$A_{9}=rsg_{R_D}\cdot A_{10}$

$A_{10}=enc_{k_{AB}}(R_D,D)\cdot A_{11}$

$A_{11}=s_{C_{AB}}(ENC_{k_{AB}}(R_D,D))\cdot A$

$A_{12}=s_{C_{AB}}(\bot)\cdot A$

Bob's state transitions described by $APTC_G$ are as follows.

$B=\{k_{AB}=NULL\}\cdot B_1+ \{k_{AB}\neq NULL\}\cdot B_8$

$B_1=r_{C_{AB}}(I,A,B,ENC_{k_{AT}}(R_A,I,A,B))\cdot B_2$

$B_2=rsg_{R_B}\cdot B_3$

$B_3=enc_{k_{BT}}(R_B,I,A,B)\cdot B_4$

$B_4=s_{C_{BT}}(I,A,B,ENC_{k_{AT}}(R_A,I,A,B),ENC_{k_{BT}}(R_B,I,A,B))\cdot B_5$

$B_5=r_{C_{TB}}(d_I,ENC_{k_{AT}}(R_A,k_{AB}),ENC_{k_{BT}}(d_{R_B},k_{AB}))\cdot B_6$

$B_6=dec_{k_{BT}}(ENC_{k_{BT}}(d_{R_B},k_{AB}))\cdot B_7$

$B_7=\{d_{R_B}=R_B\cdot d_I=I\}\cdot s_{C_{BA}}(I,ENC_{k_{AB}}(R_A,k_{AB}))\cdot B_8+\{d_{R_B}\neq R_B+ d_I\neq I\}\cdot s_{C_{AB}}(\bot)\cdot B_8$

$B_8=r_{C_{AB}}(d_{AB})\cdot B_9$

$B_9=\{d_{AB}=\bot\}\cdot s_{C_{BO}}(\bot)\cdot B+\{d_{AB}\neq \bot\}\cdot B_{10}$

$B_{10}=dec_{k_{AB}}(ENC_{k_{AB}}(R_D,D))\cdot B_{11}$

$B_{11}=\{isFresh(R_D)=TRUE\}\cdot B_{12}+\{isFresh(R_D)=FLASE\}\cdot s_{C_{BO}}(\bot)\cdot B$

$B_{12}=s_{C_{BO}}(D)\cdot B$

Trent's state transitions described by $APTC_G$ are as follows.

$T=r_{C_{BT}}(I,A,B,ENC_{k_{AT}}(R_A,I,A,B),ENC_{k_{BT}}(R_B,I,A,B))\cdot T_2$

$T_2=dec_{k_{AT}}(ENC_{k_{AT}}(R_A,I,A,B))\cdot T_3$

$T_3=dec_{k_{BT}}(ENC_{k_{BT}}(R_B,I,A,B))\cdot T_4$

$T_4=rsg_{k_{AB}}\cdot T_5$

$T_5=enc_{k_{AT}}(R_A,k_{AB})\cdot T_6$

$T_6=enc_{k_{BT}}(R_B,k_{AB})\cdot T_7$

$T_7=s_{C_{TB}}(I,ENC_{k_{AT}}(R_A,k_{AB}),ENC_{k_{BT}}(R_B,k_{AB}))\cdot T$

The sending action and the reading action of the same type data through the same channel can communicate with each other, otherwise, will cause a deadlock $\delta$. We define the following
communication functions.

$\gamma(r_{C_{AB}}(I,A,B,ENC_{k_{AT}}(R_A,I,A,B)),s_{C_{AB}}(I,A,B,ENC_{k_{AT}}(R_A,I,A,B)))\\
\triangleq c_{C_{AB}}(I,A,B,ENC_{k_{AT}}(R_A,I,A,B))$

$\gamma(r_{C_{BA}}(d_{BA}),s_{C_{BA}}(d_{BA}))\triangleq c_{C_{BA}}(d_{BA})$

$\gamma(r_{C_{BT}}(I,A,B,ENC_{k_{AT}}(R_A,I,A,B),ENC_{k_{BT}}(R_B,I,A,B)),\\
s_{C_{BT}}(I,A,B,ENC_{k_{AT}}(R_A,I,A,B),ENC_{k_{BT}}(R_B,I,A,B)))\\
\triangleq c_{C_{BT}}(I,A,B,ENC_{k_{AT}}(R_A,I,A,B),ENC_{k_{BT}}(R_B,I,A,B))$

$\gamma(r_{C_{TB}}(d_I,ENC_{k_{AT}}(R_A,k_{AB}),ENC_{k_{BT}}(d_{R_B},k_{AB})),\\
s_{C_{TB}}(d_I,ENC_{k_{AT}}(R_A,k_{AB}),ENC_{k_{BT}}(d_{R_B},k_{AB})))\\
\triangleq c_{C_{TB}}(d_I,ENC_{k_{AT}}(R_A,k_{AB}),ENC_{k_{BT}}(d_{R_B},k_{AB}))$

$\gamma(r_{C_{AB}}(d_{AB}),s_{C_{AB}}(d_{AB}))\triangleq c_{C_{AB}}(d_{AB})$

Let all modules be in parallel, then the protocol $A\quad B\quad T$ can be presented by the following process term.

$$\tau_I(\partial_H(\Theta(A\between B\between T)))=\tau_I(\partial_H(A\between B\between T))$$

where $H=\{r_{C_{AB}}(I,A,B,ENC_{k_{AT}}(R_A,I,A,B)),s_{C_{AB}}(I,A,B,ENC_{k_{AT}}(R_A,I,A,B)),\\
r_{C_{BA}}(d_{BA}),s_{C_{BA}}(d_{BA}),r_{C_{AB}}(d_{AB}),s_{C_{AB}}(d_{AB}),\\
r_{C_{BT}}(I,A,B,ENC_{k_{AT}}(R_A,I,A,B),ENC_{k_{BT}}(R_B,I,A,B)),\\
s_{C_{BT}}(I,A,B,ENC_{k_{AT}}(R_A,I,A,B),ENC_{k_{BT}}(R_B,I,A,B)),\\
r_{C_{TB}}(d_I,ENC_{k_{AT}}(R_A,k_{AB}),ENC_{k_{BT}}(d_{R_B},k_{AB})),\\
s_{C_{TB}}(d_I,ENC_{k_{AT}}(R_A,k_{AB}),ENC_{k_{BT}}(d_{R_B},k_{AB}))|D\in\Delta\}$,

$I=\{c_{C_{AB}}(I,A,B,ENC_{k_{AT}}(R_A,I,A,B)),c_{C_{BA}}(d_{BA}),c_{C_{AB}}(d_{AB}),\\
c_{C_{BT}}(I,A,B,ENC_{k_{AT}}(R_A,I,A,B),ENC_{k_{BT}}(R_B,I,A,B)),\\
c_{C_{TB}}(d_I,ENC_{k_{AT}}(R_A,k_{AB}),ENC_{k_{BT}}(d_{R_B},k_{AB})),\\
\{k_{AB}=NULL\}, rsg_I, rsg_{R_A},\{k_{AB}\neq NULL\},enc_{k_{AT}}(R_A,I,A,B),\\
\{d_{BA}\neq \bot\},\{d_{BA}= \bot\},dec_{k_{AT}}(ENC_{k_{AT}}(R_A,k_{AB})),\\
\{d_{R_A}=R_A\cdot d_I=I\},\{d_{R_A}\neq R_A+ d_I\neq I\},rsg_{R_D},\\
enc_{k_{AB}}(R_D,D),rsg_{R_B},enc_{k_{BT}}(R_B,I,A,B),\\
dec_{k_{BT}}(ENC_{k_{BT}}(d_{R_B},k_{AB})),\{d_{R_B}=R_B\cdot d_I=I\},\\
\{d_{R_B}\neq R_B+ d_I\neq I\},\{d_{AB}=\bot\},\{d_{AB}\neq \bot\},\\
dec_{k_{AB}}(ENC_{k_{AB}}(R_D,D)),\{isFresh(R_D)=TRUE\},\{isFresh(R_D)=FALSE\},\\
dec_{k_{AT}}(ENC_{k_{AT}}(R_A,I,A,B)),dec_{k_{BT}}(ENC_{k_{BT}}(R_B,I,A,B)),\\
rsg_{k_{AB}},enc_{k_{AT}}(R_A,k_{AB}),enc_{k_{BT}}(R_B,k_{AB})|D\in\Delta\}$.

Then we get the following conclusion on the protocol.

\begin{theorem}
The Otway-Rees protocol in Figure \ref{ORP7} is secure.
\end{theorem}

\begin{proof}
Based on the above state transitions of the above modules, by use of the algebraic laws of $APTC_G$, we can prove that

$\tau_I(\partial_H(A\between B\between T))=\sum_{D\in\Delta}(r_{C_{AI}}(D)\cdot (s_{C_{BO}}(\bot)+s_{C_{BO}}(D)))\cdot
\tau_I(\partial_H(A\between B\between T))$.

For the details of proof, please refer to section \ref{app}, and we omit it.

That is, the Otway-Rees protocol in Figure \ref{ORP7} $\tau_I(\partial_H(A\between B\between T))$ can exhibit desired external behaviors:

\begin{enumerate}
  \item For information leakage, because $k_{AT}$ is privately shared only between Alice and Trent, $k_{BT}$ is privately shared only between Bob and Trent, $k_{AB}$ is privately shared
  only among Trent, Alice and Bob. For the modeling of confidentiality, it is similar to the protocol in section \ref{confi}, the Otway-Rees protocol is confidential;
  \item For the man-in-the-middle attack, because $k_{AT}$ is privately shared only between Alice and Trent, $k_{BT}$ is privately shared only between Bob and Trent, $k_{AB}$ is privately shared
  only among Trent, Alice and Bob, and the use of the random numbers $I$, $R_A$ and $R_B$, the protocol would be $\tau_I(\partial_H(A\between B\between T))=\sum_{D\in\Delta}(r_{C_{AI}}(D)\cdot s_{C_{BO}}(\bot))\cdot
  \tau_I(\partial_H(A\between B\between T))$, it is desired, the Otway-Rees protocol can be against the
  man-in-the-middle attack;
  \item For replay attack, the using of the random numbers $I$, $R_A$ and $R_B$, makes that $\tau_I(\partial_H(A\between B\between T))=\sum_{D\in\Delta}(r_{C_{AI}}(D)\cdot s_{C_{BO}}(\bot))\cdot
  \tau_I(\partial_H(A\between B\between T))$, it is desired;
  \item Without man-in-the-middle and replay attack, the protocol would be $\tau_I(\partial_H(A\between B\between T))=\sum_{D\in\Delta}(r_{C_{AI}}(D)\cdot s_{C_{BO}}(D))\cdot
  \tau_I(\partial_H(A\between B\between T))$, it is desired;
  \item For the unexpected and non-technical leaking of $k_{AT}$, $k_{BT}$, $k_{AB}$, or they being not strong enough, or Trent being dishonest, they are out of the scope of analyses of security protocols;
  \item For malicious tampering and transmission errors, they are out of the scope of analyses of security protocols.
\end{enumerate}
\end{proof}

\subsection{Kerberos Protocol}\label{kp}

The Kerberos protocol shown in Figure \ref{KP7} uses symmetric keys for secure communication, that is, the key $k_{AB}$ between Alice and Bob is privately shared to Alice and Bob,
Alice, Bob have shared keys with Trent $k_{AT}$ and $k_{BT}$ already.

\begin{figure}
    \centering
    \includegraphics{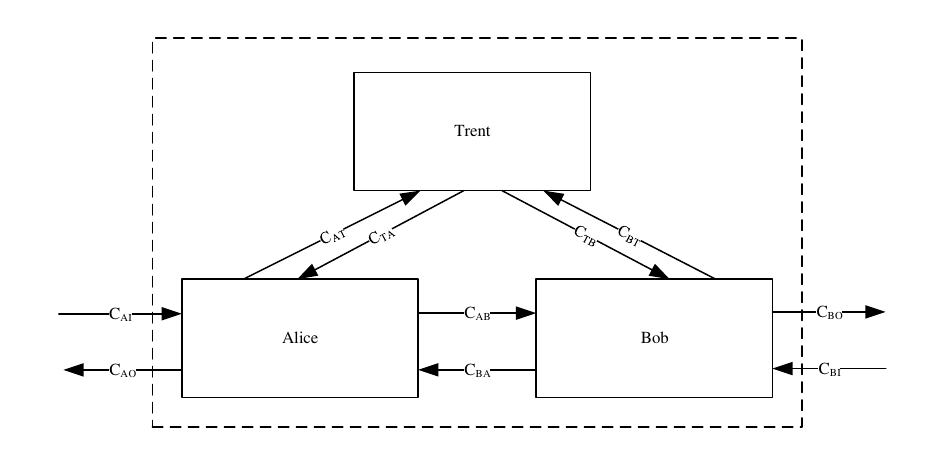}
    \caption{Kerberos protocol}
    \label{KP7}
\end{figure}

The process of the protocol is as follows.

\begin{enumerate}
  \item Alice receives some messages $D$ from the outside through the channel $C_{AI}$ (the corresponding reading action is denoted $r_{C_{AI}}(D)$), if $k_{AB}$ is not established,
  she sends $A, B$ to trent through the channel $C_{AT}$ (the corresponding sending action is denoted
  $s_{C_{AT}}(A,B)$);
  \item Trent receives $A,B$ from Alice through the channel $C_{AT}$ (the corresponding reading action is denoted $r_{C_{AT}}(A,B)$), he generates a random session key $k_{AB}$ through
  an action $rsg_{k_{AB}}$, get time stamp $T$ and lifetime $L$, then he encrypts $T,L,k_{AB},B$ by $k_{AT}$ through an action
  $enc_{k_{AT}}(T,L,k_{AB},B)$, encrypts $T,L,k_{AB},A$ by $k_{BT}$ through an action $enc_{k_{BT}}(T,L,k_{AB},A)$, and sends them to Alice through the channel $C_{TA}$ (the corresponding sending
  action is denoted
  $s_{C_{TA}}(ENC_{k_{AT}}(T,L,k_{AB},B),ENC_{k_{BT}}(K,L,k_{AB},A))$);
  \item Alice receives the message from Trent through the channel $C_{TA}$ (the corresponding reading action is denoted $r_{C_{TA}}(ENC_{k_{AT}}(T,L,k_{AB},B),ENC_{k_{BT}}(K,L,k_{AB},A))$),
  she decrypts
  $ENC_{k_{AT}}(T,L,k_{AB},B)$ by $k_{AT}$ through an action
  $dec_{k_{AT}}(ENC_{k_{AT}}(T,L,k_{AB},B))$, encrypts $A,T$ by $k_{AB}$ through an action $enc_{k_{AB}}(A,T)$, she sends \\
  $ENC_{k_{AB}}(A,T),ENC_{k_{BT}}(T,L,k_{AB},A)$
  to Bob through the channel $C_{AB}$ (the corresponding sending action is denoted $s_{C_{AB}}(ENC_{k_{AB}}(A,T),ENC_{k_{BT}}(T,L,k_{AB},A))$);
  \item Bob receives $ENC_{k_{AB}}(A,T),ENC_{k_{BT}}(T,L,k_{AB},A)$ from Alice (the corresponding reading action is denoted $r_{C_{AB}}(ENC_{k_{AB}}(A,T),ENC_{k_{BT}}(T,L,k_{AB},A))$),
  he decrypts \\
  $ENC_{k_{BT}}(T,L,k_{AB},A)$ by $k_{BT}$ through an action
  $dec_{k_{BT}}(ENC_{k_{BT}}(T,L,k_{AB},A))$ to ensure the message
  is from Alice and get $k_{AB}$, and decrypts $ENC_{k_{AB}}(A,T)$ by $k_{AB}$ through an action $dec_{k_{AB}}(ENC_{k_{AB}}(A,T))$ to get $A$ and $T$, then he
  encrypts $T+1$ by $k_{AB}$ through an action $enc_{k_{AB}}(T+1)$, and sends
  $ENC_{k_{AB}}(T+1)$ to Alice through the channel $C_{BA}$ (the corresponding sending action is denoted $s_{C_{BA}}(ENC_{k_{AB}}(T+1))$);
  \item Alice receives $ENC_{k_{AB}}(d_{T+1})$ from Bob through the channel $C_{BA}$ (the corresponding reading action is denoted $r_{C_{BA}}(ENC_{k_{AB}}(d_{T+1}))$),
  he decrypts $ENC_{k_{AB}}(d_{T+1})$ by $k_{AB}$ through an action $dec_{k_{AB}}(ENC_{k_{AB}}(d_{T+1}))$, if $d_{T+1}=T+1$, she encrypts $T+2,D$ through an action $enc_{k_{AB}}(T+2,D)$,
  and sends $ENC_{k_{AB}}(T+2,D)$ to Bob through the channel $C_{AB}$ (the corresponding sending action is denoted
  $s_{C_{AB}}(ENC_{k_{AB}}(T+2,D))$); else if $d_{T+1}\neq T+1$, she sends $\bot$ to Bob through the channel $C_{AB}$ (the corresponding sending action is denoted $s_{C_{AB}}(\bot)$);
  \item Bob receives $d_{AB}$ from Alice through the channel $C_{AB}$ (the corresponding reading action is denoted $r_{C_{AB}}(d_{AB})$), if $d_{AB}\neq \bot$, he
  decrypts $ENC_{k_{AB}}(d_{T+2},D)$ by $k_{AB}$ through an action $dec_{k_{AB}}(ENC_{k_{AB}}(d_{T+2},D))$, if $d_{T+2}=T+2$, he sends $D$ to the outside through the channel $C_{BO}$
  (the corresponding sending action is denoted $s_{C_{BO}}(D)$), else if
  $d_{T+2}\neq T+2$, he sends $\bot$ to the outside through the channel $C_{BO}$ (the corresponding sending action is denoted $s_{C_{BO}}(\bot)$);else if
  $d_{AB}=\bot$, he sends $\bot$ to the outside through the channel $C_{BO}$ (the corresponding sending action is denoted $s_{C_{BO}}(\bot)$).
\end{enumerate}

Where $D\in\Delta$, $\Delta$ is the set of data.

Alice's state transitions described by $APTC_G$ are as follows.

$A=\sum_{D\in\Delta}r_{C_{AI}}(D)\cdot A_2$

$A_2=\{k_{AB}=NULL\}\cdot A_3+ \{k_{AB}\neq NULL\}\cdot A_{12}$

$A_3=s_{C_{AT}}(A,B)\cdot A_4$

$A_4=r_{C_{TA}}(ENC_{k_{AT}}(T,L,k_{AB},B),ENC_{k_{BT}}(K,L,k_{AB},A))\cdot A_5$

$A_5=dec_{k_{AT}}(ENC_{k_{AT}}(T,L,k_{AB},B))\cdot A_6$

$A_6=enc_{k_{AB}}(A,T)\cdot A_{7}$

$A_7=s_{C_{AB}}(ENC_{k_{AB}}(A,T),ENC_{k_{BT}}(T,L,k_{AB},A))\cdot A_8$

$A_8=r_{C_{BA}}(ENC_{k_{AB}}(d_{T+1}))\cdot A_9$

$A_9=dec_{k_{AB}}(ENC_{k_{AB}}(d_{T+1}))\cdot A_{10}$

$A_{10}=dec_{k_{AB}}(ENC_{k_{AB}}(R_B))\cdot A_{11}$

$A_{11}=\{d_{T+1}=T+1\}\cdot A_{12}+\{d_{T+1}\neq T+1\}\cdot A_{14}$

$A_{12}=enc_{k_{AB}}(T+2,D)\cdot A_{13}$

$A_{13}=s_{C_{AB}}(ENC_{k_{AB}}(R_B-1,D))\cdot A$

$A_{14}=s_{C_{AB}}(\bot)\cdot A$

Bob's state transitions described by $APTC_G$ are as follows.

$B=\{k_{AB}=NULL\}\cdot B_1+ \{k_{AB}\neq NULL\}\cdot B_6$

$B_1=r_{C_{AB}}(ENC_{k_{AB}}(A,T),ENC_{k_{BT}}(T,L,k_{AB},A))\cdot B_2$

$B_2=dec_{k_{BT}}(ENC_{k_{BT}}(T,L,k_{AB},A))\cdot B_3$

$B_3=dec_{k_{AB}}(ENC_{k_{AB}}(A,T))\cdot B_4$

$B_4=enc_{k_{AB}}(T+1)\cdot B_5$

$B_5=s_{C_{BA}}(ENC_{k_{AB}}(T+1))\cdot B_6$

$B_6=r_{C_{AB}}(d_{AB})\cdot B_7$

$B_7=\{d_{AB}\neq \bot\}\cdot B_8+\{d_{AB}= \bot\}\cdot s_{C_{BO}}(\bot)\cdot B$

$B_8=dec_{k_{AB}}(ENC_{k_{AB}}(d_{T+2},D))\cdot B_9$

$B_9=\{d_{T+2}=T+2\}\cdot B_{10}+\{d_{T+2}\neq T+2\}\cdot s_{C_{BO}}(\bot)\cdot B$

$B_{10}=s_{C_{BO}}(D)\cdot B$

Trent's state transitions described by $APTC_G$ are as follows.

$T=r_{C_{AT}}(A,B)\cdot T_2$

$T_2=rsg_{k_{AB}}\cdot T_3$

$T_3=enc_{k_{AT}}(T,L,k_{AB},B)\cdot T_4$

$T_4=enc_{k_{BT}}(T,L,k_{AB},A)\cdot T_5$

$T_5=s_{C_{TA}}(ENC_{k_{AT}}(T,L,k_{AB},B),ENC_{k_{BT}}(K,L,k_{AB},A))\cdot T$

The sending action and the reading action of the same type data through the same channel can communicate with each other, otherwise, will cause a deadlock $\delta$. We define the following
communication functions.

$\gamma(r_{C_{AT}}(A,B),s_{C_{AT}}(A,B))\triangleq c_{C_{AT}}(A,B)$

$\gamma(r_{C_{TA}}(ENC_{k_{AT}}(T,L,k_{AB},B),ENC_{k_{BT}}(K,L,k_{AB},A)),\\
s_{C_{TA}}(ENC_{k_{AT}}(T,L,k_{AB},B),ENC_{k_{BT}}(K,L,k_{AB},A)))\\
\triangleq c_{C_{TA}}(ENC_{k_{AT}}(T,L,k_{AB},B),ENC_{k_{BT}}(K,L,k_{AB},A))$

$\gamma(r_{C_{AB}}(ENC_{k_{AB}}(A,T),ENC_{k_{BT}}(T,L,k_{AB},A)),s_{C_{AB}}(ENC_{k_{AB}}(A,T),ENC_{k_{BT}}(T,L,k_{AB},A)))\\
\triangleq c_{C_{AB}}(ENC_{k_{AB}}(A,T),ENC_{k_{BT}}(T,L,k_{AB},A))$

$\gamma(r_{C_{BA}}(ENC_{k_{AB}}(d_{T+1})),s_{C_{BA}}(ENC_{k_{AB}}(d_{T+1})))\triangleq c_{C_{BA}}(ENC_{k_{AB}}(d_{T+1}))$

$\gamma(r_{C_{AB}}(d_{AB}),s_{C_{AB}}(d_{AB}))\triangleq c_{C_{AB}}(d_{AB})$

Let all modules be in parallel, then the protocol $A\quad B\quad T$ can be presented by the following process term.

$$\tau_I(\partial_H(\Theta(A\between B\between T)))=\tau_I(\partial_H(A\between B\between T))$$

where $H=\{r_{C_{AT}}(A,B),s_{C_{AT}}(A,B),r_{C_{AB}}(d_{AB}),s_{C_{AB}}(d_{AB}),\\
r_{C_{TA}}(ENC_{k_{AT}}(T,L,k_{AB},B),ENC_{k_{BT}}(K,L,k_{AB},A)),\\
s_{C_{TA}}(ENC_{k_{AT}}(T,L,k_{AB},B),ENC_{k_{BT}}(K,L,k_{AB},A)),\\
r_{C_{AB}}(ENC_{k_{AB}}(A,T),ENC_{k_{BT}}(T,L,k_{AB},A)),\\
s_{C_{AB}}(ENC_{k_{AB}}(A,T),ENC_{k_{BT}}(T,L,k_{AB},A)),\\
r_{C_{BA}}(ENC_{k_{AB}}(d_{T+1})),s_{C_{BA}}(ENC_{k_{AB}}(d_{T+1}))|D\in\Delta\}$,

$I=\{c_{C_{AT}}(A,B),c_{C_{BA}}(ENC_{k_{AB}}(d_{T+1})),c_{C_{AB}}(d_{AB}),\\
c_{C_{TA}}(ENC_{k_{AT}}(T,L,k_{AB},B),ENC_{k_{BT}}(K,L,k_{AB},A)),\\
c_{C_{AB}}(ENC_{k_{AB}}(A,T),ENC_{k_{BT}}(T,L,k_{AB},A)),\\
\{k_{AB}=NULL\},\{k_{AB}\neq NULL\},dec_{k_{AT}}(ENC_{k_{AT}}(T,L,k_{AB},B)),\\
enc_{k_{AB}}(A,T),dec_{k_{AB}}(ENC_{k_{AB}}(d_{T+1})),dec_{k_{AB}}(ENC_{k_{AB}}(R_B)),\\
\{d_{T+1}=T+1\},\{d_{T+1}\neq T+1\},enc_{k_{AB}}(T+2,D),\\
dec_{k_{BT}}(ENC_{k_{BT}}(T,L,k_{AB},A)),dec_{k_{AB}}(ENC_{k_{AB}}(A,T)),\\
enc_{k_{AB}}(T+1),\{d_{AB}\neq \bot\},\{d_{AB}= \bot\},dec_{k_{AB}}(ENC_{k_{AB}}(d_{T+2},D)),\\
\{d_{T+2}=T+2\},\{d_{T+2}\neq T+2\},rsg_{k_{AB}},enc_{k_{AT}}(T,L,k_{AB},B),\\
enc_{k_{BT}}(T,L,k_{AB},A)|D\in\Delta\}$.

Then we get the following conclusion on the protocol.

\begin{theorem}
The Kerberos protocol in Figure \ref{KP7} is secure.
\end{theorem}

\begin{proof}
Based on the above state transitions of the above modules, by use of the algebraic laws of $APTC_G$, we can prove that

$\tau_I(\partial_H(A\between B\between T))=\sum_{D\in\Delta}(r_{C_{AI}}(D)\cdot (s_{C_{BO}}(\bot)+s_{C_{BO}}(D)))\cdot
\tau_I(\partial_H(A\between B\between T))$.

For the details of proof, please refer to section \ref{app}, and we omit it.

That is, the Kerberos protocol in Figure \ref{KP7} $\tau_I(\partial_H(A\between B\between T))$ can exhibit desired external behaviors:

\begin{enumerate}
  \item For information leakage, because $k_{AT}$ is privately shared only between Alice and Trent, $k_{BT}$ is privately shared only between Bob and Trent, $k_{AB}$ is privately shared
  only among Trent, Alice and Bob. For the modeling of confidentiality, it is similar to the protocol in section \ref{confi}, the Kerberos protocol is confidential;
  \item For replay attack, the using of the time stamp $T$, makes that $\tau_I(\partial_H(A\between B\between T))=\sum_{D\in\Delta}(r_{C_{AI}}(D)\cdot s_{C_{BO}}(\bot))\cdot
\tau_I(\partial_H(A\between B\between T))$, it is desired;
  \item Without replay attack, the protocol would be $\tau_I(\partial_H(A\between B\between T))=\sum_{D\in\Delta}(r_{C_{AI}}(D)\cdot s_{C_{BO}}(D))\cdot
\tau_I(\partial_H(A\between B\between T))$, it is desired;
  \item For the man-in-the-middle attack, because $k_{AT}$ is privately shared only between Alice and Trent, $k_{BT}$ is privately shared only between Bob and Trent, $k_{AB}$ is privately shared
  only among Trent, Alice and Bob. For the modeling of the man-in-the-middle attack, it is similar to the protocol in section \ref{MIM}, the Kerberos protocol can be against the
  man-in-the-middle attack;
  \item For the unexpected and non-technical leaking of $k_{AT}$, $k_{BT}$, $k_{AB}$, or they being not strong enough, or Trent being dishonest, they are out of the scope of analyses of security protocols;
  \item For malicious tampering and transmission errors, they are out of the scope of analyses of security protocols.
\end{enumerate}
\end{proof}

\subsection{Neuman-Stubblebine Protocol}\label{nsp2}

The Neuman-Stubblebine protocol shown in Figure \ref{NSP27} uses symmetric keys for secure communication, that is, the key $k_{AB}$ between Alice and Bob is privately shared to Alice and Bob,
Alice, Bob have shared keys with Trent $k_{AT}$ and $k_{BT}$ already.

\begin{figure}
    \centering
    \includegraphics{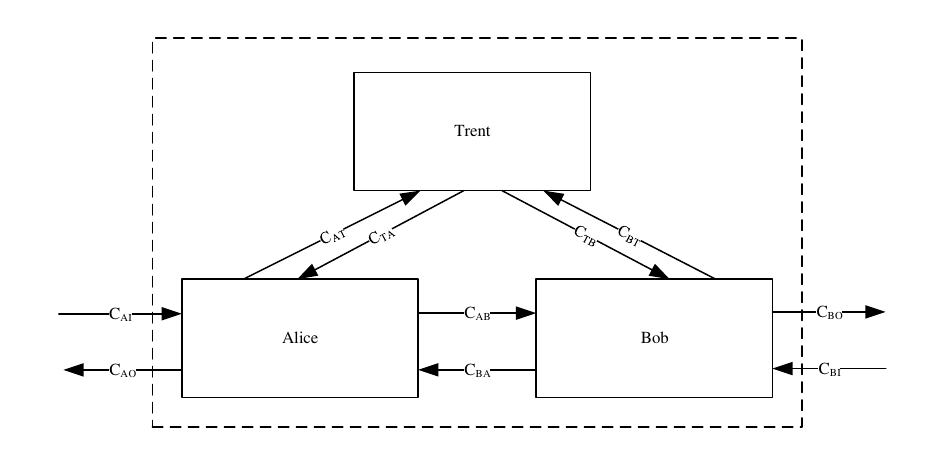}
    \caption{Neuman-Stubblebine protocol}
    \label{NSP27}
\end{figure}

The process of the protocol is as follows.

\begin{enumerate}
  \item Alice receives some messages $D$ from the outside through the channel $C_{AI}$ (the corresponding reading action is denoted $r_{C_{AI}}(D)$), if $k_{AB}$ is not established,
  she generates a random number $R_A$ through the action $rsg_{R_A}$, sends $R_A, A$ to Bob through the channel $C_{AB}$ (the corresponding sending action is denoted
  $s_{C_{AB}}(R_A, A)$);
  \item Bob receives $R_A, A$ from Alice through the channel $C_{AB}$ (the corresponding reading action is denoted $r_{C_{AB}}(R_A, A)$),
  he generates a random number $R_B$ through
  an action $rsg_{R_B}$, encrypts $R_A,A,T_B$ by $k_{BT}$ through an action $enc_{k_{BT}}(R_A,A,T_B)$, and sends \\
  $R_B,B,ENC_{k_{BT}}(R_A,A,T_B)$ to Trent
  through the channel $C_{BT}$
  (the corresponding sending action is denoted
  $s_{C_{BT}}(R_B,B,ENC_{k_{BT}}(R_A,A,T_B))$);
  \item Trent receives $R_B,B,ENC_{k_{BT}}(R_A,A,T_B)$ through the channel $C_{BT}$ (the corresponding reading action is denoted
  $r_{C_{BT}}(R_B,B,ENC_{k_{BT}}(R_A,A,T_B))$), he decrypts the message
  $ENC_{k_{BT}}(R_A,A,T_B)$
  through an action $dec_{k_{BT}}(ENC_{k_{BT}}(R_A,A,T_B))$, generates a random session
  key $k_{AB}$ through an action $rsg_{k_{AB}}$, then he encrypts $B,R_A,k_{AB},T_B$ by $k_{AT}$ through an action
  $enc_{k_{AT}}(B,R_A,k_{AB},T_B)$, encrypts $A,k_{AB},T_B$ by $k_{BT}$ through an action $enc_{k_{BT}}(A,k_{AB},T_B)$, and sends them to Alice through the channel $C_{TA}$ (the corresponding
  sending action is denoted $s_{C_{TA}}(ENC_{k_{AT}}(B,R_A,k_{AB},T_B),ENC_{k_{BT}}(A,k_{AB},T_B),R_B)$);
  \item Alice receives the message from Trent through the channel $C_{TA}$ (the corresponding reading action is denoted $r_{C_{TA}}(ENC_{k_{AT}}(B,d_{R_A},k_{AB},T_B),ENC_{k_{BT}}(A,k_{AB},T_B),R_B)$),
  she decrypts $ENC_{k_{AT}}(B,d_{R_A},k_{AB},T_B)$ by $k_{AT}$ through an action $dec_{k_{AT}}(ENC_{k_{AT}}(B,d_{R_A},k_{AB},T_B))$, if $d_{R_A}=R_A$, he encrypts $R_B$ by $k_{AB}$
  through an action $enc_{k_{AB}}(R_B)$, and sends \\
  $ENC_{k_{BT}}(A,k_{AB},T_B),ENC_{k_{AB}}(R_B)$
  to Bob through the channel $C_{AB}$ (the corresponding sending action is denoted $s_{C_{AB}}(ENC_{k_{BT}}(A,k_{AB},T_B),ENC_{k_{AB}}(R_B))$); else if $d_{R_A}\neq R_A$, she sends $\bot$ to Bob
  through the channel $C_{AB}$ (the corresponding sending action is denoted $s_{C_{AB}}(\bot)$);
  \item Bob receives $d_{AB}$ from Alice (the corresponding reading action is denoted $r_{C_{AB}}(d_{AB})$), if $d_{AB}=\bot$, he sends $\bot$ to Alice through the channel $C_{BA}$
  (the corresponding sending action is denoted $s_{C_{BA}}(\bot)$); else if $d_{AB}\neq \bot$, she decrypts $ENC_{k_{BT}}(A,k_{AB},T_B)$ by $k_{BT}$ through an action
  $dec_{k_{BT}}(ENC_{k_{BT}}(A,k_{AB},T_B))$, decrypts $ENC_{k_{AB}}(d_{R_B})$ by $k_{AB}$ through an action $dec_{k_{AB}}(ENC_{K_{AB}}(d_{R_B}))$, if $d_{R_B}=R_B$,
  he generates a random number $R_D$ through an action $rsg_{R_D}$, encrypts $R_D$ by $k_{AB}$ through an
  action $enc_{k_{AB}}(R_D)$, and sends it to Alice through the channel $C_{BA}$ (the corresponding sending action is denoted $s_{C_{BA}}(ENC_{k_{AB}}(R_D))$), else if
  $d_{R_B}\neq R_B$, he sends $\bot$ to Alice through the channel $C_{BA}$ (the corresponding sending action is denoted $s_{C_{BA}}(\bot)$);
  \item Alice receives $d_{BA}$ from Bob (the corresponding reading action is denoted $r_{C_{BA}}(d_{BA})$), if $d_{BA}=\bot$, she sends $\bot$ to Bob through the channel $C_{AB}$
  (the corresponding sending action is denoted $s_{C_{AB}}(\bot)$); else if $d_{BA}\neq \bot$, she decrypts $ENC_{k_{AB}}(R_D)$ by $k_{AB}$ through an action $dec_{k_{AB}}(ENC_{k_{AB}}(R_D))$,
  if $isFresh(R_D)=TRUE$, she generates a random number $R_D'$ through an action $rsg_{R_D'}$, encrypts $R_D',D$ by $k_{AB}$ through an
  action $enc_{k_{AB}}(R_D',D)$, and sends it to Bob through the channel $C_{AB}$ (the corresponding sending action is denoted $s_{C_{AB}}(ENC_{k_{AB}}(R_D',D))$), else if
  $isFresh(R_D)=FALSE$, he sends $\bot$ to Bob through the channel $C_{AB}$ (the corresponding sending action is denoted $s_{C_{AB}}(\bot)$);
  \item Bob receives $d_{AB}'$ from Alice (the corresponding reading action is denoted $r_{C_{AB}}(d_{AB}')$), if $d_{AB}'=\bot$, he sends $\bot$ to the outside through the channel $C_{BO}$
  (the corresponding sending action is denoted $s_{C_{BO}}(\bot)$); else if $d_{AB}'\neq \bot$, she decrypts $ENC_{k_{AB}}(R_D',D)$ by $k_{AB}$ through an action
  $dec_{k_{AB}}(ENC_{k_{AB}}(R_D',D))$, if $isFresh(d_{R_D'})=TRUE$, she sends $D$ to the outside through the channel $C_{BO}$ (the corresponding sending action is denoted $s_{C_{BO}}(D)$), else if
  $isFresh(d_{R_D}')=FALSE$, he sends $\bot$ to the outside through the channel $C_{BO}$ (the corresponding sending action is denoted $s_{C_{BO}}(\bot)$).
\end{enumerate}

Where $D\in\Delta$, $\Delta$ is the set of data.

Alice's state transitions described by $APTC_G$ are as follows.

$A=\sum_{D\in\Delta}r_{C_{AI}}(D)\cdot A_2$

$A_2=\{k_{AB}=NULL\}\cdot rsg_{R_A}\cdot A_3+ \{k_{AB}\neq NULL\}\cdot A_13$

$A_3=s_{C_{AB}}(R_A, A)\cdot A_4$

$A_4=r_{C_{TA}}(ENC_{k_{AT}}(B,d_{R_A},k_{AB},T_B),ENC_{k_{BT}}(A,k_{AB},T_B),R_B)\cdot A_5$

$A_5=dec_{k_{AT}}(ENC_{k_{AT}}(B,d_{R_A},k_{AB},T_B))\cdot A_6$

$A_6=\{d_{R_A}=R_A\}\cdot A_7+\{d_{R_A}\neq R_A\}\cdot s_{C_{AB}}(\bot)\cdot A_9$

$A_7=enc_{k_{AB}}(R_B)\cdot A_8$

$A_8=s_{C_{AB}}(ENC_{k_{BT}}(A,k_{AB},T_B),ENC_{k_{AB}}(R_B))\cdot A_{9}$

$A_{9}=r_{C_{BA}}(d_{BA})\cdot A_{10}$

$A_{10}=\{d_{BA}\neq\bot\}\cdot A_{11}+\{d_{BA}= \bot\}\cdot s_{C_{AB}}(\bot)\cdot A$

$A_{11}=dec_{k_{AB}}(ENC_{k_{AB}}(R_D))\cdot A_{12}$

$A_{12}=\{isFresh(d_{R_D})=TRUE\}\cdot A_{13}+\{isFresh(d_{R_D})=FALSE\}\cdot s_{C_{AB}}(\bot)\cdot A$

$A_{13}=rsg_{R_D'}\cdot A_{14}$

$A_{14}=enc_{k_{AB}}(R_D',D)\cdot A_{15}$

$A_{15}=s_{C_{AB}}(ENC_{k_{AB}}(R_D',D))\cdot A$

Bob's state transitions described by $APTC_G$ are as follows.

$B=\{k_{AB}=NULL\}\cdot B_1+ \{k_{AB}\neq NULL\}\cdot B_{13}$

$B_1=r_{C_{AB}}(R_A, A)\cdot B_2$

$B_2=rsg_{R_B}\cdot B_3$

$B_3=enc_{k_{BT}}(R_A,A,T_B)\cdot B_4$

$B_4=s_{C_{BT}}(R_B,B,ENC_{k_{BT}}(R_A,A,T_B))\cdot B_5$

$B_5=r_{C_{AB}}(d_{AB})\cdot B_6$

$B_6=\{d_{AB}\neq\bot\}\cdot B_7+\{d_{AB}=\bot\}\cdot s_{C_{BA}}(\bot)\cdot B_{13}$

$B_7=dec_{k_{BT}}(ENC_{k_{BT}}(A,k_{AB},T_B))\cdot B_8$

$B_8=dec_{k_{AB}}(ENC_{K_{AB}}(d_{R_B}))\cdot B_9$

$B_9=\{d_{R_B}\neq R_B\}\cdot s_{C_{BO}}(\bot)\cdot B_{13}+\{d_{R_B}=R_B\}\cdot B_{10}$

$B_{10}=rsg_{R_D}\cdot B_{11}$

$B_{11}=enc_{k_{AB}}(R_D)\cdot B_{12}$

$B_{12}=s_{C_{BA}}(ENC_{k_{AB}}(R_D))\cdot B_{13}$

$B_{13}=r_{C_{AB}}(d_{AB}')\cdot B_{14}$

$B_{14}=\{d_{AB}'=\bot\}\cdot s_{C_{BO}}(\bot)\cdot B+\{d_{AB}'\neq\bot\}\cdot B_{15}$

$B_{15}=dec_{k_{AB}}(ENC_{k_{AB}}(R_D',D))\cdot B_{16}$

$B_{16}=\{isFresh(R_D')=FLASE\}\cdot s_{C_{BO}}(\bot) \cdot B+\{isFresh(R_D')=TRUE\}\cdot B_{17}$

$B_{17}=s_{C_{BO}}(D)\cdot B$

Trent's state transitions described by $APTC_G$ are as follows.

$T=r_{C_{BT}}(R_B,B,ENC_{k_{BT}}(R_A,A,T_B))\cdot T_2$

$T_2=dec_{k_{BT}}(ENC_{k_{BT}}(R_A,A,T_B))\cdot T_3$

$T_3=rsg_{k_{AB}}\cdot T_4$

$T_4=enc_{k_{AT}}(B,R_A,k_{AB},T_B)\cdot T_5$

$T_5=enc_{k_{BT}}(A,k_{AB},T_B)\cdot T_6$

$T_6=s_{C_{TA}}(ENC_{k_{AT}}(B,R_A,k_{AB},T_B),ENC_{k_{BT}}(A,k_{AB},T_B),R_B)\cdot T$

The sending action and the reading action of the same type data through the same channel can communicate with each other, otherwise, will cause a deadlock $\delta$. We define the following
communication functions.

$\gamma(r_{C_{AB}}(R_A, A),s_{C_{AB}}(R_A, A))\triangleq c_{C_{AB}}(R_A, A)$

$\gamma(r_{C_{BT}}(R_B,B,ENC_{k_{BT}}(R_A,A,T_B)),s_{C_{BT}}(R_B,B,ENC_{k_{BT}}(R_A,A,T_B)))\\
\triangleq c_{C_{BT}}(R_B,B,ENC_{k_{BT}}(R_A,A,T_B))$

$\gamma(r_{C_{TA}}(ENC_{k_{AT}}(B,d_{R_A},k_{AB},T_B),ENC_{k_{BT}}(A,k_{AB},T_B),R_B),\\
s_{C_{TA}}(ENC_{k_{AT}}(B,d_{R_A},k_{AB},T_B),ENC_{k_{BT}}(A,k_{AB},T_B),R_B))\triangleq\\
c_{C_{TA}}(ENC_{k_{AT}}(B,d_{R_A},k_{AB},T_B),ENC_{k_{BT}}(A,k_{AB},T_B),R_B)$

$\gamma(r_{C_{AB}}(d_{AB}),s_{C_{AB}}(d_{AB}))\triangleq c_{C_{AB}}(d_{AB})$

$\gamma(r_{C_{BA}}(d_{BA}),s_{C_{BA}}(d_{BA}))\triangleq c_{C_{BA}}(d_{BA})$

$\gamma(r_{C_{AB}}(d_{AB}'),s_{C_{AB}}(d_{AB}'))\triangleq c_{C_{AB}}(d_{AB}')$

Let all modules be in parallel, then the protocol $A\quad B\quad T$ can be presented by the following process term.

$$\tau_I(\partial_H(\Theta(A\between B\between T)))=\tau_I(\partial_H(A\between B\between T))$$

where $H=\{r_{C_{AB}}(R_A, A),s_{C_{AB}}(R_A, A),r_{C_{AB}}(d_{AB}),s_{C_{AB}}(d_{AB}),\\
r_{C_{BA}}(d_{BA}),s_{C_{BA}}(d_{BA}),r_{C_{AB}}(d_{AB}'),s_{C_{AB}}(d_{AB}')'\\
r_{C_{BT}}(R_B,B,ENC_{k_{BT}}(R_A,A,T_B)),s_{C_{BT}}(R_B,B,ENC_{k_{BT}}(R_A,A,T_B)),\\
r_{C_{TA}}(ENC_{k_{AT}}(B,d_{R_A},k_{AB},T_B),ENC_{k_{BT}}(A,k_{AB},T_B),R_B),\\
s_{C_{TA}}(ENC_{k_{AT}}(B,d_{R_A},k_{AB},T_B),ENC_{k_{BT}}(A,k_{AB},T_B),R_B)|D\in\Delta\}$,

$I=\{c_{C_{AB}}(R_A, A),c_{C_{BT}}(R_B,B,ENC_{k_{BT}}(R_A,A,T_B)),\\
c_{C_{TA}}(ENC_{k_{AT}}(B,d_{R_A},k_{AB},T_B),ENC_{k_{BT}}(A,k_{AB},T_B),R_B),\\
c_{C_{AB}}(d_{AB}),c_{C_{BA}}(d_{BA}),c_{C_{AB}}(d_{AB}'),\\
\{k_{AB}=NULL\}, rsg_{R_A},\{k_{AB}\neq NULL\},dec_{k_{AT}}(ENC_{k_{AT}}(B,d_{R_A},k_{AB},T_B)),\\
\{d_{R_A}=R_A\},\{d_{R_A}\neq R_A\},enc_{k_{AB}}(R_B),\{d_{BA}=\bot\},\{d_{BA}\neq\bot\},\\
dec_{k_{AB}}(ENC_{k_{AB}}(R_D)),\{isFresh(R_D)=TRUE\},\{isFresh(R_D)=FALSE\},\\
rsg_{R_D'},enc_{k_{AB}}(R_D',D),rsg_{R_B},enc_{k_{BT}}(R_A,A,T_B),\{d_{AB}\neq\bot\},\\
\{d_{AB}=\bot\},dec_{k_{BT}}(ENC_{k_{BT}}(A,k_{AB},T_B)),dec_{k_{AB}}(ENC_{K_{AB}}(d_{R_B})),\\
\{d_{R_B}\neq R_B\},\{d_{R_B}= R_B\},rsg_{R_D},enc_{k_{AB}}(R_D),\{d_{AB}'=\bot\},\{d_{AB}'\neq\bot\},\\
dec_{k_{AB}}(ENC_{k_{AB}}(R_D',D)),\{isFresh(d_{R_D'})=FLASE\},\{isFresh(d_{R_D'})=TRUE\},\\
dec_{k_{BT}}(ENC_{k_{BT}}(R_A,A,T_B)),rsg_{k_{AB}},enc_{k_{AT}}(B,R_A,k_{AB},T_B),\\
enc_{k_{BT}}(A,k_{AB},T_B)|D\in\Delta\}$.

Then we get the following conclusion on the protocol.

\begin{theorem}
The Neuman-Stubblebine protocol in Figure \ref{NSP27} is secure.
\end{theorem}

\begin{proof}
Based on the above state transitions of the above modules, by use of the algebraic laws of $APTC_G$, we can prove that

$\tau_I(\partial_H(A\between B\between T))=\sum_{D\in\Delta}(r_{C_{AI}}(D)\cdot (s_{C_{BO}}(\bot)+s_{C_{BO}}(D)))\cdot
\tau_I(\partial_H(A\between B\between T))$.

For the details of proof, please refer to section \ref{app}, and we omit it.

That is, the Neuman-Stubblebine protocol in Figure \ref{NSP27} $\tau_I(\partial_H(A\between B\between T))$ can exhibit desired external behaviors:

\begin{enumerate}
  \item For information leakage, because $k_{AT}$ is privately shared only between Alice and Trent, $k_{BT}$ is privately shared only between Bob and Trent, $k_{AB}$ is privately shared
  only among Trent, Alice and Bob. For the modeling of confidentiality, it is similar to the protocol in section \ref{confi}, the Neuman-Stubblebine protocol is confidential;
  \item For the man-in-the-middle attack, because $k_{AT}$ is privately shared only between Alice and Trent, $k_{BT}$ is privately shared only between Bob and Trent, $k_{AB}$ is privately shared
  only among Trent, Alice and Bob, and the use of the random numbers $R_A$, $R_B$, $R_D$ and $R_D'$, the protocol would be $\tau_I(\partial_H(A\between B\between T))=\sum_{D\in\Delta}(r_{C_{AI}}(D)\cdot s_{C_{BO}}(\bot))\cdot
  \tau_I(\partial_H(A\between B\between T))$, it is desired, the Neuman-Stubblebine protocol can be against the
  man-in-the-middle attack;
  \item For replay attack, the using of the random numbers $T$, $R_A$, $R_B$, $R_D$ and $R_D'$, makes that $\tau_I(\partial_H(A\between B\between T))=\sum_{D\in\Delta}(r_{C_{AI}}(D)\cdot s_{C_{BO}}(\bot))\cdot
  \tau_I(\partial_H(A\between B\between T))$, it is desired;
  \item Without man-in-the-middle and replay attack, the protocol would be $\tau_I(\partial_H(A\between B\between T))=\sum_{D\in\Delta}(r_{C_{AI}}(D)\cdot s_{C_{BO}}(D))\cdot
  \tau_I(\partial_H(A\between B\between T))$, it is desired;
  \item For the unexpected and non-technical leaking of $k_{AT}$, $k_{BT}$, $k_{AB}$, or they being not strong enough, or Trent being dishonest, they are out of the scope of analyses of security protocols;
  \item For malicious tampering and transmission errors, they are out of the scope of analyses of security protocols.
\end{enumerate}
\end{proof}

\subsection{Denning-Sacco Protocol}\label{dsp}

The Denning-Sacco protocol shown in Figure \ref{DSP7} uses asymmetric keys and symmetric keys for secure communication, that is, the key $k_{AB}$ between Alice and Bob is privately shared to Alice and Bob,
Alice's, Bob's and Trent's public keys $pk_{A}$, $pk_{B}$ and $pk_{T}$ can be publicly gotten.

\begin{figure}
    \centering
    \includegraphics{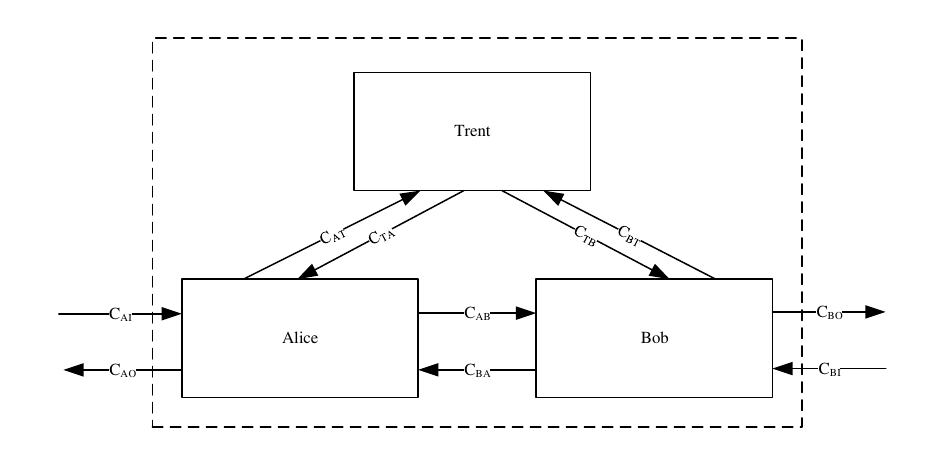}
    \caption{Denning-Sacco protocol}
    \label{DSP7}
\end{figure}

The process of the protocol is as follows.

\begin{enumerate}
  \item Alice receives some messages $D$ from the outside through the channel $C_{AI}$ (the corresponding reading action is denoted $r_{C_{AI}}(D)$), if $k_{AB}$ is not established,
  she sends $A,B$ to Trent through the channel $C_{AT}$ (the corresponding sending action is denoted
  $s_{C_{AT}}(A,B)$);
  \item Trent receives $A,B$ through the channel $C_{AT}$ (the corresponding reading action is denoted
  $r_{C_{AT}}(A,B)$), he signs Alice's and Bob's public keys $pk_A$ and $pk_B$ through the actions $sign_{sk_T}(A,pk_A)$ and $sign_{sk_T}(B,pk_B)$, and sends the signatures to Alice
  through the channel $C_{TA}$ (the corresponding sending action is denoted $s_{C_{TA}}(SIGN_{sk_T}(A,pk_A),SIGN_{sk_T}(B,pk_B))$);
  \item Alice receives the message from Trent through the channel $C_{TA}$ (the corresponding reading action is denoted $r_{C_{TA}}(SIGN_{sk_T}(A,pk_A),SIGN_{sk_T}(B,pk_B))$),
  she de-signs $SIGN_{sk_T}(B,pk_B)$ through an action $de\textrm{-}sign_{pk_T}(SIGN_{sk_T}(B,pk_B))$ to get $pk_B$, generates a random session key $k_{AB}$ through an action
  $rsg_{k_{AB}}$, signs $A,B,k_{AB},T_A$
  through an action $sign_{sk_A}(A,B,k_{AB},T_A)$, and encrypts the signature by $pk_B$ through an action $enc_{pk_B}(SIGN_{sk_A}(A,B,k_{AB},T_A))$, then sends
  $ENC_{pk_B}(SIGN_{sk_A}(A,B,k_{AB},T_A)),SIGN_{sk_T}(A,pk_A),SIGN_{sk_T}(B,pk_B)$
  to Bob through the channel $C_{AB}$ (the corresponding sending action is denoted \\
  $s_{C_{AB}}(ENC_{pk_B}(SIGN_{sk_A}(A,B,k_{AB},T_A)),SIGN_{sk_T}(A,pk_A),SIGN_{sk_T}(B,pk_B))$);
  \item Bob receives $ENC_{pk_B}(SIGN_{sk_A}(A,B,k_{AB},T_A)),SIGN_{sk_T}(A,pk_A),SIGN_{sk_T}(B,pk_B)$ from Alice (the corresponding reading action is denoted\\
  $r_{C_{AB}}(ENC_{pk_B}(SIGN_{sk_A}(A,B,k_{AB},T_A)),SIGN_{sk_T}(A,pk_A),SIGN_{sk_T}(B,pk_B))$), he de-signs $SIGN_{sk_T}(A,pk_A)$ through an action
  $de\textrm{-}sign_{pk_T}(SIGN_{sk_T}(A,pk_A))$
  to get $pk_A$, decrypts $ENC_{pk_B}(SIGN_{sk_A}(A,B,k_{AB},T_A))$ through an action \\
  $dec_{sk_B}(ENC_{pk_B}(SIGN_{sk_A}(A,B,k_{AB},T_A)))$ and de-sign $SIGN_{sk_A}(A,B,k_{AB},T_A)$
  through an action $de\textrm{-}sign_{pk_A}(SIGN_{sk_A}(A,B,k_{AB},T_A))$ to get $k_{AB}$ and $T_A$, if $isValid(T_A)=TRUE$,
  he generates a random number $R_D$ through an action $rsg_{R_D}$, encrypts $R_D$ by $k_{AB}$ through an
  action $enc_{k_{AB}}(R_D)$, and sends it to Alice through the channel $C_{BA}$ (the corresponding sending action is denoted $s_{C_{BA}}(ENC_{k_{AB}}(R_D))$), else if
  $isValid(T_A)=FALSE$, he sends $ENC_{k_{AB}}(\bot)$ to Alice through the channel $C_{BA}$ (the corresponding sending action is denoted $s_{C_{BA}}(ENC_{k_{AB}}(\bot))$);
  \item Alice receives $ENC_{k_{AB}}(d_{BA})$ from Bob (the corresponding reading action is denoted $r_{C_{BA}}(ENC_{k_{AB}}(d_{BA}))$), if $d_{BA}=\bot$, she sends $ENC_{k_{AB}}(\bot)$ to Bob through the channel $C_{AB}$
  (the corresponding sending action is denoted $s_{C_{AB}}(ENC_{k_{AB}}(\bot))$); else if $d_{BA}\neq \bot$, if $isFresh(d_{BA})=TRUE$, she generates a random number $R_D'$
  through an action $rsg_{R_D'}$, encrypts $R_D',D$ by $k_{AB}$ through an
  action $enc_{k_{AB}}(R_D',D)$, and sends it to Bob through the channel $C_{AB}$ (the corresponding sending action is denoted $s_{C_{AB}}(ENC_{k_{AB}}(R_D',D))$), else if
  $isFresh(d_{BA})=FALSE$, he sends $ENC_{k_{AB}}(\bot)$ to Bob through the channel $C_{AB}$ (the corresponding sending action is denoted $s_{C_{AB}}(ENC_{k_{AB}}(\bot))$);
  \item Bob receives $ENC_{k_{AB}}(d_{AB}')$ from Alice (the corresponding reading action is denoted $r_{C_{AB}}(ENC_{k_{AB}}(d_{AB}'))$), if $d_{AB}'=\bot$, he sends $\bot$
  to the outside through the channel $C_{BO}$
  (the corresponding sending action is denoted $s_{C_{BO}}(\bot)$); else if $d_{AB}'\neq \bot$, if $isFresh(d_{R_D'})=TRUE$, she sends $D$ to the outside through the channel $C_{BO}$ (the corresponding sending action is denoted $s_{C_{BO}}(D)$), else if
  $isFresh(d_{R_D}')=FALSE$, he sends $\bot$ to the outside through the channel $C_{BO}$ (the corresponding sending action is denoted $s_{C_{BO}}(\bot)$).
\end{enumerate}

Where $D\in\Delta$, $\Delta$ is the set of data.

Alice's state transitions described by $APTC_G$ are as follows.

$A=\sum_{D\in\Delta}r_{C_{AI}}(D)\cdot A_2$

$A_2=\{k_{AB}=NULL\}\cdot A_3+ \{k_{AB}\neq NULL\}\cdot A_{13}$

$A_3=s_{C_{AT}}(A,B)\cdot A_4$

$A_4=r_{C_{TA}}(SIGN_{sk_T}(A,pk_A),SIGN_{sk_T}(B,pk_B))\cdot A_5$

$A_5=de\textrm{-}sign_{pk_T}(SIGN_{sk_T}(B,pk_B))\cdot A_6$

$A_6=rsg_{k_{AB}}\cdot A_7$

$A_7=sign_{sk_A}(A,B,k_{AB},T_A)\cdot A_8$

$A_8=enc_{pk_B}(SIGN_{sk_A}(A,B,k_{AB},T_A))\cdot A_{9}$

$A_9=s_{C_{AB}}(ENC_{pk_B}(SIGN_{sk_A}(A,B,k_{AB},T_A)),SIGN_{sk_T}(A,pk_A),SIGN_{sk_T}(B,pk_B))\cdot A_{10}$

$A_{10}=r_{C_{BA}}(ENC_{k_{AB}}(d_{BA}))\cdot A_{11}$

$A_{11}=\{d_{BA}\neq\bot\}\cdot A_{12}+\{d_{BA}= \bot\}\cdot s_{C_{AB}}(ENC_{k_{AB}}(\bot))\cdot A$

$A_{12}=\{isFresh(d_{BA})=TRUE\}\cdot A_{13}+\{isFresh(d_{BA})=FALSE\}\cdot s_{C_{AB}}(ENC_{k_{AB}}(\bot))\cdot A$

$A_{13}=rsg_{R_D'}\cdot A_{14}$

$A_{14}=enc_{k_{AB}}(R_D',D)\cdot A_{15}$

$A_{15}=s_{C_{AB}}(ENC_{k_{AB}}(R_D',D))\cdot A$

Bob's state transitions described by $APTC_G$ are as follows.

$B=\{k_{AB}=NULL\}\cdot B_1+ \{k_{AB}\neq NULL\}\cdot B_{9}$

$B_1=r_{C_{AB}}(ENC_{pk_B}(SIGN_{sk_A}(A,B,k_{AB},T_A)),SIGN_{sk_T}(A,pk_A),SIGN_{sk_T}(B,pk_B))\cdot B_2$

$B_2=de\textrm{-}sign_{pk_T}(SIGN_{sk_T}(A,pk_A))\cdot B_3$

$B_3=dec_{sk_B}(ENC_{pk_B}(SIGN_{sk_A}(A,B,k_{AB},T_A)))\cdot B_4$

$B_4=de\textrm{-}sign_{pk_A}(SIGN_{sk_A}(A,B,k_{AB},T_A))\cdot B_5$

$B_5=\{isValid(T_A)=TRUE\}\cdot B_6+\{isValid(T_A)=FALSE\}\cdot s_{C_{BA}}(ENC_{k_{AB}}(\bot))\cdot B_9$

$B_6=rsg_{R_D}\cdot B_7$

$B_7=enc_{k_{AB}}(R_D)\cdot B_8$

$B_8=s_{C_{BA}}(ENC_{k_{AB}}(d_{BA}))\cdot B_9$

$B_9=r_{C_{AB}}(ENC_{k_{AB}}(d_{AB}'))\cdot B_{10}$

$B_{10}=dec_{k_{AB}}(ENC_{k_{AB}}(d_{AB}'))\cdot B_{11}$

$B_{11}=\{d_{AB}'=\bot\}\cdot s_{C_{BO}}(\bot)\cdot B+\{d_{AB}'\neq\bot\}\cdot B_{12}$

$B_{12}=\{isFresh(d_{R_D'})=FLASE\}\cdot s_{C_{BO}}(\bot) B+\{isFresh(d_{R_D'})=TRUE\}\cdot B_{13}$

$B_{13}=s_{C_{BO}}(D)\cdot B$

Trent's state transitions described by $APTC_G$ are as follows.

$T=r_{C_{AT}}(A,B)\cdot T_2$

$T_2=sign_{sk_T}(A,pk_A)\cdot T_3$

$T_3=sign_{sk_T}(B,pk_B)\cdot T_4$

$T_4=s_{C_{TA}}(SIGN_{sk_T}(A,pk_A),SIGN_{sk_T}(B,pk_B))\cdot T$

The sending action and the reading action of the same type data through the same channel can communicate with each other, otherwise, will cause a deadlock $\delta$. We define the following
communication functions.

$\gamma(r_{C_{AT}}(A,B),s_{C_{AT}}(A,B))\triangleq c_{C_{AT}}(A,B)$

$\gamma(r_{C_{TA}}(SIGN_{sk_T}(A,pk_A),SIGN_{sk_T}(B,pk_B)),s_{C_{TA}}(SIGN_{sk_T}(A,pk_A),SIGN_{sk_T}(B,pk_B)))\\
\triangleq c_{C_{TA}}(SIGN_{sk_T}(A,pk_A),SIGN_{sk_T}(B,pk_B))$

$\gamma(r_{C_{AB}}(ENC_{pk_B}(SIGN_{sk_A}(A,B,k_{AB},T_A)),SIGN_{sk_T}(A,pk_A),SIGN_{sk_T}(B,pk_B)),\\
s_{C_{AB}}(ENC_{pk_B}(SIGN_{sk_A}(A,B,k_{AB},T_A)),SIGN_{sk_T}(A,pk_A),SIGN_{sk_T}(B,pk_B)))\\
\triangleq c_{C_{AB}}(ENC_{pk_B}(SIGN_{sk_A}(A,B,k_{AB},T_A)),SIGN_{sk_T}(A,pk_A),SIGN_{sk_T}(B,pk_B))$

$\gamma(r_{C_{BA}}(ENC_{k_{AB}}(d_{BA})),s_{C_{BA}}(ENC_{k_{AB}}(d_{BA})))\triangleq c_{C_{BA}}(ENC_{k_{AB}}(d_{BA}))$

$\gamma(r_{C_{AB}}(ENC_{k_{AB}}(d_{AB}')),s_{C_{AB}}(ENC_{k_{AB}}(d_{AB}')))\triangleq c_{C_{AB}}(ENC_{k_{AB}}(d_{AB}'))$

Let all modules be in parallel, then the protocol $A\quad B\quad T$ can be presented by the following process term.

$$\tau_I(\partial_H(\Theta(A\between B\between T)))=\tau_I(\partial_H(A\between B\between T))$$

where $H=\{r_{C_{AT}}(A,B),s_{C_{AT}}(A,B),r_{C_{BA}}(ENC_{k_{AB}}(d_{BA})),s_{C_{BA}}(ENC_{k_{AB}}(d_{BA})),\\
r_{C_{AB}}(ENC_{k_{AB}}(d_{AB}')),s_{C_{AB}}(ENC_{k_{AB}}(d_{AB}')),\\
r_{C_{TA}}(SIGN_{sk_T}(A,pk_A),SIGN_{sk_T}(B,pk_B)),s_{C_{TA}}(SIGN_{sk_T}(A,pk_A),SIGN_{sk_T}(B,pk_B)),\\
r_{C_{AB}}(ENC_{pk_B}(SIGN_{sk_A}(A,B,k_{AB},T_A)),SIGN_{sk_T}(A,pk_A),SIGN_{sk_T}(B,pk_B)),\\
s_{C_{AB}}(ENC_{pk_B}(SIGN_{sk_A}(A,B,k_{AB},T_A)),SIGN_{sk_T}(A,pk_A),SIGN_{sk_T}(B,pk_B))|D\in\Delta\}$,

$I=\{c_{C_{AT}}(A,B),c_{C_{BA}}(ENC_{k_{AB}}(d_{BA})),c_{C_{AB}}(ENC_{k_{AB}}(d_{AB}')),\\
c_{C_{TA}}(SIGN_{sk_T}(A,pk_A),SIGN_{sk_T}(B,pk_B)),\\
c_{C_{AB}}(ENC_{pk_B}(SIGN_{sk_A}(A,B,k_{AB},T_A)),SIGN_{sk_T}(A,pk_A),SIGN_{sk_T}(B,pk_B)),\\
\{k_{AB}=NULL\},\{k_{AB}\neq NULL\},de\textrm{-}sign_{pk_T}(SIGN_{sk_T}(B,pk_B)),\\
rsg_{k_{AB}},sign_{sk_A}(A,B,k_{AB},T_A),enc_{pk_B}(SIGN_{sk_A}(A,B,k_{AB},T_A)),\\
\{isFresh(d_{BA})=TRUE\},\{isFresh(d_{BA})=FALSE\},\{d_{BA}\neq\bot\},\{d_{BA}=\bot\},\\
rsg_{R_D'},enc_{k_{AB}}(R_D',D),de\textrm{-}sign_{pk_T}(SIGN_{sk_T}(A,pk_A)),\\
dec_{sk_B}(ENC_{pk_B}(SIGN_{sk_A}(A,B,k_{AB},T_A))),de\textrm{-}sign_{pk_A}(SIGN_{sk_A}(A,B,k_{AB},T_A)),\\
\{isValid(T_A)=TRUE\},\{isValid(T_A)=FALSE\},rsg_{R_D},enc_{k_{AB}}(R_D),\\
dec_{k_{AB}}(ENC_{k_{AB}}(d_{AB}')),\{d_{AB}'=\bot\},\{d_{AB}'\neq\bot\},\\
\{isFresh(d_{R_D'})=TRUE\},\{isFresh(d_{R_D'})=FLASE\},sign_{sk_T}(A,pk_A),sign_{sk_T}(B,pk_B)|D\in\Delta\}$.

Then we get the following conclusion on the protocol.

\begin{theorem}
The Denning-Sacco protocol in Figure \ref{DSP7} is secure.
\end{theorem}

\begin{proof}
Based on the above state transitions of the above modules, by use of the algebraic laws of $APTC_G$, we can prove that

$\tau_I(\partial_H(A\between B\between T))=\sum_{D\in\Delta}(r_{C_{AI}}(D)\cdot (s_{C_{BO}}(\bot)+s_{C_{BO}}(D)))\cdot
\tau_I(\partial_H(A\between B\between T))$.

For the details of proof, please refer to section \ref{app}, and we omit it.

That is, the Denning-Sacco protocol in Figure \ref{DSP7} $\tau_I(\partial_H(A\between B\between T))$ can exhibit desired external behaviors:

\begin{enumerate}
  \item For the modeling of confidentiality, it is similar to the protocol in section \ref{confi}, the Denning-Sacco protocol is confidential;
  \item For the man-in-the-middle attack, because $pk_A$ and $pk_B$ are signed by Trent, the protocol would be $\tau_I(\partial_H(A\between B\between T))=\sum_{D\in\Delta}(r_{C_{AI}}(D)\cdot s_{C_{BO}}(\bot))\cdot
  \tau_I(\partial_H(A\between B\between T))$, it is desired, the Denning-Sacco protocol can be against the
  man-in-the-middle attack;
  \item For replay attack, the using of the time stamp $T_A$, random numbers $R_D$ and $R_D'$, makes that $\tau_I(\partial_H(A\between B\between T))=\sum_{D\in\Delta}(r_{C_{AI}}(D)\cdot s_{C_{BO}}(\bot))\cdot
  \tau_I(\partial_H(A\between B\between T))$, it is desired;
  \item Without man-in-the-middle and replay attack, the protocol would be $\tau_I(\partial_H(A\between B\between T))=\sum_{D\in\Delta}(r_{C_{AI}}(D)\cdot s_{C_{BO}}(D))\cdot
  \tau_I(\partial_H(A\between B\between T))$, it is desired;
  \item For the unexpected and non-technical leaking of $sk_{A}$, $sk_{B}$, $k_{AB}$, or they being not strong enough, or Trent being dishonest, they are out of the scope of analyses of security protocols;
  \item For malicious tampering and transmission errors, they are out of the scope of analyses of security protocols.
\end{enumerate}
\end{proof}

\subsection{DASS Protocol}\label{dass}

The DASS (Distributed Authentication Security Service) protocol shown in Figure \ref{DASS7} uses asymmetric keys and symmetric keys for secure communication, that is, the key $k_{AB}$ between Alice and Bob is privately shared to Alice and Bob,
Alice's, Bob's and Trent's public keys $pk_{A}$, $pk_{B}$ and $pk_{T}$ can be publicly gotten.

\begin{figure}
    \centering
    \includegraphics{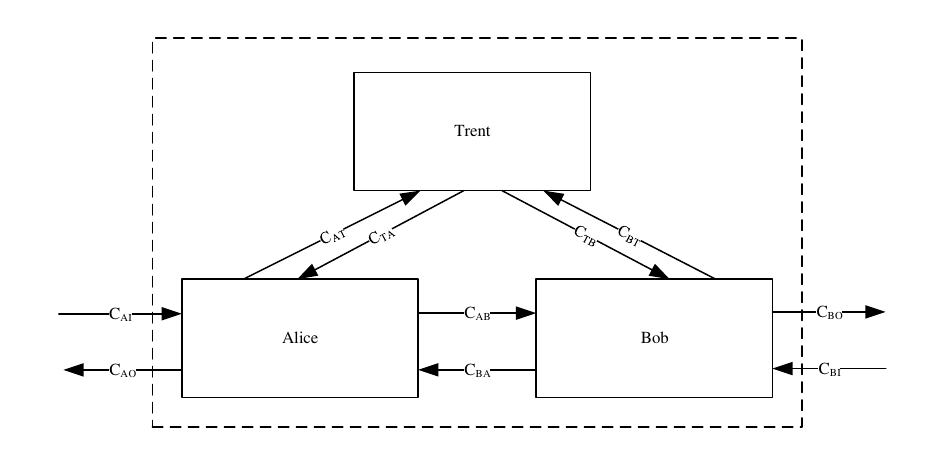}
    \caption{DASS protocol}
    \label{DASS7}
\end{figure}

The process of the protocol is as follows.

\begin{enumerate}
  \item Alice receives some messages $D$ from the outside through the channel $C_{AI}$ (the corresponding reading action is denoted $r_{C_{AI}}(D)$), if $k_{AB}$ is not established,
  she sends $B$ to Trent through the channel $C_{AT}$ (the corresponding sending action is denoted
  $s_{C_{AT}}(B)$);
  \item Trent receives $B$ through the channel $C_{AT}$ (the corresponding reading action is denoted
  $r_{C_{AT}}(B)$), he signs Bob's public key $pk_B$ through the action $sign_{sk_T}(B,pk_B)$, and sends the signature to Alice
  through the channel $C_{TA}$ (the corresponding sending action is denoted $s_{C_{TA}}(SIGN_{sk_T}(B,pk_B))$);
  \item Alice receives the message from Trent through the channel $C_{TA}$ (the corresponding reading action is denoted $r_{C_{TA}}(SIGN_{sk_T}(B,pk_B))$),
  she de-signs $SIGN_{sk_T}(B,pk_B)$ through an action $de\textrm{-}sign_{pk_T}(SIGN_{sk_T}(B,pk_B))$ to get $pk_B$, generates a random session key $k_{AB}$ through an action
  $rsg_{k_{AB}}$, generates a public key $pk_P$ through an action $rsg_{pk_P}$ and generates a private key $sk_P$ through an action $rsg_{sk_P}$, signs $L,A,k_{AB},sk_P,pk_P$
  through an action $sign_{sk_A}(L,A,k_{AB},sk_P,pk_P)$ where $L$ is the life cycle of $k_{AB}$, and encrypts the time stamp $T_A$ by $k_{AB}$ through an action
  $enc_{k_{AB}}(T_A)$, encrypts $k_{AB}$ by $pk_B$ through an action $enc_{pk_B}(k_{AB})$ and then re-encrypts it by $sk_P$ through an action $enc_{sk_P}(ENC_{pk_B}(k_{AB}))$,
  then sends $ENC_{k_{AB}}(T_A), SIGN_{sk_A}(L,A,k_{AB},sk_P,pk_P),ENC_{sk_P}(ENC_{pk_B}(k_{AB}))$
  to Bob through the channel $C_{AB}$ (the corresponding sending action is denoted \\
  $s_{C_{AB}}(ENC_{k_{AB}}(T_A), SIGN_{sk_A}(L,A,k_{AB},sk_P,pk_P),ENC_{sk_P}(ENC_{pk_B}(k_{AB})))$);
  \item Bob receives $ENC_{k_{AB}}(T_A), SIGN_{sk_A}(L,A,k_{AB},sk_P,pk_P),ENC_{sk_P}(ENC_{pk_B}(k_{AB}))$ from Alice (the corresponding reading action is denoted \\
  $r_{C_{AB}}(ENC_{k_{AB}}(T_A), SIGN_{sk_A}(L,A,k_{AB},sk_P,pk_P),ENC_{sk_P}(ENC_{pk_B}(k_{AB})))$), he sends the name of Alice $A$ to Trent through the channel $C_{BT}$
  (the corresponding sending action is denoted $s_{C_{BT}}(A)$);
  \item Trent receives the name of Alice $A$ from Bob through the channel $C_{BT}$ (the corresponding reading action is denoted $r_{C_{BT}}(A)$), signs $A$ and $pk_A$ through an
  action $sign_{sk_T}(A,pk_A)$, and sends the signature $SIGN_{sk_T}(A,pk_A)$ to Bob through the channel $C_{TB}$ (the corresponding sending action is denoted
  $s_{C_{TB}}(SIGN_{sk_T}(A,pk_A))$);
  \item Bob receives the signature from Trent through the channel $C_{TB}$ (the corresponding reading action is denoted $r_{C_{TB}}(SIGN_{sk_T}(A,pk_A))$),
  he de-signs $SIGN_{sk_T}(A,pk_A)$ through an action $de\textrm{-}sign_{pk_T}(SIGN_{sk_T}(A,pk_A))$
  to get $pk_A$, de-signs $SIGN_{sk_A}(L,A,k_{AB},sk_P,pk_P)$ through an action $de\textrm{-}sign{pk_A}(SIGN_{sk_A}(L,A,k_{AB},sk_P,pk_P))$ and decrypts $ENC_{sk_P}(ENC_{pk_B}(k_{AB}))$
  and $ENC_{k_{AB}}(T_A)$ through an action $dec_{pk_P}(ENC_{sk_P}(ENC_{pk_B}(k_{AB})))$ and an action $dec_{sk_B}(ENC_{pk_B}(k_{AB}))$ and an action $dec_{k_{AB}}(T_A)$
  to get $k_{AB}$ and $T_A$, if $isValid(T_A)=TRUE$, he encrypts the time stamp $T_B$ by $k_{AB}$ through an
  action $enc_{k_{AB}}(T_B)$, and sends it to Alice through the channel $C_{BA}$ (the corresponding sending action is denoted $s_{C_{BA}}(ENC_{k_{AB}}(T_B))$), else if
  $isValid(T_A)=FALSE$, he sends $ENC_{k_{AB}}(\bot)$ to Alice through the channel $C_{BA}$ (the corresponding sending action is denoted $s_{C_{BA}}(ENC_{k_{AB}}(\bot))$);
  \item Alice receives $ENC_{k_{AB}}(d_{BA})$ from Bob (the corresponding reading action is denoted $r_{C_{BA}}(ENC_{k_{AB}}(d_{BA}))$), if $d_{BA}=\bot$, she sends $ENC_{k_{AB}}(\bot)$ to Bob through the channel $C_{AB}$
  (the corresponding sending action is denoted $s_{C_{AB}}(ENC_{k_{AB}}(\bot))$); else if $d_{BA}\neq \bot$, if $isFresh(d_{BA})=TRUE$, she generates a random number $R_D$
  through an action $rsg_{R_D}$, encrypts $R_D,D$ by $k_{AB}$ through an
  action $enc_{k_{AB}}(R_D,D)$, and sends it to Bob through the channel $C_{AB}$ (the corresponding sending action is denoted $s_{C_{AB}}(ENC_{k_{AB}}(R_D,D))$), else if
  $isFresh(d_{BA})=FALSE$, he sends $ENC_{k_{AB}}(\bot)$ to Bob through the channel $C_{AB}$ (the corresponding sending action is denoted $s_{C_{AB}}(ENC_{k_{AB}}(\bot))$);
  \item Bob receives $ENC_{k_{AB}}(d_{AB})$ from Alice (the corresponding reading action is denoted $r_{C_{AB}}(ENC_{k_{AB}}(d_{AB}))$), if $d_{AB}=\bot$, he sends $\bot$
  to the outside through the channel $C_{BO}$
  (the corresponding sending action is denoted $s_{C_{BO}}(\bot)$); else if $d_{AB}\neq \bot$, if $isFresh(d_{R_D})=TRUE$, she sends $D$ to the outside through the channel $C_{BO}$ (the corresponding sending action is denoted $s_{C_{BO}}(D)$), else if
  $isFresh(d_{R_D})=FALSE$, he sends $\bot$ to the outside through the channel $C_{BO}$ (the corresponding sending action is denoted $s_{C_{BO}}(\bot)$).
\end{enumerate}

Where $D\in\Delta$, $\Delta$ is the set of data.

Alice's state transitions described by $APTC_G$ are as follows.

$A=\sum_{D\in\Delta}r_{C_{AI}}(D)\cdot A_2$

$A_2=\{k_{AB}=NULL\}\cdot A_3+ \{k_{AB}\neq NULL\}\cdot A_{13}$

$A_3=s_{C_{AT}}(B)\cdot A_4$

$A_4=r_{C_{TA}}(SIGN_{sk_T}(B,pk_B))\cdot A_5$

$A_5=de\textrm{-}sign_{pk_T}(SIGN_{sk_T}(B,pk_B))\cdot A_6$

$A_6=(rsg_{k_{AB}}\parallel rsg_{pk_P}\parallel rsg_{sk_P})\cdot A_7$

$A_7=sign_{sk_A}(L,A,k_{AB},sk_P,pk_P)\cdot A_8$

$A_8=(enc_{sk_P}(ENC_{pk_B}(k_{AB}))\parallel enc_{k_{AB}}(T_A))\cdot A_{9}$

$A_9=s_{C_{AB}}(ENC_{k_{AB}}(T_A), SIGN_{sk_A}(L,A,k_{AB},sk_P,pk_P),ENC_{sk_P}(ENC_{pk_B}(k_{AB})))\cdot A_{10}$

$A_{10}=r_{C_{BA}}(ENC_{k_{AB}}(d_{BA}))\cdot A_{11}$

$A_{11}=\{d_{BA}\neq\bot\}\cdot A_{12}+\{d_{BA}= \bot\}\cdot s_{C_{AB}}(ENC_{k_{AB}}(\bot))\cdot A$

$A_{12}=\{isFresh(d_{BA})=TRUE\}\cdot A_{13}+\{isFresh(d_{BA})=FALSE\}\cdot s_{C_{AB}}(ENC_{k_{AB}}(\bot))\cdot A$

$A_{13}=rsg_{R_D}\cdot A_{14}$

$A_{14}=enc_{k_{AB}}(R_D,D)\cdot A_{15}$

$A_{15}=s_{C_{AB}}(ENC_{k_{AB}}(R_D,D))\cdot A$

Bob's state transitions described by $APTC_G$ are as follows.

$B=\{k_{AB}=NULL\}\cdot B_1+ \{k_{AB}\neq NULL\}\cdot B_{11}$

$B_1=r_{C_{AB}}(ENC_{k_{AB}}(T_A), SIGN_{sk_A}(L,A,k_{AB},sk_P,pk_P),ENC_{sk_P}(ENC_{pk_B}(k_{AB})))\cdot B_2$

$B_2=s_{C_{BT}}(A)\cdot B_3$

$B_3=r_{C_{TB}}(SIGN_{sk_T}(A,pk_A))\cdot B_4$

$B_4=de\textrm{-}sign_{pk_T}(SIGN_{sk_T}(A,pk_A))\cdot B_5$

$B_5=de\textrm{-}sign{pk_A}(SIGN_{sk_A}(L,A,k_{AB},sk_P,pk_P))\cdot B_6$

$B_6=(dec_{pk_P}(ENC_{sk_P}(ENC_{pk_B}(k_{AB})))\parallel dec_{sk_B}(ENC_{pk_B}(k_{AB}))\parallel dec_{k_{AB}}(T_A))\cdot B_7$

$B_7=\{isValid(T_A)=TRUE\}\cdot B_8+\{isValid(T_A)=FALSE\}\cdot s_{C_{BA}}(ENC_{k_{AB}}(\bot))\cdot B_{11}$

$B_8=rsg_{T_B}\cdot B_{9}$

$B_{9}=enc_{k_{AB}}(T_B)\cdot B_{10}$

$B_{10}=s_{C_{BA}}(ENC_{k_{AB}}(d_{BA}))\cdot B_{11}$

$B_{11}=r_{C_{AB}}(ENC_{k_{AB}}(d_{AB}))\cdot B_{12}$

$B_{12}=dec_{k_{AB}}(ENC_{k_{AB}}(d_{AB}))\cdot B_{13}$

$B_{13}=\{d_{AB}=\bot\}\cdot s_{C_{BO}}(\bot)\cdot B+\{d_{AB}\neq\bot\}\cdot B_{14}$

$B_{14}=\{isFresh(d_{R_D})=FLASE\}\cdot s_{C_{BO}}(\bot) B+\{isFresh(d_{R_D})=TRUE\}\cdot B_{15}$

$B_{15}=s_{C_{BO}}(D)\cdot B$

Trent's state transitions described by $APTC_G$ are as follows.

$T=r_{C_{AT}}(B)\cdot T_2$

$T_2=sign_{sk_T}(B,pk_B)\cdot T_3$

$T_3=s_{C_{TA}}(SIGN_{sk_T}(B,pk_B))\cdot T_4$

$T_4=r_{C_{BT}}(A)\cdot T_5$

$T_5=sign_{sk_T}(A,pk_A)\cdot T_6$

$T_6=s_{C_{TA}}(SIGN_{sk_T}(A,pk_A))\cdot T$

The sending action and the reading action of the same type data through the same channel can communicate with each other, otherwise, will cause a deadlock $\delta$. We define the following
communication functions.

$\gamma(r_{C_{AT}}(B),s_{C_{AT}}(B))\triangleq c_{C_{AT}}(B)$

$\gamma(r_{C_{BT}}(A),s_{C_{AT}}(A))\triangleq c_{C_{AT}}(A)$

$\gamma(r_{C_{TA}}(SIGN_{sk_T}(B,pk_B)),s_{C_{TA}}(SIGN_{sk_T}(B,pk_B)))\\
\triangleq c_{C_{TA}}(SIGN_{sk_T}(B,pk_B))$

$\gamma(r_{C_{TB}}(SIGN_{sk_T}(A,pk_A)),s_{C_{TB}}(SIGN_{sk_T}(A,pk_A)))\\
\triangleq c_{C_{TB}}(SIGN_{sk_T}(A,pk_A))$

$\gamma(r_{C_{AB}}(ENC_{k_{AB}}(T_A), SIGN_{sk_A}(L,A,k_{AB},sk_P,pk_P),ENC_{sk_P}(ENC_{pk_B}(k_{AB}))),\\
s_{C_{AB}}(ENC_{k_{AB}}(T_A), SIGN_{sk_A}(L,A,k_{AB},sk_P,pk_P),ENC_{sk_P}(ENC_{pk_B}(k_{AB})))\\
\triangleq c_{C_{AB}}(ENC_{k_{AB}}(T_A), SIGN_{sk_A}(L,A,k_{AB},sk_P,pk_P),ENC_{sk_P}(ENC_{pk_B}(k_{AB})))$

$\gamma(r_{C_{BA}}(ENC_{k_{AB}}(d_{BA})),s_{C_{BA}}(ENC_{k_{AB}}(d_{BA})))\triangleq c_{C_{BA}}(ENC_{k_{AB}}(d_{BA}))$

$\gamma(r_{C_{AB}}(ENC_{k_{AB}}(d_{AB})),s_{C_{AB}}(ENC_{k_{AB}}(d_{AB})))\triangleq c_{C_{AB}}(ENC_{k_{AB}}(d_{AB}))$

Let all modules be in parallel, then the protocol $A\quad B\quad T$ can be presented by the following process term.

$$\tau_I(\partial_H(\Theta(A\between B\between T)))=\tau_I(\partial_H(A\between B\between T))$$

where $H=\{r_{C_{AT}}(B),s_{C_{AT}}(B),r_{C_{BT}}(A),s_{C_{BT}}(A),r_{C_{BA}}(ENC_{k_{AB}}(d_{BA})),s_{C_{BA}}(ENC_{k_{AB}}(d_{BA})),\\
r_{C_{AB}}(ENC_{k_{AB}}(d_{AB})),s_{C_{AB}}(ENC_{k_{AB}}(d_{AB})),\\
r_{C_{TA}}(SIGN_{sk_T}(B,pk_B)),s_{C_{TA}}(SIGN_{sk_T}(B,pk_B)),\\
r_{C_{TB}}(SIGN_{sk_T}(A,pk_A)),s_{C_{TB}}(SIGN_{sk_T}(A,pk_A)),\\
r_{C_{AB}}(ENC_{k_{AB}}(T_A), SIGN_{sk_A}(L,A,k_{AB},sk_P,pk_P),ENC_{sk_P}(ENC_{pk_B}(k_{AB}))),\\
s_{C_{AB}}(ENC_{k_{AB}}(T_A), SIGN_{sk_A}(L,A,k_{AB},sk_P,pk_P),ENC_{sk_P}(ENC_{pk_B}(k_{AB})))|D\in\Delta\}$,

$I=\{c_{C_{AT}}(B),c_{C_{BT}}(A),c_{C_{BA}}(ENC_{k_{AB}}(d_{BA})),c_{C_{AB}}(ENC_{k_{AB}}(d_{AB})),\\
c_{C_{TA}}(SIGN_{sk_T}(B,pk_B)),c_{C_{TB}}(SIGN_{sk_T}(A,pk_A))\\
c_{C_{AB}}(ENC_{k_{AB}}(T_A), SIGN_{sk_A}(L,A,k_{AB},sk_P,pk_P),ENC_{sk_P}(ENC_{pk_B}(k_{AB}))),\\
\{k_{AB}=NULL\},\{k_{AB}\neq NULL\},de\textrm{-}sign_{pk_T}(SIGN_{sk_T}(B,pk_B)),\\
rsg_{k_{AB}},rsg_{pk_p},rsg_{sk_P},sign_{sk_A}(L,A,k_{AB},sk_P,pk_P),enc_{sk_P}(ENC_{pk_B}(k_{AB})), enc_{k_{AB}}(T_A),\\
\{isFresh(d_{BA})=TRUE\},\{isFresh(d_{BA})=FALSE\},\{d_{BA}\neq\bot\},\{d_{BA}=\bot\},\\
rsg_{R_D},enc_{k_{AB}}(R_D,D),de\textrm{-}sign_{pk_T}(SIGN_{sk_T}(A,pk_A)),\\
dec_{sk_B}(ENC_{pk_B}(SIGN_{sk_A}(L,A,k_{AB},sk_P,pk_P))),dec_{pk_P}(ENC_{sk_P}(ENC_{pk_B}(k_{AB}))),\\
dec_{sk_B}(ENC_{pk_B}(k_{AB})), dec_{k_{AB}}(T_A),\{isValid(T_A)=TRUE\},\\
\{isValid(T_A)=FALSE\},rsg_{R_D},enc_{k_{AB}}(R_D),\\
dec_{k_{AB}}(ENC_{k_{AB}}(d_{AB})),\{d_{AB}'=\bot\},\{d_{AB}\neq\bot\},\\
\{isFresh(d_{R_D})=TRUE\},\{isFresh(d_{R_D})=FLASE\},sign_{sk_T}(A,pk_A),sign_{sk_T}(B,pk_B)|D\in\Delta\}$.

Then we get the following conclusion on the protocol.

\begin{theorem}
The DASS protocol in Figure \ref{DASS7} is secure.
\end{theorem}

\begin{proof}
Based on the above state transitions of the above modules, by use of the algebraic laws of $APTC_G$, we can prove that

$\tau_I(\partial_H(A\between B\between T))=\sum_{D\in\Delta}(r_{C_{AI}}(D)\cdot (s_{C_{BO}}(\bot)+s_{C_{BO}}(D)))\cdot
\tau_I(\partial_H(A\between B\between T))$.

For the details of proof, please refer to section \ref{app}, and we omit it.

That is, the DASS protocol in Figure \ref{DASS7} $\tau_I(\partial_H(A\between B\between T))$ can exhibit desired external behaviors:

\begin{enumerate}
  \item For the modeling of confidentiality, it is similar to the protocol in section \ref{confi}, the DASS protocol is confidential;
  \item For the man-in-the-middle attack, because $pk_A$ and $pk_B$ are signed by Trent, the protocol would be $\tau_I(\partial_H(A\between B\between T))=\sum_{D\in\Delta}(r_{C_{AI}}(D)\cdot s_{C_{BO}}(\bot))\cdot
  \tau_I(\partial_H(A\between B\between T))$, it is desired, the DASS protocol can be against the
  man-in-the-middle attack;
  \item For replay attack, the using of the time stamp $T_A$, $T_B$, and random number $R_D$, makes that $\tau_I(\partial_H(A\between B\between T))=\sum_{D\in\Delta}(r_{C_{AI}}(D)\cdot s_{C_{BO}}(\bot))\cdot
  \tau_I(\partial_H(A\between B\between T))$, it is desired;
  \item Without man-in-the-middle and replay attack, the protocol would be $\tau_I(\partial_H(A\between B\between T))=\sum_{D\in\Delta}(r_{C_{AI}}(D)\cdot s_{C_{BO}}(D))\cdot
  \tau_I(\partial_H(A\between B\between T))$, it is desired;
  \item For the unexpected and non-technical leaking of $sk_{A}$, $sk_{B}$, $k_{AB}$, or they being not strong enough, or Trent being dishonest, they are out of the scope of analyses of security protocols;
  \item For malicious tampering and transmission errors, they are out of the scope of analyses of security protocols.
\end{enumerate}
\end{proof}

\subsection{Woo-Lam Protocol}\label{wlp}

The Woo-Lam protocol shown in Figure \ref{WLP7} uses asymmetric keys and symmetric keys for secure communication, that is, the key $k_{AB}$ between Alice and Bob is privately shared to Alice and Bob,
Alice's, Bob's and Trent's public keys $pk_{A}$, $pk_{B}$ and $pk_{T}$ can be publicly gotten.

\begin{figure}
    \centering
    \includegraphics{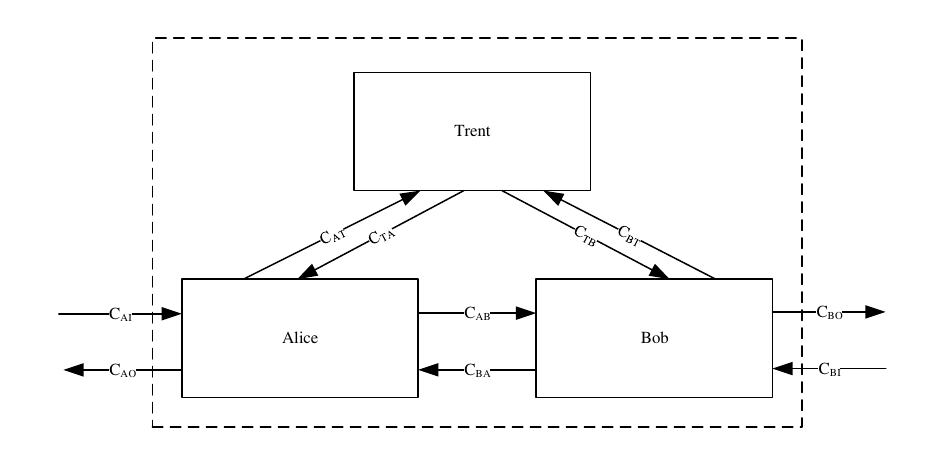}
    \caption{Woo-Lam protocol}
    \label{WLP7}
\end{figure}

The process of the protocol is as follows.

\begin{enumerate}
  \item Alice receives some messages $D$ from the outside through the channel $C_{AI}$ (the corresponding reading action is denoted $r_{C_{AI}}(D)$), if $k_{AB}$ is not established,
  she sends $A,B$ to Trent through the channel $C_{AT}$ (the corresponding sending action is denoted
  $s_{C_{AT}}(A,B)$);
  \item Trent receives $A,B$ through the channel $C_{AT}$ (the corresponding reading action is denoted
  $r_{C_{AT}}(A,B)$), he signs Bob's public key $pk_B$ through the action $sign_{sk_T}(pk_B)$, and sends the signature to Alice
  through the channel $C_{TA}$ (the corresponding sending action is denoted $s_{C_{TA}}(SIGN_{sk_T}(pk_B))$);
  \item Alice receives the message from Trent through the channel $C_{TA}$ (the corresponding reading action is denoted $r_{C_{TA}}(SIGN_{sk_T}(pk_B))$),
  she de-signs $SIGN_{sk_T}(pk_B)$ through an action $de\textrm{-}sign_{pk_T}(SIGN_{sk_T}(pk_B))$ to get $pk_B$, generates a random number $R_A$ through an action $rsg_{R_A}$
  and encrypts $A,R_A$ by $pk_B$
  through an action $enc_{pk_B}(A,R_A)$, and sends $ENC_{pk_B}(A,R_A)$
  to Bob through the channel $C_{AB}$ (the corresponding sending action is denoted $s_{C_{AB}}(ENC_{pk_B}(A,R_A))$);
  \item Bob receives $ENC_{pk_B}(A,R_A)$ from Alice (the corresponding reading action is denoted
  $r_{C_{AB}}(ENC_{pk_B}(A,R_A))$), he decrypts $ENC_{pk_B}(A,R_A)$ through an action $dec_{sk_B}(ENC_{pk_B}(A,R_A))$ to get $A$ and $R_A$, encrypts $R_A$ by $pk_T$ through an action
  $enc_{pk_T}(R_A)$, then sends $A, B, ENC_{pk_T}(R_A)$ to Trent through the channel $C_{BT}$
  (the corresponding sending action is denoted \\
  $s_{C_{BT}}(A, B, ENC_{pk_T}(R_A))$);
  \item Trent receives $A, B, ENC_{pk_T}(R_A)$ from Bob through the channel $C_{BT}$ (the corresponding reading action is denoted $r_{C_{BT}}(A, B, ENC_{pk_T}(R_A))$),
  he decrypts the message through an action $dec_{sk_T}(ENC_{pk_T}(R_A))$, signs $pk_A$ through an
  action $sign_{sk_T}(pk_A)$, generates a random session key $k_{AB}$ through an action $rsg_{k_{AB}}$ and signs $R_A,k_{AB},A,B$ through an action $sign_{sk_T}(R_A,k_{AB},A,B)$,
  encrypts $SIGN_{sk_T}(R_A,k_{AB},A,B)$ through an action $enc_{pk_B}(SIGN_{sk_T}(R_A,k_{AB},A,B))$ and sends the them to Bob through the channel $C_{TB}$
  (the corresponding sending action is denoted \\
  $s_{C_{TB}}(SIGN_{sk_T}(pk_A),ENC_{pk_B}(SIGN_{sk_T}(R_A,k_{AB},A,B)))$);
  \item Bob receives the signatures from Trent through the channel $C_{TB}$ (the corresponding reading action is denoted $r_{C_{TB}}(SIGN_{sk_T}(pk_A),ENC_{pk_B}(SIGN_{sk_T}(R_A,k_{AB},A,B)))$),
  he de-signs $SIGN_{sk_T}(pk_A)$ through an action $de\textrm{-}sign_{pk_T}(SIGN_{sk_T}(pk_A))$
  to get $pk_A$, decrypts $ENC_{pk_B}(SIGN_{sk_T}(R_A,k_{AB},A,B))$ through an action $dec_{sk_B}(ENC_{pk_B}(SIGN_{sk_T}(R_A,k_{AB},A,B)))$, generates a random number $R_B$ through an action $rsg_{R_B}$, encrypts them through an action
  $enc_{pk_A}(SIGN_{sk_T}(R_A,k_{AB},A,B),R_B)$ and sends $ENC_{pk_A}(SIGN_{sk_T}(R_A,k_{AB},A,B),R_B)$ to Alice through the channel $C_{BA}$ (the corresponding sending action is
  denoted \\
  $s_{C_{BA}}(ENC_{pk_A}(SIGN_{sk_T}(R_A,k_{AB},A,B),R_B))$);
  \item Alice receives $ENC_{pk_A}(SIGN_{sk_T}(d_{R_A},k_{AB},A,B),R_B)$ from Bob (the corresponding reading action is denoted $r_{C_{BA}}(ENC_{pk_A}(SIGN_{sk_T}(d_{R_A},k_{AB},A,B),R_B))$),
  she decrypts the message through an action $dec_{sk_A}(ENC_{pk_A}(SIGN_{sk_T}(R_A,k_{AB},A,B),R_B))$, \\de-sign $SIGN_{sk_T}(R_A,k_{AB},A,B)$ through an action
  $de\textrm{-}sign_{pk_T}(SIGN_{sk_T}(R_A,k_{AB},A,B))$,
  if $d_{R_A}\neq R_A$, she sends $ENC_{k_{AB}}(\bot)$ to Bob through the channel $C_{AB}$
  (the corresponding sending action is denoted $s_{C_{AB}}(ENC_{k_{AB}}(\bot))$); else if $d_{R_A}=R_A$, encrypts $R_B,D$ by $k_{AB}$ through an
  action $enc_{k_{AB}}(R_B,D)$, and sends it to Bob through the channel $C_{AB}$ (the corresponding sending action is denoted $s_{C_{AB}}(ENC_{k_{AB}}(R_B,D))$);
  \item Bob receives $ENC_{k_{AB}}(d_{AB})$ from Alice (the corresponding reading action is denoted $r_{C_{AB}}(ENC_{k_{AB}}(d_{AB}))$), if $d_{AB}=\bot$, he sends $\bot$
  to the outside through the channel $C_{BO}$
  (the corresponding sending action is denoted $s_{C_{BO}}(\bot)$); else if $d_{AB}\neq \bot$, if $d_{R_B}=R_B$, she sends $D$ to the outside through the channel $C_{BO}$ (the corresponding sending action is denoted $s_{C_{BO}}(D)$), else if
  $d_{R_B}\neq R_B$, he sends $\bot$ to the outside through the channel $C_{BO}$ (the corresponding sending action is denoted $s_{C_{BO}}(\bot)$).
\end{enumerate}

Where $D\in\Delta$, $\Delta$ is the set of data.

Alice's state transitions described by $APTC_G$ are as follows.

$A=\sum_{D\in\Delta}r_{C_{AI}}(D)\cdot A_2$

$A_2=\{k_{AB}=NULL\}\cdot A_3+ \{k_{AB}\neq NULL\}\cdot A_9$

$A_3=s_{C_{AT}}(A,B)\cdot A_4$

$A_4=r_{C_{TA}}(SIGN_{sk_T}(pk_B))\cdot A_5$

$A_5=de\textrm{-}sign_{pk_T}(SIGN_{sk_T}(pk_B))\cdot A_6$

$A_6=rsg_{R_A}\cdot A_7$

$A_7=enc_{sk_P}(A,R_A)\cdot A_{8}$

$A_8=s_{C_{AB}}(ENC_{sk_P}(A,R_A))\cdot A_{9}$

$A_{9}=r_{C_{BA}}(ENC_{pk_A}(SIGN_{sk_T}(d_{R_A},k_{AB},A,B),R_B))\cdot A_{10}$

$A_{10}=\{d_{R_A}=R_A\}\cdot A_{11}+\{d_{R_A}\neq R_A\}\cdot s_{C_{AB}}(ENC_{k_{AB}}(\bot))\cdot A$

$A_{11}=enc_{k_{AB}}(R_B,D)\cdot A_{12}$

$A_{12}=s_{C_{AB}}(ENC_{k_{AB}}(R_B,D))\cdot A$

Bob's state transitions described by $APTC_G$ are as follows.

$B=\{k_{AB}=NULL\}\cdot B_1+ \{k_{AB}\neq NULL\}\cdot B_{10}$

$B_1=r_{C_{AB}}(ENC_{sk_P}(A,R_A))\cdot B_2$

$B_2=dec_{sk_B}(ENC_{pk_B}(A,R_A))\cdot B_3$

$B_3=s_{C_{BT}}(A, B, ENC_{pk_T}(R_A))\cdot B_4$

$B_4=r_{C_{TB}}(SIGN_{sk_T}(pk_A),ENC_{pk_B}(SIGN_{sk_T}(R_A,k_{AB},A,B)))\cdot B_5$

$B_5=de\textrm{-}sign_{pk_T}(SIGN_{sk_T}(pk_A))\cdot B_6$

$B_6=dec_{sk_B}(ENC_{pk_B}(SIGN_{sk_T}(R_A,k_{AB},A,B)))\cdot B_7$

$B_7=rsg_{R_B}\cdot B_8$

$B_8=enc_{pk_A}(SIGN_{sk_T}(R_A,k_{AB},A,B),R_B)\cdot B_{9}$

$B_9=s_{C_{BA}}(ENC_{pk_A}(SIGN_{sk_T}(R_A,k_{AB},A,B),R_B))\cdot B_{10}$

$B_{10}=r_{C_{AB}}(ENC_{k_{AB}}(d_{AB}))\cdot B_{11}$

$B_{11}=dec_{k_{AB}}(ENC_{k_{AB}}(d_{AB}))\cdot B_{12}$

$B_{12}=\{d_{AB}=\bot\}\cdot s_{C_{BO}}(\bot)\cdot B+\{d_{AB}\neq\bot\}\cdot B_{13}$

$B_{13}=\{d_{R_B}\neq R_B\}\cdot s_{C_{BO}}(\bot) B+\{d_{R_B}=R_B\}\cdot B_{14}$

$B_{14}=s_{C_{BO}}(D)\cdot B$

Trent's state transitions described by $APTC_G$ are as follows.

$T=r_{C_{AT}}(A,B)\cdot T_2$

$T_2=sign_{sk_T}(pk_B)\cdot T_3$

$T_3=s_{C_{TA}}(SIGN_{sk_T}(pk_B))\cdot T_4$

$T_4=r_{C_{BT}}(A, B, ENC_{pk_T}(R_A))\cdot T_5$

$T_5=dec_{sk_T}(ENC_{pk_T}(R_A))\cdot T_6$

$T_6=sign_{sk_T}(pk_A)\cdot T_7$

$T_7=rsg_{k_{AB}}\cdot T_8$

$T_8=sign_{sk_T}(R_A,k_{AB},A,B)\cdot T_9$

$T_9=enc_{pk_B}(SIGN_{sk_T}(R_A,k_{AB},A,B))\cdot T_{10}$

$T_{10}=s_{C_{TB}}(SIGN_{sk_T}(pk_A),ENC_{pk_B}(SIGN_{sk_T}(R_A,k_{AB},A,B)))\cdot T$

The sending action and the reading action of the same type data through the same channel can communicate with each other, otherwise, will cause a deadlock $\delta$. We define the following
communication functions.

$\gamma(r_{C_{AT}}(A,B),s_{C_{AT}}(A,B))\triangleq c_{C_{AT}}(A,B)$

$\gamma(r_{C_{BT}}(A, B, ENC_{pk_T}(R_A)),s_{C_{AT}}(A, B, ENC_{pk_T}(R_A)))\triangleq c_{C_{AT}}(A, B, ENC_{pk_T}(R_A))$

$\gamma(r_{C_{TA}}(SIGN_{sk_T}(pk_B)),s_{C_{TA}}(SIGN_{sk_T}(pk_B)))\\
\triangleq c_{C_{TA}}(SIGN_{sk_T}(pk_B))$

$\gamma(r_{C_{TB}}(SIGN_{sk_T}(pk_A),ENC_{pk_B}(SIGN_{sk_T}(R_A,k_{AB},A,B))),\\
s_{C_{TB}}(SIGN_{sk_T}(pk_A),ENC_{pk_B}(SIGN_{sk_T}(R_A,k_{AB},A,B))))\\
\triangleq c_{C_{TB}}(SIGN_{sk_T}(A,pk_A))$

$\gamma(r_{C_{AB}}(ENC_{sk_P}(A,R_A)),s_{C_{AB}}(ENC_{sk_P}(A,R_A)))\\
\triangleq c_{C_{AB}}(ENC_{sk_P}(A,R_A))$

$\gamma(r_{C_{BA}}(ENC_{pk_A}(SIGN_{sk_T}(d_{R_A},k_{AB},A,B),R_B)),\\
s_{C_{BA}}(ENC_{pk_A}(SIGN_{sk_T}(d_{R_A},k_{AB},A,B),R_B)))\\
\triangleq c_{C_{BA}}(ENC_{pk_A}(SIGN_{sk_T}(d_{R_A},k_{AB},A,B),R_B))$

$\gamma(r_{C_{AB}}(ENC_{k_{AB}}(R_B,D)),s_{C_{AB}}(ENC_{k_{AB}}(R_B,D)))\triangleq c_{C_{AB}}(ENC_{k_{AB}}(R_B,D))$

Let all modules be in parallel, then the protocol $A\quad B\quad T$ can be presented by the following process term.

$$\tau_I(\partial_H(\Theta(A\between B\between T)))=\tau_I(\partial_H(A\between B\between T))$$

where $H=\{r_{C_{AT}}(A,B),s_{C_{AT}}(A,B),r_{C_{BT}}(A, B, ENC_{pk_T}(R_A)),s_{C_{AT}}(A, B, ENC_{pk_T}(R_A)),\\
r_{C_{TA}}(SIGN_{sk_T}(pk_B)),s_{C_{TA}}(SIGN_{sk_T}(pk_B)),\\
r_{C_{TB}}(SIGN_{sk_T}(pk_A),ENC_{pk_B}(SIGN_{sk_T}(R_A,k_{AB},A,B))),\\
s_{C_{TB}}(SIGN_{sk_T}(pk_A),ENC_{pk_B}(SIGN_{sk_T}(R_A,k_{AB},A,B))),\\
r_{C_{AB}}(ENC_{sk_P}(A,R_A)),s_{C_{AB}}(ENC_{sk_P}(A,R_A)),\\
r_{C_{BA}}(ENC_{pk_A}(SIGN_{sk_T}(d_{R_A},k_{AB},A,B),R_B)),\\
s_{C_{BA}}(ENC_{pk_A}(SIGN_{sk_T}(d_{R_A},k_{AB},A,B),R_B)),\\
r_{C_{AB}}(ENC_{k_{AB}}(R_B,D)),s_{C_{AB}}(ENC_{k_{AB}}(R_B,D))|D\in\Delta\}$,

$I=\{c_{C_{AT}}(A,B),c_{C_{AT}}(A, B, ENC_{pk_T}(R_A)),c_{C_{TA}}(SIGN_{sk_T}(pk_B)),\\
c_{C_{TB}}(SIGN_{sk_T}(A,pk_A)),c_{C_{AB}}(ENC_{sk_P}(A,R_A)),\\
c_{C_{BA}}(ENC_{pk_A}(SIGN_{sk_T}(d_{R_A},k_{AB},A,B),R_B)),c_{C_{AB}}(ENC_{k_{AB}}(R_B,D)),\\
\{k_{AB}=NULL\},\{k_{AB}\neq NULL\},de\textrm{-}sign_{pk_T}(SIGN_{sk_T}(pk_B)),\\
rsg_{R_A},enc_{sk_P}(A,R_A),\{d_{R_A}=R_A\},\{d_{R_A}\neq R_A\},enc_{k_{AB}}(R_B,D),\\
dec_{sk_B}(ENC_{pk_B}(A,R_A)),dec_{sk_B}(ENC_{pk_B}(A,R_A)),de\textrm{-}sign_{pk_T}(SIGN_{sk_T}(pk_A)),\\
dec_{sk_B}(ENC_{pk_B}(SIGN_{sk_T}(R_A,k_{AB},A,B))),rsg_{R_B},enc_{pk_A}(SIGN_{sk_T}(R_A,k_{AB},A,B),R_B),\\
dec_{k_{AB}}(ENC_{k_{AB}}(d_{AB})),\{d_{AB}=\bot\},\{d_{AB}\neq\bot\},\\
\{d_{R_B}= R_B\},\{d_{R_B}\neq R_B\},sign_{sk_T}(pk_B),dec_{sk_T}(ENC_{pk_T}(R_A)),\\
sign_{sk_T}(pk_A),rsg_{k_{AB}},sign_{sk_T}(R_A,k_{AB},A,B),enc_{pk_B}(SIGN_{sk_T}(R_A,k_{AB},A,B))|D\in\Delta\}$.

Then we get the following conclusion on the protocol.

\begin{theorem}
The Woo-Lam protocol in Figure \ref{WLP7} is secure.
\end{theorem}

\begin{proof}
Based on the above state transitions of the above modules, by use of the algebraic laws of $APTC_G$, we can prove that

$\tau_I(\partial_H(A\between B\between T))=\sum_{D\in\Delta}(r_{C_{AI}}(D)\cdot (s_{C_{BO}}(\bot)+s_{C_{BO}}(D)))\cdot
\tau_I(\partial_H(A\between B\between T))$.

For the details of proof, please refer to section \ref{app}, and we omit it.

That is, the Woo-Lam protocol in Figure \ref{WLP7} $\tau_I(\partial_H(A\between B\between T))$ can exhibit desired external behaviors:

\begin{enumerate}
  \item For the modeling of confidentiality, it is similar to the protocol in section \ref{confi}, the Woo-Lam protocol is confidential;
  \item For the man-in-the-middle attack, because $pk_A$ and $pk_B$ are signed by Trent, the protocol would be $\tau_I(\partial_H(A\between B\between T))=\sum_{D\in\Delta}(r_{C_{AI}}(D)\cdot s_{C_{BO}}(\bot))\cdot
  \tau_I(\partial_H(A\between B\between T))$, it is desired, the Woo-Lam protocol can be against the
  man-in-the-middle attack;
  \item For replay attack, the using of the random number $R_A$, $R_B$, makes that $\tau_I(\partial_H(A\between B\between T))=\sum_{D\in\Delta}(r_{C_{AI}}(D)\cdot s_{C_{BO}}(\bot))\cdot
  \tau_I(\partial_H(A\between B\between T))$, it is desired;
  \item Without man-in-the-middle and replay attack, the protocol would be $\tau_I(\partial_H(A\between B\between T))=\sum_{D\in\Delta}(r_{C_{AI}}(D)\cdot s_{C_{BO}}(D))\cdot
  \tau_I(\partial_H(A\between B\between T))$, it is desired;
  \item For the unexpected and non-technical leaking of $sk_{A}$, $sk_{B}$, $k_{AB}$, or they being not strong enough, or Trent being dishonest, they are out of the scope of analyses of security protocols;
  \item For malicious tampering and transmission errors, they are out of the scope of analyses of security protocols.
\end{enumerate}
\end{proof}

\newpage\section{Analyses of Other Protocols}\label{aoop}

In this chapter, we will introduce some other useful security protocols, including secret splitting protocols in section \ref{assp}, bit
commitment protocols in section \ref{aobcp1}, anonymous key distribution protocols in section \ref{aoakdp}.

\subsection{Analyses of Secret Splitting Protocols}\label{assp}

The hypothetical secret splitting protocol is shown in Figure \ref{SS8}. Trent receives a message, splits into four parts, and each part is sent to Alice, Bob, Carol and Dave. Then Trent
gathers the four parts from Alice, Bob, Carol and Dave, combines into a message. If the combined message is the original message, then sends out the message.

\begin{figure}
    \centering
    \includegraphics{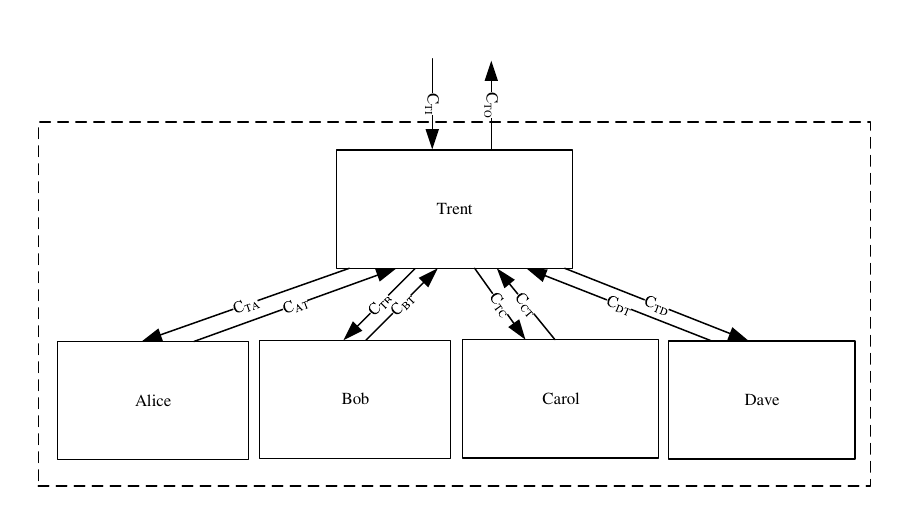}
    \caption{Secret splitting protocol}
    \label{SS8}
\end{figure}

The process of the protocol is as follows.

\begin{enumerate}
  \item Trent receives some messages $D$ from the outside through the channel $C_{TI}$ (the corresponding reading action is denoted $r_{C_{TI}}(D)$), he generates three random numbers
  $R_1,R_2,R_3$ of equal lengths to $D$ through three actions $rsg_{R_1}$, $rsg_{R_2}$ and $rsg_{R_3}$ respectively. Then he does an XOR operation to the data $D$, $R_1$, $R_2$ and $R_3$
  through an XOR action $xor(R_1,R_2,R_3,D)$ to get $R_4=XOR(R_1,R_2,R_3,D)$, he sends $R_1$, $R_2$, $R_3$, $R_4$ to Alice, Bob, Carol, and Dave through the channels $C_{TA}$, $C_{TB}$,
  $C_{TC}$ and $C_{TD}$ respectively (the corresponding sending actions are denoted $s_{C_{TA}}(R_1)$, $s_{C_{TB}}(R_2)$, $s_{C_{TC}}(R_3)$, $s_{C_{TD}}(R_4)$);
  \item Alice receives $R_1$ from Trent through the channel $C_{TA}$ (the corresponding reading action is denoted $r_{C_{TA}}(R_1)$), she may store $R_1$, we assume that she sends $R_1$
  to Trent immediately through the channel $C_{AT}$ (the corresponding sending action is denoted $s_{C_{AT}}(R_1)$);
  \item Bob receives $R_2$ from Trent through the channel $C_{TB}$ (the corresponding reading action is denoted $r_{C_{TB}}(R_2)$), he may store $R_2$, we assume that he sends $R_2$
  to Trent immediately through the channel $C_{BT}$ (the corresponding sending action is denoted $s_{C_{BT}}(R_2)$);
  \item Carol receives $R_3$ from Trent through the channel $C_{TC}$ (the corresponding reading action is denoted $r_{C_{TC}}(R_3)$), he may store $R_3$, we assume that he sends $R_3$
  to Trent immediately through the channel $C_{CT}$ (the corresponding sending action is denoted $s_{C_{CT}}(R_3)$);
  \item Dave receives $R_4$ from Trent through the channel $C_{TD}$ (the corresponding reading action is denoted $r_{C_{TD}}(R_4)$), she may store $R_4$, we assume that he sends $R_4$
  to Trent immediately through the channel $C_{DT}$ (the corresponding sending action is denoted $s_{C_{DT}}(R_4)$);
  \item Trent receives $d_{R_1}$, $d_{R_3}$, $d_{R_3}$, and $d_{R_4}$ from Alice, Bob, Carol and Dave through the channel $C_{AT}$, $C_{BT}$, $C_{CT}$, and $C_{DT}$ respectively
  (the corresponding reading actions are denoted $r_{C_{AT}}(d_{R_1})$, $r_{C_{BT}}(d_{R_2})$, $r_{C_{CT}}(d_{R_3})$, $r_{C_{DT}}(d_{R_4})$), he does an XOR operation to the data $d_{R_1}$,
  $d_{R_2}$, $d_{R_3}$ and $d_{R_4}$ through an XOR action $xor(d_{R_1},d_{R_2},d_{R_3},d_{R_4})$ to get $D'=XOR(d_{R_1},d_{R_2},d_{R_3},d_{R_4})$, if $D=D'$, he sends $D$
  to the outside through the channel $C_{TO}$ (the corresponding sending action is denoted $s_{C_{TO}}(D)$).
\end{enumerate}

Where $D\in\Delta$, $\Delta$ is the set of data.

Trent's state transitions described by $APTC_G$ are as follows.

$T=\sum_{D\in\Delta}r_{C_{TI}}(D)\cdot T_2$

$T_2=(rsg_{R_1}\parallel rsg_{R_2}\parallel rsg_{R_3})\cdot T_3$

$T_3=xor(R_1,R_2,R_3,D)\cdot T_4$

$T_4=(s_{C_{TA}}(R_1)\parallel s_{C_{TB}}(R_2)\parallel s_{C_{TC}}(R_3)\parallel s_{C_{TD}}(R_4))\cdot T_5$

$T_5=(r_{C_{AT}}(d_{R_1})\parallel r_{C_{BT}}(d_{R_2})\parallel r_{C_{CT}}(d_{R_3})\parallel r_{C_{DT}}(d_{R_4}))\cdot T_6$

$T_6=xor(d_{R_1},d_{R_2},d_{R_3},d_{R_4})\cdot T_7$

$T_7=\{D=D'\}\cdot s_{C_{TO}}(D)\cdot T$

Alice's state transitions described by $APTC_G$ are as follows.

$A=r_{C_{TB}}(R_2)\cdot A_2$

$A_2=s_{C_{BT}}(R_2)\cdot A$

Bob's state transitions described by $APTC_G$ are as follows.

$B=r_{C_{TB}}(R_2)\cdot B_2$

$B_2=s_{C_{BT}}(R_2)\cdot B$

Carol's state transitions described by $APTC_G$ are as follows.

$C=r_{C_{TB}}(R_2)\cdot C_2$

$C_2=s_{C_{BT}}(R_2)\cdot C$

Dave's state transitions described by $APTC_G$ are as follows.

$Da=r_{C_{TB}}(R_2)\cdot Da_2$

$Da_2=s_{C_{BT}}(R_2)\cdot Da$

The sending action and the reading action of the same type data through the same channel can communicate with each other, otherwise, will cause a deadlock $\delta$. We define the following
communication functions.

$\gamma(r_{C_{TA}}(R_1),s_{C_{TA}}(R_1))\triangleq c_{C_{TA}}(R_1)$

$\gamma(r_{C_{TB}}(R_2),s_{C_{TB}}(R_2))\triangleq c_{C_{TB}}(R_2)$

$\gamma(r_{C_{TC}}(R_3),s_{C_{TC}}(R_3))\triangleq c_{C_{TC}}(R_3)$

$\gamma(r_{C_{TD}}(R_4),s_{C_{TD}}(R_4))\triangleq c_{C_{TD}}(R_4)$

$\gamma(r_{C_{AT}}(d_{R_1}),s_{C_{AT}}(d_{R_1}))\triangleq c_{C_{AT}}(d_{R_1})$

$\gamma(r_{C_{BT}}(d_{R_2}),s_{C_{BT}}(d_{R_2}))\triangleq c_{C_{BT}}(d_{R_2})$

$\gamma(r_{C_{CT}}(d_{R_3}),s_{C_{CT}}(d_{R_3}))\triangleq c_{C_{CT}}(d_{R_3})$

$\gamma(r_{C_{DT}}(d_{R_4}),s_{C_{DT}}(d_{R_4}))\triangleq c_{C_{DT}}(d_{R_4})$

Let all modules be in parallel, then the protocol $A\quad B\quad C\quad Da\quad T$ can be presented by the following process term.

$$\tau_I(\partial_H(\Theta(A\between B\between C\between Da\between T)))=\tau_I(\partial_H(A\between B\between C\between Da\between T))$$

where $H=\{r_{C_{TA}}(R_1),s_{C_{TA}}(R_1),r_{C_{TB}}(R_2),s_{C_{TB}}(R_2),r_{C_{TC}}(R_3),s_{C_{TC}}(R_3),\\
r_{C_{TD}}(R_4),s_{C_{TD}}(R_4),r_{C_{AT}}(d_{R_1}),s_{C_{AT}}(d_{R_1}),r_{C_{BT}}(d_{R_2}),s_{C_{BT}}(d_{R_2}),\\
r_{C_{CT}}(d_{R_3}),s_{C_{CT}}(d_{R_3}),r_{C_{DT}}(d_{R_4}),s_{C_{DT}}(d_{R_4})|D\in\Delta\}$,

$I=\{c_{C_{TA}}(R_1),c_{C_{TB}}(R_2),c_{C_{TC}}(R_3),c_{C_{TD}}(R_4),c_{C_{AT}}(d_{R_1}),c_{C_{BT}}(d_{R_2}),\\
c_{C_{CT}}(d_{R_3}),c_{C_{DT}}(d_{R_4}),rsg_{R_1}, rsg_{R_2}, rsg_{R_3},xor(R_1,R_2,R_3,D),\\
xor(d_{R_1},d_{R_2},d_{R_3},d_{R_4}),\{D=D'\}|D\in\Delta\}$.

Then we get the following conclusion on the protocol.

\begin{theorem}
The secret splitting protocol in Figure \ref{SS8} is secure.
\end{theorem}

\begin{proof}
Based on the above state transitions of the above modules, by use of the algebraic laws of $APTC_G$, we can prove that

$\tau_I(\partial_H(A\between B\between C\between Da\between T))=\sum_{D\in\Delta}(r_{C_{TI}}(D)\cdot s_{C_{TO}}(D))\cdot
\tau_I(\partial_H(A\between B\between C\between Da\between T))$.

For the details of proof, please refer to section \ref{app}, and we omit it.

That is, the protocol in Figure \ref{SS8} $\tau_I(\partial_H(A\between B\between C\between Da\between T))$ can exhibit desired external behaviors, and satisfies the main goal of secret splitting.
It must be noted that the distribution and gathering of $R_1, R_2, R_3, R_4$ have not any cryptographic assurance, they can be made an information leakage.
\end{proof}

\subsection{Analyses of Bit Commitment Protocols}\label{aobcp1}

In this chapter, we will introduce analyses of bit commitment protocols. We introduce analyses of bit commitment protocol based on symmetric cryptography in section \ref{bcp1},
and bit commitment protocol based on one-way function in section \ref{bcp2}.

\subsubsection{Bit Commitment Protocol 1}\label{bcp1}

The protocol shown in Figure \ref{BCP1} uses symmetric cryptography to implement bit commitment.

\begin{figure}
    \centering
    \includegraphics{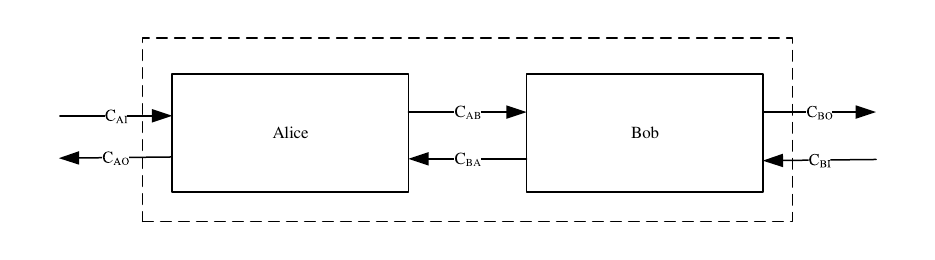}
    \caption{Bit commitment protocol 1}
    \label{BCP1}
\end{figure}

The process of the protocol is as follows.

\begin{enumerate}
  \item Bob receives some requests $D$ from the outside through the channel $C_{BI}$ (the corresponding reading action is denoted $r_{C_{BI}}(D)$), he generates a random sequence $R$
  through an action $rsg_R$,
  then Bob sends $R$ to Alice through the channel $C_{BA}$ (the corresponding sending action is denoted $s_{C_{BA}}(R)$);
  \item Alice receives $R$ from Bob through the channel $C_{BA}$ (the corresponding reading action is denoted $r_{C_{BA}}(R)$), she generates the commitment $b$ through an action
  $rsg_b$ and generate a random key $k$ through an action $rsg_k$, encrypts $b$ and $R$ by $k$ through an action $enc_k(R,b)$, and sends $ENC_k(R,b)$ to Bob through the channel $C_{AB}$
  (the corresponding sending action is denoted $s_{C_{AB}}(ENC_k(R,b))$);
  \item Bob receives the message $ENC_k(R,b)$ from Alice through the channel $C_{AB}$ (the corresponding reading action is denoted $r_{C_{AB}}(ENC_k(R,b))$), he cannot decrypt the message
  for the absence of $k$; after some time, he sends a commitment release request $r$ to Alice through the channel $C_{BA}$ (the corresponding sending action is denoted $s_{C_{BA}}(r)$);
  \item Alice receives $r$ from Bob through the channel $C_{BA}$ (the corresponding reading action is denoted $r_{C_{BA}}(r)$), she sends $k$ to Bob through the channel $C_{AB}$
  (the corresponding sending action is denoted $s_{C_{AB}}(k)$);
  \item Bob receives $k$ from Alice through the channel $C_{AB}$ (the corresponding reading action is denoted $r_{C_{AB}}(k)$), he decrypts $ENC_{k}(d_R,b)$ through an action
  $dec_k(ENC_k(d_R,b))$, if $d_R=R$, he sends $b$
  to the outside through the channel $C_{BO}$ (the corresponding sending action is denoted $s_{C_{BO}}(b)$); else if $d_R\neq R$, he sends $\bot$
  to the outside through the channel $C_{BO}$ (the corresponding sending action is denoted $s_{C_{BO}}(\bot)$).
\end{enumerate}

Where $D\in\Delta$, $\Delta$ is the set of data.

Alice's state transitions described by $APTC_G$ are as follows.

$A=r_{C_{BA}}(R)\cdot A_2$

$A_2=rsg_b\cdot A_3$

$A_3=rsg_k\cdot A_4$

$A_4=s_{C_{AB}}(ENC_k(R,b))\cdot A_5$

$A_5=r_{C_{BA}}(r)\cdot A_6$

$A_6=s_{C_{AB}}(k)\cdot A$

Bob's state transitions described by $APTC_G$ are as follows.

$B=\sum_{D\in\Delta}r_{C_{BI}}(D)\cdot B_2$

$B_2=rsg_R\cdot B_3$

$B_3=s_{C_{BA}}(R)\cdot B_4$

$B_4=r_{C_{AB}}(ENC_k(R,b))\cdot B_5$

$B_5=s_{C_{BA}}(r)\cdot B_6$

$B_6=r_{C_{AB}}(k)\cdot B_7$

$B_7=dec_k(ENC_k(d_R,b))\cdot B_8$

$B_8=\{d_R=R\}\cdot B_9+\{d_R\neq R\}\cdot B_{10}$

$B_9=s_{C_{BO}}(b)\cdot B$

$B_{10}=s_{C_{BO}}(\bot)\cdot B$

The sending action and the reading action of the same type data through the same channel can communicate with each other, otherwise, will cause a deadlock $\delta$. We define the following
communication functions.

$\gamma(r_{C_{BA}}(R),s_{C_{BA}}(R))\triangleq c_{C_{BA}}(R)$

$\gamma(r_{C_{AB}}(ENC_k(R,b)),s_{C_{AB}}(ENC_k(R,b)))\triangleq c_{C_{AB}}(ENC_k(R,b))$

$\gamma(r_{C_{BA}}(r),s_{C_{BA}}(r))\triangleq c_{C_{BA}}(r)$

$\gamma(r_{C_{AB}}(k),s_{C_{AB}}(k))\triangleq c_{C_{AB}}(k)$

Let all modules be in parallel, then the protocol $A\quad B$ can be presented by the following process term.

$$\tau_I(\partial_H(\Theta(A\between B)))=\tau_I(\partial_H(A\between B))$$

where $H=\{r_{C_{BA}}(R),s_{C_{BA}}(R),r_{C_{AB}}(ENC_k(R,b)),s_{C_{AB}}(ENC_k(R,b)),\\
r_{C_{BA}}(r),s_{C_{BA}}(r),r_{C_{AB}}(k),s_{C_{AB}}(k)|D\in\Delta\}$,

$I=\{c_{C_{BA}}(R),c_{C_{AB}}(ENC_k(R,b)),c_{C_{BA}}(r),c_{C_{AB}}(k),\\
rsg_b,rsg_k,rsg_R,dec_k(ENC_k(d_R,b)),\{d_R=R\},\{d_R\neq R\}|D\in\Delta\}$.

Then we get the following conclusion on the protocol.

\begin{theorem}
The bit commitment protocol 1 in Figure \ref{BCP1} is secure.
\end{theorem}

\begin{proof}
Based on the above state transitions of the above modules, by use of the algebraic laws of $APTC_G$, we can prove that

$\tau_I(\partial_H(A\between B))=\sum_{D\in\Delta}(r_{C_{BI}}(D)\cdot (s_{C_{BO}}(b)+s_{C_{BO}}(\bot)))\cdot
\tau_I(\partial_H(A\between B))$.

For the details of proof, please refer to section \ref{app}, and we omit it.

That is, the protocol in Figure \ref{BCP1} $\tau_I(\partial_H(A\between B))$ can exhibit desired external behaviors, that is, if the bits are committed, the system would be
$\tau_I(\partial_H(A\between B))=\sum_{D\in\Delta}(r_{C_{BI}}(D)\cdot s_{C_{BO}}(b))\cdot
\tau_I(\partial_H(A\between B))$; otherwise, the system would be $\tau_I(\partial_H(A\between B))=\sum_{D\in\Delta}(r_{C_{BI}}(D)\cdot s_{C_{BO}}(\bot))\cdot
\tau_I(\partial_H(A\between B))$.

Note that, the main security goals are bit commitment, the the protocol in Figure \ref{BCP1} cannot satisfy other security goals, such as confidentiality.
\end{proof}

\subsubsection{Bit Commitment Protocol 2}\label{bcp2}

The protocol shown in Figure \ref{BCP2} uses one-way function to implement bit commitment.

\begin{figure}
    \centering
    \includegraphics{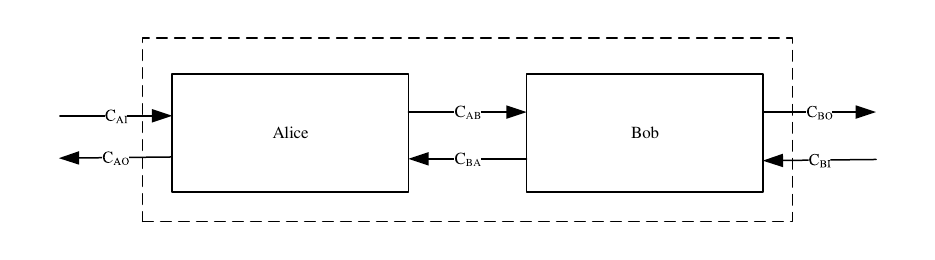}
    \caption{Bit commitment protocol 2}
    \label{BCP2}
\end{figure}

The process of the protocol is as follows.

\begin{enumerate}
  \item Alice receives some requests $D$ from the outside through the channel $C_{AI}$ (the corresponding reading action is denoted $r_{C_{AI}}(D)$), she generates a random sequence $R_1$
  through an action $rsg_{R_1}$, and a random sequence $R_2$ through an action $rsg_{R_2}$, generates the commitment $b$ through an action $rsg_b$, computes the hash of $R_1,R_2,b$ through
  an action $hash(R_1,R_2,b)$, and sends $HASH(R_1,R_2,b),R_1$ to Bob through the channel $C_{AB}$ (the corresponding sending action is denoted $s_{C_{AB}}(HASH(R_1,R_2,b),R_1)$);
  \item Bob receives the message $HASH(R_1,R_2,b),R_1)$ from Alice through the channel $C_{AB}$ (the corresponding reading action is denoted $r_{C_{AB}}(HASH(R_1,R_2,b),R_1)$), after
  some time, he sends a commitment release request $r$ to Alice through the channel $C_{BA}$ (the corresponding sending action is denoted $s_{C_{BA}}(r)$);
  \item Alice receives $r$ from Bob through the channel $C_{BA}$ (the corresponding reading action is denoted $r_{C_{BA}}(r)$), she sends $R_1,R_2,b$ to Bob through the channel $C_{AB}$
  (the corresponding sending action is denoted $s_{C_{AB}}(R_1,R_2,b)$);
  \item Bob receives $d_{R_1},R_2,b$ from Alice through the channel $C_{AB}$ (the corresponding reading action is denoted $r_{C_{AB}}(d_{R_1},R_2,b)$), if $d_{R_1}=R_1$ and
  $HASH(R_1,R_2,b)=HASH(d_{R_1},R_2,b)$, he sends $b$
  to the outside through the channel $C_{BO}$ (the corresponding sending action is denoted $s_{C_{BO}}(b)$); else if $d_{R_1}\neq R_1$ or
  $HASH(R_1,R_2,b)\neq HASH(d_{R_1},R_2,b)$, he sends $\bot$
  to the outside through the channel $C_{BO}$ (the corresponding sending action is denoted $s_{C_{BO}}(\bot)$).
\end{enumerate}

Where $D\in\Delta$, $\Delta$ is the set of data.

Alice's state transitions described by $APTC_G$ are as follows.

$A=\sum_{D\in\Delta}r_{C_{BI}}(D)\cdot A_2$

$A_2=rsg_{R_1}\cdot A_3$

$A_3=rsg_{R_2}\cdot A_4$

$A_4=rsg_b\cdot A_5$

$A_5=hash(R_1,R_2,b)\cdot A_6$

$A_6=s_{C_{AB}}(HASH(R_1,R_2,b),R_1)\cdot A_7$

$A_7=r_{C_{BA}}(r)\cdot A_8$

$A_8=s_{C_{AB}}(R_1,R_2,b)\cdot A$

Bob's state transitions described by $APTC_G$ are as follows.

$B=r_{C_{AB}}(HASH(R_1,R_2,b),R_1)\cdot B_2$

$B_2=s_{C_{BA}}(r)\cdot B_3$

$B_3=r_{C_{AB}}(d_{R_1},R_2,b)\cdot B_4$

$B_4=\{d_{R_1}=R_1\}\cdot \{HASH(R_1,R_2,b)=HASH(d_{R_1},R_2,b)\}\cdot B_5+(\{d_{R_1}\neq R_1\}+ \{HASH(R_1,R_2,b)\neq HASH(d_{R_1},R_2,b)\})\cdot B_{6}$

$B_5=s_{C_{BO}}(b)\cdot B$

$B_{6}=s_{C_{BO}}(\bot)\cdot B$

The sending action and the reading action of the same type data through the same channel can communicate with each other, otherwise, will cause a deadlock $\delta$. We define the following
communication functions.

$\gamma(r_{C_{AB}}(HASH(R_1,R_2,b),R_1),s_{C_{AB}}(HASH(R_1,R_2,b),R_1))\triangleq c_{C_{AB}}(HASH(R_1,R_2,b),R_1)$

$\gamma(r_{C_{BA}}(r),s_{C_{BA}}(r))\triangleq c_{C_{BA}}(r)$

$\gamma(r_{C_{AB}}(d_{R_1},R_2,b),s_{C_{AB}}(d_{R_1},R_2,b))\triangleq c_{C_{AB}}(d_{R_1},R_2,b)$

Let all modules be in parallel, then the protocol $A\quad B$ can be presented by the following process term.

$$\tau_I(\partial_H(\Theta(A\between B)))=\tau_I(\partial_H(A\between B))$$

where $H=\{r_{C_{AB}}(HASH(R_1,R_2,b),R_1),s_{C_{AB}}(HASH(R_1,R_2,b),R_1),r_{C_{BA}}(r),s_{C_{BA}}(r),\\
r_{C_{AB}}(d_{R_1},R_2,b),s_{C_{AB}}(d_{R_1},R_2,b)|D\in\Delta\}$,

$I=\{c_{C_{AB}}(HASH(R_1,R_2,b),R_1),c_{C_{BA}}(r),c_{C_{AB}}(d_{R_1},R_2,b),\\
rsg_{R_1},rsg_{R_2},rsg_{b},hash(R_1,R_2,b),\{d_{R_1}=R_1\}, \{HASH(R_1,R_2,b)=HASH(d_{R_1},R_2,b),\\
\{d_{R_1}\neq R_1\}, \{HASH(R_1,R_2,b)\neq HASH(d_{R_1},R_2,b)\}\}|D\in\Delta\}$.

Then we get the following conclusion on the protocol.

\begin{theorem}
The bit commitment protocol 2 in Figure \ref{BCP2} is secure.
\end{theorem}

\begin{proof}
Based on the above state transitions of the above modules, by use of the algebraic laws of $APTC_G$, we can prove that

$\tau_I(\partial_H(A\between B))=\sum_{D\in\Delta}(r_{C_{AI}}(D)\cdot (s_{C_{BO}}(b)+s_{C_{BO}}(\bot)))\cdot
\tau_I(\partial_H(A\between B))$.

For the details of proof, please refer to section \ref{app}, and we omit it.

That is, the protocol in Figure \ref{BCP2} $\tau_I(\partial_H(A\between B))$ can exhibit desired external behaviors, that is, if the bits are committed, the system would be
$\tau_I(\partial_H(A\between B))=\sum_{D\in\Delta}(r_{C_{BI}}(D)\cdot s_{C_{BO}}(b))\cdot
\tau_I(\partial_H(A\between B))$; otherwise, the system would be $\tau_I(\partial_H(A\between B))=\sum_{D\in\Delta}(r_{C_{BI}}(D)\cdot s_{C_{BO}}(\bot))\cdot
\tau_I(\partial_H(A\between B))$.

Note that, the main security goals are bit commitment, the the protocol in Figure \ref{BCP2} cannot satisfy other security goals, such as confidentiality.
\end{proof}

\subsection{Analyses of Anonymous Key Distribution Protocols}\label{aoakdp}

The protocol shown in Figure \ref{AKDP} uses asymmetric cryptography to implement anonymous key distribution.

\begin{figure}
    \centering
    \includegraphics{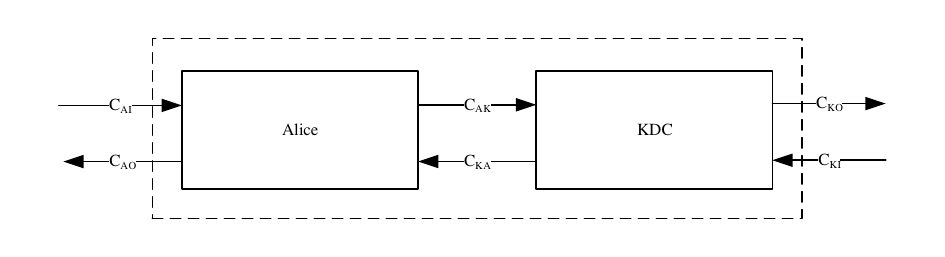}
    \caption{Anonymous key distribution protocol}
    \label{AKDP}
\end{figure}

The process of the protocol is as follows.

\begin{enumerate}
  \item Alice receives some requests $D$ from the outside through the channel $C_{BI}$ (the corresponding reading action is denoted $r_{C_{BI}}(D)$), she generates a public/private key pair
  through an action $rsg_{pk_A,sk_A}$, and sends the key request $r$ to KDC through the channel $C_{AK}$ (the corresponding sending action is denoted $s_{C_{AK}}(r)$);
  \item The KDC receives the key request $r$ from Alice through the channel $C_{AK}$ (the corresponding reading action is denoted $r_{C_{AK}}(r)$), he generates a series of keys $k_i$
  through actions $rsg_{k_i}$ for $1\leq i\leq n$, and encrypts these keys by his public key $pk_K$ through actions $enc_{pk_K}(k_i)$ for $1\leq i\leq n$, then sends these encrypted
  keys to Alice through the channel $C_{KA}$ (the corresponding sending action is denoted $s_{C_{KA}}(ENC_{pk_K}(k_1),\cdots,ENC_{pk_K}(k_n))$);
  \item Alice receives the encrypted keys from the KDC through the channel $C_{KA}$ (the corresponding reading action is denoted $r_{C_{KA}}(ENC_{pk_K}(k_1),\cdots,ENC_{pk_K}(k_n))$),
  she randomly selects one $ENC_{pk_K}(k_j)$ for $1\leq j\leq n$, encrypts it by her public key $pk_A$ through an action $enc_{pk_A}(ENC_{pk_K}(k_j))$, and sends the doubly encrypted
  key $ENC_{pk_A}(ENC_{pk_K}(k_j))$ to the KDC through the channel $C_{AK}$ (the corresponding sending action is denoted $s_{C_{AK}}(ENC_{pk_A}(ENC_{pk_K}(k_j)))$);
  \item The KDC receives the doubly encrypted key from Alice through the channel $C_{AK}$ (the corresponding reading action is denoted $r_{C_{AK}}(ENC_{pk_A}(ENC_{pk_K}(k_j)))$),
  he decrypts it by his private key $sk_K$ through an action $dec_{sk_K}(ENC_{pk_A}(ENC_{pk_K}(k_j)))$ to get $ENC_{pk_A}(k_j)$, and sends $ENC_{pk_A}(k_j)$ to Alice through the channel
  $C_{KA}$ (the corresponding sending action is denoted $s_{C_{KA}}(ENC_{pk_A}(k_j))$);
  \item Alice receives $ENC_{pk_A}(k_j)$ from the KDC through the channel $C_{KA}$ (the corresponding reading action is denoted $r_{C_{KA}}(ENC_{pk_A}(k_j))$), she decrypts it by her private
  key $sk_A$ through an action $dec_{sk_A}(ENC_{pk_A}(k_j))$ to get $k_j$, and sends $k_j$ to the outside through the channel $C_{AO}$ (the corresponding sending action is denoted
  $s_{C_AO}(k_j)$).
\end{enumerate}

Where $D\in\Delta$, $\Delta$ is the set of data.

Alice's state transitions described by $APTC_G$ are as follows.

$A=\sum_{D\in\Delta}r_{C_{AI}}(D)\cdot A_2$

$A_2=rsg_{pk_A,sk_A}\cdot A_3$

$A_3=s_{C_{AK}}(r)\cdot A_4$

$A_4=r_{C_{KA}}(ENC_{pk_K}(k_1),\cdots,ENC_{pk_K}(k_n))\cdot A_5$

$A_5=enc_{pk_A}(ENC_{pk_K}(k_j))\cdot A_6$

$A_6=s_{C_{AK}}(ENC_{pk_A}(ENC_{pk_K}(k_j)))\cdot A_7$

$A_7=r_{C_{KA}}(ENC_{pk_A}(k_j))\cdot A_8$

$A_8=dec_{sk_A}(ENC_{pk_A}(k_j))\cdot A_9$

$A_9=s_{C_AO}(k_j)\cdot A$

The KDC's state transitions described by $APTC_G$ are as follows.

$K=r_{C_{AK}}(r)\cdot K_2$

$K_2=rsg_{k_1}\parallel\cdots\parallel rsg_{k_n}\cdot K_3$

$K_3=enc_{pk_K}(k_1)\parallel\cdots\parallel enc_{pk_K}(k_n)\cdot K_4$

$K_4=s_{C_{KA}}(ENC_{pk_K}(k_1),\cdots,ENC_{pk_K}(k_n))\cdot K_5$

$K_5=r_{C_{AK}}(ENC_{pk_A}(ENC_{pk_K}(k_j)))\cdot K_6$

$K_6=dec_{sk_K}(ENC_{pk_A}(ENC_{pk_K}(k_j)))\cdot K_7$

$K_7=s_{C_{KA}}(ENC_{pk_A}(k_j))\cdot K$

The sending action and the reading action of the same type data through the same channel can communicate with each other, otherwise, will cause a deadlock $\delta$. We define the following
communication functions.

$\gamma(r_{C_{AK}}(r),s_{C_{AK}}(r))\triangleq c_{C_{AK}}(r)$

$\gamma(r_{C_{KA}}(ENC_{pk_K}(k_1),\cdots,ENC_{pk_K}(k_n)),s_{C_{KA}}(ENC_{pk_K}(k_1),\cdots,ENC_{pk_K}(k_n)))\\
\triangleq c_{C_{KA}}(ENC_{pk_K}(k_1),\cdots,ENC_{pk_K}(k_n))$

$\gamma(r_{C_{AK}}(ENC_{pk_A}(ENC_{pk_K}(k_j))),s_{C_{AK}}(ENC_{pk_A}(ENC_{pk_K}(k_j))))\triangleq c_{C_{AK}}(ENC_{pk_A}(ENC_{pk_K}(k_j)))$

$\gamma(r_{C_{KA}}(ENC_{pk_A}(k_j)),s_{C_{KA}}(ENC_{pk_A}(k_j)))\triangleq c_{C_{KA}}(ENC_{pk_A}(k_j))$

Let all modules be in parallel, then the protocol $A\quad K$ can be presented by the following process term.

$$\tau_I(\partial_H(\Theta(A\between K)))=\tau_I(\partial_H(A\between K))$$

where $H=\{r_{C_{AK}}(r),s_{C_{AK}}(r),r_{C_{KA}}(ENC_{pk_A}(k_j)),s_{C_{KA}}(ENC_{pk_A}(k_j)),\\
r_{C_{KA}}(ENC_{pk_K}(k_1),\cdots,ENC_{pk_K}(k_n)),s_{C_{KA}}(ENC_{pk_K}(k_1),\cdots,ENC_{pk_K}(k_n)),\\
r_{C_{AK}}(ENC_{pk_A}(ENC_{pk_K}(k_j))),s_{C_{AK}}(ENC_{pk_A}(ENC_{pk_K}(k_j)))|D\in\Delta\}$,

$I=\{c_{C_{AK}}(r),c_{C_{KA}}(ENC_{pk_A}(k_j)),c_{C_{AK}}(ENC_{pk_A}(ENC_{pk_K}(k_j))),\\
c_{C_{KA}}(ENC_{pk_K}(k_1),\cdots,ENC_{pk_K}(k_n)),rsg_{pk_A,sk_A},enc_{pk_A}(ENC_{pk_K}(k_j)),\\
dec_{sk_A}(ENC_{pk_A}(k_j)),rsg_{k_1},\cdots, rsg_{k_n},enc_{pk_K}(k_1),\cdots, enc_{pk_K}(k_n),\\
dec_{sk_K}(ENC_{pk_A}(ENC_{pk_K}(k_j)))|D\in\Delta\}$.

Then we get the following conclusion on the protocol.

\begin{theorem}
The anonymous key distribution protocol in Figure \ref{AKDP} is secure.
\end{theorem}

\begin{proof}
Based on the above state transitions of the above modules, by use of the algebraic laws of $APTC_G$, we can prove that

$\tau_I(\partial_H(A\between K))=\sum_{D\in\Delta}(r_{C_{AI}}(D)\cdot s_{C_{AO}}(k_j))\cdot
\tau_I(\partial_H(A\between K))$.

For the details of proof, please refer to section \ref{app}, and we omit it.

That is, the protocol in Figure \ref{AKDP} $\tau_I(\partial_H(A\between K))$ can exhibit desired external behaviors, and is secure.
\end{proof}

\newpage\section{Analyses of Digital Cash Protocols}\label{aodcp}

Digital cash makes it possible to use cash digitally. Digital cash maybe have the following six properties:
\begin{enumerate}
  \item Independence. The digital cash is independent on the location, and can be used through the network;
  \item Security. The digital cash cannot be copied and reused;
  \item Privacy. The privacy of the owner of the digital cash is protected;
  \item Off-line payment. The digital cash can be used off line;
  \item Transferability. The digital cash can be transferred to the other users;
  \item Divisibility. The digital cash can be divided into small pieces of digital cashes.
\end{enumerate}

In this chapter, we will introduce four digital cash protocols in the following sections. In the analyses of these four protocols, we will mainly analyze the security and privacy
properties.

\subsection{Digital Cash Protocol 1}\label{dcp1}

The Digital Cash Protocol 1 shown in Figure \ref{DCP1} is the basic digital cash protocol to ensure the anonymity.

\begin{figure}
    \centering
    \includegraphics{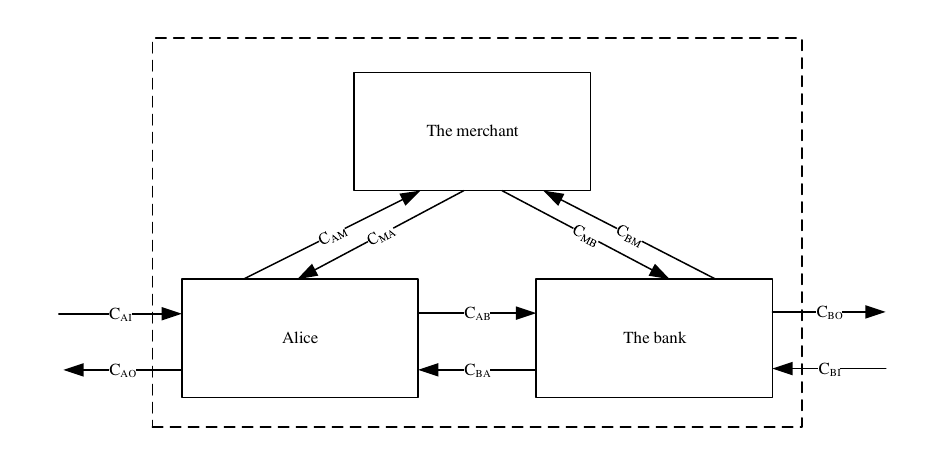}
    \caption{Digital Cash Protocol 1}
    \label{DCP1}
\end{figure}

The process of the protocol is as follows.

\begin{enumerate}
  \item Alice receives some requests $D$ from the outside through the channel $C_{AI}$ (the corresponding reading action is denoted $r_{C_{AI}}(D)$), she generates $n$ \$$m_i$ orders
  with each order encrypted by the bank's public key $pk_B$ through actions $enc_{pk_B}(m_i)$ for $1\leq i\leq n$, and sends them to the bank through the channel $C_{AB}$ (the corresponding
  sending action is denoted $s_{C_{AB}}(ENC_{pk_B}(m_1),\cdots,ENC_{pk_B}(m_n))$);
  \item The bank receives these orders from Alice through the channel $C_{AB}$ (the corresponding reading action is denoted $r_{C_{AB}}(ENC_{pk_B}(m_1),\cdots,ENC_{pk_B}(m_n))$), he
  randomly selects $n-1$ orders and decrypts them through actions $dec_{sk_B}(ENC_{pk_B}(m_j))$ for $1\leq j\leq n-1$ to ensure that each $m_j=m$. Then he sign the left $ENC_{pk_B}(m_k)$
  through an action $sign_{sk_B}(ENC_{pk_B}(m_k))$, checks the identity of Alice and deducts \$$m$ from Alice's account through an action $s_{C_{BO}}(-m)$, then sends \\
  $SIGN_{sk_B}(ENC_{pk_B}(m_k))$ to Alice through the channel $C_{BA}$ (the corresponding
  sending action is denoted $s_{C_{BA}}(SIGN_{sk_B}(ENC_{pk_B}(m_k)))$);
  \item Alice receives the signed order $SIGN_{sk_B}(ENC_{pk_B}(m_k))$ from the bank through the channel $C_{BA}$ (the corresponding reading action is denoted $r_{C_{BA}}(SIGN_{sk_B}(ENC_{pk_B}(m_k)))$),
  she may send the sighed order to some merchant through the channel $C_{AM}$ (the corresponding sending action is denoted $s_{C_{AM}}(SIGN_{sk_B}(ENC_{pk_B}(m_k)))$);
  \item The merchant receives the signed cash from Alice through the channel $C_{AM}$ (the corresponding reading action is denoted $r_{C_{AM}}(SIGN_{sk_B}(ENC_{pk_B}(m_k)))$), he sends it to
  the bank through the channel $C_{MB}$ (the corresponding sending action is denoted \\
  $s_{C_{MB}}(SIGN_{sk_B}(ENC_{pk_B}(m_k)))$);
  \item The bank receives the signed cash from the merchant through the channel $C_{MB}$ (the corresponding reading action is denoted $r_{C_{MB}}(SIGN_{sk_B}(ENC_{pk_B}(m_k)))$), he de-signs
  the cash through an action $de\textrm{-}sign(SIGN_{sk_B}(ENC_{pk_B}(m_k)))$, then decrypts it through an action $dec_{sk_B}(ENC_{pk_B}(m_k))$,  checks the identity of the merchant and
  credits \$$m$ to the merchant's account through an action $s_{C_{BO}}(+m)$.
\end{enumerate}

Where $D\in\Delta$, $\Delta$ is the set of data.

Alice's state transitions described by $APTC_G$ are as follows.

$A=\sum_{D\in\Delta}r_{C_{AI}}(D)\cdot A_2$

$A_2=enc_{pk_B}(m_1)\parallel\cdots\parallel enc_{pk_B}(m_n)\cdot A_3$

$A_3=s_{C_{AB}}(ENC_{pk_B}(m_1),\cdots,ENC_{pk_B}(m_n))\cdot A_4$

$A_4=r_{C_{BA}}(SIGN_{sk_B}(ENC_{pk_B}(m_k)))\cdot A_5$

$A_5=s_{C_{AM}}(SIGN_{sk_B}(ENC_{pk_B}(m_k)))\cdot A$

The bank's state transitions described by $APTC_G$ are as follows.

$B=r_{C_{AB}}(ENC_{pk_B}(m_1),\cdots,ENC_{pk_B}(m_n))\cdot B_2$

$B_2=dec_{sk_B}(ENC_{pk_B}(m_1))\parallel\cdots \parallel dec_{sk_B}(ENC_{pk_B}(m_{n-1}))\cdot B_3$

$B_3=sign_{sk_B}(ENC_{pk_B}(m_k))\cdot B_4$

$B_4=s_{C_{BO}}(-m)\cdot B_5$

$B_5=s_{C_{BA}}(SIGN_{sk_B}(ENC_{pk_B}(m_k)))\cdot B_6$

$B_6=r_{C_{MB}}(SIGN_{sk_B}(ENC_{pk_B}(m_k)))\cdot B_7$

$B_7=de\textrm{-}sign(SIGN_{sk_B}(ENC_{pk_B}(m_k)))\cdot B_8$

$B_8=s_{C_{BO}}(+m)\cdot B$

The merchant's state transitions described by $APTC_G$ are as follows.

$M=r_{C_{AM}}(SIGN_{sk_B}(ENC_{pk_B}(m_k)))\cdot M_2$

$M_2=s_{C_{MB}}(SIGN_{sk_B}(ENC_{pk_B}(m_k)))\cdot M$

The sending action and the reading action of the same type data through the same channel can communicate with each other, otherwise, will cause a deadlock $\delta$. We define the following
communication functions.

$\gamma(r_{C_{AB}}(ENC_{pk_B}(m_1),\cdots,ENC_{pk_B}(m_n)),s_{C_{AB}}(ENC_{pk_B}(m_1),\cdots,ENC_{pk_B}(m_n)))\\
\triangleq c_{C_{AB}}(ENC_{pk_B}(m_1),\cdots,ENC_{pk_B}(m_n))$

$\gamma(r_{C_{BA}}(SIGN_{sk_B}(ENC_{pk_B}(m_k))),s_{C_{BA}}(SIGN_{sk_B}(ENC_{pk_B}(m_k))))\\
\triangleq c_{C_{BA}}(SIGN_{sk_B}(ENC_{pk_B}(m_k)))$

$\gamma(r_{C_{MB}}(SIGN_{sk_B}(ENC_{pk_B}(m_k))),s_{C_{MB}}(SIGN_{sk_B}(ENC_{pk_B}(m_k))))\\
\triangleq c_{C_{MB}}(SIGN_{sk_B}(ENC_{pk_B}(m_k)))$

$\gamma(r_{C_{AM}}(SIGN_{sk_B}(ENC_{pk_B}(m_k))),s_{C_{AM}}(SIGN_{sk_B}(ENC_{pk_B}(m_k))))\\
\triangleq c_{C_{AM}}(SIGN_{sk_B}(ENC_{pk_B}(m_k)))$

Let all modules be in parallel, then the protocol $A\quad B\quad M$ can be presented by the following process term.

$$\tau_I(\partial_H(\Theta(A\between B\between M)))=\tau_I(\partial_H(A\between B\between M))$$

where $H=\{r_{C_{AB}}(ENC_{pk_B}(m_1),\cdots,ENC_{pk_B}(m_n)),s_{C_{AB}}(ENC_{pk_B}(m_1),\cdots,ENC_{pk_B}(m_n)),\\
r_{C_{BA}}(SIGN_{sk_B}(ENC_{pk_B}(m_k))),s_{C_{BA}}(SIGN_{sk_B}(ENC_{pk_B}(m_k))),\\
r_{C_{MB}}(SIGN_{sk_B}(ENC_{pk_B}(m_k))),s_{C_{MB}}(SIGN_{sk_B}(ENC_{pk_B}(m_k))),\\
r_{C_{AM}}(SIGN_{sk_B}(ENC_{pk_B}(m_k))),s_{C_{AM}}(SIGN_{sk_B}(ENC_{pk_B}(m_k)))|D\in\Delta\}$,

$I=\{c_{C_{AB}}(ENC_{pk_B}(m_1),\cdots,ENC_{pk_B}(m_n)),c_{C_{BA}}(SIGN_{sk_B}(ENC_{pk_B}(m_k))),\\
c_{C_{MB}}(SIGN_{sk_B}(ENC_{pk_B}(m_k))),c_{C_{AM}}(SIGN_{sk_B}(ENC_{pk_B}(m_k))),\\
enc_{pk_B}(m_1),\cdots,enc_{pk_B}(m_n),dec_{sk_B}(ENC_{pk_B}(m_1)),\cdots, dec_{sk_B}(ENC_{pk_B}(m_{n-1})),\\
sign_{sk_B}(ENC_{pk_B}(m_k)),de\textrm{-}sign(SIGN_{sk_B}(ENC_{pk_B}(m_k)))|D\in\Delta\}$.

Then we get the following conclusion on the protocol.

\begin{theorem}
The Digital Cash Protocol 1 in Figure \ref{DCP1} is anonymous.
\end{theorem}

\begin{proof}
Based on the above state transitions of the above modules, by use of the algebraic laws of $APTC_G$, we can prove that

$\tau_I(\partial_H(A\between B\between M))=\sum_{D\in\Delta}(r_{C_{AI}}(D)\cdot s_{C_{BO}}(-m)\cdot s_{C_{BO}}(+m))\cdot
\tau_I(\partial_H(A\between B\between M))$.

For the details of proof, please refer to section \ref{app}, and we omit it.

That is, the Digital Cash Protocol 1 in Figure \ref{DCP1} $\tau_I(\partial_H(A\between B\between M))$ can exhibit desired external behaviors:

\begin{enumerate}
  \item The digital cash of Alice $SIGN_{sk_B}(ENC_{pk_B}(m_k))$ is anonymous for the merchant and the bank;
  \item The protocol cannot resist replay attack, for digital cash, this is the so-called double spending problem, either for Alice or the merchant. The system would be
  $\tau_I(\partial_H(A\between B\between M))=\sum_{D\in\Delta}(r_{C_{AI}}(D)\cdot s_{C_{BO}}(-m)\cdot s_{C_{BO}}(+m)\cdot s_{C_{BO}}(+m))\cdot
\tau_I(\partial_H(A\between B\between M))$.
\end{enumerate}
\end{proof}

\subsection{Digital Cash Protocol 2}\label{dcp2}

The Digital Cash Protocol 2 shown in Figure \ref{DCP2} is the basic digital cash protocol to ensure the anonymity and resist replay attacks.

\begin{figure}
    \centering
    \includegraphics{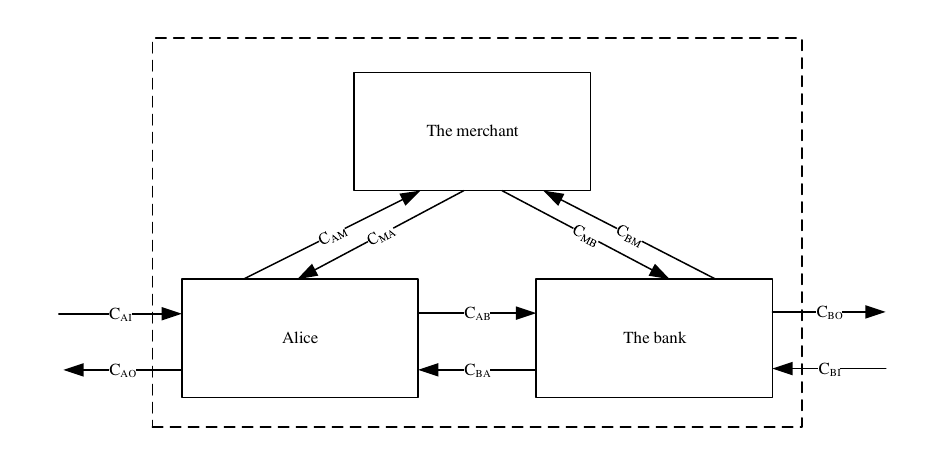}
    \caption{Digital Cash Protocol 2}
    \label{DCP2}
\end{figure}

The process of the protocol is as follows.

\begin{enumerate}
  \item Alice receives some requests $D$ from the outside through the channel $C_{AI}$ (the corresponding reading action is denoted $r_{C_{AI}}(D)$), she generates $n$ \$$m_i$ orders containing
  a random number $R_i$
  with each order encrypted by the bank's public key $pk_B$ through actions $enc_{pk_B}(m_i,R_i)$ for $1\leq i\leq n$, and sends them to the bank through the channel $C_{AB}$ (the corresponding
  sending action is denoted $s_{C_{AB}}(ENC_{pk_B}(m_1,R_1),\cdots,ENC_{pk_B}(m_n,R_n))$);
  \item The bank receives these orders from Alice through the channel $C_{AB}$ (the corresponding reading action is denoted $r_{C_{AB}}(ENC_{pk_B}(m_1,R_1),\cdots,ENC_{pk_B}(m_n,R_n))$), he
  randomly selects $n-1$ orders and decrypts them through actions $dec_{sk_B}(ENC_{pk_B}(m_j,R_j))$ for $1\leq j\leq n-1$ to ensure that each $m_j=m$ and $R_j$ is fresh. Then he sign the left $ENC_{pk_B}(m_k,R_k)$
  through an action $sign_{sk_B}(ENC_{pk_B}(m_k,R_k))$, checks the identity of Alice and deducts \$$m$ from Alice's account through an action $s_{C_{BO}}(-m)$, then sends \\
  $SIGN_{sk_B}(ENC_{pk_B}(m_k,R_k))$ to Alice through the channel $C_{BA}$ (the corresponding
  sending action is denoted $s_{C_{BA}}(SIGN_{sk_B}(ENC_{pk_B}(m_k,R_k)))$);
  \item Alice receives the signed order $SIGN_{sk_B}(ENC_{pk_B}(m_k,R_k))$ from the bank through the channel $C_{BA}$ (the corresponding reading action is denoted $r_{C_{BA}}(SIGN_{sk_B}(ENC_{pk_B}(m_k,R_k)))$),
  she may send the sighed order to some merchant through the channel $C_{AM}$ (the corresponding sending action is denoted $s_{C_{AM}}(SIGN_{sk_B}(ENC_{pk_B}(m_k,R_k)))$);
  \item The merchant receives the signed cash from Alice through the channel $C_{AM}$ (the corresponding reading action is denoted $r_{C_{AM}}(SIGN_{sk_B}(ENC_{pk_B}(m_k,R_k)))$), he sends it to
  the bank through the channel $C_{MB}$ (the corresponding sending action is denoted \\
  $s_{C_{MB}}(SIGN_{sk_B}(ENC_{pk_B}(m_k,R_k)))$);
  \item The bank receives the signed cash from the merchant through the channel $C_{MB}$ (the corresponding reading action is denoted $r_{C_{MB}}(SIGN_{sk_B}(ENC_{pk_B}(m_k,R_k)))$), he de-signs
  the cash through an action $de\textrm{-}sign(SIGN_{sk_B}(ENC_{pk_B}(m_k,R_k)))$, then decrypts it through an action $dec_{sk_B}(ENC_{pk_B}(m_k,R_k))$, if $isFresh(R_k)=TRUE$, he checks the identity of the merchant and
  credits \$$m$ to the merchant's account through an action $s_{C_{BO}}(+m)$; else if $isFresh(R_k)=FALSE$, he sends $\bot$ to the outside through the channel $C_{BO}$ (the corresponding
  sending action is denoted $s_{C_{BO}}(\bot)$).
\end{enumerate}

Where $D\in\Delta$, $\Delta$ is the set of data.

Alice's state transitions described by $APTC_G$ are as follows.

$A=\sum_{D\in\Delta}r_{C_{AI}}(D)\cdot A_2$

$A_2=enc_{pk_B}(m_1,R_1)\parallel\cdots\parallel enc_{pk_B}(m_n,R_n)\cdot A_3$

$A_3=s_{C_{AB}}(ENC_{pk_B}(m_1,R_1),\cdots,ENC_{pk_B}(m_n,R_n))\cdot A_4$

$A_4=r_{C_{BA}}(SIGN_{sk_B}(ENC_{pk_B}(m_k,R_k)))\cdot A_5$

$A_5=s_{C_{AM}}(SIGN_{sk_B}(ENC_{pk_B}(m_k,R_k)))\cdot A$

The bank's state transitions described by $APTC_G$ are as follows.

$B=r_{C_{AB}}(ENC_{pk_B}(m_1,R_1),\cdots,ENC_{pk_B}(m_n,R_n))\cdot B_2$

$B_2=dec_{sk_B}(ENC_{pk_B}(m_1,R_1))\parallel\cdots \parallel dec_{sk_B}(ENC_{pk_B}(m_{n-1},R_{n-1}))\cdot B_3$

$B_3=sign_{sk_B}(ENC_{pk_B}(m_k,R_k))\cdot B_4$

$B_4=s_{C_{BO}}(-m)\cdot B_5$

$B_5=s_{C_{BA}}(SIGN_{sk_B}(ENC_{pk_B}(m_k,R_k)))\cdot B_6$

$B_6=r_{C_{MB}}(SIGN_{sk_B}(ENC_{pk_B}(m_k,R_k)))\cdot B_7$

$B_7=de\textrm{-}sign(SIGN_{sk_B}(ENC_{pk_B}(m_k,R_k)))\cdot B_8$

$B_8=\{isFresh(R_k)=TRUE\}\cdot s_{C_{BO}}(+m)\cdot B+\{isFresh(R_k)=FALSE\}\cdot s_{C_{BO}}(\bot)\cdot B$

The merchant's state transitions described by $APTC_G$ are as follows.

$M=r_{C_{AM}}(SIGN_{sk_B}(ENC_{pk_B}(m_k,R_k)))\cdot M_2$

$M_2=s_{C_{MB}}(SIGN_{sk_B}(ENC_{pk_B}(m_k,R_k)))\cdot M$

The sending action and the reading action of the same type data through the same channel can communicate with each other, otherwise, will cause a deadlock $\delta$. We define the following
communication functions.

$\gamma(r_{C_{AB}}(ENC_{pk_B}(m_1,R_1),\cdots,ENC_{pk_B}(m_n,R_n)),s_{C_{AB}}(ENC_{pk_B}(m_1,R_1),\cdots,ENC_{pk_B}(m_n,R_n)))\\
\triangleq c_{C_{AB}}(ENC_{pk_B}(m_1,R_1),\cdots,ENC_{pk_B}(m_n,R_n))$

$\gamma(r_{C_{BA}}(SIGN_{sk_B}(ENC_{pk_B}(m_k,R_k))),s_{C_{BA}}(SIGN_{sk_B}(ENC_{pk_B}(m_k,R_k))))\\
\triangleq c_{C_{BA}}(SIGN_{sk_B}(ENC_{pk_B}(m_k,R_k)))$

$\gamma(r_{C_{MB}}(SIGN_{sk_B}(ENC_{pk_B}(m_k,R_k))),s_{C_{MB}}(SIGN_{sk_B}(ENC_{pk_B}(m_k,R_k))))\\
\triangleq c_{C_{MB}}(SIGN_{sk_B}(ENC_{pk_B}(m_k,R_k)))$

$\gamma(r_{C_{AM}}(SIGN_{sk_B}(ENC_{pk_B}(m_k,R_k))),s_{C_{AM}}(SIGN_{sk_B}(ENC_{pk_B}(m_k,R_k))))\\
\triangleq c_{C_{AM}}(SIGN_{sk_B}(ENC_{pk_B}(m_k,R_k)))$

Let all modules be in parallel, then the protocol $A\quad B\quad M$ can be presented by the following process term.

$$\tau_I(\partial_H(\Theta(A\between B\between M)))=\tau_I(\partial_H(A\between B\between M))$$

where $H=\{r_{C_{AB}}(ENC_{pk_B}(m_1,R_1),\cdots,ENC_{pk_B}(m_n,R_n)),\\
s_{C_{AB}}(ENC_{pk_B}(m_1,R_1),\cdots,ENC_{pk_B}(m_n,R_n)),\\
r_{C_{BA}}(SIGN_{sk_B}(ENC_{pk_B}(m_k,R_k))),s_{C_{BA}}(SIGN_{sk_B}(ENC_{pk_B}(m_k,R_k))),\\
r_{C_{MB}}(SIGN_{sk_B}(ENC_{pk_B}(m_k,R_k))),s_{C_{MB}}(SIGN_{sk_B}(ENC_{pk_B}(m_k,R_k))),\\
r_{C_{AM}}(SIGN_{sk_B}(ENC_{pk_B}(m_k,R_k))),s_{C_{AM}}(SIGN_{sk_B}(ENC_{pk_B}(m_k,R_k)))|D\in\Delta\}$,

$I=\{c_{C_{AB}}(ENC_{pk_B}(m_1,R_1),\cdots,ENC_{pk_B}(m_n,R_n)),c_{C_{BA}}(SIGN_{sk_B}(ENC_{pk_B}(m_k,R_k))),\\
c_{C_{MB}}(SIGN_{sk_B}(ENC_{pk_B}(m_k,R_k))),c_{C_{AM}}(SIGN_{sk_B}(ENC_{pk_B}(m_k,R_k))),\\
enc_{pk_B}(m_1,R_1),\cdots,enc_{pk_B}(m_n,R_n),dec_{sk_B}(ENC_{pk_B}(m_1,R_1)),\cdots, dec_{sk_B}(ENC_{pk_B}(m_{n-1},R_{n-1})),\\
sign_{sk_B}(ENC_{pk_B}(m_k,R_k)),de\textrm{-}sign(SIGN_{sk_B}(ENC_{pk_B}(m_k,R_k))),\\
\{isFresh(R_k)=TRUE\},\{isFresh(R_k)=FALSE\}|D\in\Delta\}$.

Then we get the following conclusion on the protocol.

\begin{theorem}
The Digital Cash Protocol 2 in Figure \ref{DCP2} is anonymous and resists replaying.
\end{theorem}

\begin{proof}
Based on the above state transitions of the above modules, by use of the algebraic laws of $APTC_G$, we can prove that

$\tau_I(\partial_H(A\between B\between M))=\sum_{D\in\Delta}(r_{C_{AI}}(D)\cdot s_{C_{BO}}(-m)\cdot (s_{C_{BO}}(+m)+s_{C_{BO}}(\bot)))\cdot
\tau_I(\partial_H(A\between B\between M))$.

For the details of proof, please refer to section \ref{app}, and we omit it.

That is, the Digital Cash Protocol 2 in Figure \ref{DCP2} $\tau_I(\partial_H(A\between B\between M))$ can exhibit desired external behaviors:

\begin{enumerate}
  \item The digital cash of Alice $SIGN_{sk_B}(ENC_{pk_B}(m_k))$ is anonymous for the merchant and the bank;
  \item The protocol can resist replay attacks, for the use of the random number in each digital cash;
  \item The bank does not know who cheats him when the double spending problem occurs, either the owner of the cash or the merchant.
\end{enumerate}
\end{proof}

\subsection{Digital Cash Protocol 3}\label{dcp3}

The Digital Cash Protocol 3 shown in Figure \ref{DCP3} is the basic digital cash protocol to ensure the anonymity, resist replay attacks and know who partly.

\begin{figure}
    \centering
    \includegraphics{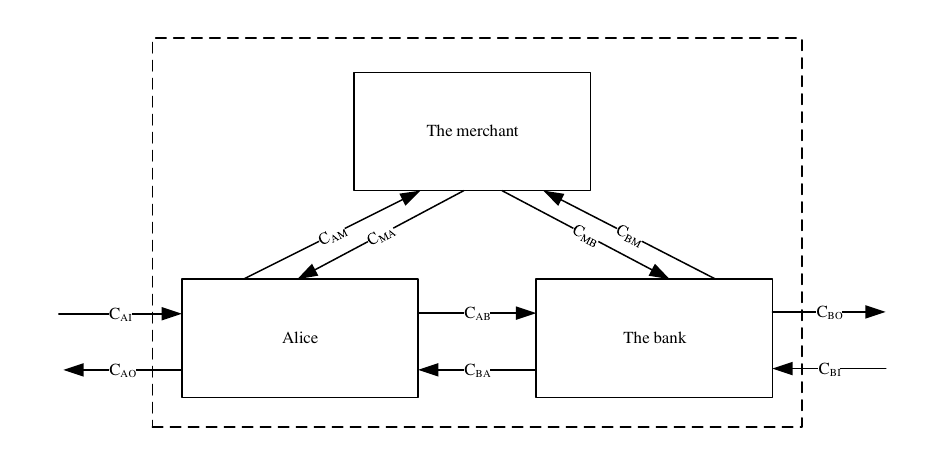}
    \caption{Digital Cash Protocol 3}
    \label{DCP3}
\end{figure}

The process of the protocol is as follows.

\begin{enumerate}
  \item Alice receives some requests $D$ from the outside through the channel $C_{AI}$ (the corresponding reading action is denoted $r_{C_{AI}}(D)$), she generates $n$ \$$m_i$ orders containing
  a random number $R_i$
  with each order encrypted by the bank's public key $pk_B$ through actions $enc_{pk_B}(m_i,R_i)$ for $1\leq i\leq n$, and sends them to the bank through the channel $C_{AB}$ (the corresponding
  sending action is denoted $s_{C_{AB}}(ENC_{pk_B}(m_1,R_1),\cdots,ENC_{pk_B}(m_n,R_n))$);
  \item The bank receives these orders from Alice through the channel $C_{AB}$ (the corresponding reading action is denoted $r_{C_{AB}}(ENC_{pk_B}(m_1,R_1),\cdots,ENC_{pk_B}(m_n,R_n))$), he
  randomly selects $n-1$ orders and decrypts them through actions $dec_{sk_B}(ENC_{pk_B}(m_j,R_j))$ for $1\leq j\leq n-1$ to ensure that each $m_j=m$ and $R_j$ is fresh. Then he sign the left $ENC_{pk_B}(m_k,R_k)$
  through an action $sign_{sk_B}(ENC_{pk_B}(m_k,R_k))$, checks the identity of Alice and deducts \$$m$ from Alice's account through an action $s_{C_{BO}}(-m)$, then sends \\
  $SIGN_{sk_B}(ENC_{pk_B}(m_k,R_k))$ to Alice through the channel $C_{BA}$ (the corresponding
  sending action is denoted $s_{C_{BA}}(SIGN_{sk_B}(ENC_{pk_B}(m_k,R_k)))$);
  \item Alice receives the signed order $SIGN_{sk_B}(ENC_{pk_B}(m_k,R_k))$ from the bank through the channel $C_{BA}$ (the corresponding reading action is denoted $r_{C_{BA}}(SIGN_{sk_B}(ENC_{pk_B}(m_k,R_k)))$),
  she generates a random string $R$ and encrypts it through an action $enc_{pk_B}(R)$,
  she may send the sighed order to some merchant through the channel $C_{AM}$ (the corresponding sending action is denoted $s_{C_{AM}}(SIGN_{sk_B}(ENC_{pk_B}(m_k,R_k)),ENC_{pk_B}(R))$);
  \item The merchant receives the signed cash from Alice through the channel $C_{AM}$ (the corresponding reading action is denoted $r_{C_{AM}}(SIGN_{sk_B}(ENC_{pk_B}(m_k,R_k)),ENC_{pk_B}(R))$), he sends it to
  the bank through the channel $C_{MB}$ (the corresponding sending action is denoted \\
  $s_{C_{MB}}(SIGN_{sk_B}(ENC_{pk_B}(m_k,R_k)))$);
  \item The bank receives the signed cash from the merchant through the channel $C_{MB}$ (the corresponding reading action is denoted $r_{C_{MB}}(SIGN_{sk_B}(ENC_{pk_B}(m_k,R_k)),ENC_{pk_B}(R))$), he de-signs
  the cash through an action $de\textrm{-}sign(SIGN_{sk_B}(ENC_{pk_B}(m_k,R_k)))$, then decrypts it through an action $dec_{sk_B}(ENC_{pk_B}(m_k,R_k))$ and $dec_{sk_B}(ENC_{pk_B}(R))$, if $isFresh(R_k)=TRUE$, he checks the identity of the merchant and
  credits \$$m$ to the merchant's account through an action $s_{C_{BO}}(+m)$; else if $isFresh(R_k)=FALSE$ and $isFresh(R)=TRUE$, he sends $\bot_A$ to the outside through the channel $C_{BO}$ (the corresponding
  sending action is denoted $s_{C_{BO}}(\bot_A)$); else if $isFresh(R_k)=FALSE$ and $isFresh(R)=FALSE$, he sends $\bot_M$ to the outside through the channel $C_{BO}$ (the corresponding
  sending action is denoted $s_{C_{BO}}(\bot_M)$).
\end{enumerate}

Where $D\in\Delta$, $\Delta$ is the set of data.

Alice's state transitions described by $APTC_G$ are as follows.

$A=\sum_{D\in\Delta}r_{C_{AI}}(D)\cdot A_2$

$A_2=enc_{pk_B}(m_1,R_1)\parallel\cdots\parallel enc_{pk_B}(m_n,R_n)\cdot A_3$

$A_3=s_{C_{AB}}(ENC_{pk_B}(m_1,R_1),\cdots,ENC_{pk_B}(m_n,R_n))\cdot A_4$

$A_4=r_{C_{BA}}(SIGN_{sk_B}(ENC_{pk_B}(m_k,R_k)))\cdot A_5$

$A_5=enc_{pk_B}(R)\cdot A_6$

$A_6=s_{C_{AM}}(SIGN_{sk_B}(ENC_{pk_B}(m_k,R_k)),ENC_{pk_B}(R))\cdot A$

The bank's state transitions described by $APTC_G$ are as follows.

$B=r_{C_{AB}}(ENC_{pk_B}(m_1,R_1),\cdots,ENC_{pk_B}(m_n,R_n))\cdot B_2$

$B_2=dec_{sk_B}(ENC_{pk_B}(m_1,R_1))\parallel\cdots \parallel dec_{sk_B}(ENC_{pk_B}(m_{n-1},R_{n-1}))\cdot B_3$

$B_3=sign_{sk_B}(ENC_{pk_B}(m_k,R_k))\cdot B_4$

$B_4=s_{C_{BO}}(-m)\cdot B_5$

$B_5=s_{C_{BA}}(SIGN_{sk_B}(ENC_{pk_B}(m_k,R_k)))\cdot B_6$

$B_6=r_{C_{MB}}(SIGN_{sk_B}(ENC_{pk_B}(m_k,R_k)),ENC_{pk_B}(R))\cdot B_7$

$B_7=de\textrm{-}sign(SIGN_{sk_B}(ENC_{pk_B}(m_k,R_k)))\cdot B_8$

$B_8=dec_{sk_B}(ENC_{pk_B}(R))\cdot B_9$

$B_9=\{isFresh(R_k)=TRUE\}\cdot s_{C_{BO}}(+m)\cdot B+\{isFresh(R_k)=FALSE\}\cdot\{isFresh(R)=TRUE\}\cdot s_{C_{BO}}(\bot_A)\cdot B\\
+\{isFresh(R_k)=FALSE\}\cdot\{isFresh(R)=FALSE\}\cdot s_{C_{BO}}(\bot_M)\cdot B$

The merchant's state transitions described by $APTC_G$ are as follows.

$M=r_{C_{AM}}(SIGN_{sk_B}(ENC_{pk_B}(m_k,R_k)),ENC_{pk_B}(R))\cdot M_2$

$M_2=s_{C_{MB}}(SIGN_{sk_B}(ENC_{pk_B}(m_k,R_k)),ENC_{pk_B}(R))\cdot M$

The sending action and the reading action of the same type data through the same channel can communicate with each other, otherwise, will cause a deadlock $\delta$. We define the following
communication functions.

$\gamma(r_{C_{AB}}(ENC_{pk_B}(m_1,R_1),\cdots,ENC_{pk_B}(m_n,R_n)),s_{C_{AB}}(ENC_{pk_B}(m_1,R_1),\cdots,ENC_{pk_B}(m_n,R_n)))\\
\triangleq c_{C_{AB}}(ENC_{pk_B}(m_1,R_1),\cdots,ENC_{pk_B}(m_n,R_n))$

$\gamma(r_{C_{BA}}(SIGN_{sk_B}(ENC_{pk_B}(m_k,R_k))),s_{C_{BA}}(SIGN_{sk_B}(ENC_{pk_B}(m_k,R_k))))\\
\triangleq c_{C_{BA}}(SIGN_{sk_B}(ENC_{pk_B}(m_k,R_k)))$

$\gamma(r_{C_{MB}}(SIGN_{sk_B}(ENC_{pk_B}(m_k,R_k)),ENC_{pk_B}(R)),\\
s_{C_{MB}}(SIGN_{sk_B}(ENC_{pk_B}(m_k,R_k)),ENC_{pk_B}(R)))\\
\triangleq c_{C_{MB}}(SIGN_{sk_B}(ENC_{pk_B}(m_k,R_k)),ENC_{pk_B}(R))$

$\gamma(r_{C_{AM}}(SIGN_{sk_B}(ENC_{pk_B}(m_k,R_k)),ENC_{pk_B}(R)),\\
s_{C_{AM}}(SIGN_{sk_B}(ENC_{pk_B}(m_k,R_k)),ENC_{pk_B}(R)))\\
\triangleq c_{C_{AM}}(SIGN_{sk_B}(ENC_{pk_B}(m_k,R_k)),ENC_{pk_B}(R))$

Let all modules be in parallel, then the protocol $A\quad B\quad M$ can be presented by the following process term.

$$\tau_I(\partial_H(\Theta(A\between B\between M)))=\tau_I(\partial_H(A\between B\between M))$$

where $H=\{r_{C_{AB}}(ENC_{pk_B}(m_1,R_1),\cdots,ENC_{pk_B}(m_n,R_n)),\\
s_{C_{AB}}(ENC_{pk_B}(m_1,R_1),\cdots,ENC_{pk_B}(m_n,R_n)),\\
r_{C_{BA}}(SIGN_{sk_B}(ENC_{pk_B}(m_k,R_k))),s_{C_{BA}}(SIGN_{sk_B}(ENC_{pk_B}(m_k,R_k))),\\
r_{C_{MB}}(SIGN_{sk_B}(ENC_{pk_B}(m_k,R_k)),ENC_{pk_B}(R)),\\
s_{C_{MB}}(SIGN_{sk_B}(ENC_{pk_B}(m_k,R_k)),ENC_{pk_B}(R)),\\
r_{C_{AM}}(SIGN_{sk_B}(ENC_{pk_B}(m_k,R_k)),ENC_{pk_B}(R)),\\
s_{C_{AM}}(SIGN_{sk_B}(ENC_{pk_B}(m_k,R_k)),ENC_{pk_B}(R))|D\in\Delta\}$,

$I=\{c_{C_{AB}}(ENC_{pk_B}(m_1,R_1),\cdots,ENC_{pk_B}(m_n,R_n)),c_{C_{BA}}(SIGN_{sk_B}(ENC_{pk_B}(m_k,R_k))),\\
c_{C_{MB}}(SIGN_{sk_B}(ENC_{pk_B}(m_k,R_k)),ENC_{pk_B}(R)),c_{C_{AM}}(SIGN_{sk_B}(ENC_{pk_B}(m_k,R_k)),ENC_{pk_B}(R)),\\
enc_{pk_B}(m_1,R_1),\cdots,enc_{pk_B}(m_n,R_n),dec_{sk_B}(ENC_{pk_B}(m_1,R_1)),\cdots, dec_{sk_B}(ENC_{pk_B}(m_{n-1},R_{n-1})),\\
sign_{sk_B}(ENC_{pk_B}(m_k,R_k)),de\textrm{-}sign(SIGN_{sk_B}(ENC_{pk_B}(m_k,R_k))),enc_{pk_B}(R),dec_{sk_B}(ENC_{pk_B}(R)),\\
\{isFresh(R_k)=TRUE\},\{isFresh(R_k)=FALSE\},\{isFresh(R)=TRUE\},\{isFresh(R)=FALSE\}|D\in\Delta\}$.

Then we get the following conclusion on the protocol.

\begin{theorem}
The Digital Cash Protocol 3 in Figure \ref{DCP3} is anonymous, and resists replaying and knowing who partly.
\end{theorem}

\begin{proof}
Based on the above state transitions of the above modules, by use of the algebraic laws of $APTC_G$, we can prove that

$\tau_I(\partial_H(A\between B\between M))=\sum_{D\in\Delta}(r_{C_{AI}}(D)\cdot s_{C_{BO}}(-m)\cdot (s_{C_{BO}}(+m)+s_{C_{BO}}(\bot_A)+s_{C_{BO}}(\bot_M)))\cdot
\tau_I(\partial_H(A\between B\between M))$.

For the details of proof, please refer to section \ref{app}, and we omit it.

That is, the Digital Cash Protocol 3 in Figure \ref{DCP3} $\tau_I(\partial_H(A\between B\between M))$ can exhibit desired external behaviors:

\begin{enumerate}
  \item The digital cash of Alice $SIGN_{sk_B}(ENC_{pk_B}(m_k))$ is anonymous for the merchant and the bank;
  \item The protocol can resist replay attack, for the use of the random number in each digital cash;
  \item The bank know who cheats him when the double spending problem occurs, either the owner of the cash or the merchant. But he does not know exactly the identity of the person.
\end{enumerate}
\end{proof}

\subsection{Digital Cash Protocol 4}\label{dcp4}

The Digital Cash Protocol 4 shown in Figure \ref{DCP4} is the basic digital cash protocol to ensure the anonymity, resist replay attacks and know who exactly.

\begin{figure}
    \centering
    \includegraphics{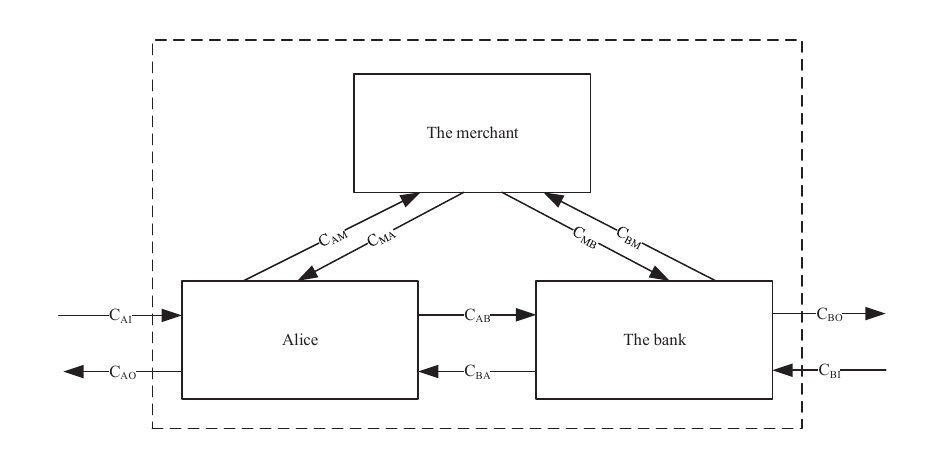}
    \caption{Digital Cash Protocol 4}
    \label{DCP4}
\end{figure}

The process of the protocol is as follows.

\begin{enumerate}
  \item Alice receives some requests $D$ from the outside through the channel $C_{AI}$ (the corresponding reading action is denoted $r_{C_{AI}}(D)$), she generates $n$ \$$m_i$ orders containing
  $m$, a random number $R_i$, and $n$ pair of string $I_{i1L},I_{i1R},\cdots,I_{inL},I_{inR}$,
  with each order blinded through actions $blind_{k_1}(m_i,R_i,I_{i1L},I_{i1R},\cdots,I_{inL},I_{inR})$ for $1\leq i\leq n$, and sends them to the bank through the channel $C_{AB}$ (the corresponding
  sending action is denoted $s_{C_{AB}}(BLIND_{k_1}(m_1,R_1,I_{1L},I_{1R}),\cdots,BLIND_{k_1}(m_n,R_n,I_{nL},I_{nR}))$);
  \item The bank receives these orders from Alice through the channel $C_{AB}$ (the corresponding reading action is denoted \\
  $r_{C_{AB}}(BLIND_{k_1}(m_1,R_1,I_{11L},I_{11R},\cdots,I_{1nL},I_{1nR}),\cdots,BLIND_{k_1}(m_n,R_n,I_{1nL},I_{1nR},\cdots,I_{nnL},I_{nnR}))$), he
  randomly selects $n-1$ orders and asks Alice to unblind them through actions \\
  $unblind_{k_1}(BLIND_{k_1}(m_j,R_j,I_{j1L},I_{j1R},\cdots,I_{jnL},I_{jnR}))$ and to reveal $I_{j1L},I_{j1R},\cdots,I_{jnL},I_{jnR}$ (see section \ref{aobcp1})
  for $1\leq j\leq n-1$ to ensure that each $m_j=m$ and $R_j$ is fresh. Then he sign the left \\
  $BLIND_{k_1}(m_k,R_k,I_{k1L},I_{k1R},\cdots,I_{knL},I_{knR})$
  through an action \\
  $sign_{sk_B}(BLIND_{k_1}(m_k,R_k,I_{k1L},I_{k1R},\cdots,I_{knL},I_{knR}))$, checks the identity of Alice and deducts \$$m$ from Alice's account through an action $s_{C_{BO}}(-m)$, then sends\\
  $SIGN_{sk_B}(BLIND_{k_1}(m_k,R_k,I_{k1L},I_{k1R},\cdots,I_{knL},I_{knR}))$ to Alice through the channel $C_{BA}$ (the corresponding
  sending action is denoted \\
  $s_{C_{BA}}(SIGN_{sk_B}(BLIND_{k_1}(m_k,R_k,I_{k1L},I_{k1R},\cdots,I_{knL},I_{knR})))$);
  \item Alice receives the signed order $SIGN_{sk_B}(BLIND_{k_1}(m_k,R_k,I_{k1L},I_{k1R},\cdots,I_{knL},I_{knR}))$ from the bank through the channel $C_{BA}$ (the corresponding
  reading action is denoted \\
  $r_{C_{BA}}(SIGN_{sk_B}(BLIND_{k_1}(m_k,R_k,I_{k1L},I_{k1R},\cdots,I_{knL},I_{knR})))$),
  she unblinds the sighed order through an action $unblind_{k_1}(SIGN_{sk_B}(BLIND_{k_1}(m_k,R_k,I_{k1L},I_{k1R},\cdots,I_{knL},I_{knR})))$ to get
  $SIGN_{sk_B}(m_k,R_k,I_{k1L},I_{k1R},\cdots,I_{knL},I_{knR})$
  she may send the sighed order to some merchant through the channel $C_{AM}$ (the corresponding sending action is denoted \\
  $s_{C_{AM}}(SIGN_{sk_B}(m_k,R_k,I_{k1L},I_{k1R},\cdots,I_{knL},I_{knR}))$);
  \item The merchant receives the signed cash from Alice through the channel $C_{AM}$ (the corresponding reading action is denoted
  $r_{C_{AM}}(SIGN_{sk_B}(m_k,R_k,I_{k1L},I_{k1R},\cdots,I_{knL},I_{knR}))$), he ask Alice to reveal half of $I_{j1L},I_{j1R},\cdots,I_{jnL},I_{jnR}$ (see section \ref{aobcp1}), and sends it to
  the bank through the channel $C_{MB}$ (the corresponding sending action is denoted \\
  $s_{C_{MB}}(SIGN_{sk_B}(m_k,R_k,I_{k1L},I_{k1R},\cdots,I_{knL},I_{knR}),I'_{k1L},I'_{k1R},\cdots,I'_{knL},I'_{knR})$);
  \item The bank receives the signed cash from the merchant through the channel $C_{MB}$ (the corresponding reading action is denoted \\
  $r_{C_{MB}}(SIGN_{sk_B}(m_k,R_k,I_{k1L},I_{k1R},\cdots,I_{knL},I_{knR}),I'_{k1L},I'_{k1R},\cdots,I'_{knL},I'_{knR})$), he de-signs
  the cash through an action $de\textrm{-}sign(SIGN_{sk_B}(m_k,R_k,I_{k1L},I_{k1R},\cdots,I_{knL},I_{knR}))$, if $isFresh(R_k)=TRUE$, he checks the identity of the merchant and
  credits \$$m$ to the merchant's account through an action $s_{C_{BO}}(+m)$; else if $isFresh(R_k)=FALSE$ and $isFresh(I'_{k1L},I'_{k1R},\cdots,I'_{knL},I'_{knR})=TRUE$,
  he gets the identity of Alice and sends $\bot_A$ to the outside through the channel $C_{BO}$ (the corresponding
  sending action is denoted $s_{C_{BO}}(\bot_A)$); else if $isFresh(R_k)=FALSE$ and $isFresh(I'_{k1L},I'_{k1R},\cdots,I'_{knL},I'_{knR})=FALSE$,
  he sends $\bot_M$ to the outside through the channel $C_{BO}$ (the corresponding sending action is denoted $s_{C_{BO}}(\bot_M)$).
\end{enumerate}

Where $D\in\Delta$, $\Delta$ is the set of data.

Alice's state transitions described by $APTC_G$ are as follows.

$A=\sum_{D\in\Delta}r_{C_{AI}}(D)\cdot A_2$

$A_2=blind_{k_1}(m_1,R_i,I_{11L},I_{11R},\cdots,I_{1nL},I_{1nR})\parallel\cdots\parallel blind_{k_1}(m_n,R_n,I_{n1L},I_{n1R},\cdots,I_{nnL},I_{nnR})\cdot A_3$

$A_3=s_{C_{AB}}(BLIND_{k_1}(m_1,R_1,I_{1L},I_{1R}),\cdots,BLIND_{k_1}(m_n,R_n,I_{nL},I_{nR}))\cdot A_4$

$A_4=r_{C_{BA}}(SIGN_{sk_B}(BLIND_{k_1}(m_k,R_k,I_{k1L},I_{k1R},\cdots,I_{knL},I_{knR})))\cdot A_5$

$A_5=unblind_{k_1}(SIGN_{sk_B}(BLIND_{k_1}(m_k,R_k,I_{k1L},I_{k1R},\cdots,I_{knL},I_{knR})))\cdot A_6$

$A_6=s_{C_{AM}}(SIGN_{sk_B}(m_k,R_k,I_{k1L},I_{k1R},\cdots,I_{knL},I_{knR}))\cdot A$

The bank's state transitions described by $APTC_G$ are as follows.

$B=r_{C_{AB}}(BLIND_{k_1}(m_1,R_1,I_{11L},I_{11R},\cdots,I_{1nL},I_{1nR}),\cdots,\\
BLIND_{k_1}(m_n,R_n,I_{1nL},I_{1nR},\cdots,I_{nnL},I_{nnR}))\cdot B_2$

$B_2=unblind_{k_1}(BLIND_{k_1}(m_1,R_1,I_{11L},I_{11R},\cdots,I_{1nL},I_{1nR}))\parallel\cdots \\
\parallel unblind_{k_1}(BLIND_{k_1}(m_{n-1},R_{n-1},I_{n-11L},I_{n-11R},\cdots,I_{n-1nL},I_{n-1nR}))\cdot B_3$

$B_3=sign_{sk_B}(BLIND_{k_1}(m_k,R_k,I_{k1L},I_{k1R},\cdots,I_{knL},I_{knR}))\cdot B_4$

$B_4=s_{C_{BO}}(-m)\cdot B_5$

$B_5=s_{C_{BA}}(SIGN_{sk_B}(BLIND_{k_1}(m_k,R_k,I_{k1L},I_{k1R},\cdots,I_{knL},I_{knR})))\cdot B_6$

$B_6=r_{C_{MB}}(SIGN_{sk_B}(m_k,R_k,I_{k1L},I_{k1R},\cdots,I_{knL},I_{knR}),I'_{k1L},I'_{k1R},\cdots,I'_{knL},I'_{knR})\cdot B_7$

$B_7=de\textrm{-}sign(SIGN_{sk_B}(m_k,R_k,I_{k1L},I_{k1R},\cdots,I_{knL},I_{knR}))\cdot B_8$

$B_8=\{isFresh(R_k)=TRUE\}\cdot s_{C_{BO}}(+m)\cdot B\\
+\{isFresh(R_k)=FALSE\}\cdot\{isFresh(I'_{k1L},I'_{k1R},\cdots,I'_{knL},I'_{knR})=TRUE\}\cdot s_{C_{BO}}(\bot_A)\cdot B\\
+\{isFresh(R_k)=FALSE\}\cdot\{isFresh(I'_{k1L},I'_{k1R},\cdots,I'_{knL},I'_{knR})=FALSE\}\cdot s_{C_{BO}}(\bot_M)\cdot B$

The merchant's state transitions described by $APTC_G$ are as follows.

$M=r_{C_{AM}}(SIGN_{sk_B}(m_k,R_k,I_{k1L},I_{k1R},\cdots,I_{knL},I_{knR}))\cdot M_2$

$M_2=s_{C_{MB}}(SIGN_{sk_B}(m_k,R_k,I_{k1L},I_{k1R},\cdots,I_{knL},I_{knR}),I'_{k1L},I'_{k1R},\cdots,I'_{knL},I'_{knR})\cdot M$

The sending action and the reading action of the same type data through the same channel can communicate with each other, otherwise, will cause a deadlock $\delta$. We define the following
communication functions.

$\gamma(r_{C_{AB}}(BLIND_{k_1}(m_1,R_1,I_{11L},I_{11R},\cdots,I_{1nL},I_{1nR}),\cdots,BLIND_{k_1}(m_n,R_n,I_{1nL},I_{1nR},\cdots,I_{nnL},I_{nnR})),\\
s_{C_{AB}}(BLIND_{k_1}(m_1,R_1,I_{11L},I_{11R},\cdots,I_{1nL},I_{1nR}),\cdots,BLIND_{k_1}(m_n,R_n,I_{1nL},I_{1nR},\cdots,I_{nnL},I_{nnR})))\\
\triangleq c_{C_{AB}}(BLIND_{k_1}(m_1,R_1,I_{11L},I_{11R},\cdots,I_{1nL},I_{1nR}),\cdots,BLIND_{k_1}(m_n,R_n,I_{1nL},I_{1nR},\cdots,I_{nnL},I_{nnR}))$

$\gamma(r_{C_{BA}}(SIGN_{sk_B}(BLIND_{k_1}(m_k,R_k,I_{k1L},I_{k1R},\cdots,I_{knL},I_{knR}))),\\
s_{C_{BA}}(SIGN_{sk_B}(BLIND_{k_1}(m_k,R_k,I_{k1L},I_{k1R},\cdots,I_{knL},I_{knR}))))\\
\triangleq c_{C_{BA}}(SIGN_{sk_B}(BLIND_{k_1}(m_k,R_k,I_{k1L},I_{k1R},\cdots,I_{knL},I_{knR})))$

$\gamma(r_{C_{MB}}(SIGN_{sk_B}(m_k,R_k,I_{k1L},I_{k1R},\cdots,I_{knL},I_{knR}),I'_{k1L},I'_{k1R},\cdots,I'_{knL},I'_{knR}),\\
s_{C_{MB}}(SIGN_{sk_B}(m_k,R_k,I_{k1L},I_{k1R},\cdots,I_{knL},I_{knR}),I'_{k1L},I'_{k1R},\cdots,I'_{knL},I'_{knR}))\\
\triangleq c_{C_{MB}}(SIGN_{sk_B}(m_k,R_k,I_{k1L},I_{k1R},\cdots,I_{knL},I_{knR}),I'_{k1L},I'_{k1R},\cdots,I'_{knL},I'_{knR})$

$\gamma(r_{C_{AM}}(SIGN_{sk_B}(m_k,R_k,I_{k1L},I_{k1R},\cdots,I_{knL},I_{knR})),\\
s_{C_{AM}}(SIGN_{sk_B}(m_k,R_k,I_{k1L},I_{k1R},\cdots,I_{knL},I_{knR})))\\
\triangleq c_{C_{AM}}(SIGN_{sk_B}(m_k,R_k,I_{k1L},I_{k1R},\cdots,I_{knL},I_{knR}))$

Let all modules be in parallel, then the protocol $A\quad B\quad M$ can be presented by the following process term.

$$\tau_I(\partial_H(\Theta(A\between B\between M)))=\tau_I(\partial_H(A\between B\between M))$$

where $H=\{r_{C_{AB}}(BLIND_{k_1}(m_1,R_1,I_{11L},I_{11R},\cdots,I_{1nL},I_{1nR}),\cdots,\\
BLIND_{k_1}(m_n,R_n,I_{1nL},I_{1nR},\cdots,I_{nnL},I_{nnR})),\\
s_{C_{AB}}(BLIND_{k_1}(m_1,R_1,I_{11L},I_{11R},\cdots,I_{1nL},I_{1nR}),\cdots,BLIND_{k_1}(m_n,R_n,I_{1nL},I_{1nR},\cdots,I_{nnL},I_{nnR})),\\
r_{C_{BA}}(SIGN_{sk_B}(BLIND_{k_1}(m_k,R_k,I_{k1L},I_{k1R},\cdots,I_{knL},I_{knR}))),\\
s_{C_{BA}}(SIGN_{sk_B}(BLIND_{k_1}(m_k,R_k,I_{k1L},I_{k1R},\cdots,I_{knL},I_{knR}))),\\
r_{C_{MB}}(SIGN_{sk_B}(m_k,R_k,I_{k1L},I_{k1R},\cdots,I_{knL},I_{knR}),I'_{k1L},I'_{k1R},\cdots,I'_{knL},I'_{knR}),\\
s_{C_{MB}}(SIGN_{sk_B}(m_k,R_k,I_{k1L},I_{k1R},\cdots,I_{knL},I_{knR}),I'_{k1L},I'_{k1R},\cdots,I'_{knL},I'_{knR}),\\
r_{C_{AM}}(SIGN_{sk_B}(m_k,R_k,I_{k1L},I_{k1R},\cdots,I_{knL},I_{knR})),\\
s_{C_{AM}}(SIGN_{sk_B}(m_k,R_k,I_{k1L},I_{k1R},\cdots,I_{knL},I_{knR}))|D\in\Delta\}$,

$I=\{c_{C_{AB}}(BLIND_{k_1}(m_1,R_1,I_{11L},I_{11R},\cdots,I_{1nL},I_{1nR}),\cdots,BLIND_{k_1}(m_n,R_n,I_{1nL},I_{1nR},\cdots,I_{nnL},I_{nnR})),\\
c_{C_{BA}}(SIGN_{sk_B}(BLIND_{k_1}(m_k,R_k,I_{k1L},I_{k1R},\cdots,I_{knL},I_{knR}))),\\
c_{C_{MB}}(SIGN_{sk_B}(m_k,R_k,I_{k1L},I_{k1R},\cdots,I_{knL},I_{knR}),I'_{k1L},I'_{k1R},\cdots,I'_{knL},I'_{knR}),\\
c_{C_{AM}}(SIGN_{sk_B}(m_k,R_k,I_{k1L},I_{k1R},\cdots,I_{knL},I_{knR})),\\
blind_{k_1}(m_1,R_i,I_{11L},I_{11R},\cdots,I_{1nL},I_{1nR}),\cdots, blind_{k_1}(m_n,R_n,I_{n1L},I_{n1R},\cdots,I_{nnL},I_{nnR}),\\
unblind_{k_1}(SIGN_{sk_B}(BLIND_{k_1}(m_k,R_k,I_{k1L},I_{k1R},\cdots,I_{knL},I_{knR}))),\\
unblind_{k_1}(BLIND_{k_1}(m_1,R_1,I_{11L},I_{11R},\cdots,I_{1nL},I_{1nR})),\\
\cdots, unblind_{k_1}(BLIND_{k_1}(m_{n-1},R_{n-1},I_{n-11L},I_{n-11R},\cdots,I_{n-1nL},I_{n-1nR})),\\
sign_{sk_B}(BLIND_{k_1}(m_k,R_k,I_{k1L},I_{k1R},\cdots,I_{knL},I_{knR})),\\
de\textrm{-}sign(SIGN_{sk_B}(m_k,R_k,I_{k1L},I_{k1R},\cdots,I_{knL},I_{knR})),\\
\{isFresh(R_k)=TRUE\},\{isFresh(R_k)=FALSE\},\\
\{isFresh(I'_{k1L},I'_{k1R},\cdots,I'_{knL},I'_{knR})=TRUE\},\{isFresh(I'_{k1L},I'_{k1R},\cdots,I'_{knL},I'_{knR})=FALSE\}|D\in\Delta\}$.

Then we get the following conclusion on the protocol.

\begin{theorem}
The Digital Cash Protocol 4 in Figure \ref{DCP4} is anonymous, resists replaying, and knowing who exactly.
\end{theorem}

\begin{proof}
Based on the above state transitions of the above modules, by use of the algebraic laws of $APTC_G$, we can prove that

$\tau_I(\partial_H(A\between B\between M))=\sum_{D\in\Delta}(r_{C_{AI}}(D)\cdot s_{C_{BO}}(-m)\cdot (s_{C_{BO}}(+m)+s_{C_{BO}}(\bot_A)+s_{C_{BO}}(\bot_M)))\cdot
\tau_I(\partial_H(A\between B\between M))$.

For the details of proof, please refer to section \ref{app}, and we omit it.

That is, the Digital Cash Protocol 4 in Figure \ref{DCP4} $\tau_I(\partial_H(A\between B\between M))$ can exhibit desired external behaviors:

\begin{enumerate}
  \item The digital cash of Alice $SIGN_{sk_B}(ENC_{pk_B}(m_k))$ is anonymous for the merchant and the bank;
  \item The protocol can resist replay attacks, for the use of the random number in each digital cash;
  \item The bank know who cheats him when the double spending problem occurs, either the owner of the cash or the merchant. And he knows exactly the identity of the person.
\end{enumerate}
\end{proof}

\newpage\section{Analyses of Secure Elections Protocols}\label{aosep}

Secure elections protocols should be able to prevent cheating and maintain the voter's privacy. An ideal secure election protocol should have the following properties:
\begin{enumerate}
  \item Legitimacy: only authorized voters can vote;
  \item Oneness: no one can vote more than once;
  \item Privacy: no one can determine for whom anyone else voted;
  \item Non-replicability: no one can duplicate anyone else's vote;
  \item Non-changeability: no one can change anyone else's vote;
  \item Validness: every voter can make sure that his vote has been taken into account in the final tabulation.
\end{enumerate}

In this chapter, we will introduce seven secure elections protocols in the following sections. In the analyses of these seven protocols, we will mainly analyze the security and privacy
properties.

\subsection{Secure Elections Protocol 1}\label{sep1}

The secure elections protocol 1 is shown in Figure \ref{SEP1}, which is a basic one to implement the basic voting function. In this protocol, there are a CTF (Central Tabulating Facility),
to collect the votes, and four voters: Alice, Bob, Carol and Dave.

\begin{figure}
    \centering
    \includegraphics{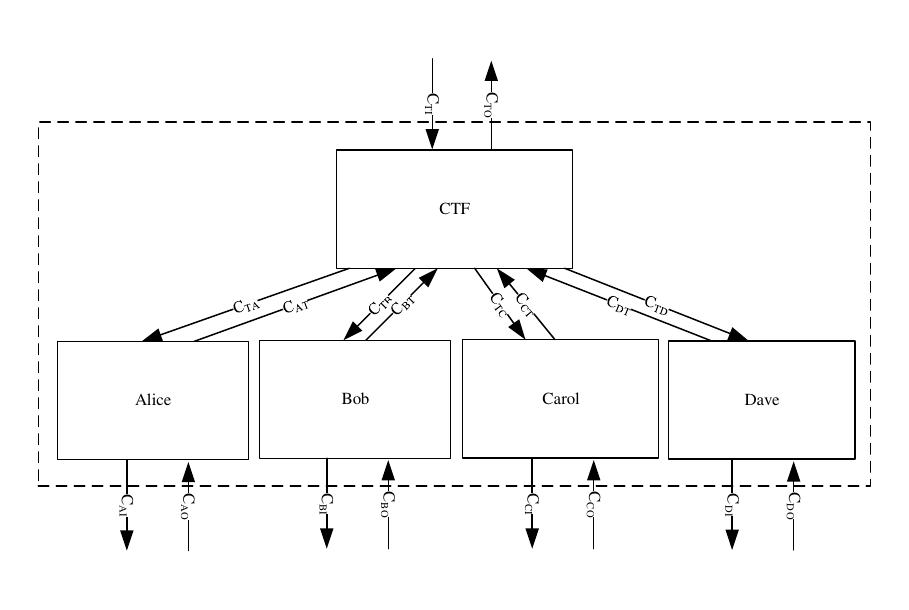}
    \caption{Secure elections protocol 1}
    \label{SEP1}
\end{figure}

The process of the protocol is as follows.

\begin{enumerate}
  \item Alice receives some voting request $D_A$ from the outside through the channel $C_{AI}$ (the corresponding reading action is denoted $r_{C_{AI}}(D_A)$), she generates the votes
  $v_A$, encrypts $v_A$ by CTF's public key $pk_T$ through an action $enc_{pk_T}(v_A)$, and sends it to CTF through the channel $C_{AT}$ (the corresponding sending action is denoted
  $s_{C_{AT}}(ENC_{pk_T}(v_A))$);
  \item Bob receives some voting request $D_B$ from the outside through the channel $C_{BI}$ (the corresponding reading action is denoted $r_{C_{BI}}(D_B)$), he generates the votes
  $v_B$, encrypts $v_B$ by CTF's public key $pk_T$ through an action $enc_{pk_T}(v_B)$, and sends it to CTF through the channel $C_{BT}$ (the corresponding sending action is denoted
  $s_{C_{BT}}(ENC_{pk_T}(v_B))$);
  \item Carol receives some voting request $D_C$ from the outside through the channel $C_{CI}$ (the corresponding reading action is denoted $r_{C_{CI}}(D_C)$), he generates the votes
  $v_C$, encrypts $v_C$ by CTF's public key $pk_T$ through an action $enc_{pk_T}(v_C)$, and sends it to CTF through the channel $C_{CT}$ (the corresponding sending action is denoted
  $s_{C_{CT}}(ENC_{pk_T}(v_C))$);
  \item Dave receives some voting request $D_D$ from the outside through the channel $C_{DI}$ (the corresponding reading action is denoted $r_{C_{DI}}(D_D)$), he generates the votes
  $v_D$, encrypts $v_D$ by CTF's public key $pk_T$ through an action $enc_{pk_T}(v_D)$, and sends it to CTF through the channel $C_{DT}$ (the corresponding sending action is denoted
  $s_{C_{DT}}(ENC_{pk_T}(v_D))$);
  \item CTF receives encrypted votes from Alice, Bob, Carol and Dave through the channels $C_{AT}$, $C_BT$, $C_{CT}$ and $C_{DT}$ (the corresponding reading actions are denoted
  $r_{C_{AT}}(ENC_{pk_T}(v_A))$, $r_{C_{BT}}(ENC_{pk_T}(v_B))$, $r_{C_{CT}}(ENC_{pk_T}(v_C))$ and $r_{C_{DT}}(ENC_{pk_T}(v_D))$ respectively), decrypts the encrypted votes through actions
  $dec_{sk_T}(ENC_{pk_T}(v_A))$, $dec_{sk_T}(ENC_{pk_T}(v_B))$, \\
  $dec_{sk_T}(ENC_{pk_T}(v_C))$, and $dec_{sk_T}(ENC_{pk_T}(v_D))$ to get $v_A$, $v_B$, $v_C$ and $v_D$, then sends
  $v_A+v_B+v_C+v_D$ to the outside through the channel $C_{TO}$ (the corresponding sending action is denoted $s_{C_{TO}}(v_A+v_B+v_C+v_D)$).
\end{enumerate}

Where $D_A,D_B,D_C,D_D\in\Delta$, $\Delta$ is the set of data.

Alice's state transitions described by $APTC_G$ are as follows.

$A=\sum_{D_A\in\Delta}r_{C_{AI}}(D_A)\cdot A_2$

$A_2=enc_{pk_T}(v_A)\cdot A_3$

$A_3=s_{C_{AT}}(ENC_{pk_T}(v_A))\cdot A$

Bob's state transitions described by $APTC_G$ are as follows.

$B=\sum_{D_B\in\Delta}r_{C_{BI}}(D_B)\cdot B_2$

$B_2=enc_{pk_T}(v_B)\cdot B_3$

$B_3=s_{C_{BT}}(ENC_{pk_T}(v_B))\cdot B$

Carol's state transitions described by $APTC_G$ are as follows.

$C=\sum_{D_C\in\Delta}r_{C_{CI}}(D_C)\cdot C_2$

$C_2=enc_{pk_T}(v_C)\cdot C_3$

$C_3=s_{C_{CT}}(ENC_{pk_T}(v_C))\cdot C$

Dave's state transitions described by $APTC_G$ are as follows.

$D=\sum_{D_D\in\Delta}r_{C_{DI}}(D_D)\cdot D_2$

$D_2=enc_{pk_T}(v_D)\cdot D_3$

$D_3=s_{C_{DT}}(ENC_{pk_T}(v_D))\cdot D$

CTF's state transitions described by $APTC_G$ are as follows.

$T=r_{C_{AT}}(ENC_{pk_T}(v_A))\parallel r_{C_{BT}}(ENC_{pk_T}(v_B))\parallel r_{C_{CT}}(ENC_{pk_T}(v_C))\parallel r_{C_{DT}}(ENC_{pk_T}(v_D))\cdot T_2$

$T_2=dec_{sk_T}(ENC_{pk_T}(v_A))\parallel dec_{sk_T}(ENC_{pk_T}(v_B))\parallel dec_{sk_T}(ENC_{pk_T}(v_C))\\
\parallel dec_{sk_T}(ENC_{pk_T}(v_D))\cdot T_3$

$T_3=s_{C_{TO}}(v_A+v_B+v_C+v_D)\cdot T$

The sending action and the reading action of the same type data through the same channel can communicate with each other, otherwise, will cause a deadlock $\delta$. We define the following
communication functions.

$\gamma(r_{C_{AT}}(ENC_{pk_T}(v_A)),s_{C_{AT}}(ENC_{pk_T}(v_A)))\triangleq c_{C_{AT}}(ENC_{pk_T}(v_A))$

$\gamma(r_{C_{BT}}(ENC_{pk_T}(v_B)),s_{C_{BT}}(ENC_{pk_T}(v_B)))\triangleq c_{C_{BT}}(ENC_{pk_T}(v_B))$

$\gamma(r_{C_{CT}}(ENC_{pk_T}(v_C)),s_{C_{CT}}(ENC_{pk_T}(v_C)))\triangleq c_{C_{CT}}(ENC_{pk_T}(v_C))$

$\gamma(r_{C_{DT}}(ENC_{pk_T}(v_D)),s_{C_{DT}}(ENC_{pk_T}(v_D)))\triangleq c_{C_{DT}}(ENC_{pk_T}(v_D))$

Let all modules be in parallel, then the protocol $A\quad B\quad C\quad D\quad T$ can be presented by the following process term.

$$\tau_I(\partial_H(\Theta(A\between B\between C\between D\between T)))=\tau_I(\partial_H(A\between B\between C\between D\between T))$$

where $H=\{r_{C_{AT}}(ENC_{pk_T}(v_A)),s_{C_{AT}}(ENC_{pk_T}(v_A)),\\
r_{C_{BT}}(ENC_{pk_T}(v_B)),s_{C_{BT}}(ENC_{pk_T}(v_B)),\\
r_{C_{CT}}(ENC_{pk_T}(v_C)),s_{C_{CT}}(ENC_{pk_T}(v_C)),\\
r_{C_{DT}}(ENC_{pk_T}(v_D)),s_{C_{DT}}(ENC_{pk_T}(v_D))|D_A,D_B,D_C,D_D\in\Delta\}$,

$I=\{c_{C_{AT}}(ENC_{pk_T}(v_A)),c_{C_{BT}}(ENC_{pk_T}(v_B)),c_{C_{CT}}(ENC_{pk_T}(v_C)),\\
c_{C_{DT}}(ENC_{pk_T}(v_D)),enc_{pk_T}(v_A),enc_{pk_T}(v_B),enc_{pk_T}(v_C),enc_{pk_T}(v_D),\\
dec_{sk_T}(ENC_{pk_T}(v_A)), dec_{sk_T}(ENC_{pk_T}(v_B)), dec_{sk_T}(ENC_{pk_T}(v_C)),\\
dec_{sk_T}(ENC_{pk_T}(v_D))|D_A,D_B,D_C,D_D\in\Delta\}$.

Then we get the following conclusion on the protocol.

\begin{theorem}
The secure elections protocol 1 in Figure \ref{SEP1} is secure, but basic.
\end{theorem}

\begin{proof}
Based on the above state transitions of the above modules, by use of the algebraic laws of $APTC_G$, we can prove that

$\tau_I(\partial_H(A\between B\between C\between D\between T))=\sum_{D_A,D_B,D_C,D_D\in\Delta}(r_{C_{AI}}(D_A)\parallel r_{C_{BI}}(D_B)\parallel r_{C_{CI}}(D_C)\parallel r_{C_{DI}}(D_D)\cdot s_{C_{TO}}(v_A+v_B+v_C+v_D))\cdot
\tau_I(\partial_H(A\between B\between C\between D\between T))$.

For the details of proof, please refer to section \ref{app}, and we omit it.

That is, the protocol in Figure \ref{SEP1} $\tau_I(\partial_H(A\between B\between C\between D\between T))$ can exhibit desired external behaviors, and is secure. But, for the properties of
secure elections protocols:
\begin{enumerate}
  \item Legitimacy: all voters can vote;
  \item Oneness: anyone can vote more than once;
  \item Privacy: no one can determine for whom anyone else voted;
  \item Non-replicability: CTF can duplicate anyone else's vote;
  \item Non-changeability: CTF can change anyone else's vote;
  \item Validness: every voter cannot make sure that his vote has been taken into account in the final tabulation.
\end{enumerate}
\end{proof}

\subsection{Secure Elections Protocol 2}\label{sep2}

The secure elections protocol 2 is shown in Figure \ref{SEP2}, which is a improved one based on the secure elections protocol 1 in section \ref{sep1}. In this protocol, there are a
CTF (Central Tabulating Facility), to check the identity of voters and collect the votes, and four voters: Alice, Bob, Carol and Dave.

\begin{figure}
    \centering
    \includegraphics{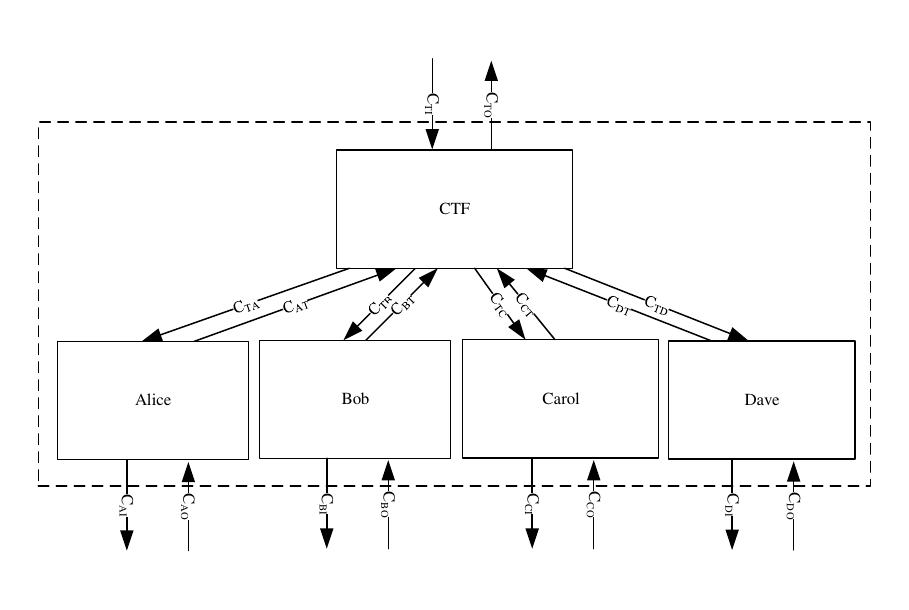}
    \caption{Secure elections protocol 2}
    \label{SEP2}
\end{figure}

The process of the protocol is as follows.

\begin{enumerate}
  \item Alice receives some voting request $D_A$ from the outside through the channel $C_{AI}$ (the corresponding reading action is denoted $r_{C_{AI}}(D_A)$), she generates the votes
  $v_A$, signs $v_A$ by her private key $sk_A$ through an action $sign_{sk_A}(v_A)$, then encrypts it by CTF's public key $pk_T$ through an action $enc_{pk_T}(SIGN_{sk_A}(v_A))$,
  and sends it to CTF through the channel $C_{AT}$ (the corresponding sending action is denoted $s_{C_{AT}}(ENC_{pk_T}(SIGN_{sk_A}(v_A)))$);
  \item Bob receives some voting request $D_B$ from the outside through the channel $C_{BI}$ (the corresponding reading action is denoted $r_{C_{BI}}(D_B)$), he generates the votes
  $v_B$, signs $v_B$ by his private key $sk_B$ through an action $sign_{sk_B}(v_B)$, then encrypts it by CTF's public key $pk_T$ through an action $enc_{pk_T}(SIGN_{sk_B}(v_B))$,
  and sends it to CTF through the channel $C_{BT}$ (the corresponding sending action is denoted $s_{C_{BT}}(ENC_{pk_T}(SIGN_{sk_B}(v_B)))$);
  \item Carol receives some voting request $D_C$ from the outside through the channel $C_{CI}$ (the corresponding reading action is denoted $r_{C_{CI}}(D_C)$), he generates the votes
  $v_C$, signs $v_C$ by his private key $sk_C$ through an action $sign_{sk_C}(v_C)$, then encrypts it by CTF's public key $pk_T$ through an action $enc_{pk_T}(SIGN_{sk_C}(v_C))$,
  and sends it to CTF through the channel $C_{CT}$ (the corresponding sending action is denoted $s_{C_{CT}}(ENC_{pk_T}(SIGN_{sk_C}(v_C)))$);
  \item Dave receives some voting request $D_D$ from the outside through the channel $C_{DI}$ (the corresponding reading action is denoted $r_{C_{DI}}(D_D)$), he generates the votes
  $v_D$, signs $v_D$ by his private key $sk_D$ through an action $sign_{sk_D}(v_D)$, then encrypts it by CTF's public key $pk_T$ through an action $enc_{pk_T}(SIGN_{sk_D}(v_D))$,
  and sends it to CTF through the channel $C_{DT}$ (the corresponding sending action is denoted $s_{C_{DT}}(ENC_{pk_T}(SIGN_{sk_D}(v_D)))$);
  \item CTF receives encrypted votes from Alice, Bob, Carol and Dave through the channels $C_{AT}$, $C_BT$, $C_{CT}$ and $C_{DT}$ (the corresponding reading actions are denoted\\
  $r_{C_{AT}}(ENC_{pk_T}(SIGN_{sk_A}(v_A)))$, $r_{C_{BT}}(ENC_{pk_T}(SIGN_{sk_B}(v_B)))$, \\
  $r_{C_{CT}}(ENC_{pk_T}(SIGN_{sk_C}(v_C)))$ and $r_{C_{DT}}(ENC_{pk_T}(SIGN_{sk_D}(v_D)))$
  respectively), decrypts the encrypted votes through actions
  $dec_{sk_T}(ENC_{pk_T}(SIGN_{sk_A}(v_A)))$, \\
  $dec_{sk_T}(ENC_{pk_T}(SIGN_{sk_B}(v_B)))$,
  $dec_{sk_T}(ENC_{pk_T}(SIGN_{sk_C}(v_C)))$, \\
  and $dec_{sk_T}(ENC_{pk_T}(SIGN_{sk_D}(v_D)))$ , then de-signs them through actions \\
  $de\textrm{-}sign_{pk_A}(SIGN_{sk_A}(v_A))$,
  $de\textrm{-}sign_{pk_B}(SIGN_{sk_B}(v_B))$, $de\textrm{-}sign_{pk_C}(SIGN_{sk_C}(v_C))$, and $de\textrm{-}sign_{pk_D}(SIGN_{sk_D}(v_D))$ to get $v_A$, $v_B$, $v_C$ and $v_D$, then sends
  $v_A+v_B+v_C+v_D,A,B,C,D$ to the outside through the channel $C_{TO}$ (the corresponding sending action is denoted $s_{C_{TO}}(v_A+v_B+v_C+v_D,A,B,C,D)$).
\end{enumerate}

Where $D_A,D_B,D_C,D_D\in\Delta$, $\Delta$ is the set of data.

Alice's state transitions described by $APTC_G$ are as follows.

$A=\sum_{D_A\in\Delta}r_{C_{AI}}(D_A)\cdot A_2$

$A_2=sign_{sk_A}(v_A)\cdot A_3$

$A_3=enc_{pk_T}(SIGN_{sk_A}(v_A))\cdot A_4$

$A_4=s_{C_{AT}}(ENC_{pk_T}(SIGN_{sk_A}(v_A)))\cdot A$

Bob's state transitions described by $APTC_G$ are as follows.

$B=\sum_{D_B\in\Delta}r_{C_{BI}}(D_B)\cdot B_2$

$B_2=sign_{sk_B}(v_B)\cdot B_3$

$B_3=enc_{pk_T}(SIGN_{sk_B}(v_B))\cdot B_4$

$B_4=s_{C_{BT}}(ENC_{pk_T}(SIGN_{sk_B}(v_B)))\cdot B$

Carol's state transitions described by $APTC_G$ are as follows.

$C=\sum_{D_C\in\Delta}r_{C_{CI}}(D_C)\cdot C_2$

$C_2=sign_{sk_C}(v_C)\cdot C_3$

$C_3=enc_{pk_T}(SIGN_{sk_C}(v_C))\cdot C_4$

$C_4=s_{C_{CT}}(ENC_{pk_T}(SIGN_{sk_C}(v_C)))\cdot C$

Dave's state transitions described by $APTC_G$ are as follows.

$D=\sum_{D_D\in\Delta}r_{C_{DI}}(D_D)\cdot D_2$

$D_2=sign_{sk_D}(v_D)\cdot D_3$

$D_3=enc_{pk_T}(SIGN_{sk_D}(v_D))\cdot D_4$

$D_4=s_{C_{DT}}(ENC_{pk_T}(SIGN_{sk_D}(v_D)))\cdot D$

CTF's state transitions described by $APTC_G$ are as follows.

$T=r_{C_{AT}}(ENC_{pk_T}(SIGN_{sk_A}(v_A)))\parallel r_{C_{BT}}(ENC_{pk_T}(SIGN_{sk_B}(v_B)))\\
\parallel r_{C_{CT}}(ENC_{pk_T}(SIGN_{sk_C}(v_C)))\parallel r_{C_{DT}}(ENC_{pk_T}(SIGN_{sk_D}(v_D)))\cdot T_2$

$T_2=dec_{sk_T}(ENC_{pk_T}(SIGN_{sk_A}(v_A)))\parallel dec_{sk_T}(ENC_{pk_T}(SIGN_{sk_B}(v_B)))\\
\parallel dec_{sk_T}(ENC_{pk_T}(SIGN_{sk_C}(v_C)))\parallel dec_{sk_T}(ENC_{pk_T}(SIGN_{sk_D}(v_D)))\cdot T_3$

$T_3=de\textrm{-}sign_{pk_A}(SIGN_{sk_A}(v_A))\parallel de\textrm{-}sign_{pk_B}(SIGN_{sk_B}(v_B))\\
\parallel de\textrm{-}sign_{pk_C}(SIGN_{sk_C}(v_C))\parallel de\textrm{-}sign_{pk_D}(SIGN_{sk_D}(v_D))\cdot T_4$

$T_4=s_{C_{TO}}(v_A+v_B+v_C+v_D,A,B,C,D)\cdot T$

The sending action and the reading action of the same type data through the same channel can communicate with each other, otherwise, will cause a deadlock $\delta$. We define the following
communication functions.

$\gamma(r_{C_{AT}}(ENC_{pk_T}(SIGN_{sk_A}(v_A))),s_{C_{AT}}(ENC_{pk_T}(SIGN_{sk_A}(v_A))))\\
\triangleq c_{C_{AT}}(ENC_{pk_T}(SIGN_{sk_A}(v_A)))$

$\gamma(r_{C_{BT}}(ENC_{pk_T}(SIGN_{sk_B}(v_B))),s_{C_{BT}}(ENC_{pk_T}(SIGN_{sk_B}(v_B))))\\
\triangleq c_{C_{BT}}(ENC_{pk_T}(SIGN_{sk_B}(v_B)))$

$\gamma(r_{C_{CT}}(ENC_{pk_T}(SIGN_{sk_C}(v_C))),s_{C_{CT}}(ENC_{pk_T}(SIGN_{sk_C}(v_C))))\\
\triangleq c_{C_{CT}}(ENC_{pk_T}(SIGN_{sk_C}(v_C)))$

$\gamma(r_{C_{DT}}(ENC_{pk_T}(SIGN_{sk_D}(v_D))),s_{C_{DT}}(ENC_{pk_T}(SIGN_{sk_D}(v_D))))\\
\triangleq c_{C_{DT}}(ENC_{pk_T}(SIGN_{sk_D}(v_D)))$

Let all modules be in parallel, then the protocol $A\quad B\quad C\quad D\quad T$ can be presented by the following process term.

$$\tau_I(\partial_H(\Theta(A\between B\between C\between D\between T)))=\tau_I(\partial_H(A\between B\between C\between D\between T))$$

where $H=\{r_{C_{AT}}(ENC_{pk_T}(SIGN_{sk_A}(v_A))),s_{C_{AT}}(ENC_{pk_T}(SIGN_{sk_A}(v_A))),\\
r_{C_{BT}}(ENC_{pk_T}(SIGN_{sk_B}(v_B))),s_{C_{BT}}(ENC_{pk_T}(SIGN_{sk_B}(v_B))),\\
r_{C_{CT}}(ENC_{pk_T}(SIGN_{sk_C}(v_C))),s_{C_{CT}}(ENC_{pk_T}(SIGN_{sk_C}(v_C))),\\
r_{C_{DT}}(ENC_{pk_T}(SIGN_{sk_D}(v_D))),s_{C_{DT}}(ENC_{pk_T}(SIGN_{sk_D}(v_D)))|D_A,D_B,D_C,D_D\in\Delta\}$,

$I=\{c_{C_{AT}}(ENC_{pk_T}(SIGN_{sk_A}(v_A))),c_{C_{BT}}(ENC_{pk_T}(SIGN_{sk_B}(v_B))),\\
c_{C_{CT}}(ENC_{pk_T}(SIGN_{sk_C}(v_C))),c_{C_{DT}}(ENC_{pk_T}(SIGN_{sk_D}(v_D))),\\
sign_{sk_A}(v_A),sign_{sk_B}(v_B),sign_{sk_C}(v_C),sign_{sk_D}(v_D),\\
enc_{pk_T}(SIGN_{sk_A}(v_A)),enc_{pk_T}(SIGN_{sk_B}(v_B)),enc_{pk_T}(SIGN_{sk_C}(v_C)),\\
enc_{pk_T}(SIGN_{sk_D}(v_D)),dec_{sk_T}(ENC_{pk_T}(SIGN_{sk_A}(v_A))),\\
dec_{sk_T}(ENC_{pk_T}(SIGN_{sk_B}(v_B))),dec_{sk_T}(ENC_{pk_T}(SIGN_{sk_C}(v_C))),\\
dec_{sk_T}(ENC_{pk_T}(SIGN_{sk_D}(v_D))),de\textrm{-}sign_{pk_A}(SIGN_{sk_A}(v_A)),\\
de\textrm{-}sign_{pk_B}(SIGN_{sk_B}(v_B)),de\textrm{-}sign_{pk_C}(SIGN_{sk_C}(v_C)),\\
de\textrm{-}sign_{pk_D}(SIGN_{sk_D}(v_D))|D_A,D_B,D_C,D_D\in\Delta\}$.

Then we get the following conclusion on the protocol.

\begin{theorem}
The secure elections protocol 2 in Figure \ref{SEP2} is improved based on the secure elections protocol 1.
\end{theorem}

\begin{proof}
Based on the above state transitions of the above modules, by use of the algebraic laws of $APTC_G$, we can prove that

$\tau_I(\partial_H(A\between B\between C\between D\between T))=\sum_{D_A,D_B,D_C,D_D\in\Delta}(r_{C_{AI}}(D_A)\parallel r_{C_{BI}}(D_B)\parallel r_{C_{CI}}(D_C)\parallel r_{C_{DI}}(D_D)\cdot s_{C_{TO}}(v_A+v_B+v_C+v_D,A,B,C,D))\cdot
\tau_I(\partial_H(A\between B\between C\between D\between T))$.

For the details of proof, please refer to section \ref{app}, and we omit it.

That is, the protocol in Figure \ref{SEP2} $\tau_I(\partial_H(A\between B\between C\between D\between T))$ can exhibit desired external behaviors, and is secure. But, for the properties of
secure elections protocols:
\begin{enumerate}
  \item Legitimacy: only authorized voters can vote;
  \item Oneness: no one can vote more than once;
  \item Privacy: CTF can determine for whom anyone else voted;
  \item Non-replicability: CTF can duplicate anyone else's vote;
  \item Non-changeability: CTF can change anyone else's vote;
  \item Validness: every voter cannot make sure that his vote has been taken into account in the final tabulation.
\end{enumerate}
\end{proof}

\subsection{Secure Elections Protocol 3}\label{sep3}

The secure elections protocol 3 is shown in Figure \ref{SEP3}, which is a improved one based on the secure elections protocol 2 in section \ref{sep2}. In this protocol, there are a
CTF (Central Tabulating Facility), to check the identity of voters and collect the votes, and four voters: Alice, Bob, Carol and Dave.

\begin{figure}
    \centering
    \includegraphics{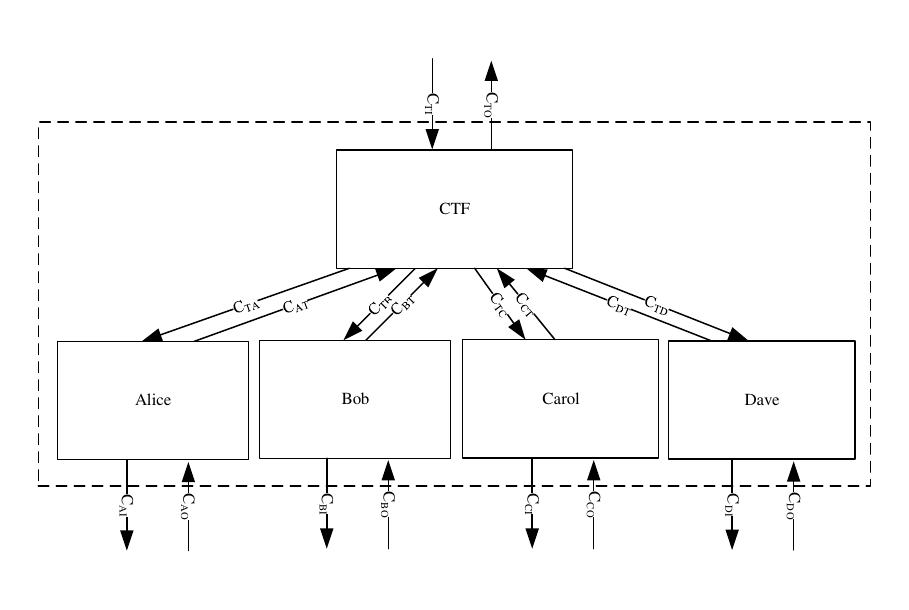}
    \caption{Secure elections protocol 3}
    \label{SEP3}
\end{figure}

The process of the protocol is as follows.

\begin{enumerate}
  \item Alice receives some voting request $D_A$ from the outside through the channel $C_{AI}$ (the corresponding reading action is denoted $r_{C_{AI}}(D_A)$), she generates a message
  containing all possible voting results $V_A$ and a random number $R_A$, blinds this message through an action $blind_{k_A}(V_A,R_A)$, totally there are 10 such messages are generated,
  then she sends these 10 messages to CTF through the channel $C_{AT}$ (the corresponding sending action is denoted $s_{C_{AT}}(10\times BLIND_{k_A}(V_A,R_A),k_A,A)$);
  \item Bob receives some voting request $D_B$ from the outside through the channel $C_{BI}$ (the corresponding reading action is denoted $r_{C_{BI}}(D_B)$), he generates a message
  containing all possible voting results $V_B$ and a random number $R_B$, blinds this message through an action $blind_{k_B}(V_B,R_B)$, totally there are 10 such messages are generated,
  then she sends these 10 messages to CTF through the channel $C_{BT}$ (the corresponding sending action is denoted $s_{C_{BT}}(10\times BLIND_{k_B}(V_B,R_B),k_B,B)$);
  \item Carol receives some voting request $D_C$ from the outside through the channel $C_{CI}$ (the corresponding reading action is denoted $r_{C_{CI}}(D_C)$), he generates a message
  containing all possible voting results $V_C$ and a random number $R_C$, blinds this message through an action $blind_{k_C}(V_C,R_C)$, totally there are 10 such messages are generated,
  then she sends these 10 messages to CTF through the channel $C_{CT}$ (the corresponding sending action is denoted $s_{C_{CT}}(10\times BLIND_{k_C}(V_C,R_C),k_C,C)$);
  \item Dave receives some voting request $D_D$ from the outside through the channel $C_{DI}$ (the corresponding reading action is denoted $r_{C_{DI}}(D_D)$), he generates a message
  containing all possible voting results $V_D$ and a random number $R_D$, blinds this message through an action $blind_{k_D}(V_D,R_D)$, totally there are 10 such messages are generated,
  then she sends these 10 messages to CTF through the channel $C_{DT}$ (the corresponding sending action is denoted $s_{C_{DT}}(10\times BLIND_{k_D}(V_D,R_D),k_D,D)$);
  \item CTF receives the messages from Alice, Bob, Carol and Dave through the channels $C_{AT}$, $C_{BT}$, $C_{CT}$ and $C_{DT}$ (the corresponding reading actions are denoted \\
  $r_{C_{AT}}(10\times BLIND_{k_A}(V_A,R_A),k_A,A)$, $s_{C_{BT}}(10\times BLIND_{k_B}(V_B,R_B),k_B,B)$, $s_{C_{CT}}(10\times BLIND_{k_C}(V_C,R_C),k_C,C)$, and
  $s_{C_{DT}}(10\times BLIND_{k_D}(V_D,R_D),k_D,D)$ respectively), he checks the names of Alice, Bob, Carol and Dave to make sure that they submit the blinded messages in the first time
  and stores the their names through the actions $check(A)$, $check(B)$, $check(C)$ and $check(D)$, then he unblinds randomly their 9 sets of messages to make sure that they are formed
  correctly through the actions $9\times unblind_{k_A}(BLIND_{k_A}(V_A,R_A))$, $9\times unblind_{k_B}(BLIND_{k_B}(V_B,R_B))$, $9\times unblind_{k_C}(BLIND_{k_C}(V_C,R_C))$, \\
  and $9\times unblind_{k_D}(BLIND_{k_D}(V_D,R_D))$. Then he sighs their left message through the actions $sign_{sk_T}(BLIND_{k_A}(V_A,R_A))$, $sign_{sk_T}(BLIND_{k_B}(V_B,R_B))$,\\
  $sign_{sk_T}(BLIND_{k_C}(V_C,R_C))$, and $sign_{sk_T}(BLIND_{k_D}(V_D,R_D))$, and sends them to Alice, Bob, Carol and Dave through the channels $C_{TA}$, $C_{TB}$, $C_{TC}$ and $C_{TD}$
  (the corresponding sending actions is denoted $s_{C_{TA}}(SIGN_{sk_T}(BLIND_{k_A}(V_A,R_A)))$, \\
  $s_{C_{TB}}(SIGN_{sk_T}(BLIND_{k_B}(V_B,R_B)))$,
  $s_{C_{TC}}(SIGN_{sk_T}(BLIND_{k_C}(V_C,R_C)))$, \\
  and $s_{C_{TD}}(SIGN_{sk_T}(BLIND_{k_D}(V_D,R_D)))$);
  \item Alice receives the signed message from CTF through the channel $C_{TA}$ (the corresponding reading action is denoted $r_{C_{TA}}(SIGN_{sk_T}(BLIND_{k_A}(V_A,R_A)))$), she unblinds
  the message through the action $unblind_{k_A}(SIGN_{sk_T}(BLIND_{k_A}(V_A,R_A)))$, selects her vote $v_A$ from $V_A$, encrypts the vote through an action $enc_{pk_T}(SIGN_{sk_T}(v_A,R_A))$,
  and sends her encrypted vote to CTF through the channel $C_{AT}$ (the corresponding sending action is denoted $s_{C_{AT}}(ENC_{pk_T}(SIGN_{sk_T}(v_A,R_A)))$);
  \item Bob receives the signed message from CTF through the channel $C_{TB}$ (the corresponding reading action is denoted $r_{C_{TB}}(SIGN_{sk_T}(BLIND_{k_B}(V_B,R_B)))$), he unblinds
  the message through the action $unblind_{k_B}(SIGN_{sk_T}(BLIND_{k_B}(V_B,R_B)))$, selects his vote $v_B$ from $V_B$, encrypts the vote through an action $enc_{pk_T}(SIGN_{sk_T}(v_B,R_B))$,
  and sends his encrypted vote to CTF through the channel $C_{BT}$ (the corresponding sending action is denoted $s_{C_{BT}}(ENC_{pk_T}(SIGN_{sk_T}(v_B,R_B)))$);
  \item Carol receives the signed message from CTF through the channel $C_{TC}$ (the corresponding reading action is denoted $r_{C_{TC}}(SIGN_{sk_T}(BLIND_{k_C}(V_C,R_C)))$), he unblinds
  the message through the action $unblind_{k_C}(SIGN_{sk_T}(BLIND_{k_C}(V_C,R_C)))$, selects his vote $v_C$ from $V_C$, encrypts the vote through an action $enc_{pk_T}(SIGN_{sk_T}(v_C,R_C))$,
  and sends his encrypted vote to CTF through the channel $C_{CT}$ (the corresponding sending action is denoted $s_{C_{CT}}(ENC_{pk_T}(SIGN_{sk_T}(v_C,R_C)))$);
  \item Dave receives the signed message from CTF through the channel $C_{TD}$ (the corresponding reading action is denoted $r_{C_{TD}}(SIGN_{sk_T}(BLIND_{k_D}(V_D,R_D)))$), he unblinds
  the message through the action $unblind_{k_D}(SIGN_{sk_T}(BLIND_{k_D}(V_D,R_D)))$, selects his vote $v_D$ from $V_D$, encrypts the vote through an action $enc_{pk_T}(SIGN_{sk_T}(v_D,R_D))$,
  and sends his encrypted vote to CTF through the channel $C_{DT}$ (the corresponding sending action is denoted $s_{C_{DT}}(ENC_{pk_T}(SIGN_{sk_T}(v_D,R_D)))$);
  \item CTF receives the votes from Alice, Bob, Carol and Dave through the channels $C_{AT}$, $C_{BT}$, $C_{CT}$ and $C_{DT}$ (the corresponding reading actions are denoted \\
  $r_{C_{AT}}(ENC_{pk_T}(SIGN_{sk_T}(v_A,R_A)))$, $r_{C_{BT}}(ENC_{pk_T}(SIGN_{sk_T}(v_B,R_B)))$, \\
  $r_{C_{CT}}(ENC_{pk_T}(SIGN_{sk_T}(v_C,R_C)))$, and
  $r_{C_{DT}}(ENC_{pk_T}(SIGN_{sk_T}(v_D,R_D)))$ respectively), he decrypts and de-signs these votes through the actions $dec_{sk_T}(ENC_{pk_T}(SIGN_{sk_T}(v_A,R_A)))$,
  $dec_{sk_T}(ENC_{pk_T}(SIGN_{sk_T}(v_B,R_B)))$, $dec_{sk_T}(ENC_{pk_T}(SIGN_{sk_T}(v_C,R_C)))$, \\
  $dec_{sk_T}(ENC_{pk_T}(SIGN_{sk_T}(v_D,R_D)))$ and
  $de\textrm{-}sign_{pk_T}(SIGN_{sk_T}(v_A,R_A))$, \\
  $de\textrm{-}sign_{pk_T}(SIGN_{sk_T}(v_B,R_B))$, $de\textrm{-}sign_{pk_T}(SIGN_{sk_T}(v_C,R_C))$, \\
  $de\textrm{-}sign_{pk_T}(SIGN_{sk_T}(v_D,R_D))$. If $isFresh(R_A)=TRUE$, he tabulates $v_A$ through an action $tab(v_A)$, else $tab(0)$; if $isFresh(R_B)=TRUE$, he tabulates $v_B$ through
  an action $tab(v_B)$, else $tab(0)$; if $isFresh(R_C)=TRUE$, he tabulates $v_C$ through an action $tab(v_C)$, else $tab(0)$; if $isFresh(R_D)=TRUE$, he tabulates $v_D$ through an action
  $tab(v_D)$, else $tab(0)$. Finally, he sends the voting results $TAB$ to the outside through the channel $C_{TO}$ (the corresponding sending action is denoted $s_{C_{TO}}(TAB)$).
\end{enumerate}

Where $D_A,D_B,D_C,D_D\in\Delta$, $\Delta$ is the set of data.

Alice's state transitions described by $APTC_G$ are as follows.

$A=\sum_{D_A\in\Delta}r_{C_{AI}}(D_A)\cdot A_2$

$A_2=blind_{k_A}(V_A,R_A)\cdot A_3$

$A_3=s_{C_{AT}}(10\times BLIND_{k_A}(V_A,R_A),k_A,A)\cdot A_4$

$A_4=r_{C_{TA}}(SIGN_{sk_T}(BLIND_{k_A}(V_A,R_A)))\cdot A_5$

$A_5=unblind_{k_A}(SIGN_{sk_T}(BLIND_{k_A}(V_A,R_A)))\cdot A_6$

$A_6=enc_{pk_T}(SIGN_{sk_T}(v_A,R_A))\cdot A_7$

$A_7=s_{C_{AT}}(ENC_{pk_T}(SIGN_{sk_T}(v_A,R_A)))\cdot A$

Bob's state transitions described by $APTC_G$ are as follows.

$B=\sum_{D_B\in\Delta}r_{C_{BI}}(D_B)\cdot B_2$

$B_2=blind_{k_B}(V_B,R_B)\cdot B_3$

$B_3=s_{C_{BT}}(10\times BLIND_{k_B}(V_B,R_B),k_B,B)\cdot B_4$

$B_4=r_{C_{TB}}(SIGN_{sk_T}(BLIND_{k_B}(V_B,R_B)))\cdot B_5$

$B_5=unblind_{k_B}(SIGN_{sk_T}(BLIND_{k_B}(V_B,R_B)))\cdot B_6$

$B_6=enc_{pk_T}(SIGN_{sk_T}(v_B,R_B))\cdot B_7$

$B_7=s_{C_{BT}}(ENC_{pk_T}(SIGN_{sk_T}(v_B,R_B)))\cdot B$

Carol's state transitions described by $APTC_G$ are as follows.

$C=\sum_{D_C\in\Delta}r_{C_{CI}}(D_C)\cdot C_2$

$C_2=blind_{k_C}(V_C,R_C)\cdot C_3$

$C_3=s_{C_{CT}}(10\times BLIND_{k_C}(V_C,R_C),k_C,C)\cdot C_4$

$C_4=r_{C_{TC}}(SIGN_{sk_T}(BLIND_{k_C}(V_C,R_C)))\cdot C_5$

$C_5=unblind_{k_C}(SIGN_{sk_T}(BLIND_{k_C}(V_C,R_C)))\cdot C_6$

$C_6=enc_{pk_T}(SIGN_{sk_T}(v_C,R_C))\cdot C_7$

$C_7=s_{C_{CT}}(ENC_{pk_T}(SIGN_{sk_T}(v_C,R_C)))\cdot C$

Dave's state transitions described by $APTC_G$ are as follows.

$D=\sum_{D_D\in\Delta}r_{C_{DI}}(D_D)\cdot D_2$

$D_2=blind_{k_D}(V_D,R_D)\cdot D_3$

$D_3=s_{C_{DT}}(10\times BLIND_{k_D}(V_D,R_D),k_D,D)\cdot D_4$

$D_4=r_{C_{TD}}(SIGN_{sk_T}(BLIND_{k_D}(V_D,R_D)))\cdot D_5$

$D_5=unblind_{k_D}(SIGN_{sk_T}(BLIND_{k_D}(V_D,R_D)))\cdot D_6$

$D_6=enc_{pk_T}(SIGN_{sk_T}(v_D,R_D))\cdot D_7$

$D_7=s_{C_{DT}}(ENC_{pk_T}(SIGN_{sk_T}(v_D,R_D)))\cdot D$

CTF's state transitions described by $APTC_G$ are as follows.

$T=r_{C_{AT}}(10\times BLIND_{k_A}(V_A,R_A),k_A,A)\parallel s_{C_{BT}}(10\times BLIND_{k_B}(V_B,R_B),k_B,B)\\
\parallel s_{C_{CT}}(10\times BLIND_{k_C}(V_C,R_C),k_C,C)
\parallel s_{C_{DT}}(10\times BLIND_{k_D}(V_D,R_D),k_D,D)\cdot T_2$

$T_2=check(A)\parallel check(B)\parallel check(C)\parallel check(D)\cdot T_3$

$T_3=9\times unblind_{k_A}(BLIND_{k_A}(V_A,R_A))\parallel 9\times unblind_{k_B}(BLIND_{k_B}(V_B,R_B))\\
\parallel 9\times unblind_{k_C}(BLIND_{k_C}(V_C,R_C))
\parallel 9\times unblind_{k_D}(BLIND_{k_D}(V_D,R_D))\cdot T_4$

$T_4=sign_{sk_T}(BLIND_{k_A}(V_A,R_A))\parallel sign_{sk_T}(BLIND_{k_B}(V_B,R_B))\\
\parallel sign_{sk_T}(BLIND_{k_C}(V_C,R_C))
\parallel sign_{sk_T}(BLIND_{k_D}(V_D,R_D))\cdot T_5$

$T_5=s_{C_{TA}}(SIGN_{sk_T}(BLIND_{k_A}(V_A,R_A)))\parallel s_{C_{TB}}(SIGN_{sk_T}(BLIND_{k_B}(V_B,R_B)))\\
\parallel s_{C_{TC}}(SIGN_{sk_T}(BLIND_{k_C}(V_C,R_C)))\parallel s_{C_{TD}}(SIGN_{sk_T}(BLIND_{k_D}(V_D,R_D)))\cdot T_6$

$T_6=r_{C_{AT}}(ENC_{pk_T}(SIGN_{sk_T}(v_A,R_A)))\parallel r_{C_{BT}}(ENC_{pk_T}(SIGN_{sk_T}(v_B,R_B)))\\
\parallel r_{C_{CT}}(ENC_{pk_T}(SIGN_{sk_T}(v_C,R_C)))\parallel r_{C_{DT}}(ENC_{pk_T}(SIGN_{sk_T}(v_D,R_D)))\cdot T_7$

$T_7=dec_{sk_T}(ENC_{pk_T}(SIGN_{sk_T}(v_A,R_A)))\parallel dec_{sk_T}(ENC_{pk_T}(SIGN_{sk_T}(v_B,R_B)))\\
\parallel dec_{sk_T}(ENC_{pk_T}(SIGN_{sk_T}(v_C,R_C)))\parallel dec_{sk_T}(ENC_{pk_T}(SIGN_{sk_T}(v_D,R_D)))\cdot T_8$

$T_8=de\textrm{-}sign_{pk_T}(SIGN_{sk_T}(v_A,R_A))\parallel de\textrm{-}sign_{pk_T}(SIGN_{sk_T}(v_B,R_B))\\
\parallel de\textrm{-}sign_{pk_T}(SIGN_{sk_T}(v_C,R_C))\parallel de\textrm{-}sign_{pk_T}(SIGN_{sk_T}(v_D,R_D))\cdot T_9$

$T_9=((\{isFresh(R_A)=TRUE\}\cdot tab(v_A)+\{isFresh(R_A)=FALSE\}\cdot tab(0))\\
\parallel (\{isFresh(R_B)=TRUE\}\cdot tab(v_B)+\{isFresh(R_B)=FALSE\}\cdot tab(0))\\
\parallel (\{isFresh(R_C)=TRUE\}\cdot tab(v_C)+\{isFresh(R_C)=FALSE\}\cdot tab(0))\\
\parallel (\{isFresh(R_D)=TRUE\}\cdot tab(v_D)+\{isFresh(R_D)=FALSE\}\cdot tab(0)))\cdot T_{10}$

$T_{10}=s_{C_{TO}}(TAB)\cdot T$

The sending action and the reading action of the same type data through the same channel can communicate with each other, otherwise, will cause a deadlock $\delta$. We define the following
communication functions.

$\gamma(r_{C_{AT}}(10\times BLIND_{k_A}(V_A,R_A),k_A,A),s_{C_{AT}}(10\times BLIND_{k_A}(V_A,R_A),k_A,A))\\
\triangleq c_{C_{AT}}(10\times BLIND_{k_A}(V_A,R_A),k_A,A)$

$\gamma(r_{C_{TA}}(SIGN_{sk_T}(BLIND_{k_A}(V_A,R_A))),s_{C_{TA}}(SIGN_{sk_T}(BLIND_{k_A}(V_A,R_A))))\\
\triangleq c_{C_{TA}}(SIGN_{sk_T}(BLIND_{k_A}(V_A,R_A)))$

$\gamma(r_{C_{AT}}(ENC_{pk_T}(SIGN_{sk_T}(v_A,R_A))),s_{C_{AT}}(ENC_{pk_T}(SIGN_{sk_T}(v_A,R_A))))\\
\triangleq c_{C_{AT}}(ENC_{pk_T}(SIGN_{sk_T}(v_A,R_A)))$

$\gamma(r_{C_{BT}}(10\times BLIND_{k_B}(V_B,R_B),k_B,B),s_{C_{BT}}(10\times BLIND_{k_B}(V_B,R_B),k_B,B))\\
\triangleq c_{C_{BT}}(10\times BLIND_{k_B}(V_B,R_B),k_B,B)$

$\gamma(r_{C_{TB}}(SIGN_{sk_T}(BLIND_{k_B}(V_B,R_B))),s_{C_{TB}}(SIGN_{sk_T}(BLIND_{k_B}(V_B,R_B))))\\
\triangleq c_{C_{TB}}(SIGN_{sk_T}(BLIND_{k_B}(V_B,R_B)))$

$\gamma(r_{C_{BT}}(ENC_{pk_T}(SIGN_{sk_T}(v_B,R_B))),s_{C_{BT}}(ENC_{pk_T}(SIGN_{sk_T}(v_B,R_B))))\\
\triangleq c_{C_{BT}}(ENC_{pk_T}(SIGN_{sk_T}(v_B,R_B)))$

$\gamma(r_{C_{CT}}(10\times BLIND_{k_C}(V_C,R_C),k_C,C),s_{C_{CT}}(10\times BLIND_{k_C}(V_C,R_C),k_C,C))\\
\triangleq c_{C_{CT}}(10\times BLIND_{k_C}(V_C,R_C),k_C,C)$

$\gamma(r_{C_{TC}}(SIGN_{sk_T}(BLIND_{k_C}(V_C,R_C))),s_{C_{TC}}(SIGN_{sk_T}(BLIND_{k_C}(V_C,R_C))))\\
\triangleq c_{C_{TC}}(SIGN_{sk_T}(BLIND_{k_C}(V_C,R_C)))$

$\gamma(r_{C_{CT}}(ENC_{pk_T}(SIGN_{sk_T}(v_C,R_C))),s_{C_{CT}}(ENC_{pk_T}(SIGN_{sk_T}(v_C,R_C))))\\
\triangleq c_{C_{CT}}(ENC_{pk_T}(SIGN_{sk_T}(v_C,R_C)))$

$\gamma(r_{C_{DT}}(10\times BLIND_{k_D}(V_D,R_D),k_D,D),s_{C_{DT}}(10\times BLIND_{k_D}(V_D,R_D),k_D,D))\\
\triangleq c_{C_{DT}}(10\times BLIND_{k_D}(V_D,R_D),k_D,D)$

$\gamma(r_{C_{TD}}(SIGN_{sk_T}(BLIND_{k_D}(V_D,R_D))),s_{C_{TD}}(SIGN_{sk_T}(BLIND_{k_D}(V_D,R_D))))\\
\triangleq c_{C_{TD}}(SIGN_{sk_T}(BLIND_{k_D}(V_D,R_D)))$

$\gamma(r_{C_{DT}}(ENC_{pk_T}(SIGN_{sk_T}(v_D,R_D))),s_{C_{DT}}(ENC_{pk_T}(SIGN_{sk_T}(v_D,R_D))))\\
\triangleq c_{C_{DT}}(ENC_{pk_T}(SIGN_{sk_T}(v_D,R_D)))$

Let all modules be in parallel, then the protocol $A\quad B\quad C\quad D\quad T$ can be presented by the following process term.

$$\tau_I(\partial_H(\Theta(A\between B\between C\between D\between T)))=\tau_I(\partial_H(A\between B\between C\between D\between T))$$

where $H=\{r_{C_{AT}}(10\times BLIND_{k_A}(V_A,R_A),k_A,A),s_{C_{AT}}(10\times BLIND_{k_A}(V_A,R_A),k_A,A),\\
r_{C_{TA}}(SIGN_{sk_T}(BLIND_{k_A}(V_A,R_A))),s_{C_{TA}}(SIGN_{sk_T}(BLIND_{k_A}(V_A,R_A))),\\
r_{C_{AT}}(ENC_{pk_T}(SIGN_{sk_T}(v_A,R_A))),s_{C_{AT}}(ENC_{pk_T}(SIGN_{sk_T}(v_A,R_A))),\\
r_{C_{BT}}(10\times BLIND_{k_B}(V_B,R_B),k_B,B),s_{C_{BT}}(10\times BLIND_{k_B}(V_B,R_B),k_B,B),\\
r_{C_{TB}}(SIGN_{sk_T}(BLIND_{k_B}(V_B,R_B))),s_{C_{TB}}(SIGN_{sk_T}(BLIND_{k_B}(V_B,R_B))),\\
r_{C_{BT}}(ENC_{pk_T}(SIGN_{sk_T}(v_B,R_B))),s_{C_{BT}}(ENC_{pk_T}(SIGN_{sk_T}(v_B,R_B))),\\
r_{C_{CT}}(10\times BLIND_{k_C}(V_C,R_C),k_C,C),s_{C_{CT}}(10\times BLIND_{k_C}(V_C,R_C),k_C,C),\\
r_{C_{TC}}(SIGN_{sk_T}(BLIND_{k_C}(V_C,R_C))),s_{C_{TC}}(SIGN_{sk_T}(BLIND_{k_C}(V_C,R_C))),\\
r_{C_{CT}}(ENC_{pk_T}(SIGN_{sk_T}(v_C,R_C))),s_{C_{CT}}(ENC_{pk_T}(SIGN_{sk_T}(v_C,R_C))),\\
r_{C_{DT}}(10\times BLIND_{k_D}(V_D,R_D),k_D,D),s_{C_{DT}}(10\times BLIND_{k_D}(V_D,R_D),k_D,D),\\
r_{C_{TD}}(SIGN_{sk_T}(BLIND_{k_D}(V_D,R_D))),s_{C_{TD}}(SIGN_{sk_T}(BLIND_{k_D}(V_D,R_D))),\\
r_{C_{DT}}(ENC_{pk_T}(SIGN_{sk_T}(v_D,R_D))),s_{C_{DT}}(ENC_{pk_T}(SIGN_{sk_T}(v_D,R_D)))|D_A,D_B,D_C,D_D\in\Delta\}$,

$I=\{c_{C_{AT}}(10\times BLIND_{k_A}(V_A,R_A),k_A,A),c_{C_{TA}}(SIGN_{sk_T}(BLIND_{k_A}(V_A,R_A))),\\
c_{C_{AT}}(ENC_{pk_T}(SIGN_{sk_T}(v_A,R_A))),c_{C_{BT}}(10\times BLIND_{k_B}(V_B,R_B),k_B,B),\\
c_{C_{TB}}(SIGN_{sk_T}(BLIND_{k_B}(V_B,R_B))),c_{C_{BT}}(ENC_{pk_T}(SIGN_{sk_T}(v_B,R_B))),\\
c_{C_{CT}}(10\times BLIND_{k_C}(V_C,R_C),k_C,C),c_{C_{TC}}(SIGN_{sk_T}(BLIND_{k_C}(V_C,R_C))),\\
c_{C_{CT}}(ENC_{pk_T}(SIGN_{sk_T}(v_C,R_C))),c_{C_{DT}}(10\times BLIND_{k_D}(V_D,R_D),k_D,D),\\
c_{C_{TD}}(SIGN_{sk_T}(BLIND_{k_D}(V_D,R_D))),c_{C_{DT}}(ENC_{pk_T}(SIGN_{sk_T}(v_D,R_D))),\\
blind_{k_A}(V_A,R_A),blind_{k_B}(V_B,R_B),blind_{k_C}(V_C,R_C),blind_{k_D}(V_D,R_D),\\
unblind_{k_A}(SIGN_{sk_T}(BLIND_{k_A}(V_A,R_A))),unblind_{k_B}(SIGN_{sk_T}(BLIND_{k_B}(V_B,R_B))),\\
unblind_{k_C}(SIGN_{sk_T}(BLIND_{k_C}(V_C,R_C))),unblind_{k_D}(SIGN_{sk_T}(BLIND_{k_D}(V_D,R_D))),\\
enc_{pk_T}(SIGN_{sk_T}(v_A,R_A)),enc_{pk_T}(SIGN_{sk_T}(v_B,R_B)),enc_{pk_T}(SIGN_{sk_T}(v_C,R_C)),\\
enc_{pk_T}(SIGN_{sk_T}(v_D,R_D)),check(A), check(B), check(C), check(D),\\
9\times unblind_{k_A}(BLIND_{k_A}(V_A,R_A)), 9\times unblind_{k_B}(BLIND_{k_B}(V_B,R_B)),\\
9\times unblind_{k_C}(BLIND_{k_C}(V_C,R_C)),9\times unblind_{k_D}(BLIND_{k_D}(V_D,R_D)),\\
sign_{sk_T}(BLIND_{k_A}(V_A,R_A)), sign_{sk_T}(BLIND_{k_B}(V_B,R_B)),\\
sign_{sk_T}(BLIND_{k_C}(V_C,R_C)),sign_{sk_T}(BLIND_{k_D}(V_D,R_D)),\\
dec_{sk_T}(ENC_{pk_T}(SIGN_{sk_T}(v_A,R_A))), dec_{sk_T}(ENC_{pk_T}(SIGN_{sk_T}(v_B,R_B))),\\
dec_{sk_T}(ENC_{pk_T}(SIGN_{sk_T}(v_C,R_C))), dec_{sk_T}(ENC_{pk_T}(SIGN_{sk_T}(v_D,R_D))),\\
de\textrm{-}sign_{pk_T}(SIGN_{sk_T}(v_A,R_A)), de\textrm{-}sign_{pk_T}(SIGN_{sk_T}(v_B,R_B)),\\
de\textrm{-}sign_{pk_T}(SIGN_{sk_T}(v_C,R_C)), de\textrm{-}sign_{pk_T}(SIGN_{sk_T}(v_D,R_D)),\\
\{isFresh(R_A)=TRUE\}, tab(v_A),\{isFresh(R_A)=FALSE\}, tab(0),\\
\{isFresh(R_B)=TRUE\}, tab(v_B),\{isFresh(R_B)=FALSE\},\\
\{isFresh(R_C)=TRUE\}, tab(v_C),\{isFresh(R_C)=FALSE\},\\
\{isFresh(R_D)=TRUE\}, tab(v_D),\{isFresh(R_D)=FALSE\}|D_A,D_B,D_C,D_D\in\Delta\}$.

Then we get the following conclusion on the protocol.

\begin{theorem}
The secure elections protocol 3 in Figure \ref{SEP3} is improved based on the secure elections protocol 2.
\end{theorem}

\begin{proof}
Based on the above state transitions of the above modules, by use of the algebraic laws of $APTC_G$, we can prove that

$\tau_I(\partial_H(A\between B\between C\between D\between T))=\sum_{D_A,D_B,D_C,D_D\in\Delta}(r_{C_{AI}}(D_A)\parallel r_{C_{BI}}(D_B)\parallel r_{C_{CI}}(D_C)\parallel r_{C_{DI}}(D_D)\cdot s_{C_{TO}}(TAB))\cdot
\tau_I(\partial_H(A\between B\between C\between D\between T))$.

For the details of proof, please refer to section \ref{app}, and we omit it.

That is, the protocol in Figure \ref{SEP3} $\tau_I(\partial_H(A\between B\between C\between D\between T))$ can exhibit desired external behaviors, and is secure. But, for the properties of
secure elections protocols:
\begin{enumerate}
  \item Legitimacy: only authorized voters can vote;
  \item Oneness: no one can vote more than once;
  \item Privacy: no one can determine for whom anyone else voted;
  \item Non-replicability: no one can duplicate anyone else's vote;
  \item Non-changeability: no one can change anyone else's vote;
  \item Validness: every voter can make sure that his vote has been taken into account in the final tabulation, if CTF is trustworthy.
\end{enumerate}

But, CTF still can make valid signatures to cheat.
\end{proof}

\subsection{Secure Elections Protocol 4}\label{sep4}

The secure elections protocol 4 is shown in Figure \ref{SEP4}, which is a improved one based on the secure elections protocol 3 in section \ref{sep3}. In this protocol, there are a CLA
(Central Legitimization Agency) to check the identity of voters and a CTF (Central Tabulating Facility) to collect the votes, and four voters: Alice, Bob, Carol and Dave.

\begin{figure}
    \centering
    \includegraphics{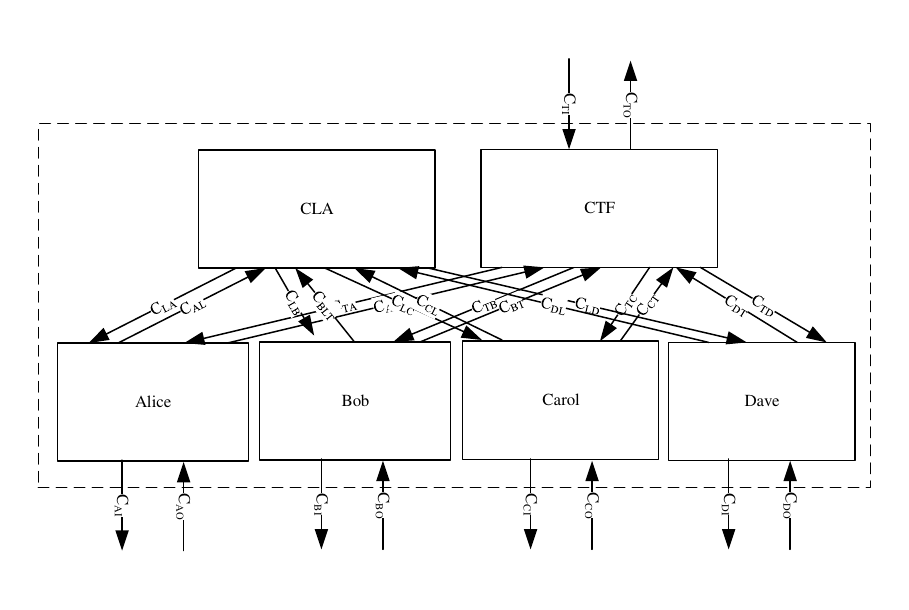}
    \caption{Secure elections protocol 4}
    \label{SEP4}
\end{figure}

The process of the protocol is as follows.

\begin{enumerate}
  \item Alice receives some voting request $D_A$ from the outside through the channel $C_{AI}$ (the corresponding reading action is denoted $r_{C_{AI}}(D_A)$), she generates a request
  $r_A$, encrypts it by CLA's public key through an action $enc_{pk_L}(r_A)$, and sends it to CLA through the channel $C_{AL}$ (the corresponding sending action is denoted
  $s_{C_{AL}}(ENC_{pk_L}(r_A))$);
  \item Bob receives some voting request $D_B$ from the outside through the channel $C_{BI}$ (the corresponding reading action is denoted $r_{C_{BI}}(D_B)$), he generates a request
  $r_B$, encrypts it by CLA's public key through an action $enc_{pk_L}(r_B)$, and sends it to CLA through the channel $C_{BL}$ (the corresponding sending action is denoted
  $s_{C_{BL}}(ENC_{pk_L}(r_B))$);
  \item Carol receives some voting request $D_C$ from the outside through the channel $C_{CI}$ (the corresponding reading action is denoted $r_{C_{CI}}(D_C)$), he generates a request
  $r_C$, encrypts it by CLA's public key through an action $enc_{pk_L}(r_C)$, and sends it to CLA through the channel $C_{CL}$ (the corresponding sending action is denoted
  $s_{C_{CL}}(ENC_{pk_L}(r_C))$);
  \item Dave receives some voting request $D_D$ from the outside through the channel $C_{DI}$ (the corresponding reading action is denoted $r_{C_{DI}}(D_D)$), he generates a request
  $r_D$, encrypts it by CLA's public key through an action $enc_{pk_L}(r_D)$, and sends it to CLA through the channel $C_{DL}$ (the corresponding sending action is denoted
  $s_{C_{DL}}(ENC_{pk_L}(r_D))$);
  \item CLA receives the requests from Alice, Bob, Carol and Dave through the channels $C_{AL}$, $C_{BL}$, $C_{CL}$ and $C_{DL}$ (the corresponding reading actions are denoted
  $s_{C_{AL}}(ENC_{pk_L}(r_A))$, $s_{C_{BL}}(ENC_{pk_L}(r_B))$, $s_{C_{CL}}(ENC_{pk_L}(r_C))$, and $s_{C_{DL}}(ENC_{pk_L}(r_D))$ respectively), he decrypts these encrypted requests
  through the actions $dec_{sk_L}(ENC_{pk_L}(r_A))$, $dec_{sk_L}(ENC_{pk_L}(r_B))$, $dec_{sk_L}(ENC_{pk_L}(r_C))$, and $dec_{sk_L}(ENC_{pk_L}(r_D))$ to get $r_A$, $r_B$, $r_C$ and $r_D$,
  records the names of Alice, Bob, Carol and Dave through actions $rec(A)$, $rec(B)$, $rec(C)$ and $rec(D)$; Both CLA and CTF maintain a table of valid numbers, and the table of CTF is
  obtained from that of CLA; then CLA randomly selects numbers $R_A$, $R_B$, $R_C$ and $R_D$, encrypts them through actions $enc_{pk_A}(R_A)$, $enc_{pk_B}(R_B)$, $enc_{pk_C}(R_C)$ and
  $enc_{pk_D}(R_D)$ and sends them to Alice, Bob, Carol and Dave through the channels $C_{LA}$, $C_{LB}$,
  $C_{LC}$ and $C_{LD}$ respectively (the corresponding sending action is denoted $s_{C_{LA}}(ENC_{pk_A}(R_A))$, $s_{C_{LB}}(ENC_{pk_B}(R_B))$, $s_{C_{LC}}(ENC_{pk_C}(R_C))$, and
  $s_{C_{LD}}(ENC_{pk_D}(R_D))$);
  \item Alice receives the encrypted number from CLA through the channel $C_{LA}$ (the corresponding reading action is denoted $r_{C_{LA}}(ENC_{pk_A}(R_A))$), she decrypts the encrypted
  number through an action $dec_{sk_A}(ENC_{pk_A}(R_A))$ to get $R_A$, generates a random identity number $I_A$ through an action $rsg_{I_A}$ and her vote $v_A$, encrypted $I_A,R_A,v_A$
  by CTF's public key through an action $enc_{pk_T}(I_A,R_A,v_A)$ and sends the encrypted message to CTF through the channel $C_{AT}$ (the corresponding sending action is denoted
  $s_{C_{AT}}(ENC_{pk_T}(I_A,R_A,v_A))$);
  \item Bob receives the encrypted number from CLA through the channel $C_{LB}$ (the corresponding reading action is denoted $r_{C_{LB}}(ENC_{pk_B}(R_B))$), he decrypts the encrypted
  number through an action $dec_{sk_B}(ENC_{pk_B}(R_B))$ to get $R_B$, generates a random identity number $I_B$ through an action $rsg_{I_B}$ and his vote $v_B$, encrypted $I_B,R_B,v_B$
  by CTF's public key through an action $enc_{pk_T}(I_B,R_B,v_B)$ and sends the encrypted message to CTF through the channel $C_{BT}$ (the corresponding sending action is denoted
  $s_{C_{BT}}(ENC_{pk_T}(I_B,R_B,v_B))$);
  \item Carol receives the encrypted number from CLA through the channel $C_{LC}$ (the corresponding reading action is denoted $r_{C_{LC}}(ENC_{pk_C}(R_C))$), he decrypts the encrypted
  number through an action $dec_{sk_C}(ENC_{pk_C}(R_C))$ to get $R_C$, generates a random identity number $I_C$ through an action $rsg_{I_C}$ and his vote $v_C$, encrypted $I_C,R_C,v_C$
  by CTF's public key through an action $enc_{pk_T}(I_C,R_C,v_C)$ and sends the encrypted message to CTF through the channel $C_{CT}$ (the corresponding sending action is denoted
  $s_{C_{CT}}(ENC_{pk_T}(I_C,R_C,v_C))$);
  \item Dave receives the encrypted number from CLA through the channel $C_{LD}$ (the corresponding reading action is denoted $r_{C_{LD}}(ENC_{pk_A}(R_D))$), he decrypts the encrypted
  number through an action $dec_{sk_D}(ENC_{pk_D}(R_D))$ to get $R_D$, generates a random identity number $I_D$ through an action $rsg_{I_D}$ and his vote $v_D$, encrypted $I_D,R_D,v_D$
  by CTF's public key through an action $enc_{pk_T}(I_D,R_D,v_D)$ and sends the encrypted message to CTF through the channel $C_{DT}$ (the corresponding sending action is denoted
  $s_{C_{DT}}(ENC_{pk_T}(I_D,R_D,v_D))$);
  \item CTF receives the encrypted messages from Alice, Bob, Carol and Dave through the channels $C_{AT}$, $C_{BT}$, $C_{CT}$ and $C_{DT}$ (the corresponding reading actions are denoted\\
  $r_{C_{AT}}(ENC_{pk_T}(I_A,R_A,v_A))$, $r_{C_{BT}}(ENC_{pk_T}(I_B,R_B,v_B))$, $r_{C_{CT}}(ENC_{pk_T}(I_C,R_C,v_C))$, and $r_{C_{DT}}(ENC_{pk_T}(I_D,R_D,v_D))$ respectively), he
  decrypts these encrypted messages through actions $dec_{sk_T}(ENC_{pk_T}(I_A,R_A,v_A))$, $dec_{sk_T}(ENC_{pk_T}(I_B,R_B,v_B))$, \\
  $dec_{sk_T}(ENC_{pk_T}(I_C,R_C,v_C))$, and
  $dec_{sk_T}(ENC_{pk_T}(I_D,R_D,v_D))$. If $isExisted(R_A)=TRUE$, he removes $R_A$ from its table through an action $remove(R_A)$, records the vote $v_A$ and the pair of $I_A$ and $v_A$
  into the voting results $TAB$ through an action $rec(I_A,v_A)$, else he does nothing; if $isExisted(R_B)=TRUE$, he removes $R_B$ from its table through an action $remove(R_B)$, records the vote $v_B$ and the pair of $I_B$ and $v_B$
  into the voting results $TAB$ through an action $rec(I_B,v_B)$, else he does nothing; if $isExisted(R_C)=TRUE$, he removes $R_C$ from its table through an action $remove(R_C)$, records the vote $v_C$ and the pair of $I_C$ and $v_C$
  into the voting results $TAB$ through an action $rec(I_C,v_C)$, else he does nothing; if $isExisted(R_D)=TRUE$, he removes $R_D$ from its table through an action $remove(R_D)$, records the vote $v_D$ and the pair of $I_D$ and $v_D$
  into the voting results $TAB$ through an action $rec(I_D,v_D)$, else he does nothing. Finally, he sends the voting results $TAB$ to the outside through the channel $C_{TO}$ (the corresponding sending action is denoted
  $s_{C_{TO}}(TAB)$).
\end{enumerate}

Where $D_A,D_B,D_C,D_D\in\Delta$, $\Delta$ is the set of data.

Alice's state transitions described by $APTC_G$ are as follows.

$A=\sum_{D_A\in\Delta}r_{C_{AI}}(D_A)\cdot A_2$

$A_2=enc_{pk_L}(r_A)\cdot A_3$

$A_3=s_{C_{AL}}(ENC_{pk_L}(r_A))\cdot A_4$

$A_4=r_{C_{LA}}(ENC_{pk_A}(R_A))\cdot A_5$

$A_5=dec_{sk_A}(ENC_{pk_A}(R_A))\cdot A_6$

$A_6=rsg_{I_A}\cdot A_7$

$A_7=enc_{pk_T}(I_A,R_A,v_A)\cdot A_8$

$A_8=s_{C_{AT}}(ENC_{pk_T}(I_A,R_A,v_A))\cdot A$

Bob's state transitions described by $APTC_G$ are as follows.

$B=\sum_{D_B\in\Delta}r_{C_{BI}}(D_B)\cdot B_2$

$B_2=enc_{pk_L}(r_B)\cdot B_3$

$B_3=s_{C_{BL}}(ENC_{pk_L}(r_B))\cdot B_4$

$B_4=r_{C_{LB}}(ENC_{pk_B}(R_B))\cdot B_5$

$B_5=dec_{sk_B}(ENC_{pk_B}(R_B))\cdot B_6$

$B_6=rsg_{I_B}\cdot B_7$

$B_7=enc_{pk_T}(I_B,R_B,v_B)\cdot B_8$

$B_8=s_{C_{BT}}(ENC_{pk_T}(I_B,R_B,v_B))\cdot B$

Carol's state transitions described by $APTC_G$ are as follows.

$C=\sum_{D_C\in\Delta}r_{C_{CI}}(D_C)\cdot C_2$

$C_2=enc_{pk_L}(r_C)\cdot C_3$

$C_3=s_{C_{CL}}(ENC_{pk_L}(r_C))\cdot C_4$

$C_4=r_{C_{LC}}(ENC_{pk_C}(R_C))\cdot C_5$

$C_5=dec_{sk_C}(ENC_{pk_C}(R_C))\cdot C_6$

$C_6=rsg_{I_C}\cdot C_7$

$C_7=enc_{pk_T}(I_C,R_C,v_C)\cdot C_8$

$C_8=s_{C_{CT}}(ENC_{pk_T}(I_C,R_C,v_C))\cdot C$

Dave's state transitions described by $APTC_G$ are as follows.

$D=\sum_{D_D\in\Delta}r_{C_{DI}}(D_D)\cdot D_2$

$D_2=enc_{pk_L}(r_D)\cdot D_3$

$D_3=s_{C_{DL}}(ENC_{pk_L}(r_D))\cdot D_4$

$D_4=r_{C_{LD}}(ENC_{pk_D}(R_D))\cdot D_5$

$D_5=dec_{sk_D}(ENC_{pk_D}(R_D))\cdot D_6$

$D_6=rsg_{I_D}\cdot D_7$

$D_7=enc_{pk_T}(I_D,R_D,v_D)\cdot D_8$

$D_8=s_{C_{DT}}(ENC_{pk_T}(I_D,R_D,v_D))\cdot D$

CLA's state transitions described by $APTC_G$ are as follows.

$L=s_{C_{AL}}(ENC_{pk_L}(r_A))\parallel s_{C_{BL}}(ENC_{pk_L}(r_B))\\
\parallel s_{C_{CL}}(ENC_{pk_L}(r_C))\parallel s_{C_{DL}}(ENC_{pk_L}(r_D))\cdot L_2$

$L_2=dec_{sk_L}(ENC_{pk_L}(r_A))\parallel dec_{sk_L}(ENC_{pk_L}(r_B))\\
\parallel dec_{sk_L}(ENC_{pk_L}(r_C))\parallel dec_{sk_L}(ENC_{pk_L}(r_D))\cdot L_3$

$L_3=rec(A)\parallel rec(B)\parallel rec(C)\parallel rec(D)\cdot L_4$

$L_4=enc_{pk_A}(R_A)\parallel enc_{pk_B}(R_B)\parallel enc_{pk_C}(R_C)\parallel enc_{pk_D}(R_D)\cdot L_5$

$L_5=s_{C_{LA}}(ENC_{pk_A}(R_A))\parallel s_{C_{LB}}(ENC_{pk_B}(R_B))\\
\parallel s_{C_{LC}}(ENC_{pk_C}(R_C))\parallel s_{C_{LD}}(ENC_{pk_D}(R_D))\cdot L$

CTF's state transitions described by $APTC_G$ are as follows.

$T=r_{C_{AT}}(ENC_{pk_T}(I_A,R_A,v_A))\parallel r_{C_{BT}}(ENC_{pk_T}(I_B,R_B,v_B))\\
\parallel r_{C_{CT}}(ENC_{pk_T}(I_C,R_C,v_C))\parallel r_{C_{DT}}(ENC_{pk_T}(I_D,R_D,v_D))\cdot T_2$

$T_2=dec_{sk_T}(ENC_{pk_T}(I_A,R_A,v_A))\parallel dec_{sk_T}(ENC_{pk_T}(I_B,R_B,v_B))\\
\parallel dec_{sk_T}(ENC_{pk_T}(I_C,R_C,v_C))\parallel dec_{sk_T}(ENC_{pk_T}(I_D,R_D,v_D))\cdot T_3$

$T_3=((\{isExisted(R_A)=TRUE\}\cdot remove(R_A)\cdot rec(I_A,v_A)+\{isExisted(R_A)=FALSE\})\\
\parallel (\{isExisted(R_B)=TRUE\}\cdot remove(R_B)\cdot rec(I_B,v_B)+\{isExisted(R_B)=FALSE\})\\
\parallel (\{isExisted(R_C)=TRUE\}\cdot remove(R_C)\cdot rec(I_C,v_C)+\{isExisted(R_C)=FALSE\})\\
\parallel (\{isExisted(R_D)=TRUE\}\cdot remove(R_D)\cdot rec(I_D,v_D)+\{isExisted(R_D)=FALSE\}))\cdot T_{4}$

$T_{4}=s_{C_{TO}}(TAB)\cdot T$

The sending action and the reading action of the same type data through the same channel can communicate with each other, otherwise, will cause a deadlock $\delta$. We define the following
communication functions.

$\gamma(r_{C_{AL}}(ENC_{pk_L}(r_A)),s_{C_{AL}}(ENC_{pk_L}(r_A)))\triangleq c_{C_{AL}}(ENC_{pk_L}(r_A))$

$\gamma(r_{C_{LA}}(ENC_{pk_A}(R_A)),s_{C_{LA}}(ENC_{pk_A}(R_A)))\triangleq c_{C_{LA}}(ENC_{pk_A}(R_A))$

$\gamma(r_{C_{AT}}(ENC_{pk_T}(I_A,R_A,v_A)),s_{C_{AT}}(ENC_{pk_T}(I_A,R_A,v_A)))\triangleq c_{C_{AT}}(ENC_{pk_T}(I_A,R_A,v_A))$

$\gamma(r_{C_{BL}}(ENC_{pk_L}(r_B)),s_{C_{BL}}(ENC_{pk_L}(r_B)))\triangleq c_{C_{BL}}(ENC_{pk_L}(r_B))$

$\gamma(r_{C_{LB}}(ENC_{pk_B}(R_B)),s_{C_{LB}}(ENC_{pk_B}(R_B)))\triangleq c_{C_{LB}}(ENC_{pk_B}(R_B))$

$\gamma(r_{C_{BT}}(ENC_{pk_T}(I_B,R_B,v_B)),s_{C_{BT}}(ENC_{pk_T}(I_B,R_B,v_B)))\triangleq c_{C_{BT}}(ENC_{pk_T}(I_B,R_B,v_B))$

$\gamma(r_{C_{CL}}(ENC_{pk_L}(r_C)),s_{C_{CL}}(ENC_{pk_L}(r_C)))\triangleq c_{C_{CL}}(ENC_{pk_L}(r_C))$

$\gamma(r_{C_{LC}}(ENC_{pk_C}(R_C)),s_{C_{LC}}(ENC_{pk_C}(R_C)))\triangleq c_{C_{LC}}(ENC_{pk_C}(R_C))$

$\gamma(r_{C_{CT}}(ENC_{pk_T}(I_C,R_C,v_C)),s_{C_{CT}}(ENC_{pk_T}(I_C,R_C,v_C)))\triangleq c_{C_{CT}}(ENC_{pk_T}(I_C,R_C,v_C))$

$\gamma(r_{C_{DL}}(ENC_{pk_L}(r_D)),s_{C_{DL}}(ENC_{pk_L}(r_D)))\triangleq c_{C_{DL}}(ENC_{pk_L}(r_D))$

$\gamma(r_{C_{LD}}(ENC_{pk_D}(R_D)),s_{C_{LD}}(ENC_{pk_D}(R_D)))\triangleq c_{C_{LD}}(ENC_{pk_D}(R_D))$

$\gamma(r_{C_{DT}}(ENC_{pk_T}(I_D,R_D,v_D)),s_{C_{DT}}(ENC_{pk_T}(I_D,R_D,v_D)))\triangleq c_{C_{DT}}(ENC_{pk_T}(I_D,R_D,v_D))$

Let all modules be in parallel, then the protocol $A\quad B\quad C\quad D\quad L\quad T$ can be presented by the following process term.

$$\tau_I(\partial_H(\Theta(A\between B\between C\between D\between L \between T)))=\tau_I(\partial_H(A\between B\between C\between D\between L\between T))$$

where $H=\{r_{C_{AL}}(ENC_{pk_L}(r_A)),s_{C_{AL}}(ENC_{pk_L}(r_A)),\\
r_{C_{LA}}(ENC_{pk_A}(R_A)),s_{C_{LA}}(ENC_{pk_A}(R_A)),\\
r_{C_{AT}}(ENC_{pk_T}(I_A,R_A,v_A)),s_{C_{AT}}(ENC_{pk_T}(I_A,R_A,v_A)),\\
r_{C_{BL}}(ENC_{pk_L}(r_B)),s_{C_{BL}}(ENC_{pk_L}(r_B)),\\
r_{C_{LB}}(ENC_{pk_B}(R_B)),s_{C_{LB}}(ENC_{pk_B}(R_B)),\\
r_{C_{BT}}(ENC_{pk_T}(I_B,R_B,v_B)),s_{C_{BT}}(ENC_{pk_T}(I_B,R_B,v_B)),\\
r_{C_{CL}}(ENC_{pk_L}(r_C)),s_{C_{CL}}(ENC_{pk_L}(r_C)),\\
r_{C_{LC}}(ENC_{pk_C}(R_C)),s_{C_{LC}}(ENC_{pk_C}(R_C)),\\
r_{C_{CT}}(ENC_{pk_T}(I_C,R_C,v_C)),s_{C_{CT}}(ENC_{pk_T}(I_C,R_C,v_C)),\\
r_{C_{DL}}(ENC_{pk_L}(r_D)),s_{C_{DL}}(ENC_{pk_L}(r_D)),\\
r_{C_{LD}}(ENC_{pk_D}(R_D)),s_{C_{LD}}(ENC_{pk_D}(R_D)),\\
r_{C_{DT}}(ENC_{pk_T}(I_D,R_D,v_D)),s_{C_{DT}}(ENC_{pk_T}(I_D,R_D,v_D))|D_A,D_B,D_C,D_D\in\Delta\}$,

$I=\{c_{C_{AL}}(ENC_{pk_L}(r_A)),c_{C_{LA}}(ENC_{pk_A}(R_A)),\\
c_{C_{AT}}(ENC_{pk_T}(I_A,R_A,v_A)),c_{C_{BL}}(ENC_{pk_L}(r_B)),\\
c_{C_{LB}}(ENC_{pk_B}(R_B)),c_{C_{BT}}(ENC_{pk_T}(I_B,R_B,v_B)),\\
c_{C_{CL}}(ENC_{pk_L}(r_C)),c_{C_{LC}}(ENC_{pk_C}(R_C)),\\
c_{C_{CT}}(ENC_{pk_T}(I_C,R_C,v_C)),c_{C_{DL}}(ENC_{pk_L}(r_D)),\\
c_{C_{LD}}(ENC_{pk_D}(R_D)),c_{C_{DT}}(ENC_{pk_T}(I_D,R_D,v_D)),\\
enc_{pk_L}(r_A),enc_{pk_L}(r_B),enc_{pk_L}(r_C),enc_{pk_L}(r_D),\\
dec_{sk_A}(ENC_{pk_A}(R_A)),dec_{sk_B}(ENC_{pk_B}(R_B)),dec_{sk_C}(ENC_{pk_C}(R_C)),\\
dec_{sk_D}(ENC_{pk_D}(R_D)),rsg_{I_A},rsg_{I_B},rsg_{I_C},rsg_{I_D},\\
enc_{pk_T}(I_A,R_A,v_A),enc_{pk_T}(I_B,R_B,v_B),enc_{pk_T}(I_C,R_C,v_C),enc_{pk_T}(I_D,R_D,v_D),\\
dec_{sk_L}(ENC_{pk_L}(r_A)), dec_{sk_L}(ENC_{pk_L}(r_B)),\\
dec_{sk_L}(ENC_{pk_L}(r_C)), dec_{sk_L}(ENC_{pk_L}(r_D)),\\
rec(A), rec(B), rec(C), rec(D),enc_{pk_A}(R_A), enc_{pk_B}(R_B),\\
enc_{pk_C}(R_C), enc_{pk_D}(R_D),dec_{sk_T}(ENC_{pk_T}(I_A,R_A,v_A)),\\
dec_{sk_T}(ENC_{pk_T}(I_B,R_B,v_B)),dec_{sk_T}(ENC_{pk_T}(I_C,R_C,v_C)),dec_{sk_T}(ENC_{pk_T}(I_D,R_D,v_D)),\\
\{isExisted(R_A)=TRUE\}, remove(R_A), rec(I_A,v_A),\{isExisted(R_A)=FALSE\},\\
\{isExisted(R_B)=TRUE\}, remove(R_B), rec(I_B,v_B),\{isExisted(R_B)=FALSE\},\\
\{isExisted(R_C)=TRUE\}, remove(R_C), rec(I_C,v_C),\{isExisted(R_C)=FALSE\},\\
\{isExisted(R_D)=TRUE\}, remove(R_D), rec(I_D,v_D),\{isExisted(R_D)=FALSE\}\\
|D_A,D_B,D_C,D_D\in\Delta\}$.

Then we get the following conclusion on the protocol.

\begin{theorem}
The secure elections protocol 4 in Figure \ref{SEP4} is improved based on the secure elections protocol 3.
\end{theorem}

\begin{proof}
Based on the above state transitions of the above modules, by use of the algebraic laws of $APTC_G$, we can prove that

$\tau_I(\partial_H(A\between B\between C\between D\between L\between T))=\sum_{D_A,D_B,D_C,D_D\in\Delta}(r_{C_{AI}}(D_A)\parallel r_{C_{BI}}(D_B)\parallel r_{C_{CI}}(D_C)\parallel r_{C_{DI}}(D_D)\cdot s_{C_{TO}}(TAB))\cdot
\tau_I(\partial_H(A\between B\between C\between D\between L\between T))$.

For the details of proof, please refer to section \ref{app}, and we omit it.

That is, the protocol in Figure \ref{SEP4} $\tau_I(\partial_H(A\between B\between C\between D\between L\between T))$ can exhibit desired external behaviors, and is secure. But, for the properties of
secure elections protocols:
\begin{enumerate}
  \item Legitimacy: only authorized voters can vote;
  \item Oneness: no one can vote more than once;
  \item Privacy: no one can determine for whom anyone else voted;
  \item Non-replicability: no one can duplicate anyone else's vote;
  \item Non-changeability: no one can change anyone else's vote;
  \item Validness: every voter can make sure that his vote has been taken into account in the final tabulation, if CLA and CTF are trustworthy.
\end{enumerate}

But, CLA and CTF still can conspire to distribute valid numbers to illegal voters.
\end{proof}

\subsection{Secure Elections Protocol 5}\label{sep5}

The secure elections protocol 5 is shown in Figure \ref{SEP5}, which is a improved one based on the secure elections protocol 4 in section \ref{sep4}. In this protocol, there are
a CTF (Central Tabulating Facility) to check the identity of voters and collect the votes, and four voters: Alice, Bob, Carol and Dave.

\begin{figure}
    \centering
    \includegraphics{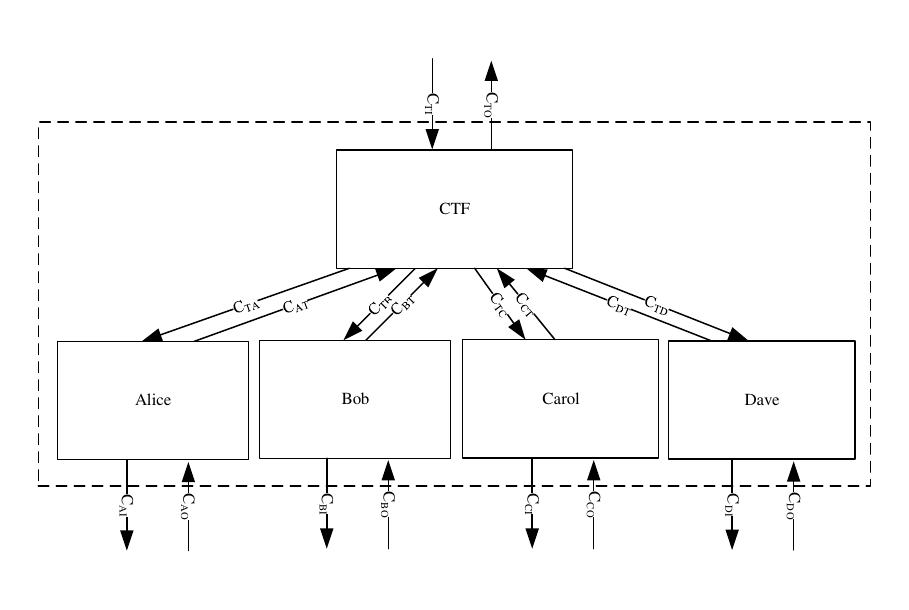}
    \caption{Secure elections protocol 5}
    \label{SEP5}
\end{figure}

The process of the protocol is as follows.

\begin{enumerate}
  \item Alice receives some voting request $D_A$ from the outside through the channel $C_{AI}$ (the corresponding reading action is denoted $r_{C_{AI}}(D_A)$), she generates a request
  $r_A$, encrypts it by CTF's public key through an action $enc_{pk_T}(r_A)$, and sends it to CTF through the channel $C_{AT}$ (the corresponding sending action is denoted
  $s_{C_{AT}}(ENC_{pk_T}(r_A))$);
  \item Bob receives some voting request $D_B$ from the outside through the channel $C_{BI}$ (the corresponding reading action is denoted $r_{C_{BI}}(D_B)$), he generates a request
  $r_B$, encrypts it by CTF's public key through an action $enc_{pk_T}(r_B)$, and sends it to CTF through the channel $C_{BT}$ (the corresponding sending action is denoted
  $s_{C_{BT}}(ENC_{pk_T}(r_B))$);
  \item Carol receives some voting request $D_C$ from the outside through the channel $C_{CI}$ (the corresponding reading action is denoted $r_{C_{CI}}(D_C)$), he generates a request
  $r_C$, encrypts it by CTF's public key through an action $enc_{pk_T}(r_C)$, and sends it to CTF through the channel $C_{CT}$ (the corresponding sending action is denoted
  $s_{C_{CT}}(ENC_{pk_T}(r_C))$);
  \item Dave receives some voting request $D_D$ from the outside through the channel $C_{DI}$ (the corresponding reading action is denoted $r_{C_{DI}}(D_D)$), he generates a request
  $r_D$, encrypts it by CTF's public key through an action $enc_{pk_T}(r_D)$, and sends it to CTF through the channel $C_{DT}$ (the corresponding sending action is denoted
  $s_{C_{DT}}(ENC_{pk_T}(r_D))$);
  \item CTF receives the requests from Alice, Bob, Carol and Dave through the channels $C_{AT}$, $C_{BT}$, $C_{CT}$ and $C_{DT}$ (the corresponding reading actions are denoted
  $s_{C_{AT}}(ENC_{pk_T}(r_A))$, $s_{C_{BT}}(ENC_{pk_T}(r_B))$, $s_{C_{CT}}(ENC_{pk_T}(r_C))$, and $s_{C_{DT}}(ENC_{pk_T}(r_D))$ respectively), he decrypts these encrypted requests
  through the actions $dec_{sk_T}(ENC_{pk_T}(r_A))$, $dec_{sk_T}(ENC_{pk_T}(r_B))$, $dec_{sk_T}(ENC_{pk_T}(r_C))$, and $dec_{sk_T}(ENC_{pk_T}(r_D))$ to get $r_A$, $r_B$, $r_C$ and $r_D$,
  records the names of Alice, Bob, Carol and Dave through actions $rec(A)$, $rec(B)$, $rec(C)$ and $rec(D)$; CTF maintain a table of valid numbers; then CTF
  encrypts all numbers $R$ through actions $enc_{pk_A}(R)$, $enc_{pk_B}(R)$, $enc_{pk_C}(R)$ and
  $enc_{pk_D}(R)$ and sends them to Alice, Bob, Carol and Dave through the channels $C_{TA}$, $C_{TB}$,
  $C_{TC}$ and $C_{TD}$ respectively (the corresponding sending action is denoted $s_{C_{TA}}(ENC_{pk_A}(R))$, $s_{C_{TB}}(ENC_{pk_B}(R))$, $s_{C_{TC}}(ENC_{pk_C}(R))$, and
  $s_{C_{TD}}(ENC_{pk_D}(R))$);
  \item Alice receives the encrypted number from CTF through the channel $C_{TA}$ (the corresponding reading action is denoted $r_{C_{TA}}(ENC_{pk_A}(R))$), she decrypts the encrypted
  number through an action $dec_{sk_A}(ENC_{pk_A}(R))$ to randomly select one $R_A$, generates a random identity number $I_A$ through an action $rsg_{I_A}$ and her vote $v_A$, encrypted $I_A,R_A,v_A$
  by CTF's public key through an action $enc_{pk_T}(I_A,R_A,v_A)$ and sends the encrypted message to CTF through the channel $C_{AT}$ (the corresponding sending action is denoted
  $s_{C_{AT}}(ENC_{pk_T}(I_A,R_A,v_A))$);
  \item Bob receives the encrypted number from CTF through the channel $C_{TB}$ (the corresponding reading action is denoted $r_{C_{TB}}(ENC_{pk_B}(R))$), he decrypts the encrypted
  number through an action $dec_{sk_B}(ENC_{pk_B}(R))$ to randomly select one $R_B$, generates a random identity number $I_B$ through an action $rsg_{I_B}$ and his vote $v_B$, encrypted $I_B,R_B,v_B$
  by CTF's public key through an action $enc_{pk_T}(I_B,R_B,v_B)$ and sends the encrypted message to CTF through the channel $C_{BT}$ (the corresponding sending action is denoted
  $s_{C_{BT}}(ENC_{pk_T}(I_B,R_B,v_B))$);
  \item Carol receives the encrypted number from CTF through the channel $C_{TC}$ (the corresponding reading action is denoted $r_{C_{TC}}(ENC_{pk_C}(R))$), he decrypts the encrypted
  number through an action $dec_{sk_C}(ENC_{pk_C}(R))$ to randomly select one $R_C$, generates a random identity number $I_C$ through an action $rsg_{I_C}$ and his vote $v_C$, encrypted $I_C,R_C,v_C$
  by CTF's public key through an action $enc_{pk_T}(I_C,R_C,v_C)$ and sends the encrypted message to CTF through the channel $C_{CT}$ (the corresponding sending action is denoted
  $s_{C_{CT}}(ENC_{pk_T}(I_C,R_C,v_C))$);
  \item Dave receives the encrypted number from CTF through the channel $C_{TD}$ (the corresponding reading action is denoted $r_{C_{TD}}(ENC_{pk_A}(R))$), he decrypts the encrypted
  number through an action $dec_{sk_D}(ENC_{pk_D}(R))$ to randomly select one $R_D$, generates a random identity number $I_D$ through an action $rsg_{I_D}$ and his vote $v_D$, encrypted $I_D,R_D,v_D$
  by CTF's public key through an action $enc_{pk_T}(I_D,R_D,v_D)$ and sends the encrypted message to CTF through the channel $C_{DT}$ (the corresponding sending action is denoted
  $s_{C_{DT}}(ENC_{pk_T}(I_D,R_D,v_D))$);
  \item CTF receives the encrypted messages from Alice, Bob, Carol and Dave through the channels $C_{AT}$, $C_{BT}$, $C_{CT}$ and $C_{DT}$ (the corresponding reading actions are denoted\\
  $r_{C_{AT}}(ENC_{pk_T}(I_A,R_A,v_A))$, $r_{C_{BT}}(ENC_{pk_T}(I_B,R_B,v_B))$, $r_{C_{CT}}(ENC_{pk_T}(I_C,R_C,v_C))$, and $r_{C_{DT}}(ENC_{pk_T}(I_D,R_D,v_D))$ respectively), he
  decrypts these encrypted messages through actions $dec_{sk_T}(ENC_{pk_T}(I_A,R_A,v_A))$, $dec_{sk_T}(ENC_{pk_T}(I_B,R_B,v_B))$, \\
  $dec_{sk_T}(ENC_{pk_T}(I_C,R_C,v_C))$, and
  $dec_{sk_T}(ENC_{pk_T}(I_D,R_D,v_D))$. If $isExisted(R_A)=TRUE$, he removes $R_A$ from its table through an action $remove(R_A)$, records the vote $v_A$ and the pair of $I_A$ and $v_A$
  into the voting results $TAB$ through an action $rec(I_A,v_A)$, else he does nothing; if $isExisted(R_B)=TRUE$, he removes $R_B$ from its table through an action $remove(R_B)$, records the vote $v_B$ and the pair of $I_B$ and $v_B$
  into the voting results $TAB$ through an action $rec(I_B,v_B)$, else he does nothing; if $isExisted(R_C)=TRUE$, he removes $R_C$ from its table through an action $remove(R_C)$, records the vote $v_C$ and the pair of $I_C$ and $v_C$
  into the voting results $TAB$ through an action $rec(I_C,v_C)$, else he does nothing; if $isExisted(R_D)=TRUE$, he removes $R_D$ from its table through an action $remove(R_D)$, records the vote $v_D$ and the pair of $I_D$ and $v_D$
  into the voting results $TAB$ through an action $rec(I_D,v_D)$, else he does nothing. Finally, he sends the voting results $TAB$ to the outside through the channel $C_{TO}$ (the corresponding sending action is denoted
  $s_{C_{TO}}(TAB)$).
\end{enumerate}

Where $D_A,D_B,D_C,D_D\in\Delta$, $\Delta$ is the set of data.

Alice's state transitions described by $APTC_G$ are as follows.

$A=\sum_{D_A\in\Delta}r_{C_{AI}}(D_A)\cdot A_2$

$A_2=enc_{pk_T}(r_A)\cdot A_3$

$A_3=s_{C_{AT}}(ENC_{pk_T}(r_A))\cdot A_4$

$A_4=r_{C_{TA}}(ENC_{pk_A}(R))\cdot A_5$

$A_5=dec_{sk_A}(ENC_{pk_A}(R))\cdot A_6$

$A_6=rsg_{I_A}\cdot A_7$

$A_7=enc_{pk_T}(I_A,R_A,v_A)\cdot A_8$

$A_8=s_{C_{AT}}(ENC_{pk_T}(I_A,R_A,v_A))\cdot A$

Bob's state transitions described by $APTC_G$ are as follows.

$B=\sum_{D_B\in\Delta}r_{C_{BI}}(D_B)\cdot B_2$

$B_2=enc_{pk_T}(r_B)\cdot B_3$

$B_3=s_{C_{BT}}(ENC_{pk_T}(r_B))\cdot B_4$

$B_4=r_{C_{TB}}(ENC_{pk_B}(R))\cdot B_5$

$B_5=dec_{sk_B}(ENC_{pk_B}(R))\cdot B_6$

$B_6=rsg_{I_B}\cdot B_7$

$B_7=enc_{pk_T}(I_B,R_B,v_B)\cdot B_8$

$B_8=s_{C_{BT}}(ENC_{pk_T}(I_B,R_B,v_B))\cdot B$

Carol's state transitions described by $APTC_G$ are as follows.

$C=\sum_{D_C\in\Delta}r_{C_{CI}}(D_C)\cdot C_2$

$C_2=enc_{pk_T}(r_C)\cdot C_3$

$C_3=s_{C_{CT}}(ENC_{pk_T}(r_C))\cdot C_4$

$C_4=r_{C_{TC}}(ENC_{pk_C}(R))\cdot C_5$

$C_5=dec_{sk_C}(ENC_{pk_C}(R))\cdot C_6$

$C_6=rsg_{I_C}\cdot C_7$

$C_7=enc_{pk_T}(I_C,R_C,v_C)\cdot C_8$

$C_8=s_{C_{CT}}(ENC_{pk_T}(I_C,R_C,v_C))\cdot C$

Dave's state transitions described by $APTC_G$ are as follows.

$D=\sum_{D_D\in\Delta}r_{C_{DI}}(D_D)\cdot D_2$

$D_2=enc_{pk_T}(r_D)\cdot D_3$

$D_3=s_{C_{DT}}(ENC_{pk_T}(r_D))\cdot D_4$

$D_4=r_{C_{TD}}(ENC_{pk_D}(R))\cdot D_5$

$D_5=dec_{sk_D}(ENC_{pk_D}(R))\cdot D_6$

$D_6=rsg_{I_D}\cdot D_7$

$D_7=enc_{pk_T}(I_D,R_D,v_D)\cdot D_8$

$D_8=s_{C_{DT}}(ENC_{pk_T}(I_D,R_D,v_D))\cdot D$

CTF's state transitions described by $APTC_G$ are as follows.

$T=s_{C_{AT}}(ENC_{pk_T}(r_A))\parallel s_{C_{BT}}(ENC_{pk_T}(r_B))\\
\parallel s_{C_{CT}}(ENC_{pk_T}(r_C))\parallel s_{C_{DT}}(ENC_{pk_T}(r_D))\cdot T_2$

$T_2=dec_{sk_T}(ENC_{pk_T}(r_A))\parallel dec_{sk_T}(ENC_{pk_T}(r_B))\\
\parallel dec_{sk_T}(ENC_{pk_T}(r_C))\parallel dec_{sk_T}(ENC_{pk_T}(r_D))\cdot T_3$

$T_3=rec(A)\parallel rec(B)\parallel rec(C)\parallel rec(D)\cdot T_4$

$T_4=enc_{pk_A}(R)\parallel enc_{pk_B}(R)\parallel enc_{pk_C}(R)\parallel enc_{pk_D}(R)\cdot T_5$

$T_5=s_{C_{TA}}(ENC_{pk_A}(R))\parallel s_{C_{TB}}(ENC_{pk_B}(R))\\
\parallel s_{C_{TC}}(ENC_{pk_C}(R))\parallel s_{C_{TD}}(ENC_{pk_D}(R))\cdot T_6$

$T_6=r_{C_{AT}}(ENC_{pk_T}(I_A,R_A,v_A))\parallel r_{C_{BT}}(ENC_{pk_T}(I_B,R_B,v_B))\\
\parallel r_{C_{CT}}(ENC_{pk_T}(I_C,R_C,v_C))\parallel r_{C_{DT}}(ENC_{pk_T}(I_D,R_D,v_D))\cdot T_7$

$T_7=dec_{sk_T}(ENC_{pk_T}(I_A,R_A,v_A))\parallel dec_{sk_T}(ENC_{pk_T}(I_B,R_B,v_B))\\
\parallel dec_{sk_T}(ENC_{pk_T}(I_C,R_C,v_C))\parallel dec_{sk_T}(ENC_{pk_T}(I_D,R_D,v_D))\cdot T_8$

$T_8=((\{isExisted(R_A)=TRUE\}\cdot remove(R_A)\cdot rec(I_A,v_A)+\{isExisted(R_A)=FALSE\})\\
\parallel (\{isExisted(R_B)=TRUE\}\cdot remove(R_B)\cdot rec(I_B,v_B)+\{isExisted(R_B)=FALSE\})\\
\parallel (\{isExisted(R_C)=TRUE\}\cdot remove(R_C)\cdot rec(I_C,v_C)+\{isExisted(R_C)=FALSE\})\\
\parallel (\{isExisted(R_D)=TRUE\}\cdot remove(R_D)\cdot rec(I_D,v_D)+\{isExisted(R_D)=FALSE\}))\cdot T_{9}$

$T_{9}=s_{C_{TO}}(TAB)\cdot T$

The sending action and the reading action of the same type data through the same channel can communicate with each other, otherwise, will cause a deadlock $\delta$. We define the following
communication functions.

$\gamma(r_{C_{AT}}(ENC_{pk_T}(r_A)),s_{C_{AT}}(ENC_{pk_T}(r_A)))\triangleq c_{C_{AT}}(ENC_{pk_T}(r_A))$

$\gamma(r_{C_{TA}}(ENC_{pk_A}(R)),s_{C_{TA}}(ENC_{pk_A}(R)))\triangleq c_{C_{TA}}(ENC_{pk_A}(R))$

$\gamma(r_{C_{AT}}(ENC_{pk_T}(I_A,R_A,v_A)),s_{C_{AT}}(ENC_{pk_T}(I_A,R_A,v_A)))\triangleq c_{C_{AT}}(ENC_{pk_T}(I_A,R_A,v_A))$

$\gamma(r_{C_{BT}}(ENC_{pk_T}(r_B)),s_{C_{BT}}(ENC_{pk_T}(r_B)))\triangleq c_{C_{BT}}(ENC_{pk_T}(r_B))$

$\gamma(r_{C_{TB}}(ENC_{pk_B}(R)),s_{C_{TB}}(ENC_{pk_B}(R)))\triangleq c_{C_{TB}}(ENC_{pk_B}(R))$

$\gamma(r_{C_{BT}}(ENC_{pk_T}(I_B,R_B,v_B)),s_{C_{BT}}(ENC_{pk_T}(I_B,R_B,v_B)))\triangleq c_{C_{BT}}(ENC_{pk_T}(I_B,R_B,v_B))$

$\gamma(r_{C_{CT}}(ENC_{pk_T}(r_C)),s_{C_{CT}}(ENC_{pk_T}(r_C)))\triangleq c_{C_{CT}}(ENC_{pk_T}(r_C))$

$\gamma(r_{C_{TC}}(ENC_{pk_C}(R)),s_{C_{TC}}(ENC_{pk_C}(R)))\triangleq c_{C_{TC}}(ENC_{pk_C}(R))$

$\gamma(r_{C_{CT}}(ENC_{pk_T}(I_C,R_C,v_C)),s_{C_{CT}}(ENC_{pk_T}(I_C,R_C,v_C)))\triangleq c_{C_{CT}}(ENC_{pk_T}(I_C,R_C,v_C))$

$\gamma(r_{C_{DT}}(ENC_{pk_T}(r_D)),s_{C_{DT}}(ENC_{pk_T}(r_D)))\triangleq c_{C_{DT}}(ENC_{pk_T}(r_D))$

$\gamma(r_{C_{TD}}(ENC_{pk_D}(R)),s_{C_{TD}}(ENC_{pk_D}(R)))\triangleq c_{C_{TD}}(ENC_{pk_D}(R))$

$\gamma(r_{C_{DT}}(ENC_{pk_T}(I_D,R_D,v_D)),s_{C_{DT}}(ENC_{pk_T}(I_D,R_D,v_D)))\triangleq c_{C_{DT}}(ENC_{pk_T}(I_D,R_D,v_D))$

Let all modules be in parallel, then the protocol $A\quad B\quad C\quad D\quad T$ can be presented by the following process term.

$$\tau_I(\partial_H(\Theta(A\between B\between C\between D\between T)))=\tau_I(\partial_H(A\between B\between C\between D\between T))$$

where $H=\{r_{C_{AT}}(ENC_{pk_T}(r_A)),s_{C_{AT}}(ENC_{pk_T}(r_A)),\\
r_{C_{TA}}(ENC_{pk_A}(R)),s_{C_{TA}}(ENC_{pk_A}(R)),\\
r_{C_{AT}}(ENC_{pk_T}(I_A,R_A,v_A)),s_{C_{AT}}(ENC_{pk_T}(I_A,R_A,v_A)),\\
r_{C_{BT}}(ENC_{pk_T}(r_B)),s_{C_{BT}}(ENC_{pk_T}(r_B)),\\
r_{C_{TB}}(ENC_{pk_B}(R)),s_{C_{TB}}(ENC_{pk_B}(R)),\\
r_{C_{BT}}(ENC_{pk_T}(I_B,R_B,v_B)),s_{C_{BT}}(ENC_{pk_T}(I_B,R_B,v_B)),\\
r_{C_{CT}}(ENC_{pk_T}(r_C)),s_{C_{CT}}(ENC_{pk_T}(r_C)),\\
r_{C_{TC}}(ENC_{pk_C}(R)),s_{C_{TC}}(ENC_{pk_C}(R)),\\
r_{C_{CT}}(ENC_{pk_T}(I_C,R_C,v_C)),s_{C_{CT}}(ENC_{pk_T}(I_C,R_C,v_C)),\\
r_{C_{DT}}(ENC_{pk_T}(r_D)),s_{C_{DT}}(ENC_{pk_T}(r_D)),\\
r_{C_{TD}}(ENC_{pk_D}(R)),s_{C_{TD}}(ENC_{pk_D}(R)),\\
r_{C_{DT}}(ENC_{pk_T}(I_D,R_D,v_D)),s_{C_{DT}}(ENC_{pk_T}(I_D,R_D,v_D))|D_A,D_B,D_C,D_D\in\Delta\}$,

$I=\{c_{C_{AT}}(ENC_{pk_T}(r_A)),c_{C_{TA}}(ENC_{pk_A}(R)),\\
c_{C_{AT}}(ENC_{pk_T}(I_A,R_A,v_A)),c_{C_{BT}}(ENC_{pk_T}(r_B)),\\
c_{C_{TB}}(ENC_{pk_B}(R)),c_{C_{BT}}(ENC_{pk_T}(I_B,R_B,v_B)),\\
c_{C_{CT}}(ENC_{pk_T}(r_C)),c_{C_{TC}}(ENC_{pk_C}(R)),\\
c_{C_{CT}}(ENC_{pk_T}(I_C,R_C,v_C)),c_{C_{DT}}(ENC_{pk_T}(r_D)),\\
c_{C_{TD}}(ENC_{pk_D}(R)),c_{C_{DT}}(ENC_{pk_T}(I_D,R_D,v_D)),\\
enc_{pk_T}(r_A),enc_{pk_T}(r_B),enc_{pk_T}(r_C),enc_{pk_T}(r_D),\\
dec_{sk_A}(ENC_{pk_A}(R)),dec_{sk_B}(ENC_{pk_B}(R)),dec_{sk_C}(ENC_{pk_C}(R)),\\
dec_{sk_D}(ENC_{pk_D}(R)),rsg_{I_A},rsg_{I_B},rsg_{I_C},rsg_{I_D},\\
enc_{pk_T}(I_A,R_A,v_A),enc_{pk_T}(I_B,R_B,v_B),enc_{pk_T}(I_C,R_C,v_C),enc_{pk_T}(I_D,R_D,v_D),\\
dec_{sk_T}(ENC_{pk_T}(r_A)), dec_{sk_T}(ENC_{pk_T}(r_B)),\\
dec_{sk_T}(ENC_{pk_T}(r_C)), dec_{sk_T}(ENC_{pk_T}(r_D)),\\
rec(A), rec(B), rec(C), rec(D),enc_{pk_A}(R), enc_{pk_B}(R),\\
enc_{pk_C}(R), enc_{pk_D}(R),dec_{sk_T}(ENC_{pk_T}(I_A,R_A,v_A)),\\
dec_{sk_T}(ENC_{pk_T}(I_B,R_B,v_B)),dec_{sk_T}(ENC_{pk_T}(I_C,R_C,v_C)),dec_{sk_T}(ENC_{pk_T}(I_D,R_D,v_D)),\\
\{isExisted(R_A)=TRUE\}, remove(R_A), rec(I_A,v_A),\{isExisted(R_A)=FALSE\},\\
\{isExisted(R_B)=TRUE\}, remove(R_B), rec(I_B,v_B),\{isExisted(R_B)=FALSE\},\\
\{isExisted(R_C)=TRUE\}, remove(R_C), rec(I_C,v_C),\{isExisted(R_C)=FALSE\},\\
\{isExisted(R_D)=TRUE\}, remove(R_D), rec(I_D,v_D),\{isExisted(R_D)=FALSE\}\\
|D_A,D_B,D_C,D_D\in\Delta\}$.

Then we get the following conclusion on the protocol.

\begin{theorem}
The secure elections protocol 5 in Figure \ref{SEP5} is improved based on the secure elections protocol 4.
\end{theorem}

\begin{proof}
Based on the above state transitions of the above modules, by use of the algebraic laws of $APTC_G$, we can prove that

$\tau_I(\partial_H(A\between B\between C\between D\between T))=\sum_{D_A,D_B,D_C,D_D\in\Delta}(r_{C_{AI}}(D_A)\parallel r_{C_{BI}}(D_B)\parallel r_{C_{CI}}(D_C)\parallel r_{C_{DI}}(D_D)\cdot s_{C_{TO}}(TAB))\cdot
\tau_I(\partial_H(A\between B\between C\between D\between T))$.

For the details of proof, please refer to section \ref{app}, and we omit it.

That is, the protocol in Figure \ref{SEP5} $\tau_I(\partial_H(A\between B\between C\between D\between T))$ can exhibit desired external behaviors, and is secure. But, for the properties of
secure elections protocols:
\begin{enumerate}
  \item Legitimacy: only authorized voters can vote;
  \item Oneness: no one can vote more than once;
  \item Privacy: no one can determine for whom anyone else voted;
  \item Non-replicability: no one can duplicate anyone else's vote;
  \item Non-changeability: no one can change anyone else's vote;
  \item Validness: every voter can make sure that his vote has been taken into account in the final tabulation, if CTF is trustworthy.
\end{enumerate}

The anonymous valid numbers distribution can avoid the distribution of valid numbers to illegal voters.
\end{proof}

\subsection{Secure Elections Protocol 6}\label{sep6}

The secure elections protocol 6 is shown in Figure \ref{SEP6}, which is a improved one based on the secure elections protocol 5 in section \ref{sep5}. In this protocol, there are
a CTF (Central Tabulating Facility) to check the identity of voters and collect the votes, and four voters: Alice, Bob, Carol and Dave.

\begin{figure}
    \centering
    \includegraphics{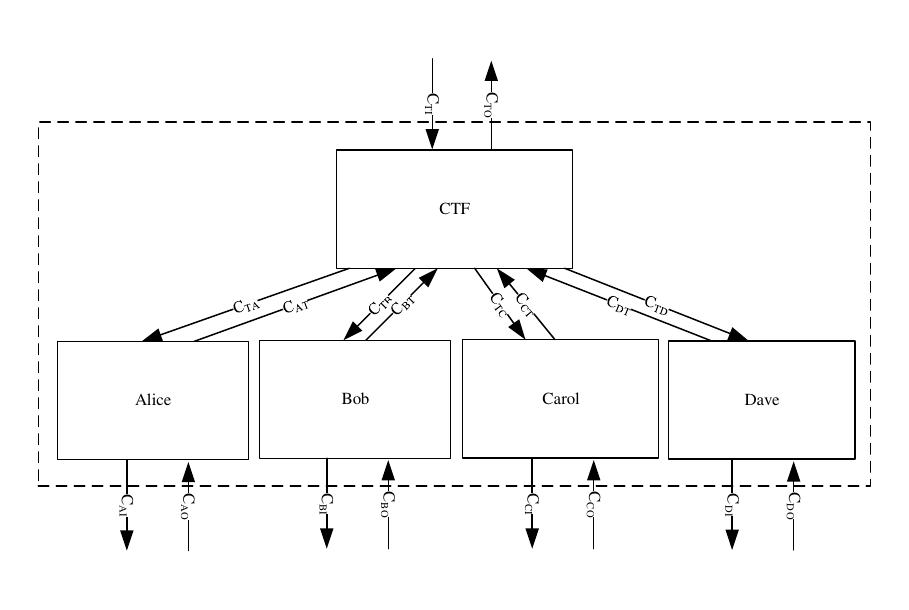}
    \caption{Secure elections protocol 6}
    \label{SEP6}
\end{figure}

The process of the protocol is as follows.

\begin{enumerate}
  \item Alice receives some voting request $D_A$ from the outside through the channel $C_{AI}$ (the corresponding reading action is denoted $r_{C_{AI}}(D_A)$), she generates a request
  $r_A$, encrypts it by CTF's public key through an action $enc_{pk_T}(r_A)$, and sends it to CTF through the channel $C_{AT}$ (the corresponding sending action is denoted
  $s_{C_{AT}}(ENC_{pk_T}(r_A))$);
  \item Bob receives some voting request $D_B$ from the outside through the channel $C_{BI}$ (the corresponding reading action is denoted $r_{C_{BI}}(D_B)$), he generates a request
  $r_B$, encrypts it by CTF's public key through an action $enc_{pk_T}(r_B)$, and sends it to CTF through the channel $C_{BT}$ (the corresponding sending action is denoted
  $s_{C_{BT}}(ENC_{pk_T}(r_B))$);
  \item Carol receives some voting request $D_C$ from the outside through the channel $C_{CI}$ (the corresponding reading action is denoted $r_{C_{CI}}(D_C)$), he generates a request
  $r_C$, encrypts it by CTF's public key through an action $enc_{pk_T}(r_C)$, and sends it to CTF through the channel $C_{CT}$ (the corresponding sending action is denoted
  $s_{C_{CT}}(ENC_{pk_T}(r_C))$);
  \item Dave receives some voting request $D_D$ from the outside through the channel $C_{DI}$ (the corresponding reading action is denoted $r_{C_{DI}}(D_D)$), he generates a request
  $r_D$, encrypts it by CTF's public key through an action $enc_{pk_T}(r_D)$, and sends it to CTF through the channel $C_{DT}$ (the corresponding sending action is denoted
  $s_{C_{DT}}(ENC_{pk_T}(r_D))$);
  \item CTF receives the requests from Alice, Bob, Carol and Dave through the channels $C_{AT}$, $C_{BT}$, $C_{CT}$ and $C_{DT}$ (the corresponding reading actions are denoted
  $s_{C_{AT}}(ENC_{pk_T}(r_A))$, $s_{C_{BT}}(ENC_{pk_T}(r_B))$, $s_{C_{CT}}(ENC_{pk_T}(r_C))$, and $s_{C_{DT}}(ENC_{pk_T}(r_D))$ respectively), he decrypts these encrypted requests
  through the actions $dec_{sk_T}(ENC_{pk_T}(r_A))$, $dec_{sk_T}(ENC_{pk_T}(r_B))$, $dec_{sk_T}(ENC_{pk_T}(r_C))$, and $dec_{sk_T}(ENC_{pk_T}(r_D))$ to get $r_A$, $r_B$, $r_C$ and $r_D$,
  records the names of Alice, Bob, Carol and Dave through actions $rec(A)$, $rec(B)$, $rec(C)$ and $rec(D)$; CTF maintain a table of valid numbers; then CTF
  encrypts all numbers $R$ through actions $enc_{pk_A}(R)$, $enc_{pk_B}(R)$, $enc_{pk_C}(R)$ and
  $enc_{pk_D}(R)$ and sends them to Alice, Bob, Carol and Dave through the channels $C_{TA}$, $C_{TB}$,
  $C_{TC}$ and $C_{TD}$ respectively (the corresponding sending action is denoted $s_{C_{TA}}(ENC_{pk_A}(R))$, $s_{C_{TB}}(ENC_{pk_B}(R))$, $s_{C_{TC}}(ENC_{pk_C}(R))$, and
  $s_{C_{TD}}(ENC_{pk_D}(R))$);
  \item Alice receives the encrypted number from CTF through the channel $C_{TA}$ (the corresponding reading action is denoted $r_{C_{TA}}(ENC_{pk_A}(R))$), she decrypts the encrypted
  number through an action $dec_{sk_A}(ENC_{pk_A}(R))$ to randomly select one $R_A$, generates a random identity number $I_A$ through an action $rsg_{I_A}$ and her vote $v_A$,
  generates a pair of public/private keys through an action $rsg_{pk'_A,sk'_A}$, encrypted $I_A,R_A,v_A$
  through an action $enc_{pk'_A}(I_A,R_A,v_A)$ and sends the encrypted message to CTF through the channel $C_{AT}$ (the corresponding sending action is denoted
  $s_{C_{AT}}(ENC_{pk'_A}(I_A,R_A,v_A))$);
  \item Bob receives the encrypted number from CTF through the channel $C_{TB}$ (the corresponding reading action is denoted $r_{C_{TB}}(ENC_{pk_B}(R))$), he decrypts the encrypted
  number through an action $dec_{sk_B}(ENC_{pk_B}(R))$ to randomly select one $R_B$, generates a random identity number $I_B$ through an action $rsg_{I_B}$ and his vote $v_B$,
  generates a pair of public/private keys through an action $rsg_{pk'_B,sk'_B}$, encrypted $I_A,R_A,v_A$
  through an action $enc_{pk'_B}(I_A,R_A,v_A)$ and sends the encrypted message to CTF through the channel $C_{BT}$ (the corresponding sending action is denoted
  $s_{C_{BT}}(ENC_{pk'_B}(I_B,R_B,v_B))$);
  \item Carol receives the encrypted number from CTF through the channel $C_{TC}$ (the corresponding reading action is denoted $r_{C_{TC}}(ENC_{pk_C}(R))$), he decrypts the encrypted
  number through an action $dec_{sk_C}(ENC_{pk_C}(R))$ to randomly select one $R_C$, generates a random identity number $I_C$ through an action $rsg_{I_C}$ and his vote $v_C$,
  generates a pair of public/private keys through an action $rsg_{pk'_C,sk'_C}$, encrypted $I_A,R_A,v_A$
  through an action $enc_{pk'_C}(I_A,R_A,v_A)$ and sends the encrypted message to CTF through the channel $C_{CT}$ (the corresponding sending action is denoted
  $s_{C_{CT}}(ENC_{pk'_C}(I_C,R_C,v_C))$);
  \item Dave receives the encrypted number from CTF through the channel $C_{TD}$ (the corresponding reading action is denoted $r_{C_{TD}}(ENC_{pk_A}(R))$), he decrypts the encrypted
  number through an action $dec_{sk_D}(ENC_{pk_D}(R))$ to randomly select one $R_D$, generates a random identity number $I_D$ through an action $rsg_{I_D}$ and his vote $v_D$,
  generates a pair of public/private keys through an action $rsg_{pk'_D,sk'_D}$, encrypted $I_A,R_A,v_A$
  through an action $enc_{pk'_D}(I_A,R_A,v_A)$ and sends the encrypted message to CTF through the channel $C_{DT}$ (the corresponding sending action is denoted
  $s_{C_{DT}}(ENC_{pk'_D}(I_D,R_D,v_D))$);
  \item CTF receives the encrypted messages from Alice, Bob, Carol and Dave through the channels $C_{AT}$, $C_{BT}$, $C_{CT}$ and $C_{DT}$ (the corresponding reading actions are denoted\\
  $r_{C_{AT}}(ENC_{pk'_A}(I_A,R_A,v_A))$, $r_{C_{BT}}(ENC_{pk'_B}(I_B,R_B,v_B))$, $r_{C_{CT}}(ENC_{pk'_C}(I_C,R_C,v_C))$, and $r_{C_{DT}}(ENC_{pk'_D}(I_D,R_D,v_D))$ respectively),
  he sends them to the outside through the channel $C_{TO}$ (the corresponding sending actions are denoted $s_{C_{TO}}(ENC_{pk'_A}(I_A,R_A,v_A))$, \\
  $s_{C_{TO}}(ENC_{pk'_B}(I_B,R_B,v_B))$,
  $s_{C_{TO}}(ENC_{pk'_C}(I_C,R_C,v_C))$, and $s_{C_{TO}}(ENC_{pk'_D}(I_D,R_D,v_D))$ respectively); then he sends the request $r_T$ to request the voter to reveal their votes through
  the channels $C_{TA}$, $C_{TB}$, $C_{TC}$ and $C_{TD}$ (the corresponding sending actions are denoted $s_{C_{TA}}(r_T)$, $s_{C_{TB}}(r_T)$, $s_{C_{TC}}(r_T)$, and $s_{C_{TD}}(r_T)$);
  \item Alice receives the request $r_T$ from CTF through the channel $C_{TA}$ (the corresponding reading action is denoted $r_{C_{TA}}(r_T)$), she sends $R_A,I_A.sk'_A$ to CTF through
  the channel $C_{AT}$ (the corresponding sending action is denoted $s_{C_{AT}}(R_A,I_A,sk'_A)$);
  \item Bob receives the request $r_T$ from CTF through the channel $C_{TB}$ (the corresponding reading action is denoted $r_{C_{TB}}(r_T)$), she sends $R_B,I_B.sk'_B$ to CTF through
  the channel $C_{BT}$ (the corresponding sending action is denoted $s_{C_{BT}}(R_B,I_B,sk'_B)$);
  \item Carol receives the request $r_T$ from CTF through the channel $C_{TC}$ (the corresponding reading action is denoted $r_{C_{TC}}(r_T)$), he sends $R_C,I_C.sk'_C$ to CTF through
  the channel $C_{CT}$ (the corresponding sending action is denoted $s_{C_{CT}}(R_C,I_C,sk'_C)$);
  \item Dave receives the request $r_T$ from CTF through the channel $C_{TD}$ (the corresponding reading action is denoted $r_{C_{TD}}(r_T)$), he sends $R_D,I_D.sk'_D$ to CTF through
  the channel $C_{DT}$ (the corresponding sending action is denoted $s_{C_{DT}}(R_D,I_D,sk'_D)$);
  \item CTF receives the message from Alice, Bob, Carol and Dave through the channels $C_{AT}$, $C_{BT}$, $C_{CT}$ and $C_{DT}$ (the corresponding reading actions are denoted
  $r_{C_{AT}}(R_A,I_A,sk'_A)$, $r_{C_{BT}}(R_B,I_B,sk'_B)$, $r_{C_{CT}}(R_C,I_C,sk'_C)$, and $r_{C_{DT}}(R_D,I_D,sk'_D)$ respectively), he decrypts the above encrypted messages
  through actions $dec_{sk'_A}(ENC_{pk'_A}(I_A,R_A,v_A))$, \\
  $dec_{sk'_B}(ENC_{pk'_B}(I_B,R_B,v_B))$, $dec_{sk'_C}(ENC_{pk'_C}(I_C,R_C,v_C))$, and\\
  $dec_{sk'_D}(ENC_{pk'_D}(I_D,R_D,v_D))$. If $isExisted(R_A)=TRUE$, he removes $R_A$ from its table through an action $remove(R_A)$, records the vote $v_A$ and the pair of $I_A$ and $v_A$
  into the voting results $TAB$ through an action $rec(I_A,v_A)$, else he does nothing; if $isExisted(R_B)=TRUE$, he removes $R_B$ from its table through an action $remove(R_B)$, records the vote $v_B$ and the pair of $I_B$ and $v_B$
  into the voting results $TAB$ through an action $rec(I_B,v_B)$, else he does nothing; if $isExisted(R_C)=TRUE$, he removes $R_C$ from its table through an action $remove(R_C)$, records the vote $v_C$ and the pair of $I_C$ and $v_C$
  into the voting results $TAB$ through an action $rec(I_C,v_C)$, else he does nothing; if $isExisted(R_D)=TRUE$, he removes $R_D$ from its table through an action $remove(R_D)$, records the vote $v_D$ and the pair of $I_D$ and $v_D$
  into the voting results $TAB$ through an action $rec(I_D,v_D)$, else he does nothing. Finally, he sends the voting results $TAB$ to the outside through the channel $C_{TO}$ (the corresponding sending action is denoted
  $s_{C_{TO}}(TAB)$).
\end{enumerate}

Where $D_A,D_B,D_C,D_D\in\Delta$, $\Delta$ is the set of data.

Alice's state transitions described by $APTC_G$ are as follows.

$A=\sum_{D_A\in\Delta}r_{C_{AI}}(D_A)\cdot A_2$

$A_2=enc_{pk_T}(r_A)\cdot A_3$

$A_3=s_{C_{AT}}(ENC_{pk_T}(r_A))\cdot A_4$

$A_4=r_{C_{TA}}(ENC_{pk_A}(R))\cdot A_5$

$A_5=dec_{sk_A}(ENC_{pk_A}(R))\cdot A_6$

$A_6=rsg_{I_A}\cdot A_7$

$A_7=rsg_{pk'_A,sk'_A}\cdot A_8$

$A_8=enc_{pk'_A}(I_A,R_A,v_A)\cdot A_9$

$A_9=s_{C_{AT}}(ENC_{pk'_A}(I_A,R_A,v_A))\cdot A_{10}$

$A_{10}=r_{C_{TA}}(r_T)\cdot A_{11}$

$A_{11}=s_{C_{AT}}(R_A,I_A,sk'_A)\cdot A$

Bob's state transitions described by $APTC_G$ are as follows.

$B=\sum_{D_B\in\Delta}r_{C_{BI}}(D_B)\cdot B_2$

$B_2=enc_{pk_T}(r_B)\cdot B_3$

$B_3=s_{C_{BT}}(ENC_{pk_T}(r_B))\cdot B_4$

$B_4=r_{C_{TB}}(ENC_{pk_B}(R))\cdot B_5$

$B_5=dec_{sk_B}(ENC_{pk_B}(R))\cdot B_6$

$B_6=rsg_{I_B}\cdot B_7$

$B_7=rsg_{pk'_B,sk'_B}\cdot B_8$

$B_8=enc_{pk'_B}(I_B,R_B,v_B)\cdot B_9$

$B_9=s_{C_{BT}}(ENC_{pk'_B}(I_B,R_B,v_B))\cdot B_{10}$

$B_{10}=r_{C_{TB}}(r_T)\cdot B_{11}$

$B_{11}=s_{C_{BT}}(R_B,I_B,sk'_B)\cdot B$

Carol's state transitions described by $APTC_G$ are as follows.

$C=\sum_{D_C\in\Delta}r_{C_{CI}}(D_C)\cdot C_2$

$C_2=enc_{pk_T}(r_C)\cdot C_3$

$C_3=s_{C_{CT}}(ENC_{pk_T}(r_C))\cdot C_4$

$C_4=r_{C_{TC}}(ENC_{pk_C}(R))\cdot C_5$

$C_5=dec_{sk_C}(ENC_{pk_C}(R))\cdot C_6$

$C_6=rsg_{I_C}\cdot C_7$

$C_7=rsg_{pk'_C,sk'_C}\cdot C_8$

$C_8=enc_{pk'_C}(I_C,R_C,v_C)\cdot C_9$

$C_9=s_{C_{CT}}(ENC_{pk'_C}(I_C,R_C,v_C))\cdot C_{10}$

$C_{10}=r_{C_{TC}}(r_T)\cdot C_{11}$

$C_{11}=s_{C_{CT}}(R_C,I_C,sk'_C)\cdot C$

Dave's state transitions described by $APTC_G$ are as follows.

$D=\sum_{D_D\in\Delta}r_{C_{DI}}(D_D)\cdot D_2$

$D_2=enc_{pk_T}(r_D)\cdot D_3$

$D_3=s_{C_{DT}}(ENC_{pk_T}(r_D))\cdot D_4$

$D_4=r_{C_{TD}}(ENC_{pk_D}(R))\cdot D_5$

$D_5=dec_{sk_D}(ENC_{pk_D}(R))\cdot D_6$

$D_6=rsg_{I_D}\cdot D_7$

$D_7=rsg_{pk'_D,sk'_D}\cdot D_8$

$D_8=enc_{pk'_D}(I_D,R_D,v_D)\cdot D_9$

$D_9=s_{C_{DT}}(ENC_{pk'_D}(I_D,R_D,v_D))\cdot D_{10}$

$D_{10}=r_{C_{TD}}(r_T)\cdot D_{11}$

$D_{11}=s_{C_{DT}}(R_D,I_D,sk'_D)\cdot D$

CTF's state transitions described by $APTC_G$ are as follows.

$T=s_{C_{AT}}(ENC_{pk_T}(r_A))\parallel s_{C_{BT}}(ENC_{pk_T}(r_B))\\
\parallel s_{C_{CT}}(ENC_{pk_T}(r_C))\parallel s_{C_{DT}}(ENC_{pk_T}(r_D))\cdot T_2$

$T_2=dec_{sk_T}(ENC_{pk_T}(r_A))\parallel dec_{sk_T}(ENC_{pk_T}(r_B))\\
\parallel dec_{sk_T}(ENC_{pk_T}(r_C))\parallel dec_{sk_T}(ENC_{pk_T}(r_D))\cdot T_3$

$T_3=rec(A)\parallel rec(B)\parallel rec(C)\parallel rec(D)\cdot T_4$

$T_4=enc_{pk_A}(R)\parallel enc_{pk_B}(R)\parallel enc_{pk_C}(R)\parallel enc_{pk_D}(R)\cdot T_5$

$T_5=s_{C_{TA}}(ENC_{pk_A}(R))\parallel s_{C_{TB}}(ENC_{pk_B}(R))\\
\parallel s_{C_{TC}}(ENC_{pk_C}(R))\parallel s_{C_{TD}}(ENC_{pk_D}(R))\cdot T_6$

$T_6=r_{C_{AT}}(ENC_{pk'_A}(I_A,R_A,v_A))\parallel r_{C_{BT}}(ENC_{pk'_B}(I_B,R_B,v_B))\\
\parallel r_{C_{CT}}(ENC_{pk'_C}(I_C,R_C,v_C))\parallel r_{C_{DT}}(ENC_{pk'_D}(I_D,R_D,v_D))\cdot T_7$

$T_7=s_{C_{TO}}(ENC_{pk'_A}(I_A,R_A,v_A))\parallel s_{C_{TO}}(ENC_{pk'_B}(I_B,R_B,v_B))\\
\parallel s_{C_{TO}}(ENC_{pk'_C}(I_C,R_C,v_C))\parallel s_{C_{TO}}(ENC_{pk'_D}(I_D,R_D,v_D))\cdot T_8$

$T_8=s_{C_{TA}}(r_T)\parallel s_{C_{TB}}(r_T)\parallel s_{C_{TC}}(r_T)\parallel s_{C_{TD}}(r_T)\cdot T_9$

$T_9=r_{C_{AT}}(R_A,I_A,sk'_A)\parallel r_{C_{BT}}(R_B,I_B,sk'_B)\\
\parallel r_{C_{CT}}(R_C,I_C,sk'_C)\parallel r_{C_{DT}}(R_D,I_D,sk'_D)\cdot T_{10}$

$T_{10}=dec_{sk'_A}(ENC_{pk'_A}(I_A,R_A,v_A))\parallel dec_{sk'_b}(ENC_{pk'_B}(I_B,R_B,v_B))\\
\parallel dec_{sk'_C}(ENC_{pk'_C}(I_C,R_C,v_C))\parallel dec_{sk'_D}(ENC_{pk'_D}(I_D,R_D,v_D))\cdot T_{11}$

$T_{11}=((\{isExisted(R_A)=TRUE\}\cdot remove(R_A)\cdot rec(I_A,v_A)+\{isExisted(R_A)=FALSE\})\\
\parallel (\{isExisted(R_B)=TRUE\}\cdot remove(R_B)\cdot rec(I_B,v_B)+\{isExisted(R_B)=FALSE\})\\
\parallel (\{isExisted(R_C)=TRUE\}\cdot remove(R_C)\cdot rec(I_C,v_C)+\{isExisted(R_C)=FALSE\})\\
\parallel (\{isExisted(R_D)=TRUE\}\cdot remove(R_D)\cdot rec(I_D,v_D)+\{isExisted(R_D)=FALSE\}))\cdot T_{12}$

$T_{12}=s_{C_{TO}}(TAB)\cdot T$

The sending action and the reading action of the same type data through the same channel can communicate with each other, otherwise, will cause a deadlock $\delta$. We define the following
communication functions.

$\gamma(r_{C_{AT}}(ENC_{pk_T}(r_A)),s_{C_{AT}}(ENC_{pk_T}(r_A)))\triangleq c_{C_{AT}}(ENC_{pk_T}(r_A))$

$\gamma(r_{C_{TA}}(ENC_{pk_A}(R)),s_{C_{TA}}(ENC_{pk_A}(R)))\triangleq c_{C_{TA}}(ENC_{pk_A}(R))$

$\gamma(r_{C_{AT}}(ENC_{pk'_A}(I_A,R_A,v_A)),s_{C_{AT}}(ENC_{pk'_A}(I_A,R_A,v_A)))\triangleq c_{C_{AT}}(ENC_{pk'_A}(I_A,R_A,v_A))$

$\gamma(r_{C_{TA}}(r_T),s_{C_{TA}}(r_T))\triangleq c_{C_{TA}}(r_T)$

$\gamma(r_{C_{AT}}(R_A,I_A,sk'_A),s_{C_{AT}}(R_A,I_A,sk'_A))\triangleq c_{C_{AT}}(R_A,I_A,sk'_A)$

$\gamma(r_{C_{BT}}(ENC_{pk_T}(r_B)),s_{C_{BT}}(ENC_{pk_T}(r_B)))\triangleq c_{C_{BT}}(ENC_{pk_T}(r_B))$

$\gamma(r_{C_{TB}}(ENC_{pk_B}(R)),s_{C_{TB}}(ENC_{pk_B}(R)))\triangleq c_{C_{TB}}(ENC_{pk_B}(R))$

$\gamma(r_{C_{BT}}(ENC_{pk'_B}(I_B,R_B,v_B)),s_{C_{BT}}(ENC_{pk'_B}(I_B,R_B,v_B)))\triangleq c_{C_{BT}}(ENC_{pk'_B}(I_B,R_B,v_B))$

$\gamma(r_{C_{TB}}(r_T),s_{C_{TB}}(r_T))\triangleq c_{C_{TB}}(r_T)$

$\gamma(r_{C_{BT}}(R_B,I_B,sk'_B),s_{C_{BT}}(R_B,I_B,sk'_B))\triangleq c_{C_{BT}}(R_B,I_B,sk'_B)$

$\gamma(r_{C_{CT}}(ENC_{pk_T}(r_C)),s_{C_{CT}}(ENC_{pk_T}(r_C)))\triangleq c_{C_{CT}}(ENC_{pk_T}(r_C))$

$\gamma(r_{C_{TC}}(ENC_{pk_C}(R)),s_{C_{TC}}(ENC_{pk_C}(R)))\triangleq c_{C_{TC}}(ENC_{pk_C}(R))$

$\gamma(r_{C_{CT}}(ENC_{pk'_C}(I_C,R_C,v_C)),s_{C_{CT}}(ENC_{pk'_C}(I_C,R_C,v_C)))\triangleq c_{C_{CT}}(ENC_{pk'_C}(I_C,R_C,v_C))$

$\gamma(r_{C_{TC}}(r_T),s_{C_{TC}}(r_T))\triangleq c_{C_{TC}}(r_T)$

$\gamma(r_{C_{CT}}(R_C,I_C,sk'_C),s_{C_{CT}}(R_C,I_C,sk'_C))\triangleq c_{C_{CT}}(R_C,I_C,sk'_C)$

$\gamma(r_{C_{DT}}(ENC_{pk_T}(r_D)),s_{C_{DT}}(ENC_{pk_T}(r_D)))\triangleq c_{C_{DT}}(ENC_{pk_T}(r_D))$

$\gamma(r_{C_{TD}}(ENC_{pk_D}(R)),s_{C_{TD}}(ENC_{pk_D}(R)))\triangleq c_{C_{TD}}(ENC_{pk_D}(R))$

$\gamma(r_{C_{DT}}(ENC_{pk'_D}(I_D,R_D,v_D)),s_{C_{DT}}(ENC_{pk'_D}(I_D,R_D,v_D)))\triangleq c_{C_{DT}}(ENC_{pk'_D}(I_D,R_D,v_D))$

$\gamma(r_{C_{TD}}(r_T),s_{C_{TD}}(r_T))\triangleq c_{C_{TD}}(r_T)$

$\gamma(r_{C_{DT}}(R_D,I_D,sk'_D),s_{C_{DT}}(R_D,I_D,sk'_D))\triangleq c_{C_{DT}}(R_D,I_D,sk'_D)$

Let all modules be in parallel, then the protocol $A\quad B\quad C\quad D\quad T$ can be presented by the following process term.

$$\tau_I(\partial_H(\Theta(A\between B\between C\between D\between T)))=\tau_I(\partial_H(A\between B\between C\between D\between T))$$

where $H=\{r_{C_{AT}}(ENC_{pk_T}(r_A)),s_{C_{AT}}(ENC_{pk_T}(r_A)),\\
r_{C_{TA}}(ENC_{pk_A}(R)),s_{C_{TA}}(ENC_{pk_A}(R)),\\
r_{C_{AT}}(ENC_{pk'_A}(I_A,R_A,v_A)),s_{C_{AT}}(ENC_{pk'_A}(I_A,R_A,v_A)),\\
r_{C_{TA}}(r_T),s_{C_{TA}}(r_T),\\
r_{C_{AT}}(R_A,I_A,sk'_A),s_{C_{AT}}(R_A,I_A,sk'_A),\\
r_{C_{BT}}(ENC_{pk_T}(r_B)),s_{C_{BT}}(ENC_{pk_T}(r_B)),\\
r_{C_{TB}}(ENC_{pk_B}(R)),s_{C_{TB}}(ENC_{pk_B}(R)),\\
r_{C_{BT}}(ENC_{pk'_B}(I_B,R_B,v_B)),s_{C_{BT}}(ENC_{pk'_B}(I_B,R_B,v_B)),\\
r_{C_{TB}}(r_T),s_{C_{TB}}(r_T),\\
r_{C_{BT}}(R_B,I_B,sk'_B),s_{C_{BT}}(R_B,I_B,sk'_B),\\
r_{C_{CT}}(ENC_{pk_T}(r_C)),s_{C_{CT}}(ENC_{pk_T}(r_C)),\\
r_{C_{TC}}(ENC_{pk_C}(R)),s_{C_{TC}}(ENC_{pk_C}(R)),\\
r_{C_{CT}}(ENC_{pk'_C}(I_C,R_C,v_C)),s_{C_{CT}}(ENC_{pk'_C}(I_C,R_C,v_C)),\\
r_{C_{TC}}(r_T),s_{C_{TC}}(r_T),\\
r_{C_{CT}}(R_C,I_C,sk'_C),s_{C_{CT}}(R_C,I_C,sk'_C),\\
r_{C_{DT}}(ENC_{pk_T}(r_D)),s_{C_{DT}}(ENC_{pk_T}(r_D)),\\
r_{C_{TD}}(ENC_{pk_D}(R)),s_{C_{TD}}(ENC_{pk_D}(R)),\\
r_{C_{DT}}(ENC_{pk'_D}(I_D,R_D,v_D)),s_{C_{DT}}(ENC_{pk'_D}(I_D,R_D,v_D)),\\
r_{C_{TD}}(r_T),s_{C_{TD}}(r_T),\\
r_{C_{DT}}(R_D,I_D,sk'_D),s_{C_{DT}}(R_D,I_D,sk'_D)|D_A,D_B,D_C,D_D\in\Delta\}$,

$I=\{c_{C_{AT}}(ENC_{pk_T}(r_A)),c_{C_{TA}}(ENC_{pk_A}(R)),\\
c_{C_{AT}}(ENC_{pk'_A}(I_A,R_A,v_A)),c_{C_{BT}}(ENC_{pk_T}(r_B)),\\
c_{C_{TB}}(ENC_{pk_B}(R)),c_{C_{BT}}(ENC_{pk'_B}(I_B,R_B,v_B)),\\
c_{C_{CT}}(ENC_{pk_T}(r_C)),c_{C_{TC}}(ENC_{pk_C}(R)),\\
c_{C_{CT}}(ENC_{pk'_C}(I_C,R_C,v_C)),c_{C_{DT}}(ENC_{pk_T}(r_D)),\\
c_{C_{TD}}(ENC_{pk_D}(R)),c_{C_{DT}}(ENC_{pk'_D}(I_D,R_D,v_D)),\\
c_{C_{TA}}(r_T),c_{C_{AT}}(R_A,I_A,sk'_A),c_{C_{TB}}(r_T),c_{C_{BT}}(R_B,I_B,sk'_B),\\
c_{C_{TC}}(r_T),c_{C_{CT}}(R_C,I_C,sk'_C),c_{C_{TD}}(r_T),c_{C_{DT}}(R_D,I_D,sk'_D),\\
enc_{pk_T}(r_A),enc_{pk_T}(r_B),enc_{pk_T}(r_C),enc_{pk_T}(r_D),\\
dec_{sk_A}(ENC_{pk_A}(R)),dec_{sk_B}(ENC_{pk_B}(R)),dec_{sk_C}(ENC_{pk_C}(R)),\\
dec_{sk_D}(ENC_{pk_D}(R)),rsg_{I_A},rsg_{I_B},rsg_{I_C},rsg_{I_D},\\
rsg_{pk'_A,sk'_A},rsg_{pk'_B,sk'_B},rsg_{pk'_C,sk'_C},rsg_{pk'_D,sk'_D},\\
enc_{pk'_A}(I_A,R_A,v_A),enc_{pk'_B}(I_B,R_B,v_B),enc_{pk'_C}(I_C,R_C,v_C),enc_{pk'_D}(I_D,R_D,v_D),\\
dec_{sk_T}(ENC_{pk_T}(r_A)), dec_{sk_T}(ENC_{pk_T}(r_B)),\\
dec_{sk_T}(ENC_{pk_T}(r_C)), dec_{sk_T}(ENC_{pk_T}(r_D)),\\
rec(A), rec(B), rec(C), rec(D),enc_{pk_A}(R), enc_{pk_B}(R),\\
enc_{pk_C}(R), enc_{pk_D}(R),dec_{sk'_A}(ENC_{pk'_B}(I_A,R_A,v_A)),\\
dec_{sk'_B}(ENC_{pk'_B}(I_B,R_B,v_B)),dec_{sk'_C}(ENC_{pk'_C}(I_C,R_C,v_C)),dec_{sk'_D}(ENC_{pk'_D}(I_D,R_D,v_D)),\\
\{isExisted(R_A)=TRUE\}, remove(R_A), rec(I_A,v_A),\{isExisted(R_A)=FALSE\},\\
\{isExisted(R_B)=TRUE\}, remove(R_B), rec(I_B,v_B),\{isExisted(R_B)=FALSE\},\\
\{isExisted(R_C)=TRUE\}, remove(R_C), rec(I_C,v_C),\{isExisted(R_C)=FALSE\},\\
\{isExisted(R_D)=TRUE\}, remove(R_D), rec(I_D,v_D),\{isExisted(R_D)=FALSE\}\\
|D_A,D_B,D_C,D_D\in\Delta\}$.

Then we get the following conclusion on the protocol.

\begin{theorem}
The secure elections protocol 6 in Figure \ref{SEP6} is improved based on the secure elections protocol 5.
\end{theorem}

\begin{proof}
Based on the above state transitions of the above modules, by use of the algebraic laws of $APTC_G$, we can prove that

$\tau_I(\partial_H(A\between B\between C\between D\between T))=\sum_{D_A,D_B,D_C,D_D\in\Delta}((r_{C_{AI}}(D_A)\parallel r_{C_{BI}}(D_B)\parallel r_{C_{CI}}(D_C)\parallel r_{C_{DI}}(D_D))
\cdot(s_{C_{TO}}(ENC_{pk'_A}(I_A,R_A,v_A))\parallel s_{C_{TO}}(ENC_{pk'_B}(I_B,R_B,v_B))\parallel s_{C_{TO}}(ENC_{pk'_C}(I_C,R_C,v_C))\parallel s_{C_{TO}}(ENC_{pk'_D}(I_D,R_D,v_D)))\cdot s_{C_{TO}}(TAB))\cdot
\tau_I(\partial_H(A\between B\between C\between D\between T))$.

For the details of proof, please refer to section \ref{app}, and we omit it.

That is, the protocol in Figure \ref{SEP6} $\tau_I(\partial_H(A\between B\between C\between D\between T))$ can exhibit desired external behaviors, and is secure. But, for the properties of
secure elections protocols:
\begin{enumerate}
  \item Legitimacy: only authorized voters can vote;
  \item Oneness: no one can vote more than once;
  \item Privacy: no one can determine for whom anyone else voted;
  \item Non-replicability: no one can duplicate anyone else's vote;
  \item Non-changeability: no one can change anyone else's vote;
  \item Validness: every voter can make sure that his vote has been taken into account in the final tabulation, if CTF is trustworthy.
\end{enumerate}

Additionally, (1) If a voter observes that his vote is not properly counted, he can protest; (2) A voter can change his votes later.
\end{proof}


\end{document}